\documentclass[a4paper,onecolumn,unpublished,10pt]{quantumarticle}
\pdfoutput=1
\usepackage[normalem]{ulem}

\let\phi\varphi

\usepackage[dvipsnames]{xcolor}
\usepackage{lipsum}
\usepackage{amsmath}
\usepackage{amsfonts}
\usepackage{amsthm}
\usepackage{amssymb}
\usepackage{cancel}
\usepackage{relsize}
\usepackage[utf8]{inputenc}
\usepackage{graphicx}
\usepackage[T1]{fontenc}
\usepackage[ colorlinks = true,
             linkcolor = MidnightBlue,
             urlcolor  = MidnightBlue,
             citecolor = orange,
             anchorcolor = ForestGreen,
]{hyperref}
\usepackage[capitalize,nameinlink]{cleveref}
\usepackage[english]{babel}
\usepackage{placeins}
\usepackage{setspace}
\usepackage[font=small,labelfont=bf]{caption}
\usepackage{url}
\usepackage[left=2cm,right=2cm,top=2.5cm,bottom=2.5cm]{geometry}
\usepackage[shortcuts]{extdash}
\usepackage{tikz}
\usetikzlibrary{arrows.meta}
\usetikzlibrary{decorations.pathmorphing}
\usetikzlibrary{shapes}
\tikzstyle{causalarrow} = [draw=black,thick,decorate,decoration={snake,amplitude=.7mm,pre length=1mm,post length=2mm}]
\tikzstyle{affectsarrow} = [draw=black,thick]
\tikzstyle{routedarrow} = [draw=black,very thick]
\usepackage{tikz-3dplot}
\tdplotsetmaincoords{80}{45}
\tdplotsetrotatedcoords{-90}{180}{-90}
\usepackage{verbatim} 
\usepackage{bbold} 

\usepackage{braket} 
\usepackage{csquotes} 
\usepackage[numbers,sort&compress]{natbib}
\usepackage{stmaryrd} 
\SetSymbolFont{stmry}{bold}{U}{stmry}{m}{n} 
\usepackage[header]{appendix}
\usepackage{tabularx} 
\usepackage{mathtools} 
\usepackage{mathrsfs} 
\usepackage{booktabs} 
\usepackage{mdframed}
\usepackage{array}
\usepackage{bm}

\numberwithin{equation}{section} 

\usepackage{xstring}
\usepackage{subcaption}
\usepackage{multirow}

\graphicspath{
	{../tikz/}
	{../pics/}
    {pics/}
}

\newcommand{\CC}{\ensuremath{\mathcal{C}}}

\newcommand{\HH}{\ensuremath{\mathcal{H}}}
\newcommand{\K}{\ensuremath{\mathcal{K}}}

\newcommand{\N}{\ensuremath{\mathcal{N}}}

\newcommand{\XX}{\ensuremath{\mathcal{X}}}

\newcommand{\abs}[1]{\left| #1 \right|}

\newcommand{\editorial}[1]{}

\newcommand{\id}{\mathbb{1}}

\newcommand{\Nodes}{\texttt{Nodes}}

\newcommand{\Happ}{\mathtt{Happ}}
\newcommand{\Ind}{\mathtt{Ind}}
\newcommand{\Bran}{\mathtt{Bran}}

\DeclareMathOperator{\Tr}{Tr}

\DeclareMathOperator{\spann}{span}

\DeclareMathOperator*{\boolSum}{\scalerel*{\mathbb{\Sigma}}{\sum}}
\usepackage{scalerel}

\def\ketbra#1#2{\mathinner{|{#1}\rangle\!\langle{#2}|}}

\def\dket#1{\mathinner{|{#1}\rangle\!\rangle}}

\usepackage{thm-restate}
\usepackage{thmtools} 

\newtheoremstyle{note}
{1.2em}                
{1.2em}                
{}                     
{}                     
{\bfseries}    
{.}                    
{.5em}                 
{}                     
\theoremstyle{note}

\newtheorem{theorem}{Theorem}[section]
\newtheorem{definition}[theorem]{Definition}

\newtheorem{corollary}[theorem]{Corollary}

\newtheorem{example}[theorem]{Example}

\newtheorem{remark}[theorem]{Remark}

\numberwithin{theorem}{section}
\numberwithin{figure}{section}

\makeatletter
\renewcommand\listoffigures{%
        \@starttoc{lof}%
}
\makeatother

\newcommand{\cG}{\mathcal{G}}
\newcommand{\cN}{\N}

\newcommand{\ravnothing}{{\mathchoice%
{\reflectbox{$\varnothing$}}%
{\reflectbox{$\varnothing$}}%
{\reflectbox{$\scriptstyle \varnothing$}}%
{\reflectbox{$\scriptscriptstyle \varnothing$}}%
}}

\begin{document}
\setlength{\parindent}{.5cm}

\author{Maarten Grothus}
    \affiliation{Univ.\ Grenoble Alpes, Inria, 38000 Grenoble, France}
    \affiliation{Univ.\ Grenoble Alpes, CNRS, Grenoble INP, Institut N\'eel, 38000 Grenoble, France}
    \orcid{0000-0002-1561-9709}
\author{Alastair A.\ Abbott}
    \affiliation{Univ.\ Grenoble Alpes, Inria, 38000 Grenoble, France}
    \orcid{0000-0002-2759-633X}
\author{Augustin Vanrietvelde}
    \affiliation{Télécom Paris-LTCI, Inria, Institut Polytechnique de Paris, 91120 Palaiseau, France}
    \orcid{0000-0001-9022-8655}
\author{Cyril Branciard}
    \affiliation{Univ.\ Grenoble Alpes, CNRS, Grenoble INP, Institut N\'eel, 38000 Grenoble, France}
    \orcid{0000-0001-9460-825X}
\title{Routing Quantum Control of Causal Order}

\maketitle

\begin{abstract}
In recent years, various frameworks have been proposed for the study of quantum processes with indefinite causal order.
In particular, quantum circuits with quantum control of causal order (QC-QCs) form a broad class of physical supermaps obtained from a bottom-up construction and are believed to represent all quantum processes physically realisable in a fixed spacetime.
Complementarily, the formalism of routed quantum circuits introduces quantum operations constrained by ``routes'' to represent processes in terms of a more fine-grained routed circuit decomposition.
This decomposition, formalised using a so-called routed graph, represents the information flow within the respective process.
However, the existence of routed circuit decompositions has only been established for a small set of processes so far, including both certain specific QC-QCs and more exotic processes as examples.

In this work, we remedy this fact by connecting these two frameworks.
We prove that for any given $N$, one can use a single routed graph to systematically obtain a routed circuit decomposition for any QC-QC with $N$ parties.
We detail this construction explicitly and contrast it with other routed circuit decompositions of QC-QCs, which we obtain from alternative routed graphs.
We conclude by pointing out how this connection can be useful to tackle various open problems in the field of indefinite causal order, particularly establishing circuit representations of subclasses of QC-QCs.
\end{abstract}

\newpage
{
	\hypersetup{linkcolor=black}
	\tableofcontents
}
\newpage

\section{Introduction}
\label{sec:introduction}

Over the course of the last fifteen years, there has been an explosion of interest in studying the quantum world through the eye of causality.
While we are used to making sense of our everyday experiences in terms of cause and effect -- a perspective which has been well-formalised in the field of classical statistics~\cite{Reichenbach1956, Spirtes1993, Pearl2009}, certain quantum phenomena defy any attempt to be explained in terms of classical causation~\cite{Wood2015}.
However, it recently has been realised that the notion of causality can be adapted
into the quantum realm, shifting from classical random variables to quantum systems as causal relata~\cite{PhysRevX.7.031021, arxiv.1906.10726}.
This has led to the increasing adoption of causal perspectives and techniques in the fields of quantum foundations and quantum information theory.
In the wake of this movement, a particularly interesting field of research that has emerged is that of indefinite causal order (ICO) in quantum processes, suggesting the necessity of cyclic quantum causal models for their description~\cite{Hardy2009, Brukner2015, Costa2016, Barrett2021, VVR, Adlam2023, Vilasini24}.

The top-down approach of \emph{quantum supermaps}~\cite{Chiribella2008_Supermap} (also known as higher-order maps and often formalised in the \emph{process matrix framework}~\cite{Oreshkov2012}) has been extensively used to model ICO phenomena.
In this framework, one considers a fixed number of agents who, as parties in a fixed process, each execute a freely chosen quantum operation precisely once.
The only condition imposed on the process is that any probabilities observed by these agents are logically consistent, but no global causal structure is assumed -- as such a structure might no longer be in place, e.g., in a theory of quantum gravity.
Such processes can, equivalently, be understood as the most general ``higher-order'' transformations that transform the parties' local operations into a valid quantum operation between a common global past and a global future.
Quantum supermaps are therefore of interest both from an informational perspective and as an extension of quantum theory into general relativistic regimes~\cite{Hardy2009, Oreshkov2012, Zych2019, Moller2021, VVR, Baumann2022, AC2024, Sahdo2024, Vilasini24}. 

Over the course of the study of ICO, several particularly interesting processes have been studied in detail,
such as the quantum switch~\cite{Chiribella2013}, in which the two possible orders of two operations are placed in a quantum superposition with the aid of a quantum control system,
as well as more exotic processes such as the Oreshkov-Costa-Brukner (OCB) process~\cite{Oreshkov2012}, for which no clear physical interpretation has been suggested so far (and which has in fact been argued to be non-physical, on the basis that it is not ``purifibable''~\cite{Araujo2017}), and the so-called Lugano (or ``AF/BW'') process~\cite{Lugano2014, Baumeler2016}, which features ICO on classical systems.
These examples and their respective properties have been the subject of intense study, and for the quantum switch~\cite{Chiribella2012,Chiribella2013,Feix2015,Guerin2016} and its natural generalisation to more operations~\cite{Colnaghi2012,Facchini2015,Araujo2014,Taddei2021}, as well as for the Lugano process~\cite{Porcreau2025}, computational and communication advantages
beyond standard causally ordered quantum information processing have been identified.
However, the top-down perspective of the process matrix framework does not provide any straightforward way to systematically establish
further processes with interesting properties, or easily assess their physical implementability. 
Accordingly, working towards a better understanding of the internal structure and physical realisability of such processes is one of the major challenges of the field.

With this goal in mind, various constructive frameworks have been proposed~\cite{Portmann2017, Wechs2021, PBS, PBSGrenoble, Purves2021, Arrighi2023, Vanrietvelde2023, Chardonnet2025}. 
Refs.~\cite{Wechs2021, Vanrietvelde2023} in particular establish formal connections with quantum supermaps, providing an avenue to categorise processes into structurally different classes.
However, research into the relation and connections between these approaches and their synergies has been rather limited so far, with a notable exception being Refs.~\cite{Salzger2022, Salzger2024} which relate the \emph{causal box} framework~\cite{Portmann2017} to the generalised circuit framework of~\cite{Wechs2021} in an effort to gauge the physical realisability of processes within a fixed spacetime.

This work contributes to closing this gap by connecting the framework of \textit{quantum circuits with quantum control of causal order} (QC-QCs)~\cite{Wechs2021} -- physically realisable processes which feature coherent control of the order of operations, and whose properties as process matrices have been well-characterised -- with \textit{routed quantum circuits}~\cite{Vanrietvelde2021, Vanrietvelde2023}, which provide a decomposition of ICO processes into more fine-grained ``routed'' circuits, with precise and intuitive diagrammatic representations in which the process validity can be checked from the circuit structure, despite the presence of feedback loops.
While for some notable examples of causally indefinite QC-QCs, such as the quantum switch and the so-called Grenoble process~\cite{Wechs2021}, explicit constructions as routed quantum circuits have been formulated~\cite{Vanrietvelde2023}, no systematic procedure has so far been established. 
Indeed, it has thus far remained unknown whether routed quantum circuits can represent all QC-QCs, even though they can also represent some more exotic purifiable processes, such as the Lugano process~\cite{Vanrietvelde2023}.

Here we establish a generic procedure for obtaining a routed quantum circuit representation of any QC-QC.
Specifically, for any given number of parties, we specify a single routed graph which will allow one to obtain any QC-QC with that number of parties, by using an appropriate ``fleshing out''. 
In \cref{sec:routed-quantum-circuits}, we provide a detailed review of the framework of routed quantum circuits and how it allows for consistent decompositions of ICO processes, which we will require to demonstrate the general validity of our construction.
We then briefly review the QC-QC framework in \cref{sec:QCQC}, before presenting our main results in \cref{sec:QCQCs_as_RQCs}, representing any QC-QC as a routed circuit:
in \cref{sec:routed-graph-qcqc}, we propose the backbone of this representation in terms of a routed graph;
after confirming the validity of this routed graph in \cref{sec:branch-graph-QCQC}, we detail the \enquote{fleshing out} of a skeletal supermap associated with this routed graph to any desired QC-QC in \cref{sec:fleshing-out}, and prove the equivalence of the resulting routed supermap to the original QC-QC construction.
We then discuss some possible alternative routed circuit decompositions for QC-QCs in \cref{sec:variations}, before concluding in \cref{sec:discussion}.

\section{Consistent Routed Quantum Circuits}
\label{sec:routed-quantum-circuits}

In this section we review the routed quantum circuit framework for ICO processes of Ref.~\cite{Vanrietvelde2023}, making a significant effort to propose a simplified yet formal presentation of the framework that slightly deviates from that of Ref.~\cite{Vanrietvelde2023} without affecting the main results.

The broader goal of \textit{routed quantum circuits} \cite{Vanrietvelde2021, Vanrietvelde2023} is to generalise quantum circuits so as to represent not only the tensor product structure of the Hilbert spaces involved in a circuit, but also their potential direct sum structure.
For instance, considering two Hilbert spaces $\HH^A = \HH^A_0\oplus\HH^A_1$ and $\HH^B=\HH^B_0\oplus\HH^B_1$, that each decomposes into the direct sum of two \emph{sectors} $\HH^A_k, \HH^B_l$, it may be relevant (and sufficient) to only refer to the composite sectorised space
\begin{equation}
	\tilde{\HH}^{AB} \ = \ \HH^A_0 \otimes \HH^B_1 \ \oplus \ \HH^A_1 \otimes \HH^B_0, \label{eq:tilde_HAB}
\end{equation}
rather than to the full space $\HH^A \otimes \HH^B \supsetneq \tilde{\HH}^{AB}$.
Sectorised spaces like this are of pivotal importance in quantum optics for instance, to describe the transmission of a photon in a superposition of two trajectories $A$ and $B$~\cite{Hler2019}.
In the framework, this is modelled by the introduction of \textit{routes}, relations between the sectors specifying which sectors of $\HH^A \otimes \HH^B$ may be mapped to each other through a given \textit{routed channel}, and possibly rendering certain sectors inaccessible, as in \cref{eq:tilde_HAB} above.

To affirm the validity of a routed quantum circuit, it is sufficient to study its route structure, which can be specified by an underlying \emph{routed graph}, which may generally be cyclic.
Accordingly, we will initially focus on establishing the route structure in isolation, before turning to the circuits themselves.

After presenting some preliminaries on relations in \cref{sec:relations}, we will
formally introduce the routed graph in \cref{sec:routed-graph}, followed by the more fine-grained \emph{branch graph} it implies in \cref{subsec:branch_graph}.
These directed graphs need to satisfy certain properties to ensure that one can build consistent routed quantum circuits on top of them:
the routed graph needs to satisfy the condition of \emph{bi-univocality} which, informally, ensures each operation is performed precisely once, while the branch graph is only allowed to contain so-called \emph{weak} loops, ensuring that these do not introduce any paradoxical feedback mechanism. 
In \cref{sec:routed-maps}, we then introduce formally the notion of \emph{routed linear maps}, before connecting it with the routed graph structure to obtain a \emph{skeletal (routed) supermap} in \cref{subsec:skeletal_fleshingout}, by associating each arrow in the routed graph to a respective Hilbert space.
This skeletal supermap will then be fleshed out to a routed or unrouted supermap by composition with compatible routed channels or superchannels at each node.
We conclude this section by recalling the main result of~\cite{Vanrietvelde2023} in \cref{subsec:main_thm}.

\subsection{Preliminaries: relations}
\label{sec:relations}

The framework heavily relies on binary relations, also called \emph{routes}. Let us start by introducing these.

A (binary) relation $\lambda$ over some sets $K$ and $L$ will be denoted $\lambda:K\to L$. 
It can be represented as a Boolean matrix $(\lambda_k^l)_{k\in K}^{l\in L}$, where the lower and upper indices represent the inputs and outputs of the relation, respectively, and with $\lambda_k^l = 1$ if $k$ and $l$ are related through $\lambda$ (which we write $k\overset{\lambda}{\sim}l$), and $\lambda_k^l = 0$ otherwise. Throughout, we will freely identify a relation with its Boolean matrix. 
A relation can also be defined by some constraint(s) on its inputs and outputs; e.g., the constraint ``$k=l$'' would define the relation with matrix elements $\lambda_k^l = \delta_{k,l}$ (where $\delta$ is the Kronecker delta).
Note also that the converse relation is represented by the transpose matrix $\lambda^\top$.

A relation $\lambda:K\to L$ can formally be seen as a function $K\to\mathcal{P}(L)$, where $\mathcal{P}(L)$ denotes the power set of $L$. We will thus use the notation $\lambda(k)$ to denote the set of all $l\in L$ that are related to $k$ through $\lambda$.
To characterise the subsets of $K$ and $L$ which are actually related by $\lambda$, we introduce the notion of a \textit{practical (co)domain}.

\begin{definition}[Practical (co)domain]
    \label{def:practical-co-domain}
    The \emph{practical domain} of a relation $\lambda: K \to L$ is defined as the subset $K_\text{in($\lambda$),prac} \coloneqq \{k \in K \mid \lambda(k) \neq \emptyset \}$.
    Similarly, the \emph{practical codomain} of $\lambda$ is defined as $L_\text{out($\lambda$),prac} \coloneqq \{ l \in L \mid \lambda^\top (l) \neq \emptyset \}$.
\end{definition}

\subsubsection{Composing relations} \label{subsubsec:compose_rel}

To build up a routed graph we will need to compose relations with one another. 
Given two relations%
\footnote{The definition of composition that we present of course also applies when $L_1$, $L_2$ and/or $L_3$ are trivial.}
$\lambda: K\to L_1\times L_2$ and $\mu:L_2\times L_3\to M$, as illustrated in \cref{fig:compose_relations}, we can define their composition through $L_2$ as the new relation $\lambda * \mu: K\times L_3\to L_1\times M$ that satisfies
\begin{align}
    (k,l_3) \overset{\lambda * \mu}{\sim} (l_1,m) \quad \text{if and only if } \quad \exists\, l_2\in L_2 \text{ such that } k \overset{\lambda}{\sim} (l_1,l_2) \text{ and } (l_2,l_3) \overset{\mu}{\sim} m.
\end{align}

\begin{figure}
    \centering
    \includegraphics[scale=1.2]{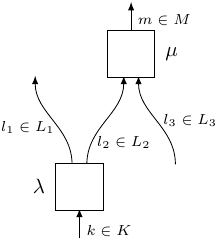}
    \caption{Two relations $\lambda:K\to L_1\times L_2$ and $\mu:L_2\times L_3\to M$ can be composed through $L_2$, resulting in a relation $\lambda * \mu: K\times L_3\to L_1\times M$ that can be written as $\lambda * \mu = \Tr_{L_2} (\lambda\otimes \mu)$, as in \cref{eq:link_prod_relations}.}
    \label{fig:compose_relations}
\end{figure}

Representing the relations $\lambda$ and $\mu$ as Boolean matrices $(\lambda_k^{l_1,l_2})_{k,l_1,l_2}$ and $(\mu_{l_2,l_3}^m)_{l_2,l_3,m}$, it is easy to verify\footnote{$(\lambda * \mu)_{k,l_3}^{l_1,m} = 1 \ \iff \ \exists\,l_2, \ \lambda_k^{l_1,l_2} \mu_{l_2,l_3}^m = 1 \ \iff \ \boolSum_{l_2} \lambda_k^{l_1,l_2} \mu_{l_2,l_3}^m = 1$.} that $\lambda * \mu$ is represented by the Boolean matrix with elements
\begin{align}
    (\lambda * \mu)_{k,l_3}^{l_1,m} = \boolSum_{l_2} \lambda_k^{l_1,l_2} \mu_{l_2,l_3}^m, \label{eq:compose_rel}
\end{align}
where $\boolSum$ denotes the Boolean sum, i.e.\ which evaluates to $1$ if at least one of the summands is $1$ (equivalent to the logical \emph{or}).
The composed relation can thus be written as a (Boolean) partial trace:\footnote{This is somewhat similar to the ``link product''~\cite{Chiribella2008,Chiribella2009} (in its ``pure version''~\cite{Wechs2021}) used to compose quantum maps in their Choi representation.}
\begin{align}
    \lambda * \mu = \Tr_{L_2} (\lambda\otimes \mu). \label{eq:link_prod_relations}
\end{align}
Remark that composing two relations $\lambda$ and $\mu$ in this manner may then generally result in a further restricted practical domain or co-domain in comparison to the individual relations.
Specifically, in the case for instance where $L_1$ and $L_3$ are trivial, $K_\text{in($\lambda * \mu$),prac}$ may be strictly contained in $K_\text{in($\lambda$),prac}$ if $\lambda(k)$ does not relate to any $l_2 \in L_2$ within the practical domain of $\mu$, and analogously for the practical co-domain $M_\text{out($\lambda * \mu$),prac}$.

This way of composing relations can be generalised to more relations, connected through more common sets of values (i.e., more linking arrows in a graphical representations as in \cref{fig:compose_relations}; see, e.g., \cref{fig:Routed Graph} below).
For the resulting relation to hold between a pair of its input and output values, it is required that there exists a value in each of the common sets, such that all individual relations hold.
The resulting relation is still obtained by tensoring all relations and partially tracing out over all common sets of indices, thus defining a (commutative and associative) multi-term ``link product''.
We will in particular use this in defining the notion of a \textit{choice relation} (\cref{def:choice_rel}) in \cref{subsubsec:choice_rel}.

\subsubsection{Branched relations}
\label{sec:branched-relations}

In the framework of routed circuits~\cite{Vanrietvelde2023} we focus on a specific type of relations, the so-called \emph{branched} relations.

\begin{definition}[Branched relation] \label{def:branched_rel}
    A relation $\lambda:K\to L$ is said to be \emph{branched} if for any $k,k'\in K$, $\lambda(k)$ and $\lambda(k')$ are either the same, or disjoint.
\end{definition}

Note that a relation $\lambda$ is branched if and only if its converse relation $\lambda^\top$ is branched. Branched relations define partitions of their (practical) domain and codomain that are related through so-called \emph{branches}.

\begin{definition}[Branch] \label{def:branch}
    A \emph{branch} $\beta = (K^\beta\to L^\beta)$ of a branched relation $\lambda:K\to L$ is a pair of nonempty sets $K^\beta \subseteq K$ and $L^\beta \subseteq L$ such that for any $k\in K^\beta$, $\lambda(k) = L^\beta$ and for any $l\in L^\beta$, $\lambda^\top(l) = K^\beta$.
\end{definition}

We denote by $\Bran(\lambda)$ the set of branches $\beta$ of $\lambda$. Note that the corresponding partitions of $K$ and $L$ need not be complete, in the sense that $\bigsqcup_{\beta\in\Bran(\lambda)} K^\beta$ might not be equal to $K$ and $\bigsqcup_{\beta\in\Bran(\lambda)} L^\beta$ might not be equal to $L$; see \cref{fig:branched_relation}.
Rather, $\bigsqcup_{\beta\in\Bran(\lambda)} K^\beta$ is precisely the practical domain of $\lambda$, while $\bigsqcup_{\beta\in\Bran(\lambda)} L^\beta$ is its practical co-domain, as introduced in \cref{def:practical-co-domain} above.

Note that generally, if the partitions $\{L^\beta \mid \beta\in\Bran(\lambda)\}$ and $\{L^\gamma \mid \gamma\in\Bran(\mu)\}$ differ, the sequential composition of two branched relations $\lambda: K \to L$ and $\mu: L \to M$ will no longer be a branched relation.\footnote{However, this is not an if and only if condition, e.g., $\lambda * \mu$ is still branched if the partition of the domain of $\mu$ is a fine-graining of the partition of the co-domain of $\lambda$ or vice versa.}

\begin{figure}
    \centering
    \includegraphics[scale=.95]{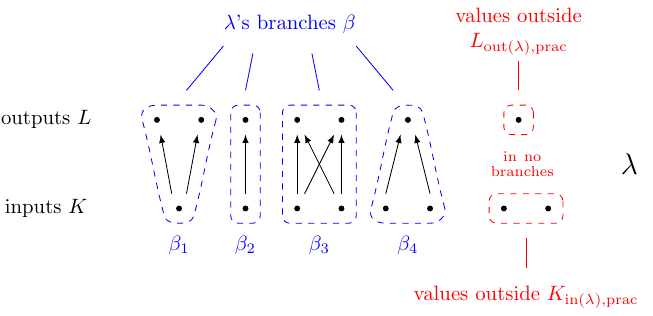}
    \caption{Example of a branched relation $\lambda: K\to L$, with its domain and codomain partitioned into branches (reproducing the example of Fig.~12 in~\cite{Vanrietvelde2023}). If $\lambda$ is not total or not surjective, then some values of $K$ or $L$ are outside of the practical domain or codomain of $\lambda$, and are in no branches.}
    \label{fig:branched_relation}
\end{figure}

\subsubsection{Augmenting relations}
\label{subsubsec:augmenting}

From a branched relation, we define its \emph{augmented} version, which will be an important technical tool in determining if the composition of some relations can describe a valid circuit structure.

\begin{definition}[{Augmented relation\protect\footnote{Note that we depart slightly here from the definition given in~\cite{Vanrietvelde2023}, which (equivalently) wrote the output of the partial function as $(l^{\bar{\beta}}, (\delta_{\beta,\bar{\beta}})_{\beta\in\Bran(\lambda)}) \in L\times\left(\bigtimes_{\beta\in\Bran(\lambda)} \mathtt{Happens}_{\beta}\right)$ (with each $\mathtt{Happens}_{\beta} \cong \{0,1\}$). Since looking at all $\delta_{\beta,\bar{\beta}}$'s, for all $\beta\in\Bran(\lambda)$ individually, will not be relevant right away when we first use augmented relations (when constructing the ``choice relation'', see \cref{subsubsec:choice_rel}), but will only be relevant later on when we need to clarify causal influences between branches (when constructing the ``branch graph'', see \cref{subsec:branch_graph}), we will keep that alternative description of $\bar{\beta}$ for later and use the simpler expression of the augmented relations' outputs described here. \label{ftn:depart_aug}}}] \label{def:augmented_rel}
    Given a branched relation $\lambda:K\to L$, its \emph{augmented} version is the partial function -- written as a branched relation -- $\lambda^\text{aug}: K\times\left(\bigtimes_{\beta\in\Bran(\lambda)} L^\beta\right) \to L\times\Bran(\lambda)$ given by
    \begin{equation}
        \lambda^\text{aug} (k, (l^\beta)_{\beta\in\Bran(\lambda)}) \coloneqq
        \begin{cases}
            (l^{\bar{\beta}}, \bar{\beta}) & \text{ if } k\in K^{\bar{\beta}} \\
            \emptyset & \text{ if } k\notin K_\text{in($\lambda$),prac} 
        \end{cases}. \label{eq:augmented_rel}
    \end{equation}
\end{definition}

That is, the augmented relation takes as input, in addition to the original input $k\in K$ of the initial branched relation, the choice of one output value $l^\beta$ -- that we call a ``bifurcation choice'' -- for each branch $\beta$ of the relation. 
It identifies whether the input $k$ belongs to the input set $K^{\bar{\beta}}$ of some branch $\bar{\beta}$. 
If it does, it then outputs the corresponding bifurcation choice for that branch, together with the name of that branch; we then say that this branch ``happens''. 
If it does not, the partial function is undefined -- or, equivalently, the input is related to the empty set. See \cref{fig:augmenting} for an illustration.

\begin{figure}
    \centering
    \includegraphics[scale=1.2]{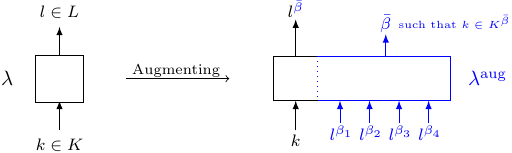}
    \caption{Augmenting the branched relation $\lambda$ of \cref{fig:branched_relation}. The augmented relation identifies in which branch $\bar{\beta}$ the input $k$ is, and outputs the choice of output $l^{\bar{\beta}}$ -- the ``bifurcation choice'' -- within that branch provided among the auxiliary inputs, together with the name $\bar{\beta}$ of the branch that thus ``happens''. If the input $k$ is in no branch, the relation has no output.}
    \label{fig:augmenting}
\end{figure}

Note that some auxiliary inputs and/or outputs may be trivial, i.e.\ they may take only one possible value. In such a case, one may remove them from the definition of the augmented relation.%
\footnote{For instance the branches $\beta_2$ and $\beta_4$ in \cref{fig:branched_relation} each have a single output value, so that the inputs $l^{\beta_2}$ and $l^{\beta_4}$ of the augmented relation can only take that single value. The auxiliary output on the other hand may be trivial if the relation has only one branch. One may choose, as e.g.\ in Fig.~13 of~\cite{Vanrietvelde2023}, not to represent these trivial inputs and outputs explicitly.}

\subsection{Routed graph}
\label{sec:routed-graph}

\subsubsection{Indexed and (branched) routed graphs}

The route structure of a routed circuit will be obtained by composing (branched) relations together to obtain a so-called \emph{routed graph}.

\begin{definition}[Indexed graph\protect\footnote{Again we slightly depart here from the definition given in~\cite{Vanrietvelde2023}, by not including the ``dimension'' of each index value in the definition of an indexed graph. We will only need to specify later on, when constructing the branch graph, which index values are considered to be ``one-dimensional''; see \cref{def:1dim}. \\
Furthermore, contrary to~\cite{Vanrietvelde2023} we do not require \emph{a priori} that for an open-ended arrow $A$, $\Ind_A$ must necessarily be a singleton; we will only impose this when this becomes necessary, namely when testing the property of ``(bi\=/)univocality'', see \cref{def:biunivocality}.
\\
Also, for simplicity we define an indexed graph as a graph rather than a multigraph, which, as noted in~\cite{Vanrietvelde2023}, does not make any difference in the framework presented here. In fact, having several arrows from a node $\mathbf{N}$ to a node $\mathbf{M}$, each with a set of index values $\Ind_{A_j}$, is equivalent to having a single arrow with the cartesian product $\bigtimes_j\Ind_{A_j}$ attached to it. Indeed, we allow index values to be
of any type, including lists of values (or sets, as we will use in our construction for QC-QCs in \cref{sec:QCQCs_as_RQCs}). \label{ftn:departures_indexed_graph}}]
    An \emph{indexed graph} $\Gamma$ is a directed graph in which each arrow $A$ is associated with a non-empty set of index values $\Ind_A$. \\
    The arrows in the graph can either connect two nodes (or possibly a node to itself) or be ``open-ended'' -- i.e., an arrow that either points to a node but ``comes from nowhere'' (from no other node), or an arrow that leaves from a node but ``goes nowhere'' (to no other node).
\end{definition}

We denote by $\Nodes_\Gamma$ the set of nodes of the graph $\Gamma$, and by $\mathtt{Arr}_\Gamma$ its set of arrows. For each node $\mathbf{N}\in\Nodes_\Gamma$, we denote by $\text{in}(\mathbf{N}) \subseteq \mathtt{Arr}_\Gamma$ the subset of arrows that point to $\mathbf{N}$, and by $\text{out}(\mathbf{N}) \subseteq \mathtt{Arr}_\Gamma$ the subset of arrows that leave from $\mathbf{N}$. We then define $\Ind^\text{in}_\mathbf{N} \coloneqq \bigtimes_{A \in \text{in}(\mathbf{N})} \Ind_A$ and $\Ind^\text{out}_\mathbf{N} \coloneqq \bigtimes_{A \in \text{out}(\mathbf{N})} \Ind_A$.

\begin{definition}[(Branched) Routed graph\protect\footnote{Once again we slightly depart and take a more general definition than that in~\cite{Vanrietvelde2023}, which required routed graphs to only involve branched relations. Here we will specify and talk about a ``branched routed graph'' when this is the case.}]
    A \emph{routed graph}
    $\cG= (\Gamma,(\lambda_\mathbf{N})_{\mathbf{N}\in\Nodes_\Gamma})$ consists of an indexed graph $\Gamma$ and, for every node $\mathbf{N}$, a relation $\lambda_\mathbf{N}:\Ind^\text{in}_\mathbf{N}\to \Ind^\text{out}_\mathbf{N}$ called the \emph{route} for node $\mathbf{N}$. \\
    A routed graph $\cG= (\Gamma,(\lambda_\mathbf{N})_{\mathbf{N}\in\Nodes_\Gamma})$ is said to be \emph{branched} if all relations $\lambda_\mathbf{N}$, for all $\mathbf{N}\in\Nodes_\Gamma$, are branched.
\end{definition}

We will often, for brevity, just write a (branched) routed graph as $\cG$ or $(\Gamma,(\lambda_\mathbf{N})_\mathbf{N})$. 
Graphically we write an index
$k_A$ (with the understanding that $k_A\in\Ind_A$) next to each arrow $A$ -- although we may drop it when $\Ind_A$ is just a singleton, or simply to lighten the graphical representation -- and we write the route $\lambda_\mathbf{N}$ next to each node $\mathbf{N}$; see \cref{fig:routed_graph}.

\begin{figure}
    \centering
    \begin{subfigure}[c]{0.4\textwidth}
        \centering
        \includegraphics[scale=.77]{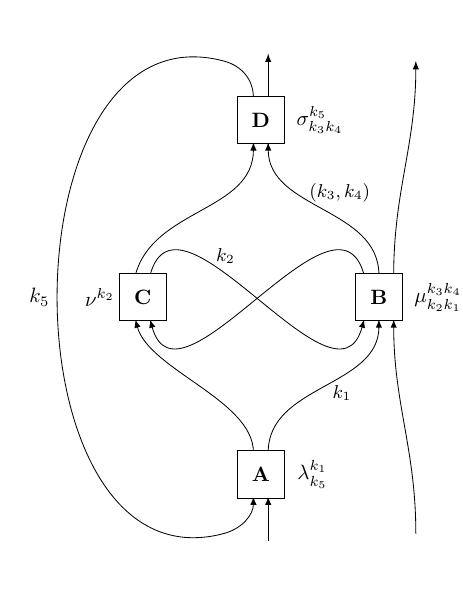} 
        \caption{A routed graph}
        \label{fig:routed_graph}
    \end{subfigure}%
    $\qquad\xrightarrow{\text{Augmenting}}\quad$%
    \begin{subfigure}[c]{0.4\textwidth}
        \centering
        \includegraphics[scale=.77]{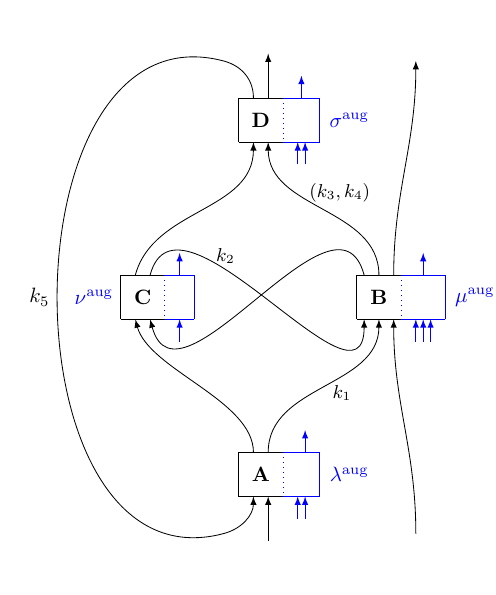}
        \caption{The associated augmented routed graph}
        \label{fig:augmented_graph}
    \end{subfigure}
    \caption{
        {\bf (a)} Example of a routed graph (reproducing that of Fig.~11(b) in~\cite{Vanrietvelde2023}). The arrows not bearing indices have a single index value attached to them (i.e.\ the corresponding $\Ind_A$ is a singleton).
        {\bf (b)}~Augmenting all of its relations (as in \cref{fig:augmenting}, with, for each node $\mathbf{N}$ the set $K$ being $\Ind^\text{in}_\mathbf{N} = \bigtimes_{A \in \text{in}(\mathbf{N})} \Ind_A$ and the set $L$ being $\Ind^\text{out}_\mathbf{N} = \bigtimes_{A \in \text{out}(\mathbf{N})} \Ind_A$) and composing the relations through the arrows of the graph (as explained in \cref{subsubsec:compose_rel}) yields the \emph{choice relation} -- a relation between (the index values attached to) all arrows with open-ended tails and all arrows with open-ended heads.}
    \label{fig:Routed Graph}
\end{figure}

In the context of a branched routed graph $(\Gamma,(\lambda_\mathbf{N})_\mathbf{N})$, the branches of a relation $\lambda_\mathbf{N}$ will generically be denoted $\mathbf{N}^\beta$ (rather than just $\beta$ as in the general preliminary section on relations); the corresponding sets of input and output values will be denoted $\Ind^\text{in}_{\mathbf{N}^\beta} \subseteq \Ind^\text{in}_\mathbf{N}$ and $\Ind^\text{out}_{\mathbf{N}^\beta} \subseteq \Ind^\text{out}_\mathbf{N}$ (instead of $K^\beta \subseteq K$ and $L^\beta \subseteq L$), respectively, so that $\mathbf{N}^\beta = (\Ind^\text{in}_{\mathbf{N}^\beta}\to\Ind^\text{out}_{\mathbf{N}^\beta})$.
By abuse of language, the inputs, outputs and branches of the relations $\lambda_\mathbf{N}$ will also be referred to as the inputs, outputs and branches of the node $\mathbf{N}$ directly.

Let us also mention the possibility to use \emph{global index constraints} as a short-hand to specify the routes $(\lambda_\mathbf{N})_\mathbf{N}$ for a routed graph, which we will extensively use throughout this work.
Indeed, just as the relation at a node can conveniently by specified by some constraint(s) on its inputs and outputs, the global routes (i.e., the multi-node relation $\bigtimes_{\mathbf{N}\in\Nodes_\Gamma} \lambda_\mathbf{N}$) can be specified indirectly by giving some constraints relating all the indices appearing in the routed graph.
Thereby, they specify the globally consistent assignments of index values across the graph.

\subsubsection{Choice relation and (bi-)univocality} \label{subsubsec:choice_rel}

Instead of attaching the original relations $\lambda_\mathbf{N}$ to each node of the branched routed graph $(\Gamma,(\lambda_\mathbf{N})_\mathbf{N})$, one may ``plug in'' their augmented versions, representing the auxiliary inputs and outputs introduced in the augmenting procedure as additional open-ended arrows (similarly to \cref{fig:augmenting}), with for each node $\mathbf{N}$ the input sets $\Ind^\text{out}_{\mathbf{N}^\beta}$ (for all $\mathbf{N}^\beta$) and the output set $\Bran(\lambda_\mathbf{N})$ attached to them.
By composing the augmented relations through the arrows of the graph (see \cref{subsubsec:compose_rel}), this allows us to define a relation between all auxiliary inputs -- together with the indices attached to arrows that ``come from nowhere'' in the original indexed graph -- and all auxiliary outputs -- together with the indices attached to arrows that ``go nowhere''; see \cref{fig:augmented_graph}.

\bigskip

Formally, denoting by $\mathtt{Arr}^\text{int}_\Gamma$ the set of all ``internal'' arrows of $\Gamma$ (which connect two nodes in the graph), by $\mathtt{Arr}^\text{in}_\Gamma$ and $\mathtt{Arr}^\text{out}_\Gamma$ the sets of all (original) arrows of $\Gamma$ that ``come from nowhere'' and that ``go nowhere'', respectively (so that $\mathtt{Arr}_\Gamma = \mathtt{Arr}^\text{int}_\Gamma \sqcup \mathtt{Arr}^\text{in}_\Gamma \sqcup \mathtt{Arr}^\text{out}_\Gamma$), we introduce the following definition.

\begin{definition}[Choice relation] \label{def:choice_rel}
    Given a branched routed graph $\cG = (\Gamma,(\lambda_\mathbf{N})_\mathbf{N})$, its \emph{choice relation} is the relation
    $\Lambda_{(\Gamma,(\lambda_\mathbf{N})_\mathbf{N})}\!: \left(\bigtimes_{A\in\mathtt{Arr}^\text{in}_\Gamma}\Ind_A\right)\times\left(\!\bigtimes_{\substack{\!\mathbf{N}\in\Nodes_\Gamma,\\ \mathbf{N}^\beta\in\Bran(\lambda_\mathbf{N})}} \!\!\!\Ind^\text{out}_{\mathbf{N}^\beta} \!\right) \to \left(\bigtimes_{A\in\mathtt{Arr}^\text{out}_\Gamma}\Ind_A\right)\times\left(\bigtimes_{\mathbf{N}\in\Nodes_\Gamma} \Bran(\lambda_\mathbf{N}) \right)$ defined 
    as
    \begin{align}
        \Lambda_{(\Gamma,(\lambda_\mathbf{N})_\mathbf{N})} \coloneqq &
        \Tr_{(\Ind_A)_{A\in\mathtt{Arr}^\text{int}_\Gamma}} \left[{\textstyle \bigotimes_{\mathbf{N}\in\Nodes_\Gamma}}\,\lambda_\mathbf{N}^\text{aug} \right] 
		= \raisebox{-1.9ex}{$\scalebox{2}{*}$}{}_{\mathbf{N}\in\Nodes_\Gamma} \ \lambda_\mathbf{N}^\text{aug} \label{eq:choice_rel}
    \end{align}
    (where \raisebox{-1.9ex}{$\scalebox{2}{*}$} denotes the multi-term ``link product'' that allows one to compose relations, generalising \cref{eq:link_prod_relations}).
\end{definition}

Thus, together with some input in $\bigtimes_{A\in\mathtt{Arr}^\text{in}_\Gamma}\Ind_A$ and some output in $\bigtimes_{A\in\mathtt{Arr}^\text{out}_\Gamma}\Ind_A$, the choice relation relates some choice of output values $(l^{\mathbf{N}^\beta})_{\mathbf{N},\mathbf{N}^\beta} \in \bigtimes_{\mathbf{N},\mathbf{N}^\beta}\Ind^\text{out}_{\mathbf{N}^\beta}$ -- some ``bifurcation choice'' -- for each of all branches of all routes in the graph, to the specific branches, for all nodes, that do ``happen'' for those bifurcation choices.

\bigskip

With all this in place, we are now ready to define the property of ``(bi-)univocality'' that we will require the branched routed graph to satisfy. This presupposes that the indexed graph has trivial input and output values attached to its open-ended arrows (i.e., that for all $A\in\mathtt{Arr}^\text{in}_\Gamma\sqcup\mathtt{Arr}^\text{out}_\Gamma$, $\Ind_A$ is a singleton),\footnote{
    \label{fn:bi-univocal-index}
    In \cite{Vanrietvelde2023}, this assumption is introduced to simplify the statement of bi-univocality, but the possibility of a direct generalisation is claimed by the authors.}
and formalises the requirement that for each possible combination of bifurcation choices, one and only one branch should happen at each node (and similarly in the adjoint graph\footnote{The adjoint of a routed graph is simply obtained by reversing all its arrows and taking the converse of all its routes. Beware that $\Lambda_{(\Gamma^\top,(\lambda^\top_\mathbf{N})_\mathbf{N})} \neq \Lambda_{(\Gamma,(\lambda_\mathbf{N})_\mathbf{N})}^\top$: the choice relation of the adjoint graph $\cG^\top$ does not even act on the same spaces as the adjoint of the choice relation on $\cG$.}).

\begin{definition}[(Bi-)Univocality] \label{def:biunivocality}
    A branched routed graph $\cG = (\Gamma,(\lambda_\mathbf{N})_\mathbf{N})$ with trivial $\Ind_A$ for all open-ended arrows $A$ is said to be \emph{univocal} if its choice relation $\Lambda_{(\Gamma,(\lambda_\mathbf{N})_\mathbf{N})}$ is a function -- which is then accordingly called the \emph{choice function}.
    \\
    Such a $\cG=(\Gamma,(\lambda_\mathbf{N})_\mathbf{N})$ is said to be \emph{bi-univocal} if both it and its adjoint $\cG^\top=(\Gamma^\top,(\lambda_\mathbf{N}^\top)_\mathbf{N})$ are univocal.
\end{definition}

Given that the definition of (bi-)univocality requires \emph{a priori} that the routed graph is branched and that $\Ind_A$ is trivial for each open-ended arrow $A$, we will no longer explicitly specify these properties when referring to a (bi-)univocal routed graph.
Furthermore, since $\bigtimes_{A\in\mathtt{Arr}^\text{in}_\Gamma}\Ind_A$ and $\bigtimes_{A\in\mathtt{Arr}^\text{out}_\Gamma}\Ind_A$ are taken to be trivial, we will make them implicit in the definition of the choice function, which will then be seen as a function
\begin{equation}
    \Lambda_{(\Gamma,(\lambda_\mathbf{N})_\mathbf{N})}: \bigtimes_{\substack{\mathbf{N}\in\Nodes_\Gamma,\\ \mathbf{N}^\beta\in\Bran(\lambda_\mathbf{N})}} \Ind^\text{out}_{\mathbf{N}^\beta} \to \bigtimes_{\mathbf{N}\in\Nodes_\Gamma} \Bran(\lambda_\mathbf{N}) \, . \label{eq:choice-function}
\end{equation}

\begin{remark}
    \label{rem:global_constraint}
    If $\Lambda_{(\Gamma,(\lambda_\mathbf{N})_\mathbf{N})}$ is not surjective, the choice function acts as a \emph{global constraint} restricting which combination of branches in $\bigtimes_{\mathbf{N}\in\Nodes_\Gamma} \Bran(\lambda_\mathbf{N})$ may jointly happen.
    If, further, the projection of the choice function onto an individual node $\mathbf{M}$, i.e.\ $
    \bigtimes_{\mathbf{N}\in\Nodes_\Gamma,\mathbf{N}^\beta\in\Bran(\lambda_\mathbf{N})} \Ind^\text{out}_{\mathbf{N}^\beta} \to \Bran(\lambda_\mathbf{M})$
    is not surjective, some branches $\mathbf{M}^\gamma$ will never happen and may be disregarded entirely.\footnote{
        For the following section, one can easily deduce that any such branch will not partake in any strong or weak parent relation.
    }
    This can be seen as a restriction of the practical input and output sets.
    In practice, we will restrict the combinations of incoming and outgoing indices $\Ind^\text{in/out}_\mathbf{M}$ considered for each node $\mathbf{M}$ \textit{a priori}, so as not to be left with unnecessary irrelevant branches $\mathbf{M}^\gamma$.
\end{remark}

\subsection{Branch graph}
\label{subsec:branch_graph}

Once we have checked that the routed graph under consideration is bi-univocal, we can build its \emph{branch graph}, on which we will check some condition (the ``only weak loops'' property) on the potential loops it may contain.
The branch graph connects the branches of the different nodes of the original routed graph, providing a more fine-grained picture of their connectivity, and differentiates between so-called \emph{strong} and \emph{weak parents}.%
\footnote{For clarity: these are two independent notions. A branch may be a strong parent but not a weak parent of another branch, or vice versa.}

\subsubsection{One-dimensional values and strong parents} \label{subsubsec:triv_dim}

Roughly speaking, the notion of \emph{strong parent} indicates whether there can be some direct transmission of information (in a way that will be specified later) from one branch at one node to another branch at some other node in the graph. For this, one first needs to characterise how -- through which index values -- the two branches under consideration are linked:

\begin{definition}[Branch-linking values]
    Let $\cG = (\Gamma,(\lambda_\mathbf{N})_\mathbf{N})$ be a branched routed graph, and let
    $\mathbf{N}^\beta\in\Bran(\lambda_\mathbf{N})$, $\mathbf{M}^\gamma\in\Bran(\lambda_\mathbf{M})$ be two branches of two nodes $\mathbf{N}, \mathbf{M}\in\Nodes_\Gamma$.
    If there is an arrow $A$ from $\mathbf{N}$ to $\mathbf{M}$ in the graph $\Gamma$, we define the \emph{set of values linking $\mathbf{N}^\beta$ to $\mathbf{M}^\gamma$} as
    \begin{align}
        & \mathtt{LinkVal}(\mathbf{N}^\beta, \mathbf{M}^\gamma) \notag \\
        & \ \coloneqq
        \left\{k_A\in \Ind_A \, \middle|\!
        \begin{array}{rl}
            \exists (k_{A'})_{A'\in\mathtt{Arr}_\Gamma\backslash\{A\}} \text{ such that}\!\!\! & \forall\,\mathbf{N}'\in\Nodes_\Gamma, \, (k_{A'})_{A'\in\text{in}(\mathbf{N}')}
            \!\overset{\lambda_{\mathbf{N}'}}{\sim}\!(k_{A'})_{A'\in\text{out}(\mathbf{N}')} \\[1mm]
            \text{and}\!\!\! & (k_{A'})_{A'\in\text{out}(\mathbf{N})} \in \Ind^\text{out}_{\mathbf{N}^\beta}, \ (k_{A'})_{A'\in\text{in}(\mathbf{M})} \in \Ind^\text{in}_{\mathbf{M}^\gamma}
        \end{array}\!\!\!
        \right\}.
    \end{align}
\end{definition}

That is, $\mathtt{LinkVal}(\mathbf{N}^\beta, \mathbf{M}^\gamma)$ is the set of index values that can be attributed to the arrow (if any) from $\mathbf{N}$ to $\mathbf{M}$, in order to satisfy the global index constraints, and belong to the branches $\mathbf{N}^\beta, \mathbf{M}^\gamma$ of the nodes $\mathbf{N}, \mathbf{M}$.
By definition, these global index constraints are satisfied if there exists a consistent assignment of all other values on all other arrows of the graph, whose appropriate combinations (corresponding to the arrows that point to, or leave from a given node) satisfy all the relations specified at all nodes in the graph.
In this case, these index combinations are related by the global routes. 

\bigskip

In fact not all index values that thus link two branches will be useful to transmit information. Indeed we shall specify and single out certain index values that will play a specific role, in not being able to transmit information by themselves.

\begin{definition}[One-dimensional index values] \label{def:1dim}
    Certain index values are singled out and said to be \emph{``one-di\-men\-sio\-nal''}.%
    \footnote{With the original definitions of~\cite{Vanrietvelde2023}, all index values are attributed a non-zero natural number called its ``dimension''. For now we actually only need to identify those whose dimension is 1.}
    \\
    Given an arrow $A$ of an indexed graph $\Gamma$, we denote by $\mathtt{1Dim}_A \subseteq \Ind_A$ the set of one-dimensional index values attached to it; 
    $(\mathtt{1Dim}_A)_{A\in\mathtt{Arr}_\Gamma}$ (which we shall abbreviate as $(\mathtt{1Dim}_A)_A$) thus denotes the set of all one-dimensional index values of $\Gamma$.
\end{definition}

Once the set of one-dimensional index values is specified, the \emph{strong parents} are then characterised as follows:

\begin{definition}[Strong parent] \label{def:strong_parent}
    Given a branched routed graph $\cG = (\Gamma,(\lambda_\mathbf{N})_\mathbf{N})$, its set of one-dimensional index values $(\mathtt{1Dim}_A)_A$, and two nodes $\mathbf{N},\mathbf{M}\in\Nodes_\Gamma$ as well as two branches $\mathbf{N}^\beta\in\Bran(\lambda_\mathbf{N})$, $\mathbf{M}^\gamma\in\Bran(\lambda_\mathbf{M})$, we say that $\mathbf{N}^\beta$ is a \emph{strong parent} of $\mathbf{M}^\gamma$ with respect to $(\Gamma,(\lambda_\mathbf{N})_\mathbf{N},(\mathtt{1Dim}_A)_A)$ 
    if there is an arrow $A$ from $\mathbf{N}$ to $\mathbf{M}$ in the graph $\Gamma$, and $\mathtt{LinkVal}(\mathbf{N}^\beta, \mathbf{M}^\gamma)$ contains 
    \begin{itemize}
        \item (at least) one element that is \emph{not} one-dimensional -- i.e., not in $\mathtt{1Dim}_A$;
        \item \emph{or} strictly more than one element.
    \end{itemize}
\end{definition}

\subsubsection{Weak parents} \label{subsubsec:weak_parents}

Recall that, according to \cref{def:augmented_rel}, the augmented version of a relation outputs the branch $\bar{\beta}$ (if any) that the input $k$ belongs to -- the branch that ``happens'' when the input is $k$.
One may then focus on a given branch $\beta$ of the relation, and ask specifically if that branch happens or not -- i.e., ask for the Boolean value of $\delta_{\beta,\bar{\beta}}$ (see \cref{ftn:depart_aug}).

Similarly, considering now the choice function $\Lambda_{(\Gamma,(\lambda_\mathbf{N})_\mathbf{N})}: \bigtimes_{\mathbf{N},\mathbf{N}^\beta} \Ind^\text{out}_{\mathbf{N}^\beta} \to \bigtimes_{\mathbf{N}} \Bran(\lambda_\mathbf{N})$ of a univocal routed graph, one may focus on a given branch $\mathbf{M}^\gamma\in\Bran(\lambda_\mathbf{M})$ of a given node $\mathbf{M}$, and ask specifically if that branch happens for some given inputs in $\bigtimes_{\mathbf{N},\mathbf{N}^\beta} \Ind^\text{out}_{\mathbf{N}^\beta}$. To formalise this, we introduce the following:

\begin{definition}[$\mathbf{M}^\gamma$-Happens function] \label{def:MBetaHapp}
    Given a univocal routed graph $\cG = (\Gamma,(\lambda_\mathbf{N})_\mathbf{N})$ and a specific node $\mathbf{M}\in\Nodes_\Gamma$ with a specific branch $\mathbf{M}^\gamma\in\Bran(\lambda_\mathbf{M})$, the \emph{$\mathbf{M}^\gamma$-Happens function} is the Boolean function\footnote{$\Happ_{\mathbf{M}^\gamma}$ here is not to be confused with the notation $\mathtt{Happens}_{\beta}$ from \cite{Vanrietvelde2023}, which was denoting the Boolean output set, $\mathtt{Happens}_{\beta} \cong \{0,1\}$ (see \cref{ftn:depart_aug}).}
    \begin{align}
        \Happ_{\mathbf{M}^\gamma}: {\textstyle \bigtimes_{\mathbf{N}\in\Nodes_\Gamma, \mathbf{N}^\beta\in\Bran(\lambda_\mathbf{N})}} \Ind^\text{out}_{\mathbf{N}^\beta} & \to \{0,1\}, \notag \\
        (l^{\mathbf{N}^\beta})_{\mathbf{N}\in\Nodes_\Gamma, \mathbf{N}^\beta\in\Bran(\lambda_\mathbf{N})} & \mapsto \delta_{\mathbf{M}^\gamma,\overline{\mathbf{M}^\gamma}}, \label{eq:MBetaHapp}
    \end{align}
    where $\overline{\mathbf{M}^\gamma}\in\Bran(\lambda_\mathbf{M})$ is the output of the choice function $\Lambda_{(\Gamma,(\lambda_\mathbf{N})_\mathbf{N})}: \bigtimes_{\mathbf{N},\mathbf{N}^\beta} \Ind^\text{out}_{\mathbf{N}^\beta} \to \bigtimes_{\mathbf{N}} \Bran(\lambda_\mathbf{N})$ of the graph corresponding to the node $\mathbf{M}$, for the same input $(l^{\mathbf{N}^\beta})_{\mathbf{N},\mathbf{N}^\beta}$.
\end{definition}

The notion of a \emph{weak parent} then characterises whether a bifurcation choice made within a given branch influences whether another branch happens or not.

\begin{definition}[Weak parent] \label{def:weak_parent}
    Given a univocal routed graph $\cG = (\Gamma,(\lambda_\mathbf{N})_\mathbf{N})$ and two nodes $\mathbf{N},\mathbf{M}\in\Nodes_\Gamma$ with two branches $\mathbf{N}^\beta\in\Bran(\lambda_\mathbf{N}), \mathbf{M}^\gamma\in\Bran(\lambda_\mathbf{M})$, we say that $\mathbf{N}^\beta$ is a \emph{weak parent} of $\mathbf{M}^\gamma$ with respect to $(\Gamma,(\lambda_\mathbf{N})_\mathbf{N})$ if the Boolean output value of the $\mathbf{M}^\gamma$-Happens function $\Happ_{\mathbf{M}^\gamma}$, as defined in \cref{eq:MBetaHapp}, depends on the input value of $l^{\mathbf{N}^\beta}$ -- i.e., the bifurcation choice -- attached to the branch $\mathbf{N}^\beta$.
\end{definition}

Note that if $\mathbf{N}^\beta$ has a single output value (i.e.\ there is no bifurcation choice to be made within that branch: $\abs{\Ind^\text{out}_{\mathbf{N}^\beta}} = 1$), then it cannot be the weak parent of any other branch.
On the other hand, if the node $\mathbf{M}$ has a single possible branch $\mathbf{M}^\gamma$, i.e.\ $\Bran(\lambda_\mathbf{M}) = \{ \mathbf{M}^\gamma \}$, then $\mathbf{M}^\gamma$ happens for sure ($\Happ_{\mathbf{M}^\gamma}$ returns the constant output 1, independently of any input values at any node of the routed graph), so that $\mathbf{M}^\gamma$ has no weak parent.

\subsubsection{Branch graph and weak loops}
\label{subsubsec:branch_graph}

\begin{definition}[Branch graph] \label{def:branch_graph}
    Given a bi-univocal routed graph $\cG = (\Gamma,(\lambda_\mathbf{N})_\mathbf{N})$ and the specification of its set of one-dimensional index values $(\mathtt{1Dim}_A)_A$, the \emph{branch graph} $\cG^\Bran$ is the directed (multi)graph in which
    \begin{itemize}
        \item the nodes are the branches of $(\Gamma,(\lambda_\mathbf{N})_\mathbf{N})$;
        \item there is a solid arrow from a branch $\mathbf{N}^\beta$ to a branch $\mathbf{M}^\gamma$ if $\mathbf{N}^\beta$ is a strong parent of $\mathbf{M}^\gamma$ with respect to $(\Gamma,(\lambda_\mathbf{N})_\mathbf{N},(\mathtt{1Dim}_A)_A)$;
        \item there is a green dashed arrow from $\mathbf{N}^\beta$ to $\mathbf{M}^\gamma$ if $\mathbf{N}^\beta$ is a weak parent of $\mathbf{M}^\gamma$ with respect to $(\Gamma,(\lambda_\mathbf{N})_\mathbf{N})$;
        \item there is a red dashed arrow from $\mathbf{N}^\beta$ to $\mathbf{M}^\gamma$ if $\mathbf{M}^\gamma$ is a weak parent of $\mathbf{N}^\beta$ with respect to the adjoint graph $(\Gamma^\top,(\lambda_\mathbf{N}^\top)_\mathbf{N})$.
    \end{itemize}
\end{definition}

Considering all types of arrows, a branch graph may contain cycles, consisting of arrows of potentially different colours. 
The framework distinguishes a certain type of loops, namely:

\begin{definition}[Weak loop]
    \label{def:weak-loops}
    A loop in a branch graph $\cG^\Bran$ is said to be \emph{weak} if it only contains green dashed arrows, or if it only contains red dashed arrows.
\end{definition}

As we will see below, such loops are not ``harmful'':
routed graphs that only contain weak loops can be used to build valid quantum processes.
This motivates the following:

\begin{definition}[Valid graph]
    \label{def:valid-routed-graph}
    We take a routed graph $\cG = (\Gamma,(\lambda_\mathbf{N})_\mathbf{N})$ with one-dimensional index values given by $(\mathtt{1Dim}_A)_A$ to be \emph{valid} if it satisfies the principle of \emph{bi-univocality} (cf.\ \cref{def:biunivocality}) and every loop in its associated branch graph is \emph{weak} (cf.\ \cref{def:weak-loops}).
\end{definition}

In particular, a bi-univocal routed graph will be valid if its associated branch graph features no loops at all.

\subsection{Routed linear maps}
\label{sec:routed-maps}

With the routed graph, we have established a graph-theoretic structure to encode -- potentially cyclic, yet consistent -- connectivity between nodes.
To bridge the gap to quantum processes, we now shift our attention to \emph{routed linear maps},
which we will use to describe quantum operations to associate with these nodes.
Introduced in \cite{Vanrietvelde2021}, routed linear maps also use relations to restrict linear maps as well as the Hilbert spaces they act upon, splitting them into multiple independent sectors.
In this work, we will only consider Hilbert spaces $\HH$ of finite dimension $\dim(\HH)$, and refrain from mentioning this fact explicitly going forward.

Consider two Hilbert spaces $\HH^K$ and $\HH^L$, which are sectorised into orthogonal subspaces $\HH^K = \bigoplus_{k\in K} \HH^K_k$ and $\HH^L = \bigoplus_{l \in L} \HH^L_l$, with $K$ and $L$ denoting finite index sets.
Then, for linear maps $A : \HH^K \to \HH^L$ (which we will also write as $A \in \mathcal{L}(\HH^K, \HH^L$)), we may introduce \textit{sectorial constraints} restricting the form of $A$.
More precisely, we specify a \emph{route} $\lambda$, given as a relation $\lambda: K \to L$ (identified with a Boolean matrix ($\lambda_k^l)_{k \in K}^{l \in L}$).
We then say that $A$ \emph{follows the route} $\lambda$ if it satisfies
\begin{equation}
     \neg (k\overset{\lambda}{\sim}l) \ \implies \ \Pi^L_l A \: \Pi^K_k = 0
     \, , \qquad
    \text{or equivalently,}  \qquad
    A = \sum_{k,l} \lambda_k^l \cdot \Pi^L_l A \: \Pi^K_k \, ,
    \label{eq:routes}
\end{equation}
where $\Pi^K_k$ is the projector from $\HH^K$ to $\HH^K_k$ and $\Pi^L_l$ is the projector from $\HH^L$ to $\HH^L_l$.
We will then write $A: \HH^K \xrightarrow{\,\lambda\,} \HH^L$ and denote the tuple $(\lambda, A)$ a \emph{routed (linear) map}.
The condition above captures the idea that the image of a state in $\HH^K_k$ has support only on those sectors $\HH^L_l$ such that $k\overset{\lambda}{\sim}l$, i.e.\ on $\bigoplus_{l \in\lambda(k)} \HH^L_l$.

These routed maps can be composed analogously to the linear maps themselves: their composition will follow the composition of the respective routes.
Consider $\lambda: K \to L_1 \times L_2$, $\mu: L_2 \times L_3 \to M$ (as in \cref{subsubsec:compose_rel}), and $A: \HH^K \xrightarrow{\lambda} \HH^{L_1 L_2}$, $B: \HH^{L_2 L_3} \xrightarrow{\mu} \HH^M$, with $\HH^{L_1 L_2} \coloneqq \HH^{L_1} \otimes \HH^{L_2}$ and $\HH^{L_2 L_3} \coloneqq \HH^{L_2} \otimes \HH^{L_3}$ being (naturally) sectorised according to $L_1 \times L_2$ and $L_2 \times L_3$, respectively.
These then compose as\footnote{
    If $A$ and $B$ follow the routes $\lambda$ and $\mu$, resp., then $\lambda_k^{l_1,l_2}=0$ implies $(\Pi_{l_1}^{L_1}\otimes\Pi_{l_2}^{L_2})A\,\Pi_k^K=0$ and $\mu_{l_2,l_3}^m=0$ implies $\Pi_m^MB(\Pi_{l_2}^{L_2}\otimes\Pi_{l_3}^{L_3})=0$. Then using \cref{eq:compose_rel}, $(\lambda * \mu)_{k,l_3}^{l_1,m} = 0 \iff \boolSum_{l_2} \lambda_k^{l_1,l_2} \mu_{l_2,l_3}^m = 0 \iff \forall\,l_2, \lambda_k^{l_1,l_2}=0$ or $\mu_{l_2,l_3}^m = 0$, from which it follows that $(\Pi_{l_1}^{L_1}\otimes\Pi_m^M)(\id^{L_1}\otimes B)(A\otimes\id^{L_3})(\Pi_k^K\otimes\Pi_{l_3}^{L_3}) = \boolSum_{l_2} (\Pi_{l_1}^{L_1}\otimes\Pi_m^M B)(\id^{L_1}\otimes\Pi_{l_2}^{L_2}\otimes\id^{L_3})(A\,\Pi_k^K\otimes\Pi_{l_3}^{L_3}) = \boolSum_{l_2} \big(\id^{L_1}\otimes\underbrace{\Pi_m^M B(\Pi_{l_2}^{L_2}\otimes\Pi_{l_3}^{L_3})}_{0}\underset{\leftarrow\text{ or }\rightarrow}{\big)\big(}\underbrace{(\Pi_{l_1}^{L_1}\otimes\Pi_{l_2}^{L_2})A\,\Pi_k^K}_{0}\otimes\id^{L_3}\big) = 0$,
    meaning that $(\id^{L_1}\otimes B)(A\otimes\id^{L_3})$ follows the route $\lambda*\mu$.}
\begin{equation}
    (\lambda, A) \ast (\mu, B) = (\lambda \ast \mu, BA) \quad \in \enspace (K \times L_3 \to L_1 \times M) \times \mathcal{L} ( \HH^K \otimes \HH^{L_3}, \HH^{L_1} \otimes \HH^{M} ) \, .
    \label{eq:compose-routed-maps}
\end{equation}
Here, we employ the composition of relations as considered in \cref{subsubsec:compose_rel}, and for matrix multiplication, assume that the involved matrices are extended with identities on the respectively missing spaces.
In the same way that the composition of unitary quantum transformations can be represented using circuit diagrams~\cite{Coecke_Kissinger_2017}, 
Theorem~1 of \cite{Vanrietvelde2021} then proves that entire routed circuits may be constructed using elementary routed maps as building blocks. 

One important consequence of this construction is that $\lambda$ may restrict the \textit{practical} input and output spaces of $A$ in contrast to its \textit{formal} input and output spaces $\HH^K$ and $\HH^L$.
This reflects the notion of practical (co-)domain for the relation $\lambda$, as introduced in \cref{def:practical-co-domain}: $\HH^K_k$ is in the practical input space of $(\lambda,A)$ if and only if $k \in K_\text{in($\lambda$),prac}$, and analogously for the practical output space.
We then define the practical in- and output spaces for $A$ as 
\begin{align}
    \HH^{K}_{\text{in}(\lambda),\text{prac}} &\coloneqq 
    \bigoplus_{k \in K_\text{in($\lambda$),prac}} \HH^K_k \ =
    \bigoplus_{k \in K} \bigg(\boolSum_l \lambda_k^l \bigg) \, \HH^K_k \:\ \subseteq \: \HH^K \notag \\
    \HH^{L}_{\text{out}(\lambda),\text{prac}} &\coloneqq
    \bigoplus_{l \in L_\text{out($\lambda$),prac}} \HH^L_l \ =
    \bigoplus_{l \in L} \bigg(\boolSum_k \lambda_k^l \bigg) \, \HH^L_l \:\ \subseteq \: \HH^L .
    \label{eq:practical}
\end{align}
This captures the idea that $A$ must annihilate sectors $\HH^K_k$ outside of its practical input space (with $\boolSum_l \lambda_k^l = 0$): the route $\lambda$ does not relate these sectors to anything.
Similarly, unrelated sectors $\HH^L_l$ at the output must be outside the image of $A$ (as $\boolSum_k \lambda_k^l =0$).
Note that under sequential composition of routed maps, the practical input and output spaces may be further reduced to smaller subspaces of the respective spaces of the individual routed maps.
This is analogous to the behaviour of the practical input and output sets of relations under composition, as outlined in \cref{subsubsec:compose_rel}.

For a \textit{branched relation} $\lambda$ in particular, we specify \textit{branch spaces} $\HH^{K^\beta}_{\text{in}}$ and $\HH^{L^\beta}_{\text{out}}$ for a given branch $\beta=(K^\beta\to L^\beta) \in \Bran(\lambda)$.
We define 
\begin{equation}
    \HH^{K^\beta}_\text{in} \coloneqq 
    \bigoplus_{k \in K^\beta} \HH^K_k
    \qquad\text{and}\qquad
    \HH^{L^\beta}_\text{out} \coloneqq
    \bigoplus_{l \in L^\beta} \HH^L_l \, .
    \label{eq:branch-spaces}
\end{equation}
Any routed map $M: \HH^K \xrightarrow{\lambda} \HH^L$ will then (uniquely) decompose into a set of routed maps
$M = \bigoplus_{\beta \in \Bran(\lambda)} M^\beta$, with $M^\beta: \HH^{K^\beta}_{\text{in}} \to \HH^{L^\beta}_{\text{out}}$ for each branch $\beta$.

We will consider a routed map a \emph{routed unitary or isometry} if it acts as a unitary or isometry, respectively, when restricted to its practical input and output space. (Due to \cref{eq:routes}, they will be $0$ outside these spaces.\footnote{
    For a routed unitary $U: \HH^K \xrightarrow{\lambda} \HH^L$, this implies that $\dim(\HH^{K,\lambda}_\text{prac,in}) = \dim(\HH^{L,\lambda}_\text{prac,out})$, while for a routed isometry $V$ on the same spaces, $\dim(\HH^{K,\lambda}_\text{prac,in}) \leq \dim(\HH^{L,\lambda}_\text{prac,out})$ is sufficient.
    If the relation is branched, then additionally for a given branch $\beta=(K^\beta\to L^\beta)$, $\dim(\HH^{K^\beta}_\text{in}) = \dim(\HH^{L^\beta}_\text{out})$ holds for unitaries and $\dim(\HH^{K^\beta}_\text{in}) \le \dim(\HH^{L^\beta}_\text{out})$ for isometries.})

\subsection{Skeletal supermap and fleshing out for a routed graph} \label{subsec:skeletal_fleshingout}

Having introduced routed linear maps, we now shift our attention to \emph{higher-order (linear) operations}, which are commonly called \emph{supermaps}~\cite{Chiribella2008_Supermap}.
These (routed) supermaps take a tuple of (routed) linear maps as an argument and yield another (routed) linear map.
Accordingly, in circuit diagrams, they are usually depicted as boxes with \enquote{slots}, allowing one to \enquote{plug in} some linear operations inside.
Anticipating our intention to associate these slots with a node in a graph later in this section, for a given slot associated with a node $\mathbf{N}$, we will generally denote its input and output space as $\HH^{\text{in}(\mathbf{N})}$ and $\HH^{\text{out}(\mathbf{N})}$ respectively, rather than the potentially more familiar syntax $\HH^{A^{I/O}}$ used e.g.\ in~\cite{Wechs2021}.
For the unrouted case, supermaps are then defined as follows.

\begin{definition}[Supermap]
    \label{def:supermap}
    A \emph{supermap} of the type
    \begin{equation}
        \bigtimes_\mathbf{N} \left(\HH^\text{in($\mathbf{N}$)} \to \HH^\text{out($\mathbf{N}$)} \right)
        \enspace \longrightarrow \enspace
        (\HH^P \,\to\, \HH^F) \, ,
    \end{equation}
    is a linear map
    $\bigotimes_\mathbf{N} \mathcal{L} (\HH^\text{in($\mathbf{N}$)},\HH^\text{out($\mathbf{N}$}) \longrightarrow \mathcal{L} (\HH^P, \HH^F)$.\footnote{
        In the literature, supermaps are usually expressed as higher-order operations of the type $\bigtimes_\mathbf{N} \big(\mathcal{L}(\HH^\text{in($\mathbf{N}$)}) \to \mathcal{L}(\HH^\text{out($\mathbf{N}$)}) \big) \enspace \longrightarrow \enspace \big(\mathcal{L}(\HH^P) \to \mathcal{L}(\HH^F)\big)$, mapping channels to channels at the level of the ``mixed'' description of quantum states and operations.
        As in this work we focus on the actions of supermaps onto ``pure'' operations, we use the simpler pure representations instead -- as was also considered in~\cite{Vanrietvelde2023}.}
\end{definition}

For a given slot in the supermap, associated with a node $\mathbf{N}$, we then say it is of type
$\HH^{\mathrm{in}(\mathbf{N})} \to \HH^{\mathrm{out}(\mathbf{N})}$.
Generally, the linear operations we compose with the supermap may feature additional input and output systems, which will denote as $\HH^\text{in($\mathbf{N}$)}_\text{anc}$ and $\HH^\text{out($\mathbf{N}$)}_\text{anc}$.
In circuit diagrams, these systems will be \enquote{sticking out} as \enquote{dangling wires}, after an operation extended in this manner has been plugged into a slot.

Being interested in quantum theory, supermaps which preserve unitarity and isometricity are of particular interest.

\begin{definition}[Superunitary, superisometry]
    A \emph{superunitary} of the type
    \begin{equation}
        \bigtimes_\mathbf{N} \left(\HH^\text{in($\mathbf{N}$)} \to \HH^\text{out($\mathbf{N}$)} \right)
        \enspace \longrightarrow \enspace
        (\HH^P \,\to\, \HH^F)
    \end{equation}
    is a supermap of the same type, such that for any choice of ancillary input and output spaces $\HH^\text{in($\mathbf{N}$)}_\text{anc}$ and $\HH^\text{out($\mathbf{N}$)}_\text{anc}$ and any choice of unitary maps
    $U_\mathbf{N}: \HH^\text{in($\mathbf{N}$)} \otimes \HH^\text{in($\mathbf{N}$)}_\text{anc} \to \HH^\text{out($\mathbf{N}$)} \otimes \HH^\text{out($\mathbf{N}$)}_\text{anc}$
    for every $\mathbf{N}$, the linear map
    \begin{equation}
        M: \ \HH^P \otimes \Big(\bigotimes_\mathbf{N} \HH^\text{in($\mathbf{N}$)}_\text{anc} \Big) \ \longrightarrow \ \HH^F \otimes \Big(\bigotimes_\mathbf{N}  \HH^\text{out($\mathbf{N}$)}_\text{anc} \Big)
    \end{equation}
    resulting from plugging the $U_\mathbf{N}$ into the slots of the supermap is unitary.\footnote{
        In order to be able to plug in unitaries $U_\mathbf{N}$ at each node $\mathbf{N}$, the dimensions of input and output systems need to match:
        $\dim(\HH^\text{in($\mathbf{N}$)})\cdot \dim(\HH^\text{in($\mathbf{N}$)}_\text{anc}) = \dim(\HH^{\text{out}(\mathbf{N})})\cdot \dim(\HH^\text{out($\mathbf{N}$)}_\text{anc})$.
        In order for the resulting map to be a unitary, one further needs to have
        $\dim(\HH^P) \cdot \prod_\mathbf{N} \dim(\HH^\text{in($\mathbf{N}$)}_\text{anc}) = \dim(\HH^F) \cdot \prod_\mathbf{N} \dim(\HH^\text{out($\mathbf{N}$)}_\text{anc})$.
        Altogether, this requires $\dim(\HH^P) \cdot \prod_\mathbf{N} \dim(\HH^\text{out($\mathbf{N}$)}) = \dim(\HH^F) \cdot \prod_\mathbf{N} \dim(\HH^\text{in($\mathbf{N}$)})$. \label{ftn:dims_U_N}}
    Analogously, if for any such choice of ancillary spaces and any corresponding choice of isometries $V_\mathbf{N}$ for every $\mathbf{N}$, the resulting linear map $M$ is an isometry, the supermap is a \emph{superisometry}.\footnote{
        Here, the ancillary systems need to be chosen such that
        $\dim(\HH^\text{in($\mathbf{N}$)}) \cdot \dim(\HH^\text{in($\mathbf{N}$)}_\text{anc}) 
        \leq \dim(\HH^\text{out($\mathbf{N}$)}) \cdot \dim(\HH^\text{out($\mathbf{N}$)}_\text{anc})$
        and $\dim(\HH^P) \cdot \prod_\mathbf{N} \dim(\HH^\text{in($\mathbf{N}$)}_\text{anc}) \leq
        \dim(\HH^F) \cdot \prod_\mathbf{N} \dim(\HH^\text{out($\mathbf{N}$)}_\text{anc})$.}
\end{definition}

Clearly, if we plug other superunitaries rather than unitaries into a superunitary, the resulting supermap will again be a superunitary, with an analogous statement holding for superisometries.

We continue by formally generalising the notions of supermaps, superunitaries and superisometries to their respective \textit{routed} analogues.

\begin{definition}[Routed supermap]
    \label{def:routed-supermap}
    Let $\lambda_\mathbf{N}$ and $\mu$ be relations
    with $\lambda_\mathbf{N}: \Ind_\mathbf{N}^\text{in} \to \Ind_\mathbf{N}^\text{out}$, and all Hilbert spaces $\HH$ be sectorised accordingly.
    Then, a routed supermap of type
    \begin{equation}
        \bigtimes_\mathbf{N} \left(\HH^\text{in($\mathbf{N}$)} \xrightarrow{\lambda_\mathbf{N}} \HH^\text{out($\mathbf{N}$)} \right)
        \enspace \longrightarrow \enspace
        (\HH^P \, \xrightarrow{\mu} \, \HH^F) \, ,
    \end{equation}
    is a supermap on the respective spaces, which is restricted to input maps in
    $\bigotimes_\mathbf{N} \mathcal{L} (\HH^\text{in($\mathbf{N}$)},\HH^\text{out($\mathbf{N}$)})$, where each tensor factor follows the associated route $\lambda_\mathbf{N}$,
    and whose output map follows the route $\mu$.
    To denote the practical input and output spaces for a node $\mathbf{N}$, we write $\HH^{\text{in/out}(\mathbf{N})}_\text{prac}$, while for individual branches $\mathbf{N}^\beta$ (if applicable), we write $\HH^{\text{in/out}(\mathbf{N})}_\beta$.

    We call a routed supermap a \emph{routed superunitary / superisometry} whenever it acts as a superunitary / superisometry when restricted to any routed unitaries $(U_\mathbf{N}, \lambda_\mathbf{N})$ / routed isometries $(V_\mathbf{N}, \lambda_\mathbf{N})$ following the route $\lambda_\mathbf{N}$ for each node $\mathbf{N}$, potentially involving unsectorised ancillary systems $\HH^\text{in($\mathbf{N}$)}_\text{anc}$ and $\HH^\text{out($\mathbf{N}$)}_\text{anc}$.%
    \footnote{
        While for unrouted superunitaries, we can adjust the dimensions of $\HH^\text{in($\mathbf{N}$)}_\text{anc}$ and $\HH^\text{out($\mathbf{N}$)}_\text{anc}$ for the input and output dimension of each unitary $U_\mathbf{N}$ to match whenever $\dim(\HH^P) \cdot \prod_\mathbf{N} \dim(\HH^\text{out($\mathbf{N}$)}) = \dim(\HH^F) \cdot \prod_\mathbf{N} \dim(\HH^\text{in($\mathbf{N}$)})$ (see \cref{ftn:dims_U_N}), this is generally not the case for a routed superunitary.
        Considering for instance a branched node $\mathbf{N}$ with sectorised practical input/output spaces $\HH^{\text{in/out}(\mathbf{N})}_\text{prac} = \bigoplus_{\mathbf{N}^\beta \in \Bran(\lambda_\mathbf{N})} \HH^{\text{in/out}(\mathbf{N})}_{\mathbf{N}^\beta}$, to which one attaches (unsectorised) ancillary systems, the extended input/output spaces are
        $\HH^{\text{in/out}(\mathbf{N})}_\text{prac} \otimes \HH^{\text{in/out}(\mathbf{N})}_\text{anc} = \bigoplus_{\mathbf{N}^\beta} (\HH_{\mathbf{N}^\beta}^{\text{in/out}(\mathbf{N})} \otimes \HH^{\text{in/out}(\mathbf{N})}_\text{anc})$.
        In order to be able to plug in unitaries that follow the branched relations, one then needs
        $\dim(\HH^{\text{in}(\mathbf{N})}_{\mathbf{N}^\beta} \otimes \HH^{\text{in}(\mathbf{N})}_\text{anc}) = \dim(\HH^{\text{out}(\mathbf{N})}_{\mathbf{N}^\beta} \otimes \HH^{\text{out}(\mathbf{N})}_\text{anc})$
        for each $\mathbf{N}^\beta$, which requires the ratio
        $\dim(\HH^{\text{in}(\mathbf{N})}_\mathbf{N^\beta}) / \dim(\HH^{\text{out}(\mathbf{N})}_\mathbf{N^\beta})$
        to be the same for all $\mathbf{N}^\beta$.
        Any routed supermap which does not satisfy this condition for all $\mathbf{N}$ will therefore \textit{trivially} be a routed superunitary, as it is unable to act on routed unitaries \textit{at all} (not even using ancillas).
        For routed super\emph{isometries}, no similar challenges arise.
    }
\end{definition}

For a given slot in the routed supermap, associated with a node $\mathbf{N}$, we then say it is of type
$\HH^{\mathrm{in}(\mathbf{N})} \xrightarrow{\lambda_\mathbf{N}} \HH^{\mathrm{out}(\mathbf{N})}$.

We continue by associating the routed graph we have constructed in \cref{sec:routed-graph} with a particular routed supermap, which we will call the \emph{skeletal supermap}.
Here, the arrows will be translated into wires in a circuit diagram as shown in \cref{fig:skeletal-supermap}, representing sectorised Hilbert spaces, while the nodes will be understood as slots for operations following their associated routes.%
\footnote{
    While supermaps are often conveniently described in the Choi picture (cf.\ \cref{sec:choi}) we will refrain from doing so for the skeletal supermap, as this would require introducing two copies of each Hilbert space, attached to each end of their supporting arrows (so as to describe the identity channels by the Choi vectors of the form $\dket{\id}^{XX'}$). We will use the Choi picture later on, when explicitly ``fleshing out'' the skeletal supermap of our generic routed graph for QC-QCs in \cref{sec:fleshing-out}.}
Thereby, a single routed graph can lead to a whole class of processes by specifying different (dimensional) Hilbert spaces and choosing different routed maps or supermaps for the fleshing out, which we will formally introduce in \cref{def:fleshing-out} below.

\begin{figure}
    \centering
    \includegraphics[scale=0.77]{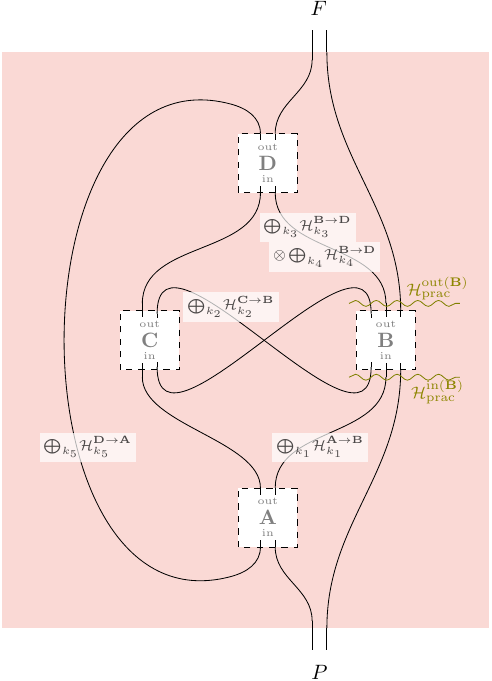}
    \caption{
        Depiction of the skeletal supermap (cf.\ \cref{def:skeletal_supermap}) associated to the routed graph shown in \cref{fig:routed_graph}.
        Each wire linking nodes $\mathbf{N}$ to $\mathbf{M}$, corresponding to an arrow $A$ from $\mathbf{N} \to \mathbf{M}$ in the routed graph, is associated with a sectorised Hilbert space $\HH^{A} = \bigoplus_{k_A \in \Ind_A} \HH^{A}_{k_A}$. 
        The dimension of these sectors is specified via a dimension assignment $\mathtt{dim}_\Gamma$.
        Due to correlations between the indices mediated through the index constraints, the physically relevant input/output space of each node is only given by a practical subspace (cf.\ \cref{sec:routed-maps}) of the tensor product of the incoming/outgoing wires.
        E.g., for the node $\mathbf{B}$, we have $\HH^{\text{in/out}(\mathbf{B})}_\text{prac} \subset \HH^{\text{in/out}(\mathbf{B})}$.\\
        To reduce clutter, we do not make explicit Hilbert spaces with a trivial sectorisation on the wires in the figure, e.g.\ $\HH^{\mathbf{B} \to \mathbf{C}}$, and omit an explicit representation of the routes in this figure.
    }
    \label{fig:skeletal-supermap}
\end{figure}
\begin{definition}[Skeletal supermap associated to a routed graph] \label{def:skeletal_supermap}
    Given a routed graph
    $\cG = (\Gamma,(\lambda_\mathbf{N})_\mathbf{N})$
    and a dimension assignment
    $\mathtt{dim}_\Gamma: \bigsqcup_{A \in \mathtt{Arr}_A} \Ind_A \to \mathbb{N}^*$, their associated \emph{skeletal supermap} is the routed supermap obtained by interpreting
    \begin{itemize}
        \item each arrow $A$ as a sectorised Hilbert space $\HH^A = \bigoplus_{k_A} \HH_{k_A}^{A}$, whose sectors are labelled by the different
        index values $k_A \in \Ind_A$ attached to this arrow and satisfy $\dim(\HH_{k_A}^A) = \mathtt{dim}_\Gamma (k_A) \ \forall\, k_A \in \Ind_A$; and
        \item each node $\mathbf{N}$ as an open slot for a routed linear map $M$, going from the tensor product of the Hilbert spaces associated to its incoming arrows in the set $\text{in} (\mathbf{N})$, to that of the Hilbert spaces associated to its outgoing arrows in the set $\text{out} (\mathbf{N})$, and following the route $\lambda_\mathbf{N}: \Ind^\text{in}_\mathbf{N} \to \Ind^\text{out}_\mathbf{N}$ associated to that node
        \begin{equation}
            \label{eq:skeletal-map}
            M: \bigotimes_{A \in \text{in}(\mathbf{N})} \bigoplus_{k_{A} \in \Ind_A} \HH_{k_A}^{A} 
            \enspace \xrightarrow{\lambda_\mathbf{N}}
            \bigotimes_{A \in \text{out}(\mathbf{N})} \bigoplus_{k_{A} \in \Ind_A} \HH_{k_A}^{A} \, .
        \end{equation}
    \end{itemize}
    The skeletal supermap acts on the respective routed linear maps $M$ as a routed supermap, composing them     as specified by the routed graph $\Gamma$, interpreting the arrows as identity channels. The open-ended arrows that ``come from nowhere'' and that ``go nowhere'' will (altogether) be identified with the ``global past'' and ``global future'' spaces $\HH^P$ and $\HH^F$ that appear the definition of a routed supermap, respectively: i.e., $\HH^P\coloneqq\bigotimes_{A\in\mathtt{Arr}^\text{in}_\Gamma}\HH^A$ and
    $\HH^F\coloneqq\bigotimes_{A\in\mathtt{Arr}^\text{out}_\Gamma}\HH^A$.
	
    Furthermore, when specifying the set of one-dimensional index values $(\mathtt{1Dim}_A)_A$ of the routed graph, one imposes that any Hilbert space sector associated to an index value in $(\mathtt{1Dim}_A)_A$ must be 1-dimensional -- i.e.\ $\forall\,A, \forall\,k_A\in\mathtt{1Dim}_A,$ $\dim(\HH_{k_A}^A) = \mathtt{dim}_\Gamma (k_A)=1$.%
    \footnote{While we will usually not consider this case, further sectors could, in general, be of dimension 1 as well.}
    If $(\mathtt{1Dim}_A)_A$ is not specified explicitly, then $(\mathtt{1Dim}_A)_A$ is implictly taken to contain precisely the indices $k$ for which $\mathtt{dim}_\Gamma (k)=1$.
\end{definition}

If a skeletal supermap is associated with a bi-univocal routed graph (\cref{def:biunivocality}), the open arrows $A$ of the routed graph must be trivially indexed.
Nonetheless, we may also consider a generalised notion of bi-univocality (see \cref{fn:bi-univocal-index}), which does feature such indices, and correspondingly, yields a non-trivially routed map $\HH^P\xrightarrow{\mu}\HH^F$.

Note that when stating the practical input and output space $\HH^{\text{in/out}(\mathbf{N})}_\text{prac}$ for each node $\mathbf{N}$, we can exclude index combinations which will never be realised due to global constraints, as outlined in \cref{rem:global_constraint}.
With respect to the associated global route, these are not part of the practical input and output spaces of the skeletal supermap slots.
Therefore, some tuples $k_\text{in/out} \in \Ind^{\text{in/out}}_\mathbf{N}$ are not practically realised due to $(\lambda_\mathbf{N})_\mathbf{N}$ restricting index combinations, and we can effectively omit the respective spaces.

With this in place, we can specify compatible routed linear (super)maps $M$ to fill the respective slots of the skeletal supermap, thereby \emph{fleshing out} the skeletal supermap.

\begin{definition}[Fleshing out]
    \label{def:fleshing-out}
    Given a skeletal supermap associated with a routed graph, its \textit{fleshing out} is specified by a collection of routed maps 
    or supermaps (with any number of slots) $\{ M_\mathbf{N} \}_{\mathbf{N} \in \Nodes_\Gamma}$,
    matching the type
    $\HH^\text{in($\mathbf{N}$)} \xrightarrow{\lambda_\mathbf{N}} \HH^\text{out($\mathbf{N}$)}$ for each $\mathbf{N}$, while allowing for ancillary systems.\footnote{
        That is, the routed maps are of type 
        $\HH^\text{in($\mathbf{N}$)}\otimes\HH^\text{in($\mathbf{N}$)}_\text{anc} \xrightarrow{\lambda_\mathbf{N}} \HH^\text{out($\mathbf{N}$)}\otimes\HH^\text{out($\mathbf{N}$)}_\text{anc}$ for each $\mathbf{N}$,
        while the routed supermaps have this type for their resulting output map, but may have (any number of) arbitrarily typed slots.
    }
\end{definition}

If a slot is fleshed out with a routed supermap, it may feature an arbitrary number of (sectorised or unsectorised) slots.

\subsection{Main theorem of Vanrietvelde \emph{et al.}~\texorpdfstring{\cite{Vanrietvelde2023}}{}}
\label{subsec:main_thm}

To conclude this review of routed quantum circuits, we state the main theorems of Ref.~\cite{Vanrietvelde2023}, which certify that when the underlying routed graph is valid (according to \cref{def:valid-routed-graph}), then the construction outlined above actually yields valid quantum processes.

\begin{theorem}[Main theorem of \cite{Vanrietvelde2023}]
    \label{def:valid-routed-superunitary}
    Let $(\Gamma,(\lambda_\mathbf{N})_\mathbf{N})$ be a valid routed graph, with a dimension assignment $\mathtt{dim}_\Gamma$.
    Then their associated skeletal supermap is a routed superunitary.
\end{theorem}

\begin{corollary}
    \label{thm:superunitary}
    Let $(\Gamma,(\lambda_\mathbf{N})_\mathbf{N})$ be a valid routed graph with a dimension assignment $\mathtt{dim}_\Gamma$.
    Then, any supermap built from their associated skeletal supermap by fleshing it out with routed unitaries at some of its nodes and routed superunitaries\footnote{
        The restriction to monopartite superunitaries made in \cite{Vanrietvelde2023} is clearly not necessary for the corollary to hold.}
    at the other nodes is a routed superunitary.
    In particular, if the routed superunitaries used in the fleshing out only have unsectorised slots, then the resulting supermap is a conventional (non-routed) superunitary.
\end{corollary}

Actually, this corollary can be extended to routed isometries and superisometries, even if these were not explicitly studied in \cite{Vanrietvelde2021}.
Note, however, that we restrict ourselves to fleshing out with \emph{monopartite} superisometries for this statement.

\begin{corollary}
    \label{thm:superisometry}
     Let $(\Gamma,(\lambda_\mathbf{N})_\mathbf{N})$ be a valid routed graph with a dimension assignment $\mathtt{dim}_\Gamma$. Then, any supermap built from their associated skeletal supermap by fleshing it out with routed isometries at some of its nodes and monopartite routed superisometries at other nodes is a routed superisometry.
     In particular, if the routed superisometries used in the fleshing out only
     have unsectorised slots, then the resulting supermap is a conventional (non-routed) superisometry.
\end{corollary}
\begin{proof}
    Any isometry $V: \HH^A \to \HH^B$ can be extended to a unitary by introducing additional ancillary input and output systems  $\HH^\text{in}_\text{anc}$ and $\HH^\text{out}_\text{anc}$, such that $\dim(\HH^\text{in}_\text{anc}) = \dim(\HH^B)$ and $\dim(\HH^\text{out}_\text{anc}) = \dim(\HH^A)$.
    We obtain a unitary $U: \HH^A \otimes \HH^\text{in}_\text{anc} \to \HH^B \otimes \HH^\text{out}_\text{anc}$, from which we can recover $V$ in two steps.
    By plugging in a suitable ancillary input state into $\HH^\text{in}_\text{anc}$,
    we obtain a map $V': \HH^A \to \HH^B \otimes \HH^\text{out}_\text{anc}$ that recovers the action of the isometry onto $\HH^B$.
    As this restricts the image of $V'$, when projected on $\HH^\text{out}_\text{anc}$, to be one-dimensional, we then recover $V$ from $V'$ by simply tracing out $\HH^\text{out}_\text{anc}$.

    Analogously, we can extend any monopartite superisometry, which may map unitaries to isometries, to a superunitary, by decomposing it into isometries first.\footnote{
        Each monopartite superisometry is a (monopartite) quantum comb \cite{Chiribella2008}, which can be decomposed into two isometries $V_1$ and $V_2$, with a memory system $\alpha$ and an open slot in-between. Extending both isometries to unitaries, we obtain a superunitary with two additional input and output wires $\HH^\text{in/out}_{1,\text{anc}}$ and $\HH^\text{in/out}_{2,\text{anc}}$, respectively.
        While this reasoning generalises to quantum combs and QC-QCs with more parties, it may not be applicable to general superisometries, for which no decomposition into isometries is established.}
    In a routed graph, these ancillary input / output systems correspond to additional unindexed open arrows (i.e.\ wires \enquote{coming from nowhere}) going into / coming out of the individual nodes $\mathbf{N}$ corresponding to the skeletal supermap slots fleshed out by the respective maps $U_\mathbf{N}$.
    As introducing these arrows impacts neither the presence of loops nor the index assignments, the (valid) routed graph still gives rise to a superunitary.
    
    We may then again compose the additional input wires with suitable input states, while tracing out the additional output wires, to obtain a routed superisometry (via the action of the respective isometries and monopartite superisometries).
\end{proof}

Inspired by this corollary, we specify when we consider a routed quantum circuit to be a satisfactory representation of the connectivity of a supermap.

\begin{definition}[Routed circuit decomposition]
    \label{def:decomposition}
    A \emph{routed circuit decomposition of a routed superisometry} is given by a routed graph $\cG = (\Gamma, (\lambda_\mathbf{N})_\mathbf{N})$, a dimension assignment $\mathtt{dim}_\Gamma$ --~which, together, define a skeletal supermap~-- and a fleshing out, given by a collection of routed isometries for some nodes $\mathbf{V}$ and monopartite routed superisometries for the remaining nodes $\mathbf{A}$, that do not feature any ancillary systems $\HH^{\mathrm{in/out}(\mathbf{A})}_\text{anc}$ and $\HH^{\mathrm{in/out}(\mathbf{V})}_\text{anc}$, respectively.
    
    Two routed circuit decompositions are considered \emph{equivalent} if they give rise to the same routed superisometry.
    
    If none of the remaining slots after the fleshing out features non-trivial routes, we simply call the decomposition a \emph{routed circuit decomposition of a superisometry}.
\end{definition}

By disallowing the use of ancillary systems within a routed circuit decomposition, we ensure that indeed all relevant (routed) information processing relevant to the process is encoded into the skeletal supermap. 
As an alternative, additional unindexed arrows may be added to some nodes in the routed graph in their stead, recovering the auxiliary systems from these arrows in a modified skeletal supermap.

Routed circuit decompositions have been given for several supermaps, including causally indefinite ones~\cite{Vanrietvelde2023, Mejdoub2025}. Indeed, thus far all the examples of purifiable supermaps~\cite{Araujo2017} that have been explicitly studied have been able to be given such decompositions, but it remains an open question whether this is possible for all such supermaps.
Note further that in a routed circuit decomposition of a supermap, only some of the slots in the skeletal supermap will be associated with slots in the decomposed supermap, while others will be fleshed out with \enquote{internal} operations with no open slots.
That is, the routed graph used in a routed circuit decomposition may, in general, have more nodes than there are slots in the resulting supermap.

Regarding the following section, \cref{thm:superisometry} allows for an independent proof for the validity of QC-QCs based on the routed quantum circuit framework.
However, this will not be our main point of attention in this work.
Rather, we will demonstrate a systematic construction that makes QC-QCs fit into the framework of~\cite{Vanrietvelde2023} in terms of a suitable routed circuit decomposition, which will then allow us to see how their properties translate into diagrammatic features.

\section{Quantum Circuits with Quantum Control of Causal Order}
\label{sec:QCQC}

\begin{figure}
	\centering
	\includegraphics[width=1\textwidth]{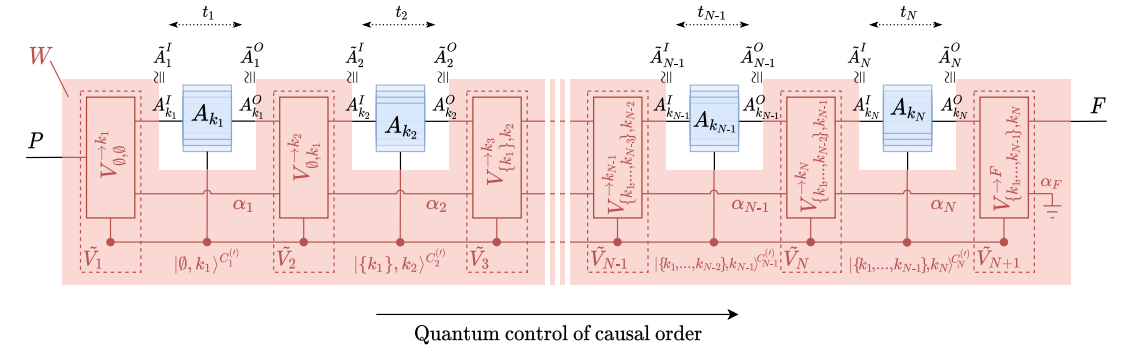}
	\caption{Generic depiction of a QC-QC. Figure reproduced from Ref.~\cite{Wechs2021}.}
	\label{fig:QCQC}
\end{figure}

\textit{Quantum circuits with quantum control of causal order} (QC-QCs) \cite{Wechs2021} are a class of supermaps with concrete physical descriptions that includes many relevant causally indefinite processes.
In an $N$-slot QC-QC, as illustrated in \cref{fig:QCQC}, the action of $N$ agents $A_k$, happening precisely once, alternate with ``internal'' transformations $\tilde{V}_n$, which control which agent $A_k$ acts at which time $t_1, \ldots, t_N$, dependent both on the agents who have acted before and on the operations they have performed.
Reflecting this, in the remainder of this work, $k$ (and $\ell$) will consistently be used to index (time-delocalised) agents or their operations, while $n$ indexes will refer to well-defined time slots.
Various results indicate that under certain conditions, the class of QC-QCs is precisely the set of processes realisable within a fixed spacetime \cite{Salzger2022, Salzger2024, VVR, Vilasini24}. It is worth noting also that QC-QCs do not allow for violations of causal inequalities~\cite{Wechs2021}.

For the formal description of QC-QCs, we write $\cN = \{ 1, \ldots, N \}$, consistently take
$k, \ell \in \cN, \K_n \subseteq \cN$ and implicitly take $\K_n$ such that $|\K_n|=n$, with $0\le n \le N$.
Furthermore, we identify a singleton $\{k\}$ with its unique element, and write for instance $\K_n\backslash k$ rather than $\K_n\backslash\{k\}$ or $\K_{n-1}\cup k$ rather than $\K_{n-1}\cup\{k\}$.

The QC-QC pattern is then realised using a quantum control system $C_n$ with Hilbert space $\HH^{C_n}$, spanned by states of the form $\ket{\K_{n} \setminus k,k}^{C_n}$.
This system coherently controls, for each time slot $t_n$, which agent $A_k$'s operation is applied, as well as the ``internal operations'' $V^{\to \ell}_{\K_{n} \setminus k,k}$ applied in-between two agent's operations.
Here, each internal $V^{\to \ell}_{\K_{n} \setminus k,k}$ sends its input from agent $A_{k}$ with $k\in\K_n$ to agent $A_{\ell}$ with $\ell \not\in \K_n$, for a given set of prior agents $\{A_j\}_{j \in \K_{n} \setminus k}$.
Jointly, they constitute an isometry $\tilde{V}_{n+1}$, specifying both the transfer
$k \to \ell$ between the agents and potential further transformations applied in-between.
In addition to the input/output systems $A_{k}^{I/O}$, which are passed to the respective agents $A_{k}$ to act upon, the isometries may also act on an ancillary system $\alpha_{n}$, which allows one to pass arbitrary information between the internal isometries without making it available to the agents.

The control system $C_n$ serves to keep track of the set of operations $\K_n$ (including $k$, for $A_k$ being applied at time $t_n$) already used, ensuring that each operation is used exactly once.
Altogether, the QC-QC construction yields a superisometry, which is why we will call the resulting circuits \enquote{pure} QC-QCs.
These have an open past $P$ and an open future $F$ system.
However, by setting the dimension of $P$ to $1$\footnote{This can be done by absorbing incoming quantum states into the isometry $\tilde V_1$.}, we may also close the past. Similarly, by partially tracing out the future $F$, we may obtain mixed supermaps with partially or fully closed future, which are no longer superisometries, but still purifiable superchannels.

\subsection{Choi picture and the link product}
\label{sec:choi}

We continue by introducing the (pure) \emph{Choi picture} for linear maps \cite{Choi1975} (on finite-dimensional spaces), which we will rely on heavily to lay out the formal QC-QC construction.
It allows one to conveniently express linear maps in the form of vectors. 
To do so, for a given Hilbert space $\HH^X$, we fix an orthonormal basis $\{ \ket{i}^X \}_i$ as the computational basis, and for any isomorphic space $\HH^{X'}$, take its computational basis states $\ket{i}^{X'}$ to be in 1-to-1 correspondence. Then, we define
\begin{equation}
    \dket{\id}^{XX'} = \sum_i \ket{i}^X \otimes \ket{i}^{X'} \in \HH^X \otimes \HH^{X'} \, .
\end{equation}
Going forward, as a general rule,
we will denote a $\ket{\psi}\in\HH^X$ as $\ket{\psi}^X$, using just the label $X$ to denote explicitly the state space. We may omit this label if it is obvious from context.
With this, we may represent a (routed) linear map $A: \HH^K \to \HH^L$ by its \emph{Choi vector},
\begin{equation}
    \dket{A} \coloneqq (\id^K \otimes A) \dket{\id}^{KK'} = \sum_i \ket{i}^K \otimes A \ket{i}^{K'} \ \in \HH^{KL} .
\end{equation}
Notice that by definition of the Choi vector, $A = \dket{A}^{\top_{K}}$ (where $^{\top_{K}}$ denotes the partial transpose on the space $\HH^K$).

Given two Choi vectors $\dket{V_1} \in \HH^{XY}$ and $\dket{V_2} \in \HH^{YZ}$, we may introduce their (pure) \emph{link product}, denoted by $\ast$, by using the partial trace over their shared systems $\{ Y \}$:
\begin{equation}
    \label{eq:link-product}
    \dket{V_1} \ast \dket{V_2} = \Tr_Y \left[ \dket{V_1}^\top \otimes \dket{V_2} \right] \ \in \HH^{XZ}.
\end{equation}
This is analogous to the composition of relations as introduced in \cref{subsubsec:compose_rel}, and provides the Choi vector for the composition of the two linear maps $V_1,V_2$~\cite{Chiribella2008,Chiribella2009,Wechs2021}.

\subsection{Formal construction of QC-QCs}
\label{sec:QCQC-formal}

We continue by outlining the formal construction of the notion outlined in the beginning of the section. For more specific details we direct the reader to Ref.~\cite{Wechs2021}.

Without loss of generality, we assume that for all agents $A_k$, their associated Hilbert spaces $\HH^{A_k^I}$ and $\HH^{A_k^O}$ are isomorphic to one another.
Therefore, they are taken to be of the same dimension $d_{A^I}$ and $d_{A^O}$.\footnote{These do not necessarily have to be equal.}
If not given by default, this can be achieved by extending the dimensions of all operations appropriately.%
\footnote{Specifically, the respective operations can be extended by adding identities as tensor factors whose additional inputs and outputs are then traced out in the process.}

To translate from the spaces of the individual agent operations to the space of all operations applied for a given time slot $t_n$, we introduce a further set of spaces $\HH^{\tilde{A}_n^{I/O}}$, which are isomorphic to $\HH^{A_k^{I/O}}$.
Going forward, we will generally use a \textasciitilde{} to indicate that an operation is expressed with respect to $\HH^{\tilde{A}_n^{I/O}}$ rather than $\HH^{A_k^{I/O}}$.
The internal isometries $V$ then act upon these spaces, as well as ancillary spaces $\HH^{\alpha_n}$, passing information amongst themselves rather than making it accessible to the agents.
Accordingly, for each internal control operator $V^{\to \ell}_{\K_{n} \setminus k,k} : \HH^{A^O_{k} \alpha_n} \to 
\HH^{A^I_{\ell} \alpha_{n+1}}$
applied between two agents $A_k$ and $A_\ell$ in $n^\text{th}$ and $(n+1)^\text{th}$ positions (with $k\in\K_n, \ell\notin\K_n$), we obtain its tilded description
$\tilde V^{\to \ell}_{\K_{n} \setminus k,k} : \HH^{\tilde A^O_{n} \alpha_n} \to 
\HH^{\tilde A^I_{n} \alpha_{n+1}}$ as follows:
\begin{equation}
    \tilde V^{\to \ell}_{\K_{n} \setminus k,k} =
        \id^{A_{k}^O \to \tilde{A}_n^O} \: 
        V^{\to \ell}_{\K_{n} \setminus k,k} \: 
        \,\id^{\tilde{A}_n^I \to A_{\ell}^I} \, ,
\end{equation}
where $\id^{X \to Y}$ is to be understood as the identity from a space $\HH^X$ to an isomorphic space $\HH^Y$.
Similarly, we may use these spaces to describe the joint input and output spaces of the superposed operations $A_k : \HH^{A_k^I} \to \HH^{A_k^O}$ being in $n^\text{th}$ position respectively.
Using the control system, we define a a joint operation $\tilde A_n : \HH^{\tilde A_n^I C_n} \to \HH^{\tilde A_n^O C_n'}$.

This gives the following formal description of the operations in the tilded spaces in terms of Choi vectors:
\begin{align}
    \dket{\tilde A_n} &= \!\!\!
        \sum_{\substack{\K_{n}, k \in \K_{n}}} \!\! (\dket{\id^{A_{k}^I \tilde{A}_n^I}} \otimes \dket{\id^{A_{k}^O \tilde{A}_n^O}}) \ast \dket{A_{k}}^{A_{k}^I A_{k}^O} \!\otimes \ket{\K_{n} \setminus k,k}^{C_{n}} \!\otimes \ket{\K_{n} \setminus k,k}^{C_{n}'}
        \ \in \, \HH^{\tilde{A}_n^I C_n \tilde{A}_n^O C_n'} \label{eq:An} , \\
    \dket{\tilde V_1} &=
        \sum_{\ell \in \cN} \dket{\tilde V^{\to \ell}_{\emptyset,\emptyset}}^{P \tilde{A}_1^{I} \alpha_1} \otimes \ket{\emptyset,\ell}^{C_{1}} 
        \!\!\quad \in \ \HH^{P \tilde{A}_1^{I} \alpha_1 C_1} \notag , \\
	\dket{\tilde V_{n+1}} &= \!\!\!
	    \sum_{\K_n, k \in \K_n, \ell \notin \K_n}  \!\!\!\!\! \!\!\!\!\!\! \dket{\tilde V^{\to \ell}_{\K_{n} \setminus k,k}}^{\tilde A_n^O \alpha_n \tilde A_{n+1}^I \alpha_{n+1}} \!\otimes \ket{\K_n\setminus k, k}^{C_n'} \!\otimes \ket{\K_n,\ell}^{C_{n+1}} \label{eq:V-choi}
        \ \in \, \HH^{\tilde{A}_n^O \alpha_n C_n^{\prime} \tilde{A}_{n+1}^I \alpha_{n+1} C_{n+1}}
        \notag , \\
    \dket{\tilde V_{N+1}} &=
        \sum_{k \in \cN} \dket{\tilde V^{\to F}_{\cN \setminus k,k}}^{\tilde{A}_N^O \alpha_N F} \otimes \ket{\cN \setminus k ,k}^{C_{N}'}
        \!\!\quad \in \ \HH^{\tilde{A}_N^O \alpha_N C_{N}' F } .
\end{align}

We then compose all of these operations to a circuit $\dket{V_{P \to F}}$ which decomposes into the (pure) QC-QC $\ket{w_\text{QC-QC}}$ and the $N$ agent operations $\dket{A_k}$.
We compose these operations as follows,
relabelling $\K_{n-1} \to \{k_1, \ldots, k_{n-1} \}, k \to k_n, \ell \to k_{n+1}$ to be able to combine the respective sums:
\begin{align}
    \dket{V_{P \to F}} & = \dket{\tilde V_1} * \dket{\tilde A_1} * \dket{\tilde V_2} * \cdots * \dket{\tilde V_N} * \dket{\tilde A_N} * \dket{\tilde V_{N+1}} \notag \\
    & = \!\!\!\sum_{(k_1,\ldots,k_N)} \!\!\! \dket{\tilde V_{\emptyset,\emptyset}^{\to k_1}} * \dket{\tilde A_{k_1}} * \dket{\tilde V_{\emptyset,k_1}^{\to k_2}} * \cdots * \dket{\tilde V_{\{k_1,\ldots,k_{N-2}\},k_{N-1}}^{\to k_N}} * \dket{\tilde A_{k_N}} * \dket{\tilde V_{\{k_1,\ldots,k_{N-1}\},k_N}^{\to F}} \notag \\[1mm]
    & = \!\!\!\sum_{(k_1,\ldots,k_N)} \!\!\! \big( \dket{A_1} \otimes \cdots \otimes \dket{A_N} \big) * \ket{w_{(k_1,\ldots,k_N,F)}} \notag \\[1mm]
    & = \big( \dket{A_1} \otimes \cdots \otimes \dket{A_N} \big) * \ket{w_\text{QC-QC}} \ \in \HH^{PF} \label{eq:Choi_V_QCQC} \, ,
\end{align}
where
\begin{equation}
    \ket{w_{(k_1,\ldots,k_N,F)}} \coloneqq \dket{V_{\emptyset,\emptyset}^{\to k_1}} * \dket{V_{\emptyset,k_1}^{\to k_2}} * \dket{V_{\{k_1\},k_2}^{\to k_3}} * \cdots * \dket{V_{\{k_1,\ldots,k_{N-2}\},k_{N-1}}^{\to k_N}} * \dket{V_{\{k_1,\ldots,k_{N-1}\},k_N}^{\to F}} \, ,
\end{equation}
and obtaining
\begin{equation}
    \ket{w_\text{QC-QC}}
    \coloneqq \sum_{(k_1,\ldots,k_N)} \ket{w_{(k_1,\ldots,k_N,F)}} \label{eq:W_QCQC}
\end{equation}
as the process vector description for a pure QC-QC.
Here, the sums are over all permutations $(k_1,\dots,k_N)$ of $\{1,\dots,N\}$ (and, more generally, when we write $(k_1,\dots,k_n)$ throughout the paper we implicitly assume that all $k_i$'s are distinct).

Finally, let us mention the generalisation from pure to mixed QC-QCs.
In the QC-QC framework, mixed processes are captured by tracing out a final ancillary system $\alpha_F$, which, in the description above, would be a subsystem of $F$.
Therefore, only this single modification is required to obtain general QC-QCs from the pure case, leaving the construction of the process vector $\ket{w_\text{QC-QC}}$ unchanged otherwise.
We obtain the process matrix representation of mixed circuits directly by splitting $\HH^F = \HH^{F'} \otimes \HH^{\alpha_F}$:
\begin{equation}
    W_\text{QC-QC} = \Tr_{\alpha_F} \ketbra{w_\text{QC-QC}}{w_\text{QC-QC}} \, .
\end{equation}

\section{QC-QCs as Routed Quantum Circuits}
\label{sec:QCQCs_as_RQCs}

In this section, we will present our main result. For each number of agents $N$, we will provide a single \enquote{generic} routed graph, whose associated skeletal supermap can be fleshed out to obtain any $N$-slot QC-QC, detailing their construction explicitly and proving the equivalence to the original QC-QC construction.
All in all, this will amount to proving the following theorem:

\begin{restatable}{theorem}{main}
    \label{thm:main}
    Any $N$-slot quantum circuit with quantum control of causal order can be implemented by fleshing out a skeletal supermap obtained from the generic routed graph $\cG_{\text{QC-QC}(N)}$ introduced in \cref{def:generic-graph}.
\end{restatable}

In this construction, many structural features of QC-QCs will be recovered:
both the internal isometries $\tilde{V}_{n+1}$ and the agent operations $A_k$ will be featured as nodes $\mathbf{V}_{n+1}$ and $\mathbf{A}_k$ in the routed graph, and the control system $C^{(\prime)}_n$ will emanate both the indices $\mathcal{X}^k_n$ featured in the routed graph and the branch structure these indices induce.
To make these correspondences as clear as possible, we choose our notation accordingly.

\subsection{Generic routed graph for QC-QCs}
\label{sec:routed-graph-qcqc}

We begin by describing in detail how the \enquote{generic routed graph} for a QC-QC will be defined.
To ease understanding of the construction, we will provide an intuitive interpretations of various elements in italic sentences.

\subsubsection{Definition of the routed graph}
\label{def:graph}

\begin{definition}[{$\mathcal{G}_{\text{QC-QC}(N)}$}]
\label{def:generic-graph}
The \emph{generic routed graph} $\mathcal{G}_{\text{QC-QC}(N)}$ for an $N$-slot QC-QC is defined as follows:

\begin{itemize}

    \item There are 2 types of nodes: ``\textbf{A}-nodes'' $\textbf{A}_k$, for $k=1,\ldots,N$, and ``\textbf{V}-nodes'' $\textbf{V}_{n+1}$, for $n=0,\ldots,N$.

    \emph{(As we will see, the external operations $A_k$, on which a QC-QC acts, will be plugged into the \textup{\textbf{A}}-nodes, while the internal operations $\tilde{V}_{n+1}$ of a QC-QC will be plugged into the \textup{\textbf{V}}-nodes.)}
    
    \item The indexed graph involves the following indexed arrows:
    \begin{itemize}
        \item An open-ended arrow \mbox{$\xrightarrow{\ \emptyset\ }\textbf{V}_1$} that ``comes from nowhere'' and points to $\textbf{V}_1$, with the single index value $\emptyset$ attached to it (which we will henceforth drop).
        
        \item $\forall\,n=1,\ldots,N$, $\forall\,k=1,\ldots,N$,
        \begin{align}
            \textbf{V}_n\xrightarrow{\ \XX_n^k\ }\textbf{A}_k \quad \text{and} \quad \textbf{A}_k\xrightarrow{\ {\XX_n^k}^{\,\prime}\ }\textbf{V}_{n+1}.
        \end{align}
        The indices ${\XX_n^k}^{\,(\prime)}$ take values among the $n$-element subsets of $\N$ that contain $k$, or the ``null value'' $\ravnothing$: either ${\XX_n^k}^{\,(\prime)} = \K_n$ for some $\K_n\ni k$ (still with the convention that $|\K_n|=n$), or ${\XX_n^k}^{\,(\prime)} = \ravnothing$.

        \item An open-ended arrow \mbox{$\textbf{V}_{N+1}\xrightarrow{\N\cup F}$} that leaves from $\textbf{V}_{N+1}$ and ``goes nowhere'', with the single index value $\N\cup F$ attached to it (which we will henceforth drop).
        
    \end{itemize}

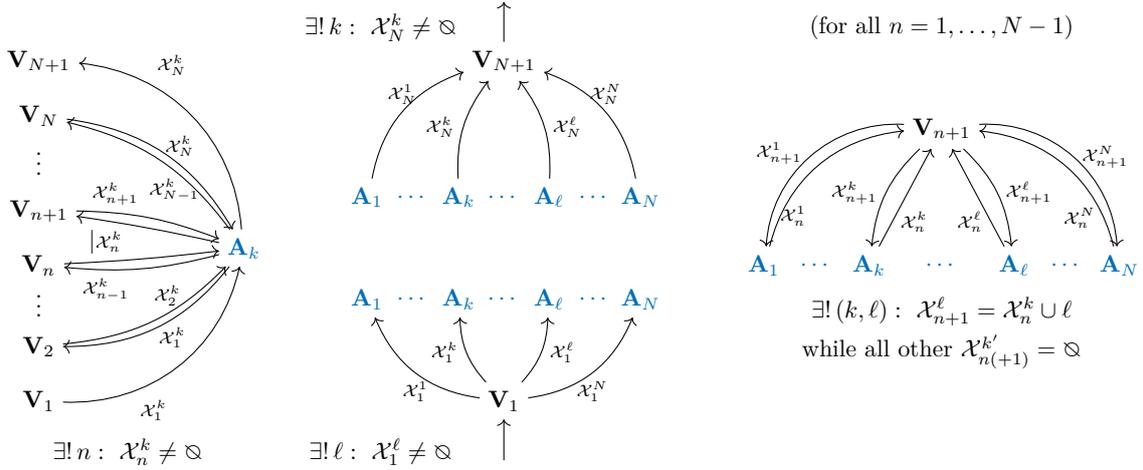
\begin{figure}[t]
\centering
\begin{tikzpicture}[scale=0.9, transform shape]
    \node (bottom) at (0,-1) {};
    \node (V0) at (0,0) {$\textbf{V}_1$};
    \node[text=NavyBlue] (A1) at (-2,1.5) {$\textbf{A}_1$};
    \node[text=NavyBlue] at (-1.33,1.5) {$\cdots$};
    \node[text=NavyBlue] (B1) at (-0.66,1.5) {$\textbf{A}_k$};
    \node[text=NavyBlue] at (0,1.5) {$\cdots$};
    \node[text=NavyBlue] (C1) at (0.66,1.5) {$\textbf{A}_\ell$};
    \node[text=NavyBlue] at (1.33,1.5) {$\cdots$};
    \node[text=NavyBlue] (D1) at (2,1.5) {$\textbf{A}_N$};
    \node (label0) at (-1.8,-0.75) {$\exists! \, \ell: \ \XX^\ell_1 \neq \ravnothing $};

    \node (top) at (0,6) {};
    \node (VN) at (0,5) {$\textbf{V}_{N+1}$};
    \node[text=NavyBlue] (A2) at (-2,3) {$\textbf{A}_1$};
    \node[text=NavyBlue] at (-1.33,3) {$\cdots$};
    \node[text=NavyBlue] (B2) at (-0.66,3) {$\textbf{A}_k$};
    \node[text=NavyBlue] at (0,3) {$\cdots$};
    \node[text=NavyBlue] (C2) at (0.66,3) {$\textbf{A}_\ell$};
    \node[text=NavyBlue] at (1.33,3) {$\cdots$};
    \node[text=NavyBlue] (D2) at (2,3) {$\textbf{A}_N$};
    \node (labelN) at (-1.8,5.5) {$\exists! \, k: \ \XX^k_N \neq \ravnothing $};

    \node at (6.4,5.5) [align=center] {
        (for all $n = 1, \ldots, N-1$)};
    \node (Vn) at (6.4,4) {$\textbf{V}_{n+1}$};
    \node[text=NavyBlue] (A3) at (3.8,2) {$\textbf{A}_1$};
    \node[text=NavyBlue] at (6-1.33-0.1,2) {$\cdots$};
    \node[text=NavyBlue] (B3) at (6-0.66,2) {$\textbf{A}_k$};
    \node[text=NavyBlue] at (6.4,2) {$\cdots$};
    \node[text=NavyBlue] (C3) at (7.46,2) {$\textbf{A}_\ell$};
    \node[text=NavyBlue] at (8.13+0.1,2) {$\cdots$};
    \node[text=NavyBlue] (D3) at (9,2) {$\textbf{A}_N$};
    \node at (6.4,1) [align=center] {
        $\exists! \, (k, \ell): \ \XX^\ell_{n+1} = \XX^k_n \cup \ell $\\[1mm]
        while all other $\XX^{k'}_{n(+1)} = \ravnothing$};

    \node[text=NavyBlue] (Ak) at (-3.8,2.25) {$\textbf{A}_k$};
    \node (VkN) at (-6.8,5) {$\textbf{V}_{N+1}$};
    \node (VkN_1) at (-6.8,4.2) {$\textbf{V}_{N}$};
    \node at (-6.8,3.6) {$\vdots$};
    \node (Vkn_1) at (-6.8,2.8) {$\textbf{V}_{n+1}$};
    \node (Vkn) at (-6.8,2) {$\textbf{V}_{n}$};
    \node at (-6.8,1.5) {$\vdots$};
    \node (Vk2) at (-6.8,0.8) {$\textbf{V}_2$};
    \node (Vk1) at (-6.8,0) {$\textbf{V}_1$};
    \node (labelk) at (-5.5,-0.75) {$\exists! \, n : \ \XX_n^k \neq \ravnothing$};
    
    \begin{scope}[
        every node/.style={fill=white,circle,inner sep=0pt}
        every edge/.style=routedarrow]
        \path [->] (bottom) edge (V0);
        \path [->] (V0) edge[bend left] node[below] {$\scriptstyle{\XX^1_1}$} (A1);
        \path [->] (V0) edge[bend left=20] node[left] {$\scriptstyle{\XX^k_1}$} (B1);
        \path [->] (V0) edge[bend right=20] node[right] {$\scriptstyle{\XX^\ell_1}$} (C1);
        \path [->] (V0) edge[bend right] node[below] {$\scriptstyle{\XX^N_1}$} (D1);

        \path [<-] (top) edge (VN);
        \path [<-] (VN) edge[bend right] node[above] {$\scriptstyle{\XX^1_N}$} (A2);
        \path [<-] (VN) edge[bend right=20] node[left] {$\scriptstyle{\XX^k_N}$} (B2);
        \path [<-] (VN) edge[bend left=20] node[right] {$\scriptstyle{\XX^\ell_N}$} (C2);
        \path [<-] (VN) edge[bend left] node[above] {$\scriptstyle{\XX^N_N}$} (D2);

        \path [<-] (Vn) edge[bend right=35] node[above left,pos=0.55] {$\scriptstyle{\XX^1_{n+1}}$} (A3);
        \path [->] (Vn) edge[bend right=45] node[right,pos=0.85] {$\scriptstyle{\XX^1_{n}}$} (A3);
        \path [->] (Vn) edge[bend right=20] node[left] {$\scriptstyle{\XX^k_{n+1}}$} (B3);
        \path [->] (Vn) edge[bend left=20] node[right] {$\scriptstyle{\XX^\ell_{n+1}}$} (C3);
        \path [<-] (Vn) edge node[right,pos=0.75] {$\scriptstyle{\XX^k_{n}}$} (B3);
        \path [<-] (Vn) edge node[left,pos=0.75] {$\scriptstyle{\XX^\ell_{n}}$} (C3);
        \path [<-] (Vn) edge[bend left=35] node[left,pos=0.85] {$\scriptstyle{\XX^N_{n}}$} (D3);
        \path [->] (Vn) edge[bend left=45] node[right] {$\scriptstyle{\XX^N_{n+1}}$}(D3);

        \path [->] (Ak) edge [bend right=40] node [above right,pos=0.72] {$\scriptstyle{\XX^k_N}$} (VkN);
        \path [<-] (Ak) edge [bend right=26] node [above,pos=0.4] {$\scriptstyle{\XX^k_{N}}$} (VkN_1);
        \path [->] (Ak) edge [bend right=20] node [below,pos=0.4] {$\scriptstyle{\XX^k_{N-1}}$} (VkN_1);
        \path [<-] (Ak) edge [bend right=10] node [above,pos=0.75] {$\scriptstyle{\XX^k_{n+1}}$} (Vkn_1);
        \path [->] (Ak) edge [bend right=2] node [below,pos=0.4] {} (Vkn_1);
        \path [<-] (Ak) edge [bend left=2] node [above,pos=0.75] {$\mid \scriptstyle{\XX^k_n}$} (Vkn);
        \path [->] (Ak) edge [bend left=10] node [below,pos=0.75] {$\scriptstyle{\XX^k_{n-1}}$} (Vkn);
        \path [<-] (Ak) edge [bend left=20] node [above,pos=0.4] {$\scriptstyle{\XX^k_{2}}$} (Vk2);
        \path [->] (Ak) edge [bend left=26] node [below,pos=0.4] {$\scriptstyle{\XX^k_{1}}$} (Vk2);
        \path [<-] (Ak) edge [bend left=40] node [below right,pos=0.72] {$\scriptstyle{\XX^k_1}$} (Vk1);
    \end{scope}
\end{tikzpicture}

    \caption{Neighbourhood of each type of node in the generic routed graph $\mathcal{G}_{\text{QC-QC}(N)}$ of a QC-QC: $\textbf{A}_k$, $\textbf{V}_1$, $\textbf{V}_{n+1}$ (for $n=1,\ldots,N-1$), and $\textbf{V}_{N+1}$.}
    \label{fig:neighbourhoods_2}
\end{figure}

\begin{figure}[t]
    \centering
    \begin{subfigure}[b]{0.4\textwidth}
		\centering
		\!\!\!\!\begin{tikzpicture}[transform shape]
			\node (bottom) at (0,0.2) {};
			\node (V0) at (0,1) {$\textbf{V}_1$};
			\node[text=NavyBlue] (A) at (-2,2) {$\textbf{A}$};
			\node (V1) at (0,2) {$\textbf{V}_2$};
			\node[text=NavyBlue] (B) at (2,2) {$\textbf{B}$};
			\node (V2) at (0,3) {$\textbf{V}_3$};
			\node (top) at (0,3.8) {};
			\node (i1) at (0,-1.2) [align=center] {
                $\exists! \, (k_1, k_2) \, :$\\[1mm]
                $\XX_1^{k_1} = \{ k_1 \}, \;\ \XX_2^{k_2} = \{k_1, k_2 \} \mathrel{(=} \{A, B \}),$\\[1mm]
                while all other $\XX^k_n = \ravnothing$};
			\node (i2) at (0,-2.5) {};
			
			\begin{scope}[
				every node/.style={fill=white,circle,inner sep=0pt}
				every edge/.style=routedarrow]
				\path [->] (bottom) edge (V0);
				\path [->] (V0) edge[bend left] node[below] {$\scriptstyle{\XX^A_1}$} (A);
				\path [->] (V0) edge[bend right] node[below] {$\scriptstyle{\XX^B_1}$} (B);
				\path [->] (A) edge[bend right=15] node[below] {$\scriptstyle{\XX^A_1}$} (V1);
				\path [->] (B) edge[bend left=15] node[below] {$\scriptstyle{\XX^B_1}$} (V1);
				\path [->] (V1) edge[bend left=15] node[above] {$\scriptstyle{\XX^B_2}$} (B);
				\path [->] (V1) edge[bend right=15] node[above] {$\scriptstyle{\XX^A_2}$} (A);
				\path [->] (B) edge[bend right] node[above] {$\scriptstyle{\XX^B_2}$} (V2);
				\path [->] (A) edge[bend left] node[above] {$\scriptstyle{\XX^A_2}$} (V2);
			    \path [->] (V2) edge (top);
			\end{scope}
		\end{tikzpicture}
		\caption{for $N=2$}
        \label{fig:routed_graphs_QCQC_N_2}
	\end{subfigure}%
	\hfill
    \begin{subfigure}[b]{0.6\textwidth}
        \centering
    	\begin{tikzpicture}[scale=0.85, transform shape]
    		\node (bottom) at (0,0) {};
    		\node (V0) at (0,0.75) {$\textbf{V}_1$};
    		\node[text=NavyBlue] (A) at (-3,3.8) {$\textbf{A}$};
    		\node (V1) at (0,2) {$\textbf{V}_2$};
    		\node[text=NavyBlue] (C) at (5,3.8) {$\textbf{C}$};
    		\node[text=NavyBlue] (B) at (1.5,3.8) {$\textbf{B}$};
    		\node (V2) at (0,5.6) {$\textbf{V}_3$};
    		\node (V3) at (0,6.85) {$\textbf{V}_4$};
    		\node (top) at (0,7.6) {};

            \node (i) at (0,-.8) [align=center] {
                $\exists! \, (k_1, k_2, k_3) \, :$
                $\XX_1^{k_1} = \{ k_1 \}, \;\ \XX_2^{k_2} = \{k_1, k_2 \},$\\[1mm]
                $ \XX_3^{k_2} = \{k_1, k_2, k_3 \} \mathrel{(=} \{A, B, C \})$,
                while all other $\XX^k_n = \ravnothing$};
    		
    		\begin{scope}[
    			every node/.style={fill=white,circle,inner sep=0pt}
    			every edge/.style={routedarrow}]
    			\path [->] (bottom) edge (V0);
    
    			\path [->] (V0) edge[bend left] node[below left] {$\scriptstyle{\XX^A_1}$} (A);
    			\path [->] (V0) edge[bend right] node[below right,pos=0.4] {$\scriptstyle{\XX^B_1}$} (B);
    			\path [->] (V0) edge[bend right] node[below right] {$\scriptstyle{\XX^C_1}$} (C);
    
    			\begin{scope}[bend angle=15]
                    \path [->] (A) edge[bend right] node[below,font=\scriptsize] {$\XX^A_1$} (V1);
    				\path [->] (B) edge node[left,pos=0.3,font=\scriptsize] {$\XX^B_1$} (V1);
    				\path [->] (C) edge[bend left] node[below,font=\scriptsize] {$\XX^C_1$} (V1);
    
    				\path [->] (V1) edge[bend left=40] node[above left,font=\scriptsize] {$\XX^B_2$} (B);
    				\path [->] (V1) edge node[above,font=\scriptsize] {$\XX^A_2$} (A);
    				\path [->] (V1) edge node[above,font=\scriptsize] {$\XX^C_2$} (C);

                    \path [->] (B) edge[bend left=40] node[below left,font=\scriptsize] {$\XX^B_2$} (V2);
    				\path [->] (A) edge node[below right,pos=0.4,font=\scriptsize] {$\XX^A_2$} (V2);
    				\path [->] (C) edge node[below left,pos=0.4,font=\scriptsize] {$\XX^C_2$} (V2);
    
    				\path [->] (V2) edge[bend right] node[above,font=\scriptsize] {$\XX^A_3$} (A);
    				\path [->] (V2) edge node[left,pos=0.7,font=\scriptsize] {$\XX^B_3$} (B);
    				\path [->] (V2) edge[bend left] node[above,font=\scriptsize] {$\XX^C_3$} (C);
    			\end{scope}
    
    			\path [->] (A) edge[bend left] node[above left] {$\scriptstyle{\XX^A_3}$} (V3);
    			\path [->] (B) edge[bend right] node[above right,pos=0.7] {$\scriptstyle{\XX^B_3}$} (V3);
    			\path [->] (C) edge[bend right] node[above right] {$\scriptstyle{\XX^C_3}$} (V3);
    
    			\path [->] (V3) edge (top);
    		\end{scope}
    	\end{tikzpicture}
        \caption{for $N=3$}
        \label{fig:routed_graphs_QCQC_N_3}
    \end{subfigure}
    \caption{The generic routed graph $\cG_{\text{QC-QC}(N)}$ for $N=2$ (left) and $N=3$ (right). Here, $\textbf{A}_1, \textbf{A}_2, \textbf{A}_3$ are denoted as $\textbf{A}$, $\textbf{B}$, $\textbf{C}$.}
    \label{fig:routed_graphs_QCQC_N_2_3}
\end{figure}

    Thus, each node $\textbf{A}_k$ has $N$ incoming arrows from the $N$ nodes $\textbf{V}_1$ to $\textbf{V}_N$, and $N$ outgoing arrows to the $N$ nodes $\textbf{V}_2$ to $\textbf{V}_{N+1}$; its input and output index values are of the form $(\XX_n^k)_{n=1}^N \in \Ind^\text{in}_{\mathbf{A}_k}$ and $({\XX_n^k}^{\,\prime})_{n=1}^N \in \Ind^\text{out}_{\mathbf{A}_k}$, resp.
    Each node $\textbf{V}_{n+1}$ on the other hand has $N$ incoming and $N$ outgoing arrows, from and to the $N$ nodes $\textbf{A}_1$ to $\textbf{A}_N$; its input and output index values are of the forms $(\XX_n^k)_{k=1}^N \in \Ind^\text{in}_{\mathbf{V}_{n+1}}$ and $({\XX_{n+1}^k})_{k=1}^N \in \Ind^\text{out}_{\mathbf{V}_{n+1}}$, respectively. 
    The neighbourhood of each type of node is represented in \cref{fig:neighbourhoods_2,fig:routed_graphs_QCQC_N_2_3}.

    \emph{(The intuitive interpretation for the values of the indices ${\XX_n^k}^{(\prime)}$ is that ${\XX_n^k}^{\,(\prime)} = \K_n$ indicates that operation $A_k$ is applied in $n^\text{th}$ position, with all other operations in $\K_n\backslash k$ being applied before; while ${\XX_n^k}^{\,(\prime)} = \ravnothing$ indicates that $A_k$ does not come in $n^\text{th}$ position.)}

    \item The routes are specified as follows (and as also stated on \cref{fig:neighbourhoods_2}):
    \begin{itemize}
        \item For each node $\textbf{A}_k$, the lists of all input and output index values $(\XX_n^k)_{n=1}^N \in \Ind^\text{in}_{\mathbf{A}_k}$ and $({\XX_n^k}^{\,\prime})_{n=1}^N \in \Ind^\text{out}_{\mathbf{A}_k}$ must be equal -- from now on we will then simply drop the primes -- and must further satisfy
        \begin{align}
            \exists!\,n, \XX_n^k\neq\ravnothing. \label{eq:A_routes_new}
        \end{align}
        Hence, the practical input and output sets of each node $\textbf{A}_k$ consist of (the same) lists of index values $(\ravnothing,\ldots,\ravnothing,\K_n,\ravnothing,\ldots,\ravnothing)$, for some set $\K_n\ni k$ at the $n^\text{th}$ position,%
        \footnote{Elements in the lists $(\XX_n^k)_{n=1}^N \in \Ind^\text{in/out}_{\mathbf{A}_k}$ are ordered from $n=1$ to $n=N$. Further below, elements in the lists $(\XX_n^k)_{k=1}^N \in \Ind^\text{in}_{\mathbf{V}_{n+1}}$ and $({\XX_{n+1}^k})_{k=1}^N \in \Ind^\text{out}_{\mathbf{V}_{n+1}}$ will be ordered from $k=1$ to $k=N$.}
        attached to the arrow coming from node $\textbf{V}_n$ and to the arrow going to node $\textbf{V}_{n+1}$.
        
        \emph{(The fact that in each admissible combination of input and output index values to the $\textup{\textbf{A}}$-nodes there must be one and only one non-null-value index translates the fact that each operation $A_k$ must be applied once and only once, at a single position $n$.)}

    \item For each node $\textbf{V}_{n+1}$, the lists of all input and output index values $(\XX_n^k)_{k=1}^N \in \Ind^\text{in}_{\mathbf{V}_{n+1}}$ and $({\XX_{n+1}^\ell})_{\ell=1}^N \in \Ind^\text{out}_{\mathbf{V}_{n+1}}$ must satisfy
        \begin{align}
            \text{for } n = 0 \ \text{ and } \ n=N: \quad & \exists!\,\ell, \XX_1^\ell\neq\ravnothing \ \text{ and } \ \exists!\,k, \XX_N^k\neq\ravnothing, \text{ resp.} \label{eq:V_routes_1_new} \\[2mm]
            \text{for } 1 \leq n \leq N-1: \quad & \exists!\,k, \XX_n^k\neq\ravnothing, \ \exists!\,\ell, \XX_{n+1}^\ell\neq\ravnothing; \ \text{ and } \XX_{n+1}^\ell=\XX_n^k\cup \ell. \label{eq:V_routes_2_new}
        \end{align}
        Hence, the practical input sets of each node $\textbf{V}_{n+1}$, for $0 < n \leq N$, consist of lists of index values of the form $(\ravnothing,\ldots,\ravnothing,\K_n,\ravnothing,\ldots,\ravnothing)$, for some set $\K_n$ at the $k^\text{th}$ position, such that $k\in\K_n$. Similarly, the practical output sets of each node $\textbf{V}_{n+1}$, for $0 \leq n < N$, consist of lists of index values of the form $(\ravnothing,\ldots,\ravnothing,\K_{n+1},\ravnothing,\ldots,\ravnothing)$, for some set $\K_{n+1}$ at the $\ell^\text{th}$ position, such that $\ell\in\K_{n+1}$. Such input and output index values are related (for $1 \leq n \leq N-1$) if and only if $\K_{n+1}=\K_n\cup \ell$.

        \emph{(These constraints translate the fact that along each route, between any two operations $\tilde{V}_n$ and $\tilde{V}_{n+1}$, one and only one operation $A_k$ is applied.)}
        
    \end{itemize}
    
    Combining the different constraints at the different nodes, we get the \emph{global index constraint} (as expressed for the cases $N=2,3$ on \cref{fig:routed_graphs_QCQC_N_2_3}) that%
    \footnote{It can indeed be seen that \cref{eq:V_routes_1_new,eq:V_routes_2_new} for all $n$ imply \cref{eq:global_cstr_QCQC}; and vice versa, that \cref{eq:global_cstr_QCQC} implies \cref{eq:A_routes_new,eq:V_routes_1_new,eq:V_routes_2_new}.}
    \begin{align}
        & \exists!\,(k_1,k_2,\ldots,k_N), \notag \\
        & \XX_1^{k_1} = \{k_1\},\ \XX_2^{k_2} = \{k_1,k_2\},\ \ldots,\ \XX_N^{k_N} = \{k_1,k_2,\ldots,k_N\}=\N, \text{ and all other } \XX_n^k=\ravnothing. \label{eq:global_cstr_QCQC}
    \end{align}
    \emph{(Here the understanding is that along each global path, all operations are applied, in a given order.)}
        
\end{itemize}
\end{definition}

\subsubsection{Identifying the branches}

\paragraph*{$\mathbf{A}$-nodes.}
The lists of index values $(\ravnothing,\ldots,\ravnothing,\K_n,\ravnothing,\ldots,\ravnothing)$ in the practical input and output sets of each node $\textbf{A}_{k}$ are in one-to-one relation.
The relation $\lambda_{\mathbf{A}_k}$ is therefore (rather trivially) branched; to help relate its branches to the QC-QC construction later on, we label them by $\K_n\setminus k$ rather than by $\K_n$ directly, i.e.\ we denote the branches as
\begin{align}
    \textbf{A}_k^{{\K_n\backslash k}} \coloneqq \Big( \big\{(\ravnothing,\ldots,\ravnothing,\K_n,\ravnothing,\ldots,\ravnothing)\big\} \to \big\{(\ravnothing,\ldots,\ravnothing,\K_n,\ravnothing,\ldots,\ravnothing)\big\} \Big). \label{eq:branches_Ak}
\end{align}
Each node $\textbf{A}_{k}$ has $2^{N-1}$ such branches, one for each subset $\K_n\subseteq\N$ that includes $k$.

\noindent\emph{(Each branch $\textup{\textbf{A}}_k^{\K_n\backslash k}$ relates to the set of parties $\K_n\setminus k$ that came before $A_k$.)}

\paragraph*{$\mathbf{V}$-nodes.}
According to the previous description, the route of each node $\textbf{V}_{n+1}$ (for $1 \leq n \leq N-1$) relates index values of the form $(\ravnothing,\ldots,\ravnothing,\K_n,\ravnothing,\ldots,\ravnothing)$ in the practical input set to index values of the form $(\ravnothing,\ldots,\ravnothing,\K_{n+1},\ravnothing,\ldots,\ravnothing)$ in the practical output set, with $\K_{n+1}=\K_n\cup \ell$ for some $\ell\notin\K_n$.
One can verify that $\lambda_{\textbf{V}_{n+1}}$ is also branched, with one branch corresponding to each subset $\K_n$, namely:%
\footnote{To see that $\lambda_{\textbf{V}_{n+1}}$ is branched, recall that its practical input set consists of lists of the form $(\ravnothing,\ldots,\ravnothing,\underset{\uparrow k}{\K_n},\ravnothing,\ldots,\ravnothing)$, for some set $\K_n$ at some position $k\in\K_n$ (as we denote by ${\uparrow}k$ below it). From \cref{eq:V_routes_2_new} it can then be seen that $\lambda_{\textbf{V}_{n+1}}\big((\ravnothing,\ldots,\ravnothing,\underset{\uparrow k}{\K_n},\ravnothing,\ldots,\ravnothing)\big) = \big\{(\ravnothing,\ldots,\ravnothing,\underset{\uparrow \ell}{\K_n\cup\ell},\ravnothing,\ldots,\ravnothing)\big|\ell\in\N\backslash\K_n\big\}$, from which one can see that $\lambda_{\textbf{V}_{n+1}}\big((\ravnothing,\ldots,\ravnothing,\underset{\uparrow k}{\K_n},\ravnothing,\ldots,\ravnothing)\big)$ and $\lambda_{\textbf{V}_{n+1}}\big((\ravnothing,\ldots,\ravnothing,\underset{\uparrow k'}{\K_n'},\ravnothing,\ldots,\ravnothing)\big)$ are either equal if $\K_n=\K_n'$ (while $k,k'\in\K_n^{(\prime)}$ may differ), or disjoint if $\K_n\neq\K_n'$.
Similarly, we find $\lambda_{\textbf{V}_{n+1}}^\top\big((\ravnothing,\ldots,\ravnothing,\underset{\uparrow \ell}{\K_n\cup\ell},\ravnothing,\ldots,\ravnothing)\big) = \big\{(\ravnothing,\ldots,\ravnothing,\underset{\uparrow k}{\K_n},\ravnothing,\ldots,\ravnothing)\big|k\in\K_n\big\}$.
From this one can indeed identify the branches of $\textbf{V}_{n+1}^{\K_n}$ as in \cref{eq:branches_Vn1}.}
\begin{align}
    \textbf{V}_{n+1}^{\K_n} \coloneqq & \Big( \big\{ (\ravnothing,\ldots,\ravnothing,\K_n,\ravnothing,\ldots,\ravnothing) \,\big|\, \K_n \text{ at the $k^\text{th}$ position}, \forall\,k\in\K_n \big\} \notag \\
    & \quad \to \big\{ (\ravnothing,\ldots,\ravnothing,\K_n\cup\ell,\ravnothing,\ldots,\ravnothing) \, \big| \,\K_n\cup\ell \text{ at the $\ell^\text{th}$ position}, \forall\,\ell\in\N\backslash\K_n \big\} \Big). \label{eq:branches_Vn1}
\end{align}

For the nodes $\textbf{V}_1$ and $\textbf{V}_{N+1}$, we find that these each have a single branch
\begin{align}
    \textbf{V}_1^\emptyset \coloneqq & \Big( \big\{ \emptyset \big\} \to \big\{ (\ravnothing,\ldots,\ravnothing,\ell,\ravnothing,\ldots,\ravnothing) \,\big| \,\ell \text{ at the $\ell^\text{th}$ position}, \forall\,\ell\in\N \big\} \Big) \notag \\
    \text{and} \quad \textbf{V}_{N+1}^{\N} \coloneqq & \Big( \big\{ (\ravnothing,\ldots,\ravnothing,\N,\ravnothing,\ldots,\ravnothing) \,\big| \,\N \text{ at the $k^\text{th}$ position}, \forall\,k\in\N \big\} \to \big\{ \N\cup F \big\} \Big), \label{eq:branches_V1_VN1}
\end{align}
respectively.

\begin{figure}[t]
    \centering
    \includegraphics{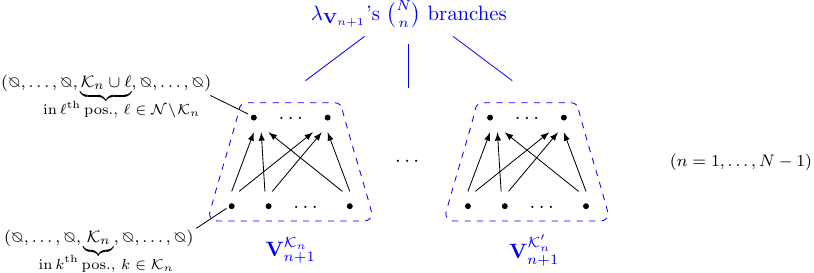}
    
    \vspace{1cm}
    
    \includegraphics{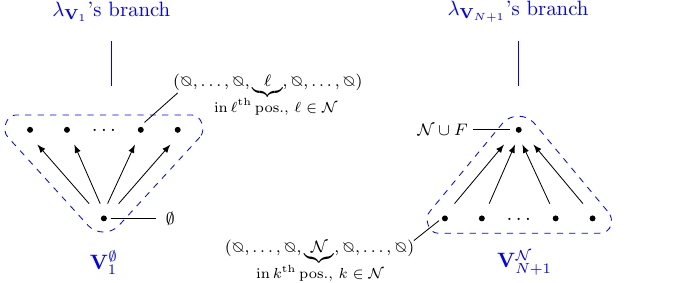}
    \caption{Branches of the $\mathbf{V}$-nodes.
    For each node $\mathbf{V}_{n+1}$, for $n=0, \ldots, N$,
    $\lambda_{\mathbf{V}_{n+1}}: \Ind^\text{in}_{\textbf{V}_{n+1}} \to \Ind^\text{out}_{\mathbf{V}_{n+1}}$
    features $\binom{N}{n}$ branches, as well as values outside its practical input and output sets (not shown).
    Each branch $V^{\K_n}_{n+1}$ is associated with $n$ input values (for $n\ge 1$), each taking the form $(\XX^k_n)_{k=1}^{N}$, where the position $k$ of the single non-null index value indicates the agent $A_k$ last visited.
    Analogously, it is associated with $N-n$ output values of the same form (for $n\le N-1$), with the position $\ell$ of the single non-null index value indicating the agent $A_\ell$ to be visited next.}
    \label{fig:V_branches}
\end{figure}

The branches of the $\textbf{V}$-nodes are illustrated on \cref{fig:V_branches}. Each node $\textbf{V}_{n+1}$ has $\binom{N}{n}$ branches $\textbf{V}_{n+1}^{\K_n}$, one for each $n$-element subset $\K_n\subseteq \N$. 

\noindent\emph{(Each branch $\textup{\textbf{V}}_{n+1}^{\K_n}$ relates to the set of parties $\K_n$ that came before $\tilde{V}_{n+1}$.)}

\subsubsection{Augmented relations}
\label{subsubsec:augmented_rel_QCQC}

Now that we have identified the branches of the relations $\lambda_{\textbf{A}_k}$ and $\lambda_{\textbf{V}_{n+1}}$, we continue by specifying their augmented versions $\lambda_{\textbf{A}_k}^\text{aug}$ and $\lambda_{\textbf{V}_{n+1}}^\text{aug}$ according to \cref{def:augmented_rel}.

\begin{figure}[t]
    \centering
    \includegraphics[scale=1]{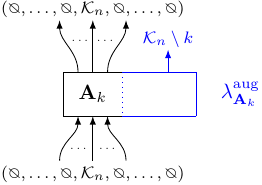}

    \vspace{.5cm}

    \includegraphics[scale=1]{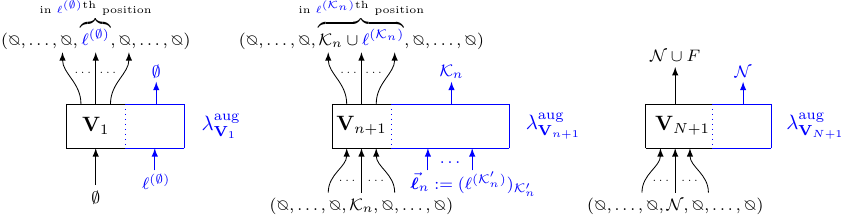}
    \caption{Augmented relations for the $\textbf{A}$- and $\textbf{V}$-nodes of the routed graph $\mathcal{G}_{\text{QC-QC}(N)}$.
    We do not show any auxiliary input for the nodes $\textbf{A}_k$ and $\textbf{V}_{N+1}$, as there are no bifurcation choices to be made. Notice also that the auxiliary output of $\textbf{V}_1$, as well as both outputs of $\textbf{V}_{N+1}$, are constant and could therefore be ignored.}
    \label{fig:augmented_rel_A_Vs_new}
\end{figure}

\paragraph*{$\textbf{A}$-nodes.}
There are no bifurcation choices to be made, hence no need to introduce auxiliary inputs (as these would anyway be trivial). The auxiliary output of some node $\textbf{A}_k$ just corresponds to the label of the branch that the input belongs to, namely $\K_n\setminus k$ as in \cref{eq:branches_Ak}. The augmented relation $\lambda_{\textbf{A}_k}^\text{aug}$, as illustrated on \cref{fig:augmented_rel_A_Vs_new}, is then defined as the (partial) function
\begin{align}
    \lambda_{\textbf{A}_k}^\text{aug}: (\ravnothing,\ldots,\ravnothing,\K_n,\ravnothing,\ldots,\ravnothing) \mapsto \Big( (\ravnothing,\ldots,\ravnothing,\K_n,\ravnothing,\ldots,\ravnothing),\; \K_n\setminus k \Big) \, . \label{eq:aug_rel_Ak_new}
\end{align}

\paragraph*{$\textbf{V}$-nodes.}
For each branch $\textbf{V}_{n+1}^{\K_n}$ of each node $\textbf{V}_{n+1}$, for $n=1,\ldots,N-1$, there is a bifurcation choice to be made, corresponding to the choice of the value $\ell\in\N\setminus\K_n$ in the output set of that branch (cf.\ \cref{eq:branches_Vn1}).%
\footnote{For $n=N-1$, the bifurcation choice within each branch $\textbf{V}_N^{\K_{N-1}}$ is trivial, as there remains only one $\ell\in\N\setminus\K_{N-1}$. In that case we may ignore the auxiliary inputs of the augmented relation.}
Hence to define the augmented relation, we add some auxiliary inputs $(\ell^{(\K_n')})_{\K_n'} \in \bigtimes_{\K_n'}\N\setminus\K_n'$ that specify these bifurcation choices, for all possible branches. 
Upon receiving an input (to the original relation $\lambda_{\textbf{V}_{n+1}}$) of the form $(\ravnothing,\ldots,\ravnothing,\K_n,\ravnothing,\ldots,\ravnothing)$, the augmented relation $\lambda_{\textbf{V}_{n+1}}^\text{aug}$ returns the element $(\ravnothing,\ldots,\ravnothing,\K_n\cup\ell^{(\K_n)},\ravnothing,\ldots,\ravnothing)$ of the output set that is determined by the value $\ell^{(\K_n)}$ in the list of auxiliary inputs, together with the value of $\K_n$ indicating the branch that happens.

The augmented relation, as illustrated on \cref{fig:augmented_rel_A_Vs_new}, can thus be written as the (partial) function
\begin{align}
    \lambda_{\textbf{V}_{n+1}}^\text{aug}: \Big( (\ravnothing,\ldots,\ravnothing,\K_n,\ravnothing,\ldots,\ravnothing), (\ell^{(\K_n')})_{\K_n'} \Big) \mapsto \ & \Big( (\ravnothing,\ldots,\ravnothing,\K_n\cup\ell^{(\K_n)},\ravnothing,\ldots,\ravnothing), \K_n \Big) \notag \\
    & \text{with $\K_n\cup\ell^{(\K_n)}$ at the ${\ell^{(\K_n)}}^\text{th}$ position} . \label{eq:aug_rel_Vn1_new}
\end{align}

The situation is similar for $\textbf{V}_1$, whose augmented relation can be written as
\begin{align}
    \lambda_{\textbf{V}_1}^\text{aug}: \big( \emptyset, \ell^{(\emptyset)} \big) \mapsto \ & \big( (\ravnothing,\ldots,\ravnothing,\ell^{(\emptyset)},\ravnothing,\ldots,\ravnothing), \emptyset \big) \ \ \text{with $\ell^{(\emptyset)}$ at the ${\ell^{(\emptyset)}}^\text{th}$ position} . \label{eq:aug_rel_V1_new}
\end{align}
For $\textbf{V}_{N+1}$ on the other hand, no bifurcation choice remains to be made, and the output of the (augmented) relation is constant:
\begin{align}
    \lambda_{\textbf{V}_{N+1}}^\text{aug}: (\ravnothing,\ldots,\ravnothing,\N,\ravnothing,\ldots,\ravnothing) \mapsto ( \N\cup F, \N ) \label{eq:aug_rel_VN1_new}
\end{align}
(see again \cref{fig:augmented_rel_A_Vs_new}).

\subsubsection{The choice function and bi-univocality}
\label{subsubsec:biunivoc_GQCQC}

By \cref{def:choice_rel}, the choice relation of $\mathcal{G}_{\text{QC-QC}(N)}$ is obtained by composing all augmented relations according to the routed graph, and is thus given, as in \cref{eq:choice_rel}, by
\begin{align}
    \Lambda_{\mathcal{G}_{\text{QC-QC}(N)}} & = \lambda_{\textbf{A}_1}^\text{aug} * \cdots * \lambda_{\textbf{A}_N}^\text{aug} * \lambda_{\textbf{V}_1}^\text{aug} * \cdots * \lambda_{\textbf{V}_{N+1}}^\text{aug} .
    \label{eq:choice_rel_QCQC_link_prod}
\end{align}
The detailed calculations, starting from this expression, are presented in \cref{sec:choice-function}. Let us describe here the result on a more intuitive level.

The choice relation $\Lambda_{\mathcal{G}_{\text{QC-QC}(N)}}$ takes as inputs all the various bifurcation choices for all possible branches $\textbf{V}_{n+1}^{\K_n}$ of all $\textbf{V}$-nodes (recalling that the $\textbf{A}$-nodes have no bifurcation choices), and outputs the possible branches that may happen at each of the $\textbf{A}$- and $\textbf{V}$-nodes.
For convenience let us denote by $\vec{\bm{\ell}}_n\coloneqq(\ell^{(\K_n')})_{\K_n'}$ the bifurcation choices for the different branches at the node $\textbf{V}_{n+1}$ (i.e.\ the auxiliary inputs of $\lambda_{\textbf{V}_{n+1}}^\text{aug}$), and by $\vec{\bm{\ell}}\coloneqq(\vec{\bm{\ell}}_n)_{n=0}^{N-1}$ the collection of all these bifurcation choices, at all \textbf{V}-nodes.
Let us then recursively define 
\begin{align}
    & \K_1^{\vec{\bm{\ell}}} \coloneqq \ell^{(\emptyset)}, \quad \K_{n+1}^{\vec{\bm{\ell}}} \coloneqq \K_n^{\vec{\bm{\ell}}}\cup \ell^{(\K_n^{\vec{\bm{\ell}}})} \ \text{ for } n=1,\ldots,N-2, \notag \\
    & \text{and for each $k\in\N$, let $n_k^{\vec{\bm{\ell}}}$ be the (single) value of $n$ such that } \ell^{(\K_{n-1}^{\vec{\bm{\ell}}})}=k. \label{eq:Kn1ell}
\end{align}

Starting from the first $\textbf{V}$-node $\textbf{V}_1$, its bifurcation choice $\ell^{(\emptyset)}$ determines the output value $(\ravnothing,\ldots,\ravnothing,\ell^{(\emptyset)},\ravnothing,\ldots,\ravnothing)=(\ravnothing,\ldots,\ravnothing,\K_1^{\vec{\bm{\ell}}},\ravnothing,\ldots,\ravnothing)$ of the augmented relation $\lambda_{\textbf{V}_1}^\text{aug}$ -- which coincides with the input of node $\textbf{V}_2$,%
\footnote{
    Recalling that the input and output index values of the $\textbf{A}$-nodes are the same (when these are in the practical input/output sets), the routes at the $\textbf{A}$-nodes that come between $\textbf{V}_n$ and $\textbf{V}_{n+1}$ are ``transparent'' for the index values of the $\textbf{V}$-nodes.
}
and which also determines the branch $\textbf{V}_2^{\K_1^{\vec{\bm{\ell}}}}$ that happens at that node. The corresponding component $\ell^{(\K_1^{\vec{\bm{\ell}}})}$ of the bifurcation choice $\vec{\bm{\ell}}_1$ at that second $\textbf{V}$-node $\textbf{V}_2$ then determines the output value $(\ravnothing,\ldots,\ravnothing,\K_1^{\vec{\bm{\ell}}}\cup \ell^{(\K_1^{\vec{\bm{\ell}}})},\ravnothing,\ldots,\ravnothing)=(\ravnothing,\ldots,\ravnothing,\K_2^{\vec{\bm{\ell}}},\ravnothing,\ldots,\ravnothing)$ of the augmented relation $\lambda_{\textbf{V}_2}^\text{aug}$ -- which coincides with the input of node $\textbf{V}_3$, and which also determines the branch $\textbf{V}_3^{\K_2^{\vec{\bm{\ell}}}}$ that happens at that node.
Continuing accordingly up to $\K_N^{\vec{\bm{\ell}}}$,
the bifurcation choices in $\vec{\bm{\ell}}$ thus determine the sequence $(\K_1^{\vec{\bm{\ell}}},\K_2^{\vec{\bm{\ell}}},\ldots,\K_{N-1}^{\vec{\bm{\ell}}},\K_N^{\vec{\bm{\ell}}})$ of (the labels of) all branches that successively happen at the $\textbf{V}$-nodes.

Note, furthermore, that each $\K_n^{\vec{\bm{\ell}}}=\K_{n-1}^{\vec{\bm{\ell}}}\cup \ell^{(\K_{n-1}^{\vec{\bm{\ell}}})}$ in the output values $(\ravnothing,\ldots,\ravnothing,\K_{n-1}^{\vec{\bm{\ell}}}\cup \ell^{(\K_{n-1}^{\vec{\bm{\ell}}})},\ravnothing,\ldots,\ravnothing)$ in the procedure above appears at a ``new'' position $k=\ell^{(\K_{n-1}^{\vec{\bm{\ell}}})}$ in the tuple of index values, ``unvisited'' before. For each $k$ there is therefore (as required) a single value of $n$ such that $\XX_n^k\neq\ravnothing$, for which $\XX_n^k=\K_n^{\vec{\bm{\ell}}}$ -- that value being $n_k^{\vec{\bm{\ell}}}$, as introduced above. This determines the branch that happens at the node $\textbf{A}_k$ to be the one labelled by $\K_{n_k^{\vec{\bm{\ell}}}}^{\vec{\bm{\ell}}}\backslash k = \K_{n_k^{\vec{\bm{\ell}}}-1}^{\vec{\bm{\ell}}}$. 

All in all, one thus finds that the collection of all bifurcation choices $\vec{\bm{\ell}}$ determines in a unique manner all the branches that happen at all nodes. The choice relation $\Lambda_{\mathcal{G}_{\text{QC-QC}(N)}}$ therefore defines a proper function, which can be written as
\begin{align}
    \Lambda_{\mathcal{G}_{\text{QC-QC}(N)}}: \Big(\underbrace{\vec{\bm{\ell}}= (\ell^{(\K_n')})_{n,\K_n'}}_{\substack{\text{auxiliary inputs}\\ \text{of the $\mathbf{V}$-nodes}}} \Big) \mapsto \Bigg(\underbrace{\Big(\K_{n_k^{\vec{\bm{\ell}}}-1}^{\vec{\bm{\ell}}}\Big)_{k\in\N}}_{\substack{\text{auxiliary outputs}\\ \text{of the $\mathbf{A}$-nodes}}} \ \ , \underbrace{\Big(\K_n^{\vec{\bm{\ell}}}\Big)_{n=1}^{N-1}}_{\substack{\text{auxiliary outputs}\\ \text{of the $\mathbf{V}$-nodes}}} \Bigg) \label{eq:choice_function_QCQC}
\end{align}
(see also \cref{eq:choice_rel_QCQC_v3} in \cref{sec:choice-function}). This means that the routed graph $\mathcal{G}_{\text{QC-QC}(N)}$ is univocal.

Conceptually, each of the inputs of the choice function says `if the set of operations up to $n$ was $\K_n$, then here is which operation should be the $(n+1)^\text{th}$', which provides an algorithm to fully determine the global order.)

Furthermore, since $\mathcal{G}_{\text{QC-QC}(N)}$ is self-adjoint -- in the sense that reversing the arrows and the routes gives the same routed graph, up to some relabelling [$\textbf{V}_{n+1}\leftrightarrow \textbf{V}_{N-n+1}$, $\XX_n^k\leftrightarrow \XX_{N-n+1}^k$ (taking values $\ravnothing\leftrightarrow\ravnothing$, $\K_n\leftrightarrow(\N\backslash\K_n)\cup k$)] -- it readily follows that $\mathcal{G}_{\text{QC-QC}(N)}^\top$ is also univocal, and therefore that $\mathcal{G}_{\text{QC-QC}(N)}$ is bi-univocal.

\subsection{Generic branch graph for QC-QCs}
\label{sec:branch-graph-QCQC}

Having constructed a bi-univocal
routed graph for QC-QCs, we will now derive its branch graph.
To do so, we need to study the information flow between the branches by determining all strong and weak parent relations.

\subsubsection{Strong parents}

To define strong parents we first need to specify the 1-dimensional values in the routed graph. We take these to be the index values $\ravnothing$ (on all arrows where these may appear), while all other index values of the form $\K_n$ are considered (a priori) to be higher-dimensional.
Hence, for all internal arrows $A \in \mathtt{Arr}^\text{int}_\Gamma$, we take $\mathtt{1Dim}_{A} = \{ \ravnothing \}$.

Recalling \cref{def:strong_parent}, given that each arrow has at most one possible 1-dimensional index value, then one can see that in order for some branch $\mathbf{V}_n^{\K_{n-1}}$ to be a strong parent of some branch $\mathbf{A}_k^{\K_{m-1}}$ (with $\K_{m-1}$, this requires that $\K_m\backslash k=\K_{n-1}$ (and hence $m=n$). Similarly, for some branch $\textbf{A}_k^{\K_n\backslash k}$ to be a strong parent of some branch $\mathbf{V}_{m+1}^{\K_m'}$, this requires that $\K_m'=\K_n$ (hence $m=n$). In both cases, $\mathtt{LinkVal}(\mathbf{V}_n^{\K_n\backslash k}, \mathbf{A}_k^{\K_n\backslash k})$ and $\mathtt{LinkVal}(\textbf{A}_k^{\K_n\backslash k}, \mathbf{V}_{n+1}^{\K_n})$ contain $\K_n$ -- which is not taken a priori to be 1-dimensional. We thus find that each $\mathbf{V}_n^{\K_n\backslash k}$ is a strong parent of each $\mathbf{A}_k^{\K_n\backslash k}$ (for all $\K_n$, $k\in\K_n$), and each $\mathbf{A}_k^{\K_n\backslash k}$ is a strong parent of $\mathbf{V}_{n+1}^{\K_n}$, and that these are the only ``strong parent'' relationships between branches.

\subsubsection{Weak parents}

Recall from \cref{subsubsec:weak_parents} that in order to determine whether a branch $\mathbf{N}^\beta$ is a weak parent of another branch $\mathbf{M}^\gamma$, we need to analyse how the bifurcation choices for $\mathbf{N}^\beta$ influence whether $\mathbf{M}^\gamma$ happens or not.
To do so, we consider the ``$\mathbf{M}^\gamma$-Happens functions'' (cf.\ \cref{def:MBetaHapp}) derived from the choice function of \cref{eq:choice_function_QCQC}.

For the branches $\textbf{A}_k^{\K_{n-1}}$ (with $\K_{n-1}$ of the form $\K_n\backslash k$) of the $\mathbf{A}$-nodes, by considering the $k^\text{th}$ element in the first tuple of outputs of the choice function, we obtain 
\begin{align}
    \Happ_{\textbf{A}_k^{\K_{n-1}}}: \ \vec{\bm{\ell}} \ \mapsto \ \delta_{\K_{n_k^{\vec{\bm{\ell}}}-1}^{\vec{\bm{\ell}}},\,\K_{n-1}}, \label{eq:Happ_Ak}
\end{align}
where $\K_m^{\vec{\bm{\ell}}}$ and $n_k^{\vec{\bm{\ell}}}$ are defined in \cref{eq:Kn1ell}, as a function of the input (the whole collection of bifurcation choices at the $\textbf{V}$-nodes) $\vec{\bm{\ell}}= (\ell^{(\K_n')})_{n,\K_n'}$.
It can be seen that $\K_{n_k^{\vec{\bm{\ell}}}-1}^{\vec{\bm{\ell}}}$, and therefore the Boolean output value of $\Happ_{\textbf{A}_k^{\K_{n-1}}}$, depends non-trivially only on%
\footnote{Notice that the $\ell^{(\K_{N-1})}$'s are excluded here, because these are trivial bifurcation choices (with only one possible values, $\ell^{(\K_{N-1})}=\N\backslash\K_{N-1}$).}
$(\ell^{(\K_{m-1})})_{m\leq n,\K_{m-1}\subseteq\K_{n-1},m<N}$ -- i.e., on the bifurcation choices made within the branches $\mathbf{V}_m^{\K_{m-1}}$, for all $m\leq \min(n,N-1)$ and all subsets $\K_{m-1}\subseteq\K_{n-1}$, so that these branches are precisely the weak parents of $\textbf{A}_k^{\K_{n-1}}$.

For the branches $\mathbf{V}_{n+1}^{\K_n}$ of the $\mathbf{V}$-nodes, by considering now the $n^\text{th}$ element (for $n=1,\ldots,N-1$) in the second tuple of outputs of the choice function, we obtain
\begin{align}
    \Happ_{\mathbf{V}_{n+1}^{\K_n}}: \ \vec{\bm{\ell}} \ \mapsto \ \delta_{\K_n^{\vec{\bm{\ell}}},\,\K_n}.
\end{align}
Here it can be seen that $\K_n^{\vec{\bm{\ell}}}$ (again defined as in \cref{eq:Kn1ell}) depends non-trivially only on $(\ell^{(\K_m)})_{m<n,\K_m\subsetneq\K_n}$
The weak parents of $\mathbf{V}_{n+1}^{\K_n}$ are thus the branches $\mathbf{V}_{m+1}^{\K_m}$ for all $m<n$ and $\K_m\subsetneq\K_n$.

As for the branches $\mathbf{V}_{1}^\emptyset$ and $\mathbf{V}_{N+1}^{\N}$, since these are the only branches of their respective nodes $\mathbf{V}_1$ and $\mathbf{V}_{N+1}$, they do not have any weak parents.

(One may notice, for completeness, that since no bifurcation choices are made at the $\mathbf{A}$-nodes, their branches are not weak parents of any branch of any other node.)

\medskip

To construct the branch graph below, we also need to characterise the weak parent relations in the adjoint graph.
Here again, we use the fact that the graph $\mathcal{G}_{\text{QC-QC}(N)}$ is self-adjoint, as already mentioned to justify its bi-univocality (end of \cref{subsubsec:biunivoc_GQCQC}).
By symmetry (through $\textbf{A}_k^{\K_n\backslash k}\leftrightarrow \textbf{A}_k^{\N\backslash\K_n}, \textbf{V}_{n+1}^{\K_n}\leftrightarrow \textbf{V}_{N-n+1}^{\N\backslash\K_n}$) we readily find that in the adjoint graph $\mathcal{G}_{\text{QC-QC}(N)}^\top$, the weak parents of the branches $\textbf{A}_k^{\K_n\backslash k}$ are the branches $\mathbf{V}_{m+1}^{\K_m}$, for all $m\geq\max(n,2)$ and all $\K_m\supseteq\K_n$, and that the weak parents of the branches $\mathbf{V}_{n+1}^{\K_n}$ are the branches $\mathbf{V}_{m+1}^{\K_m}$, for all $m>n$ and all $\K_m\supsetneq\K_n$ (and for $n=1,\ldots,N-1$; $\mathbf{V}_{1}^\emptyset$ and $\mathbf{V}_{N+1}^{\N}$ also do not have any weak parents in the adjoint graph).

\subsubsection{Constructing the branch graph and verifying the absence of loops}

The nodes of the branch graph are the various branches $\mathbf{A}_k^{\K_n\backslash k}$ and $\mathbf{V}_{n+1}^{\K_n}$. According to \cref{def:branch_graph} and to our identification of the strong and weak parent relations above, these are connected by the following arrows:
\begin{itemize}
    \item solid arrows $\mathbf{V}_n^{\K_n\backslash k} \to \mathbf{A}_k^{\K_n\backslash k}$ and $\mathbf{A}_k^{\K_n\backslash k} \to \mathbf{V}_{n+1}^{\K_n}$ for all $n,k=1,\ldots,N$ and all $\K_n\ni k$;
    
    \item green dashed arrows $\mathbf{V}_m^{\K_{m-1}} \to \textbf{A}_k^{\K_n\backslash k}$ for $1\leq m\leq n\leq N$,  $m<N$ and all $\K_{m-1}$, $k\notin\K_{m-1}$, $\K_n\supseteq\K_{m-1}\cup k$;
    
    \item red dashed arrows $\textbf{A}_k^{\K_n\backslash k} \to \mathbf{V}_{m+1}^{\K_m}$ for $1\leq n\leq m\leq N$,  $m>1$ and all $k$, $\K_n\ni k$, $\K_m\supseteq\K_n$;
    
    \item green dashed arrows $\mathbf{V}_{m+1}^{\K_m} \to \mathbf{V}_{n+1}^{\K_n}$ for $0\leq m < n < N$ and all $\K_m\subsetneq\K_n$;
    
    \item red dashed arrows $\mathbf{V}_{n+1}^{\K_n} \to \mathbf{V}_{m+1}^{\K_m}$ for $0< n < m \leq N$ and all $\K_n\subsetneq\K_m$.

\end{itemize}

\cref{fig:branch-graphs} shows the generic branch graphs for
$\mathcal{G}_{\text{QC-QC}(N)}$ with $N=2$ and $N=3$.
One can clearly see the symmetry of the graph, due to the above-mentioned self-adjointness of $\mathcal{G}_{\text{QC-QC}(N)}$: flipping the branch graph upside-down, relabelling $\textbf{A}_k^{\K_n\backslash k}\leftrightarrow \textbf{A}_k^{\N\backslash\K_n}$ and $\textbf{V}_{n+1}^{\K_n}\leftrightarrow \textbf{V}_{N-n+1}^{\N\backslash\K_n}$, flipping the orientation of the arrows and swapping the green and red colours gives the same branch graph.

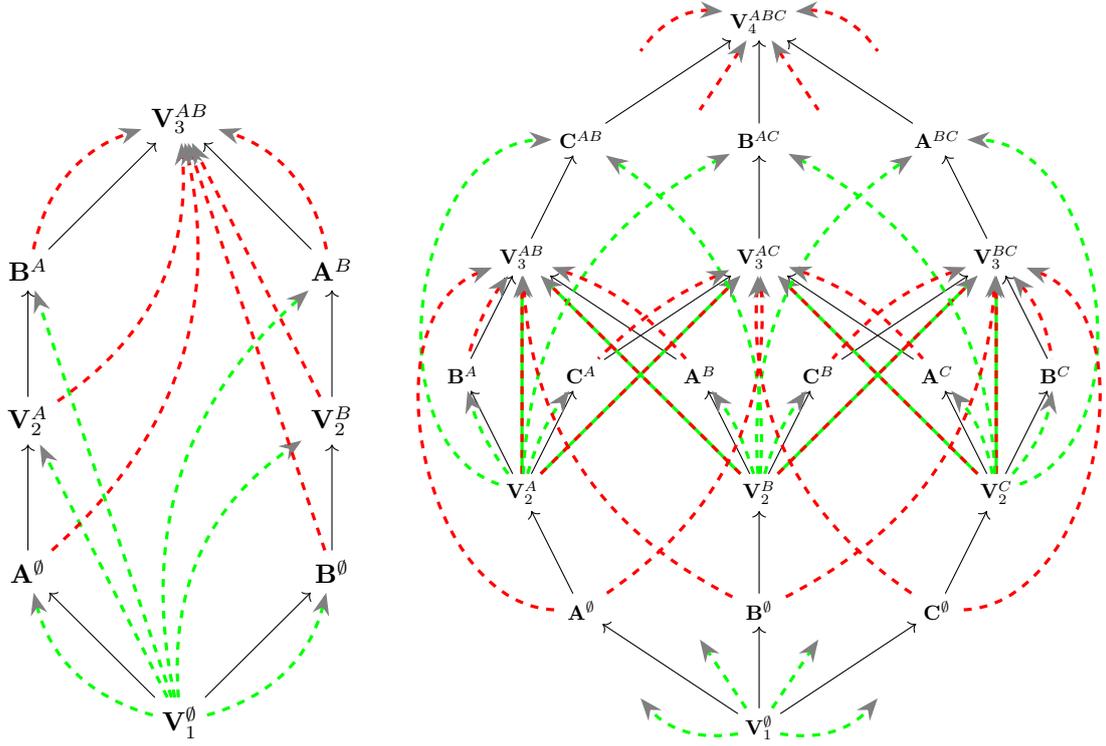
\begin{figure}
	\centering
    \begin{subfigure}[b]{0.32\textwidth}
        \centering
        \begin{tikzpicture}
            \node (V0) at (0,1) {$\mathbf{V}_1^\emptyset$};
            \node (A1) at (-2,3) {$\mathbf{A}^\emptyset$};
            \node (B1) at (2,3) {$\mathbf{B}^\emptyset$};
            \node (V1A) at (-2,5) {$\mathbf{V}_2^A$};
            \node (V1B) at (2,5) {$\mathbf{V}_2^B$};
            \node (B2) at (-2,7) {$\mathbf{B}^A$};
            \node (A2) at (2,7) {$\mathbf{A}^B$};
            \node (V2) at (0,9) {$\mathbf{V}_3^{AB}$};

            \path [->] (V0) edge (A1);
            \path [->] (A1) edge (V1A);
            \path [->] (V1A) edge (B2);
            \path [->] (B2) edge (V2);
            \path [->] (V0) edge (B1);
            \path [->] (B1) edge (V1B);
            \path [->] (V1B) edge (A2);
            \path [->] (A2) edge (V2);

            \begin{scope}[
    			>={Stealth[gray]},
    			every edge/.style={draw=green,dashed,very thick}]
                \path [->] (V0) edge[bend left] (A1);
                \path [->] (V0) edge (V1A);
                \path [->] (V0) edge (B2);
                \path [->] (V0) edge[bend right] (B1);
                \path [->] (V0) edge[bend left] (V1B);
                \path [->] (V0) edge[bend left] (A2);
            \end{scope}

            \begin{scope}[
    			>={Stealth[gray]},
    			every edge/.style={draw=red,dashed,very thick}]
                \path [->] (A1) edge[bend right] (V2);
                \path [->] (V1A) edge[bend right] (V2);
                \path [->] (B2) edge[bend left] (V2);
                \path [->] (B1) edge (V2);
                \path [->] (V1B) edge (V2);
                \path [->] (A2) edge[bend right] (V2);
            \end{scope}
        \end{tikzpicture}
    \end{subfigure}%
    \begin{subfigure}[b]{0.64\textwidth}
        \centering
    	\begin{tikzpicture}[scale=0.78, transform shape]
    		\node (V0) at (0,0) {$\mathbf{V}_1^\emptyset$};
    
    		\node (A1) at (-3,2) {$\mathbf{A}^\emptyset$};
    		\node (B1) at (0,2) {$\mathbf{B}^\emptyset$};
    		\node (C1) at (3,2) {$\mathbf{C}^\emptyset$};
    
    		\node (V1A) at (-4,4) {$\mathbf{V}_2^A$};
    		\node (V1B) at (0,4) {$\mathbf{V}_2^B$};
    		\node (V1C) at (4,4) {$\mathbf{V}_2^C$};
    
    		\node (BA) at (-5,6) {$\mathbf{B}^A$};
    		\node (CA) at (-3,6) {$\mathbf{C}^A$};
    		\node (AB) at (-1,6) {$\mathbf{A}^B$};
    		\node (CB) at (1,6) {$\mathbf{C}^B$};
    		\node (AC) at (3,6) {$\mathbf{A}^C$};
    		\node (BC) at (5,6) {$\mathbf{B}^C$};
    
    		\node (V2AB) at (-4,8) {$\mathbf{V}_3^{AB}$};
    		\node (V2BC) at (4,8) {$\mathbf{V}_3^{BC}$};
    		\node (V2AC) at (0,8) {$\mathbf{V}_3^{AC}$};
    
    		\node (C3) at (-3,10) {$\mathbf{C}^{AB}$};
    		\node (A3) at (3,10) {$\mathbf{A}^{BC}$};
    		\node (B3) at (0,10) {$\mathbf{B}^{AC}$};
    
    		\node (V3) at (0,12) {$\mathbf{V}_4^{ABC}$};
    
    		\begin{scope}[
    			every node/.style={fill=white,circle,inner sep=0pt}
    			every edge/.style=routedarrow]
    
    			\path [->] (V0) edge (A1);
    			\path [->] (V0) edge (B1);
    			\path [->] (V0) edge (C1);
    
    			\path [->] (A1) edge (V1A);
    			\path [->] (B1) edge (V1B);
    			\path [->] (C1) edge (V1C);
    
    			\path [->] (V1A) edge (BA);
    			\path [->] (V1A) edge (CA);
    			\path [->] (V1B) edge (AB);
    			\path [->] (V1B) edge (CB);
    			\path [->] (V1C) edge (AC);
    			\path [->] (V1C) edge (BC);
    
    			\path [->] (AB) edge (V2AB);
    			\path [->] (AC) edge (V2AC);
    			\path [->] (BA) edge (V2AB);
    			\path [->] (BC) edge (V2BC);
    			\path [->] (CA) edge (V2AC);
    			\path [->] (CB) edge (V2BC);
    
    			\path [->] (V2AB) edge (C3);
    			\path [->] (V2BC) edge (A3);
    			\path [->] (V2AC) edge (B3);
    
    			\path [->] (C3) edge (V3);
    			\path [->] (A3) edge (V3);
    			\path [->] (B3) edge (V3);
    		\end{scope}
    
    		\begin{scope}[
    			>={Stealth[gray]},
    			every edge/.style={draw=green,dashed,very thick}]
    
    			\path [->] (V0) edge ++(1,1.5);
    			\path [->] (V0) edge ++(-1,1.5);
    			\path [->, bend right] (V0) edge ++(2,0.5);
    			\path [->, bend left] (V0) edge ++(-2,0.5);
    		\end{scope}
    
    		\begin{scope}[
    			>={Stealth[gray]},
    			every edge/.style={draw=red,dashed,very thick}]
    
    			\path [<-] (V3) edge ++(1,-1.5);
    			\path [<-] (V3) edge ++(-1,-1.5);
    			\path [<-, bend left] (V3) edge ++(2,-0.5);
    			\path [<-, bend right] (V3) edge ++(-2,-0.5);
    		\end{scope}

    		\begin{scope}[
    					  >={Stealth[gray]},
    					  every edge/.style={draw=green,very thick,postaction={draw=red,very thick,dashed}}]
    			\path [->] (V1A) edge (V2AB);
    			\path [->] (V1A) edge (V2AC);
    			\path [->] (V1B) edge (V2BC);
    			\path [->] (V1B) edge (V2AB);
    			\path [->] (V1C) edge (V2AC);
    			\path [->] (V1C) edge (V2BC);
    		\end{scope}
    
    		\begin{scope}[
    			>={Stealth[gray]},
    			every edge/.style={{draw=green,very thick,dashed}}]
    			\path [->, bend left=10] (V1A) edge (BA);
    			\path [->, bend left=10] (V1A) edge (CA);
    			\path [->, bend left=80] (V1A) edge (C3);
    			\path [->, bend left=30] (V1A) edge (B3);
    
    			\path [->, bend right=10] (V1B) edge (AB);
    			\path [->, bend left=10] (V1B) edge (CB);
    			\path [->, bend right] (V1B) edge (C3);
    			\path [->, bend left] (V1B) edge (A3);

    			\path [->, bend right=15] (V1C) edge (AC);
    			\path [->, bend right=15] (V1C) edge (BC);
    			\path [->, bend right=80] (V1C) edge (A3);
    			\path [->, bend right=30] (V1C) edge (B3);
    		\end{scope}
    
    		\begin{scope}[
    			>={Stealth[gray]},
    			every edge/.style={{draw=red,very thick,dashed}}]
    			\path [<-, bend right=10] (V2AB) edge (BA);
    			\path [<-, bend right=10] (V2AC) edge (CA);
    			\path [<-, bend left=80] (V2BC) edge (C1);
    			\path [<-, bend left=30] (V2BC) edge (B1);
    
    			\path [<-, bend left=15] (V2AB) edge (AB);
    			\path [<-, bend right=15] (V2BC) edge (CB);
    			\path [<-, bend right] (V2AC) edge (C1);
    			\path [<-, bend left] (V2AC) edge (A1);

    			\path [<-, bend left=15] (V2AC) edge (AC);
    			\path [<-, bend left=15] (V2BC) edge (BC);
    			\path [<-, bend right=80] (V2AB) edge (A1);
    			\path [<-, bend right=30] (V2AB) edge (B1);
    		\end{scope}
    		
    	\end{tikzpicture}
    \end{subfigure}
    \caption{Generic branch graphs for $\mathcal{G}_{\text{QC-QC}(N)}$ with $N=2$ (left) and with $N=3$ (right). For simplicity of notation, $\{A, B\}$ and similar in the superscripts are written as $AB$.
    In the graph on the right, the unconnected green dashed arrows leaving from $\mathbf{V}_1^\emptyset$ are understood as being connected to all other nodes, except for $\mathbf{V}_4^{ABC}$; similarly the unconnected red dashed arrows pointing to $\mathbf{V}_4^{ABC}$ are understood as being connected to all other nodes, except for $\mathbf{V}_1^\emptyset$.
    As discussed in \cref{sec:grouping}, these branch graphs are closely related to Figures~9 and~24 in \cite{Vanrietvelde2023}, and may be obtained by identifying $P$ with $\mathbf{V}_1^\emptyset$ and $F$ with $\mathbf{V}_3^{AB}$/$\mathbf{V}_4^{ABC}$, merging all nodes $\mathbf{V}_2^{\K_1}$ with their unique parent and all nodes $\mathbf{V}_3^{\K_2}$ with their unique child.}
	\label{fig:branch-graphs}
\end{figure}

\medskip

Notice that along each (solid or dashed) arrow described above, either the size of the branch indicator $\K_n(\backslash k)$ strictly increases, or it stays the same and the arrow goes from from a $\mathbf{V}$-node to a $\mathbf{A}$-node.
This implies that there cannot be any loop in the graph: indeed, since the cardinality of the branch indicators cannot decrease, a loop would also require at least one arrow between two nodes of the same type, or from an $\mathbf{A}$-node to a $\mathbf{V}$-node, with the same branch indicator cardinality.

It then follows that the routed graph $\mathcal{G}_{\text{QC-QC}(N)}$, with its 1-dimensional index values taken to be $\ravnothing$, is valid (according to \cref{def:valid-routed-graph}).
According to \cref{def:valid-routed-superunitary}, its skeletal supermap is therefore a routed superunitary.
We will now show how to flesh it out and obtain any QC-QC.

\subsection{Obtaining any QC-QC: Fleshing out the skeletal supermap for \texorpdfstring{$\mathcal{G}_{\text{QC-QC}(N)}$}{the QC-QC generic routed graph}}
\label{sec:fleshing-out}

Having demonstrated the validity of the routed graph that we defined in \cref{sec:routed-graph-qcqc}, we now finally demonstrate that and how we can use it to represent QC-QCs.
We begin by obtaining a skeletal supermap for the routed graph $\cG_{\text{QC-QC}(N)}$, associating each arrow with a Hilbert space of suitable dimension.

Then, we provide a fleshing out of both the $\mathbf{V}$- and the $\mathbf{A}$-nodes.
The general technique we will use for this purpose is to specify general isomorphisms $J^\text{in/out}_{\mathbf{V}_{n+1}}$ and $J^\text{in/out}_{\mathbf{A}_k}$, which will be wrapped around all transformations performed in the slots of the skeletal supermap.
They will serve to translate between spaces and representations natural to routed quantum circuits, e.g.\ $\HH^{\text{in/out}(\mathbf{V}_{n+1})}$ or $\HH^{\text{in/out}(\mathbf{A}_{k})}$, and the spaces known and used in the QC-QC framework, e.g.\ $\HH^{A^{I/O}_k}$, $\HH^{\tilde A^{I/O}_n}$, $\HH^{\alpha_n}$ or $\HH^{C^{(\prime)}_n}$.

This will enable the use of the original QC-QC construction in the slots of the skeletal supermap, and leave us with a fleshing out for the $\mathbf{V}$-nodes by plugging in the (routed) isometries $\tilde{V}_{n+1}$ for the \enquote{internal} nodes, dependent on the particular QC-QC, as well as with a fleshing out for the $\mathbf{A}$-nodes, with monopartite superunitaries independent of the particular QC-QC, maintaining an unsectorised slot for each agent operation $A_k$.

We conclude this section by demonstrating that the fleshing out of the skeletal supermap indeed yields the same supermap
as the QC-QC construction stated in \cref{sec:QCQC}, and straightforwardly generalise to mixed QC-QCs.

\subsubsection{The skeletal supermap}
\label{subsubsec:skeletal_QCQC}

As a first step to obtain a supermap associated with the generic routed graph stated in \cref{def:generic-graph}, we use \cref{def:skeletal_supermap} to obtain a skeletal supermap.
Here, each arrow $A$ connecting two nodes $\mathbf{M} \to \mathbf{N}$ is associated with a sectorised Hilbert space $\HH^{\mathbf{M} \to \mathbf{N}} = \bigoplus_{k_A} \HH^{\mathbf{M} \to \mathbf{N}}_{k_A}$.

As established in \cref{sec:routed-graph-qcqc}, for each internal (i.e.\ not open-ended) arrow, its index ${\XX_n^k}$ takes values among the $n$-element subsets $\K_n$ of $\N$ that contain $k$, or the null value $\ravnothing$.
As stated in \cref{sec:branch-graph-QCQC}, we take $\mathtt{1Dim}_{M \to N} = \{ \ravnothing \}$,
implying the associated sector to be one-dimensional and be denoted as $\HH^{\mathbf{M} \to \mathbf{N}}_\ravnothing \cong \mathbb{C}$.
While for the higher-dimensional sectors, our generic notation for an index $\XX_n^k = \K_n \in \Ind_{\mathbf{M} \to \mathbf{N}}$ would be $\HH^{\mathbf{M} \to \mathbf{N}}_{\K_n}$, we will deviate from this notation going forward, writing $\HH^{\mathbf{M} \to \mathbf{N}}_{\K_n \setminus k, k}$.
Thereby, we reintroduce and single out $k$ (which disappears when writing $\XX_n^k = \K_n$) and relate more directly to the notations of the control state $\ket{\K_n \setminus k, k}^{C^{(\prime)}_n}$ for QC-QCs.
Altogether, we endow the arrows with a sectorised Hilbert space as follows for $n,k=1,\ldots,N$:
\begin{align}
    \textbf{V}_n\xrightarrow{\ \XX_n^k\ }\textbf{A}_k: \quad &
    \HH^{\mathbf{V}_n \to \mathbf{A}_k} \coloneqq
    \bigoplus_{\K_n\ni k} \HH_{\K_n\backslash k,k}^{\textbf{V}_n\to\textbf{A}_k} \oplus \HH_\ravnothing^{\textbf{V}_n\to\textbf{A}_k}, \notag \\
    \textbf{A}_k\xrightarrow{\ \XX_n^k\ }\textbf{V}_{n+1}: \quad & 
    \HH^{\mathbf{A}_k \to \mathbf{V}_{n+1}} \coloneqq
    \bigoplus_{\K_n\ni k} \HH_{\K_n\backslash k,k}^{\textbf{A}_k\to\textbf{V}_{n+1}} \oplus \HH_\ravnothing^{\textbf{A}_k\to\textbf{V}_{n+1}}.
\end{align}

Specifically, we choose these spaces (and thereby, specify a dimension assignment) such that
$\HH_{\K_n\backslash k,k}^{\textbf{V}_n\to\textbf{A}_k}\cong\HH^{A_k^I}\otimes\HH^{\alpha_n} \cong\HH^{\tilde{A}_n^I}\otimes\HH^{\alpha_n}$, $\HH_{\K_n\backslash k,k}^{\textbf{A}_k\to\textbf{V}_{n+1}}\cong\HH^{A_k^O}\otimes\HH^{\alpha_n} \cong\HH^{\tilde{A}_n^O}\otimes\HH^{\alpha_n}$
where the respective sectors each are a copy of their respective isomorphic space introduced in \cref{sec:QCQC-formal}.

To the two open-ended arrows we attach non-sectorised spaces:
\begin{align}
    \xrightarrow{\ \emptyset\ }\textbf{V}_1: \quad \HH^P, & \quad\qquad
    \textbf{V}_{N+1}\xrightarrow{\N\cup F}\ : \quad \HH^F .
\end{align}
Thereby, we obtain the following \enquote{formal} input/output spaces, as introduced in \cref{sec:routed-maps},
for each node (for $n,k=1,\ldots,N$):
\begin{align}
    \HH^{\text{in}(\textbf{V}_{n+1})} & =
    \bigotimes_{k=1}^N \HH^{\mathbf{A}_k \to \mathbf{V}_{n+1}} =
    \bigotimes_{k=1}^N \Big( \bigoplus_{\K_n\ni k} \HH_{\K_n\backslash k,k}^{\textbf{A}_k\to\textbf{V}_{n+1}} \oplus \HH_\ravnothing^{\textbf{A}_k\to\textbf{V}_{n+1}} \Big), \qquad
    \HH^{\text{in}(\textbf{V}_{1})} \eqqcolon \HH^P, \notag \\
    \HH^{\text{out}(\textbf{V}_{n})} & = 
    \bigotimes_{k=1}^N \HH^{\mathbf{V}_n \to \mathbf{A}_k} =
    \bigotimes_{k=1}^N \Big( \bigoplus_{\K_n\ni k} \HH_{\K_n\backslash k,k}^{\textbf{V}_n\to\textbf{A}_k} \oplus \HH_\ravnothing^{\textbf{V}_n\to\textbf{A}_k} \Big), \qquad
    \HH^{\text{out}(\textbf{V}_{N+1})} \eqqcolon \HH^F \, , \notag \\
    \HH^{\text{in}(\textbf{A}_{k})} & = 
    \bigotimes_{n=1}^N \HH^{\mathbf{V}_n \to \mathbf{A}_k} =
    \bigotimes_{n=1}^N \Big( \bigoplus_{\K_n\ni k} \HH_{\K_n\backslash k,k}^{\textbf{V}_n\to\textbf{A}_k} \oplus \HH_\ravnothing^{\textbf{V}_n\to\textbf{A}_k} \Big), \notag \\
    \HH^{\text{out}(\textbf{A}_{k})} & = 
    \bigotimes_{n=1}^N \HH^{\mathbf{A}_k \to \mathbf{V}_{n+1}} =
    \bigotimes_{n=1}^N \Big( \bigoplus_{\K_n\ni k} \HH_{\K_n\backslash k,k}^{\textbf{A}_k\to\textbf{V}_{n+1}} \oplus \HH_\ravnothing^{\textbf{A}_k\to\textbf{V}_{n+1}} \Big). \label{eq:H_in_out_nodes_formal}
\end{align}
These spaces are restricted by index constraints on the routed graph, which imply sectorial correlations between the incoming/outgoing spaces at each slot of the skeletal supermap.
Specifically, for each valid index value assignment, the parties $A_k$ act in a certain order.
For each space $\HH^{\textbf{V}_n\to\textbf{A}_k}$ (or $\HH^{\textbf{A}_k\to\textbf{V}_{n+1}}$), its higher-dimensional sector\footnote{Composed of multiple subsectors, for each set of ``previous parties'' $\K_n\!\setminus k$.} $\bigoplus_{\K_n\ni k} \HH_{\K_n\backslash k,k}^{\textbf{V}_n\to\textbf{A}_k}$ (or $\bigoplus_{\K_n\ni k} \HH_{\K_n\backslash k,k}^{\textbf{A}_k\to\textbf{V}_{n+1}}$)
is associated with to a party $A_k$ being in $n^\text{th}$ position.
For each $n' \neq n$, we then find a sectorial correlation between $\HH^{\textbf{V}_n\to\textbf{A}_k}$ and $\HH^{\textbf{V}_{n'}\to\textbf{A}_k}$, selecting the respective 1-dimensional sector $\HH^{\textbf{V}_{n'} \to\textbf{A}_k}_\ravnothing$  (or $\HH^{\textbf{A}_k\to\textbf{V}_{n'+1}}_\ravnothing$) for $\HH^{\textbf{V}_{n'} \to\textbf{A}_k}$ (or $\HH^{\textbf{A}_k\to\textbf{V}_{n'+1}}$).
For each $k' \neq k$ and $\HH^{\textbf{V}_{n} \to\textbf{A}_{k'}}$ (or $\HH^{\textbf{A}_{k'}\to\textbf{V}_{n+1}}$), we find analogous results.
In a slight abuse of notation, we will drop these 1-dimensional tensor factors when denoting the practical subspaces.\footnote{
    Specifically, we will identify a Hilbert space $\HH$ with its tensor product with 1-dimensional Hilbert spaces. For instance, we more explicitly have
    \begin{equation}
        \label{eq:1d-factor}
        \HH^{\text{in}(\textbf{A}_{k})} \supset
        \tilde\HH^{\text{in}(\textbf{A}_{k})}_\text{prac}
        = \bigoplus_{n=1}^N \Bigg[ \bigoplus_{\K_n\ni k} \HH_{\K_n\backslash k,k}^{\textbf{V}_n\to\textbf{A}_k} \otimes \bigotimes_{n' \neq n} \HH_\ravnothing^{\textbf{V}_{n'}\to\textbf{A}_k} \Bigg]
        \cong \bigoplus_{n=1}^N \bigoplus_{\K_n\ni k} \HH_{\K_n\backslash k,k}^{\textbf{V}_n\to\textbf{A}_k}
        = \HH^{\text{in}(\textbf{A}_{k})}_\text{prac} \, .
    \end{equation}}
Therefore, the practical input/output subspaces for each node are given by:
\begin{align}
    \HH^{\text{in}(\textbf{V}_{n+1})}_\text{prac} & = \bigoplus_{k=1}^N \bigoplus_{\K_n\ni k} \HH_{\K_n\backslash k,k}^{\textbf{A}_k\to\textbf{V}_{n+1}},
    \qquad\HH^{\text{in}(\textbf{V}_{1})}_\text{prac} = \HH^P, \notag \\
    \HH^{\text{out}(\textbf{V}_{n})}_\text{prac} & = \bigoplus_{k=1}^N \bigoplus_{\K_n\ni k} \HH_{\K_n\backslash k,k}^{\textbf{V}_n\to\textbf{A}_k}, \qquad
    \HH^{\text{out}(\textbf{V}_{N+1})}_\text{prac} = \HH^F \notag \\
    \HH^{\text{in}(\textbf{A}_{k})}_\text{prac} & = \bigoplus_{n=1}^N \bigoplus_{\K_n\ni k} \HH_{\K_n\backslash k,k}^{\textbf{V}_n\to\textbf{A}_k}, \notag \\
    \HH^{\text{out}(\textbf{A}_{k})}_\text{prac} & = \bigoplus_{n=1}^N \bigoplus_{\K_n\ni k} \HH_{\K_n\backslash k,k}^{\textbf{A}_k\to\textbf{V}_{n+1}}.
    \label{eq:H_in_out_nodes_prac}
\end{align}

\begin{remark}
    The QC-QC control systems $C_{n}$ can be seen as providing an isomorphic embedding of $\HH^{\text{out}(\textbf{V}_{n})}_\text{prac}$ into a joint space, as 
    $
        \HH^{C_n} \cong \bigoplus_{k=1}^N \bigoplus_{\K_{n} \ni k} \spann \{ \;\!\ket{\K_{n}\setminus k,k}^{C_n} \}
    $.
    Similarly, the control systems $C'_{n}$ correspond to $\HH^{\text{in}(\textbf{V}_{n+1})}_\text{prac}$.
\end{remark}

\medskip

With these, we have introduced (and simplified) the Hilbert spaces needed to construct the skeletal supermap, as explained in \cref{subsec:skeletal_fleshingout},
with the nodes replaced by empty slots connected by identities following the respective routes.
In order to obtain the desired QC-QC, we will then flesh out the slots for both the $\mathbf{V}$- and $\mathbf{A}$-nodes, as specified in the following subsections.

\subsubsection{Fleshing out for the \textbf{V}-nodes}

\begin{figure}[p]
    \hspace*{0.2\textwidth}\includegraphics[scale=0.95]{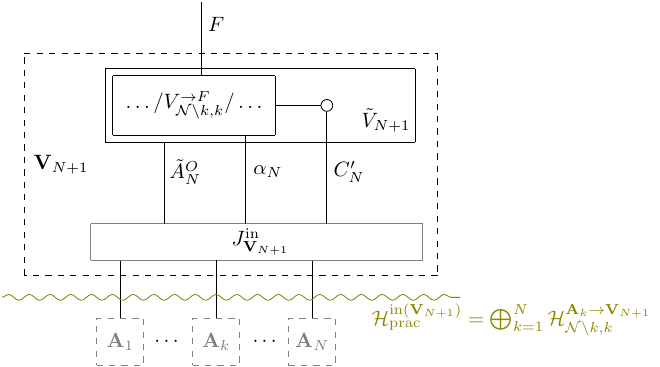}

    \vspace{.6cm}
    
     \hspace*{0.2\textwidth}\includegraphics[scale=0.95]{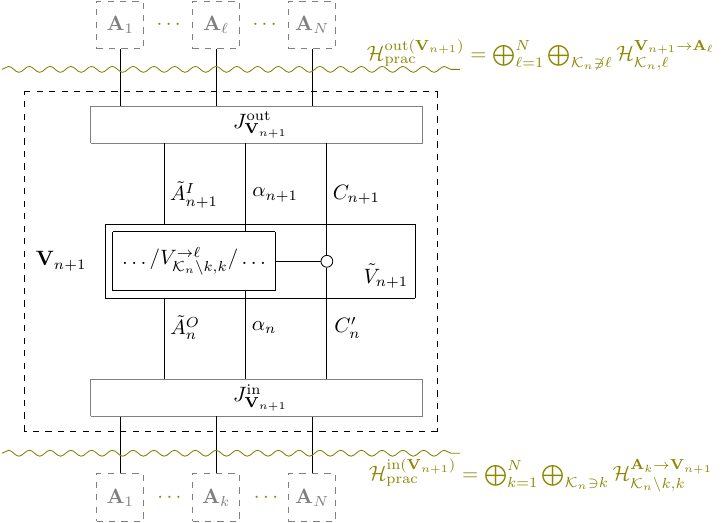}

    \vspace{.6cm}
    
    \hspace*{0.2\textwidth}\includegraphics[scale=0.95]{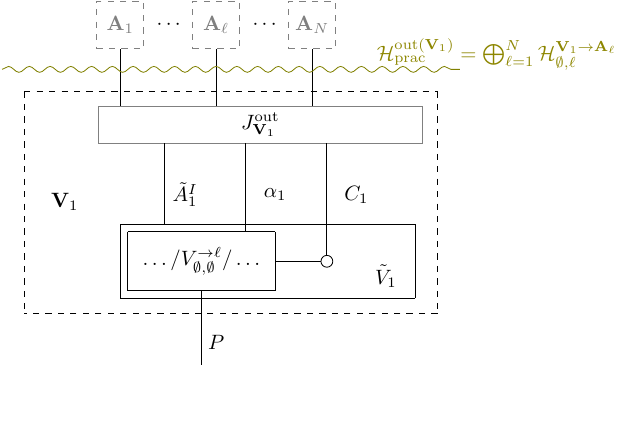}

    \vspace{-.5cm}
    
    \caption{
        Fleshing out for the \textbf{V}-nodes (with a channel).
        For each slot $\mathbf{V}_{n+1}$, isomorphisms $J^\text{in/out}_{\mathbf{V}_{n+1}}$ translate between the skeletal supermap and its associated spaces in the neighbourhood of the slot, and the QC-QC systems $\tilde A_n^{I/O}$, $\alpha_{n+1}$ and $C^{(\prime)}_{n+1}$ within the slot.
        In-between these isomorphisms, the QC-QC isometries $\tilde{V}_{n+1}$ control whose agent $A_\ell$ acts next dependent on the previous agents.}
    \label{fig:Fleshing_out_V_nodes}
\end{figure}

Having constructed the skeletal supermap associated with the generic routed graph, we first proceed by fleshing out the $\textbf{V}$-nodes with routed isometries.
In order to obtain the desired QC-QC, the idea of this fleshing out is to use precisely the internal isometries $\tilde{V}_{n+1}$ of the QC-QC in question, acting on appropriate subspaces of the Hilbert spaces associated with the $\textbf{V}$-nodes.
To this end, we first introduce two isomorphisms $J^{\text{in}}_{\mathbf{V}_{n+1}}$ and $J^{\text{out}}_{\mathbf{V}_{n+1}}$ for each node $\mathbf{V}_{n+1}$.
They identify explicitly the subspaces of the practical spaces onto which the internal isometries $\tilde{V}_{n+1}$ can be applied -- recalling that each $\HH_{\K_n\backslash k,k}^{\textbf{A}_k\to\textbf{V}_{n+1}}\cong\HH^{\tilde{A}_n^O}\otimes\HH^{\alpha_n}$ and each%
\footnote{As we consider here the node $\textbf{V}_{n+1}$, then rather than considering $\HH^{\text{out}(\textbf{V}_{n})}_{\text{prac}}$ as in \cref{eq:H_in_out_nodes_prac}, we will consider $\HH^{\text{out}(\textbf{V}_{n+1})}_{\text{prac}}$ -- hence the change in the notations: $n\to n+1$, $(\K_n\backslash k,k)\to(\K_n,\ell)$.}
$\HH_{\K_n,\ell}^{\textbf{V}_{n+1}\to\textbf{A}_k}\cong\HH^{\tilde{A}_{n+1}^I}\otimes\HH^{\alpha_{n+1}}$:
\begin{align*}
    (\text{for }n=1,\ldots,N)
    \quad J_{\textbf{V}_{n+1}}^\text{in} : \ \HH^{\text{in}(\textbf{V}_{n+1})}_\text{prac}
    = \bigoplus_{k=1}^N \bigoplus_{\K_n\ni k} \HH_{\K_n\backslash k,k}^{\textbf{A}_k\to\textbf{V}_{n+1}}
    \ \to & \quad
    \HH^{\tilde A_n^O}\otimes\HH^{\alpha_n}\otimes\HH^{C_n'}, \notag \\
    \ket{\psi} \in \HH_{\K_n\backslash k,k}^{\textbf{A}_k\to\textbf{V}_{n+1}}
    \ \mapsto & \quad
    \ket{\psi}^{\tilde A_n^O\alpha_n}\otimes\ket{\K_n\backslash k,k}^{C_n'},
\end{align*}
\vspace{-.2cm}
\begin{align}
    (\text{for }n=0,\ldots,N-1)
    \quad J_{\textbf{V}_{n+1}}^\text{out} : \ \HH^{\tilde A_{n+1}^I}\otimes\HH^{\alpha_{n+1}}\otimes\HH^{C_{n+1}}
    \ \to & \quad
   \HH^{\text{out}(\textbf{V}_{n+1})}_\text{prac} = \bigoplus_{\ell=1}^N \bigoplus_{\K_n\not\ni \ell} \HH_{\K_n,\ell}^{\textbf{V}_{n+1}\to\textbf{A}_\ell}, \notag \\
    \ket{\psi}^{\tilde A_{n+1}^I\alpha_{n+1}}\otimes\ket{\K_n,\ell}^{C_{n+1}}
    \ \mapsto & \quad
    \ket{\psi} \in \HH_{\K_n,\ell}^{\textbf{V}_{n+1}\to\textbf{A}_\ell}. \label{eq:Js_V}
\end{align}
These isomorphisms are represented on \cref{fig:Fleshing_out_V_nodes}. As Choi vectors, they can be written as
\begin{align}
    \dket{J_{\textbf{V}_{n+1}}^\text{in}} &= \bigoplus_{k=1}^N \bigoplus_{\K_n\ni k} \dket{\id}^{\HH_{\K_n\backslash k,k}^{\textbf{A}_k\to\textbf{V}_{n+1}},\tilde A_n^O\alpha_n}\otimes\ket{\K_n\backslash k,k}^{C_n'}, \notag \\
    \dket{J_{\textbf{V}_{n+1}}^\text{out}} &= \bigoplus_{\ell=1}^N \ \bigoplus_{\K_n\not\ni \ell} \ \dket{\id}^{\tilde A_{n+1}^I\alpha_{n+1},\HH_{\K_n,\ell}^{\textbf{V}_{n+1}\to\textbf{A}_\ell}} \otimes\ket{\K_n,\ell}^{C_{n+1}} \, .
\end{align}
Here, we have kept the whole sectorised Hilbert spaces explicit in the superscripts of $\dket{\id}$ to be clear with regard to the upper and lower indices.
Using \cref{eq:V-choi} for $\dket{\tilde V_{n+1}}$, the Choi vector of the overall operation that fleshes out the slot for $\textbf{V}_{n+1}$ (where we abuse slightly the notation by also calling this operation $\textbf{V}_{n+1}$), obtained by composing the above two isomorphisms with $\tilde{V}_{n+1}$ (for $n=1,\ldots,N-1$), is then given by (cf.\ \cref{fig:Fleshing_out_V_nodes})
\begin{align}
    \dket{\textbf{V}_{n+1}} & = \dket{J_{\textbf{V}_{n+1}}^\text{in}} * \dket{\tilde{V}_{n+1}} * \dket{J_{\textbf{V}_{n+1}}^\text{out}} \notag \\
    & = \dket{J_{\textbf{V}_{n+1}}^\text{in}} \ast \Bigg( \sum_{\substack{\K_{n},k\in\K_n,\\ \ell\notin\K_n}} \dket{\tilde V_{\K_{n}\setminus k,k}^{\to \ell}}^{\tilde A_n^O \alpha_n\tilde A_{n+1}^I \alpha_{n+1}} \otimes \ket{\K_{n}\setminus k,k}^{C_n'} \otimes \ket{\K_{n},\ell}^{C_{n+1}} \Bigg) \ast \dket{J_{\textbf{V}_{n+1}}^\text{out}} \notag \\
    & = \bigoplus_{\substack{\K_{n},k\in\K_n,\\ \ell\notin\K_n}} \dket{\id}^{\HH_{\K_n\backslash k,k}^{\textbf{A}_k\to\textbf{V}_{n+1}},\tilde A_n^O\alpha_n} * \dket{\tilde V_{\K_{n}\setminus k,k}^{\to \ell}}^{\tilde A_n^O \alpha_n\tilde A_{n+1}^I \alpha_{n+1}} * \dket{\id}^{\tilde A_{n+1}^I\alpha_{n+1},\HH_{\K_n,\ell}^{\textbf{V}_{n+1}\to\textbf{A}_\ell}} \notag \\
    & = \bigoplus_{\substack{\K_{n},k\in\K_n,\\ \ell\notin\K_n}} \dket{V_{\K_{n}\setminus k,k}^{\to \ell}}^{\HH_{\K_n\backslash k,k}^{\textbf{A}_k\to\textbf{V}_{n+1}},\HH_{\K_n,\ell}^{\textbf{V}_{n+1}\to\textbf{A}_\ell}}, \label{eq:dket_node_Vn1}
\end{align}
where
$\ \dket{V_{\K_{n}\setminus k,k}^{\to \ell}}^{\HH_{\K_{n}\setminus k,k}^{\textbf{A}_k\to\textbf{V}_{n+1}},\HH_{\K_{n},\ell}^{\textbf{V}_{n+1}\to\textbf{A}_\ell}}\ $
is formally the same as
$\ \dket{\tilde V_{\K_{n}\setminus k,k}^{\to \ell}}^{\tilde A_n^O \alpha_n\tilde A_{n+1}^I \alpha_{n+1}}\ $
(or as $\dket{V_{\K_{n}\setminus k,k}^{\to \ell}}^{A_k^O \alpha_n A_\ell^I \alpha_{n+1}}$),
just written in the (isomorphic) spaces
$\HH_{\K_{n}\setminus k,k}^{\textbf{A}_k\to\textbf{V}_{n+1}}$, $\HH_{\K_{n},\ell}^{\textbf{V}_{n+1}\to\textbf{A}_\ell}$ instead of $\HH^{\tilde A_n^O \alpha_n}$, $\HH^{\tilde A_{n+1}^I \alpha_{n+1}}$ (or of $\HH^{A_k^O \alpha_n}$, $\HH^{A_\ell^I \alpha_{n+1}}$), as introduced in \cref{sec:QCQC-formal}.
For the nodes $\textbf{V}_1$ and $\textbf{V}_{N+1}$, the fleshing-out operations are
\begin{align}
    \dket{\textbf{V}_1} &= 
    \dket{\tilde{V}_{1}} * \dket{J_{\textbf{V}_{1}}^\text{out}} =
    \bigoplus_{\ell\in\N} \dket{V_{\emptyset,\emptyset}^{\to \ell}}^{P,\HH_{\emptyset,\ell}^{\textbf{V}_1\to\textbf{A}_\ell}}, \label{eq:dket_node_V1} \\
    \dket{\textbf{V}_{N+1}} &=
    \dket{J_{\textbf{V}_{N+1}}^\text{in}} * \dket{\tilde{V}_{N+1}} =
    \bigoplus_{k\in\N} \dket{V_{\N\backslash k,k}^{\to F}}^{\HH_{\N\backslash k,k}^{\textbf{A}_k\to\textbf{V}_{N+1}},F}, \label{eq:dket_node_VN1}
\end{align}
respectively. Analogously, 
$\dket{V_{\emptyset,\emptyset}^{\to \ell}}^{P,\HH_{\emptyset,\ell}^{\textbf{V}_1\to\textbf{A}_\ell}}$
and
$\dket{V_{\N\backslash k,k}^{\to F}}^{\HH_{\N\backslash k,k}^{\textbf{A}_k\to\textbf{V}_{N+1}},F}$
are formally the same as
$\dket{\tilde V_{\emptyset,\emptyset}^{\to \ell}}^{P,\tilde A_1^I \alpha_1}$
and
$\dket{\tilde V_{\N\backslash k,k}^{\to F}}^{\tilde A_N^O \alpha_N,F}$ (or as $\dket{V_{\emptyset,\emptyset}^{\to \ell}}^{P,A_\ell^I \alpha_1}$
and
$\dket{V_{\N\backslash k,k}^{\to F}}^{A_k^O \alpha_N,F}$), respectively.

Note that these operations are indeed routed isometries on the (practical) spaces: they explicitly follow the routes at each node, and they are isometries due to the fact that the $\tilde{V}_{n+1}$ of QC-QCs are isometries (on the \enquote{practical} spaces of the QC-QC).
In fact, \emph{every} fleshing out for the $\textbf{V}$-nodes specifies a QC-QC, given that the decomposition in \cref{eq:V-choi} precisely ensures that each $\tilde{V}_n$ follows the relevant routes.

\subsubsection{Fleshing out for the $\mathbf{A}$-nodes}

\begin{figure}
    \centering
    \includegraphics[scale=0.95]{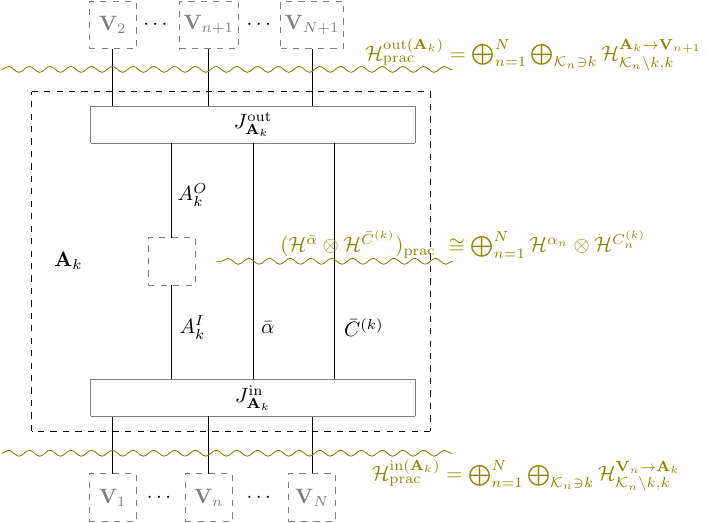}
    \caption{
        Fleshing out for the \textbf{A}-nodes (with a 1-slot comb).
        After the fleshing out, an empty (unsectorised) slot $A_k^I \to A_k^O$ remains, associated with the free operation of an agent $A_k$, while in parallel, system $\bar{C}^{(k)}$ maintains the control system of the QC-QC.
        Two isomorphisms $J^\text{in/out}_{\mathbf{A}_{k}}$ translate between the skeletal supermap and its associated spaces in the neighbourhood of the slot, and the QC-QC systems $A_k^{I/O}$, $\alpha_{n}$ and $C^{(k)}_{n}$ within the slot.
        Altogether, the fleshing out of $\mathbf{A}_k$ is independent of the QC-QC under consideration.}        
    \label{fig:Fleshing_out_A_nodes}
\end{figure}

As the next step, as illustrated in \cref{fig:Fleshing_out_A_nodes}, we flesh out the slots for the $\textbf{A}$-nodes using 1-slot superunitaries of the type
\begin{equation}
    (\HH^{A^I_k}\to\HH^{A^O_k})
    \quad \longrightarrow \quad
    (\HH^{\text{in}(\textbf{A}_k)}_\text{prac} \to \HH^{\text{out}(\textbf{A}_k)}_\text{prac}) \, ,
\end{equation}
which will take the routed slot of type $\HH^{\text{in}(\textbf{A}_k)}_\text{prac} \to \HH^{\text{out}(\textbf{A}_k)}_\text{prac}$ to a non-routed slot of type $\HH^{A^I_k}\to\HH^{A^O_k}$.

For this purpose, we need to consider the subspace associated with the specific agent node $\mathbf{A}_k$,
rather than the tilded spaces $\HH^{\tilde{A}_n^{I/O}}$, which are employed in \cref{eq:An} to obtain a QC-QC implementation.
To do so, we introduce a Hilbert space $\HH^{\bar{C}^{(k)}}$, serving as the control space for the causal order, reorganised to be grouped by the \enquote{current} agent $A_k$:
\begin{equation}
     \HH^{\bar{C}^{(k)}} \coloneqq \bigoplus_{n=1}^N \HH^{C_n^{(k)}} \quad \text{with} \quad 
     \HH^{C_n^{(k)}} \coloneqq \bigoplus_{\K_{n}\ni k} \spann \{ \ket{\K_{n}\setminus k,k}^{C_n} \}\,.
\end{equation}
Due to the tensor product with the ancillary space, this decomposition by $k$ carries over to the ancillary system $\alpha_n$, which takes different values for $n$ dependent on the system $C_n^{(k)}$ we consider.
Accordingly, we introduce
$\HH^{\bar{\alpha}} \coloneqq \bigoplus_{n=1}^N \HH^{\alpha_n}$.
Following the treatment of indexed wires in the skeletal supermap picture, we differentiate between the formal tensor product of $\HH^{\bar{C}^{(k)}}$ and $\HH^{\bar{\alpha}}$ and their more restricted practical space
\begin{equation}
    (\HH^{\bar{\alpha}}\otimes\HH^{\bar{C}^{(k)}})_\text{prac} \cong \bigoplus_{n=1}^N \HH^{\alpha_n}\otimes\HH^{C_n^{(k)}} \, .
\end{equation}
With this in place, we introduce two isomorphisms $J^{\text{in}}_{\mathbf{A}_k}$ and $J^{\text{out}}_{\mathbf{A}_k}$ for each node $\mathbf{A}_k$, transforming the incoming and outgoing spaces $\HH^{\text{in}(\mathbf{A}_k)}_\text{prac}$ and $\HH^{\text{out}(\mathbf{A}_k)}_\text{prac}$ to the spaces used by the QC-QC formalism, as introduced in \cref{sec:QCQC}:
\begin{align}
    J_{\textbf{A}_k}^\text{in} &:
    & \HH^{\text{in}(\textbf{A}_k)}_\text{prac}
    = \bigoplus_{n=1}^N \bigoplus_{\K_{n}\ni k} \HH_{\K_{n}\setminus k,k}^{\textbf{V}_n\to\textbf{A}_k}
    \ \to & \
    \HH^{A_k^I}\otimes(\HH^{\bar{\alpha}}\otimes\HH^{\bar{C}^{(k)}})_\text{prac}, \notag \\
    & & \ket{\psi} \in \HH_{\K_{n}\setminus k,k}^{\textbf{V}_n\to\textbf{A}_k}
    \ \mapsto & \
    \ket{\psi}^{A_k^I\alpha_n}\otimes\ket{\K_{n}\setminus k,k}^{C_n^{(k)}}, \notag \\[3mm]
    J_{\textbf{A}_k}^\text{out} &: & \HH^{A_k^O}\otimes(\HH^{\bar{\alpha}}\otimes\HH^{\bar{C}^{(k)}})_\text{prac}
    \ \to & \
    \HH^{\text{out}(\textbf{A}_k)}_\text{prac}
    = \bigoplus_{n=1}^N \bigoplus_{\K_{n}\ni k} \HH_{\K_{n}\setminus k,k}^{\textbf{A}_k\to\textbf{V}_{n+1}}, \notag \\
    & & \ket{\psi}^{A_k^O\alpha_n}\otimes\ket{\K_{n}\setminus k,k}^{C_n^{(k)}}
    \ \mapsto & \
    \ket{\psi} \in \HH_{\K_{n}\setminus k,k}^{\textbf{A}_k\to\textbf{V}_{n+1}}. \label{eq:def_J_Ak_out}
\end{align}
As Choi vectors, these isomorphisms can be represented as:
\begin{align}
    \dket{J_{\textbf{A}_k}^\text{in}} &= \bigoplus_{n=1}^N \bigoplus_{\K_{n}\ni k} \dket{\id}^{\HH_{\K_{n}\setminus k,k}^{\textbf{V}_n\to\textbf{A}_k},A_k^I\alpha_n}\otimes\ket{\K_{n}\setminus k,k}^{C_n^{(k)}}, \notag \\
    \dket{J_{\textbf{A}_k}^\text{out}} &= \bigoplus_{n=1}^N \bigoplus_{\K_{n}\ni k} \dket{\id}^{A_k^O\alpha_n,\HH_{\K_{n}\setminus k,k}^{\textbf{A}_k\to\textbf{V}_{n+1}}} \otimes\ket{\K_{n}\setminus k,k}^{C_n^{(k)}},
\end{align}
using the same notational conventions as in the previous section. So the Choi vector of the overall 1-slot superunitary that fleshes out the slot for $\textbf{A}_k$ (where again, we abuse slightly the notation by also calling this operation $\textbf{A}_{k}$), as pictured in \cref{fig:Fleshing_out_A_nodes}, is directly given as
\begin{align}
    \dket{\textbf{A}_k} & = \dket{J_{\textbf{A}_k}^\text{in}} * \dket{J_{\textbf{A}_k}^\text{out}} = \bigoplus_{n=1}^N \bigoplus_{\K_{n}\ni k} \dket{\id}^{\HH_{\K_{n}\setminus k,k}^{\textbf{V}_n\to\textbf{A}_k},A_k^I\alpha_n} * \dket{\id}^{A_k^O\alpha_n,\HH_{\K_{n}\setminus k,k}^{\textbf{A}_k\to\textbf{V}_{n+1}}}. \label{eq:dket_node_Ak}
\end{align}

\subsubsection{Composing the process vector and obtaining any QC-QC}
Altogether, fleshing out the slots for all \textbf{V}-nodes and all \textbf{A}-nodes, with isometries and 1-slot superunitaries respectively, for the supermap we obtain the Choi vector
\begin{align}
    & \Big( \bigotimes_{n=0}^N \dket{\textbf{V}_{n+1}} \Big) * \Big( \bigotimes_{k\in\N} \dket{\textbf{A}_k} \Big) \notag \\
    & = \!\Bigg( \!\! \Big( \! \bigoplus_{\ell\in\N} \dket{V_{\emptyset,\emptyset}^{\to \ell}}^{P,\HH_{\emptyset,\ell}^{\textbf{V}_1\to\textbf{A}_\ell}} \!\Big) \!\otimes\! \bigotimes_{n=1}^{N-1} \! \Big( \!\!\! \bigoplus_{\substack{\K_{n},\\ k\in\K_n,\\ \ell\notin\K_n}} \!\! \dket{V_{\K_{n}\setminus k,k}^{\to \ell}}^{\HH_{\K_{n} \setminus k,k}^{\textbf{A}_k\to\textbf{V}_{n+1}},\HH_{\K_{n},\ell}^{\textbf{V}_{n+1}\to\textbf{A}_\ell}} \!\Big) \!\otimes\! \Big( \! \bigoplus_{k\in\N} \dket{V_{\N\backslash k,k}^{\to F}}^{\HH_{\N\backslash k,k}^{\textbf{A}_k\to\textbf{V}_{N+1}},F} \Big) \!\! \Bigg) \notag \\[-3mm]
    & \hspace{60mm} * \Bigg( \bigotimes_{k\in\N} \! \Big( \bigoplus_{n=1}^N \bigoplus_{\K_{n} \ni k} \!\dket{\id}^{\HH_{\K_{n}\setminus k,k}^{\textbf{V}_n\to\textbf{A}_k},A_k^I\alpha_n} * \dket{\id}^{A_k^O\alpha_n,\HH_{\K_{n}\setminus k,k}^{\textbf{A}_k\to\textbf{V}_{n+1}}} \Big) \!\Bigg) \notag \\[2mm]
    & = \sum_{(k_1,k_2,\ldots,k_N)} \dket{V_{\emptyset,\emptyset}^{\to k_1}}^{P,A_{k_1}^I\alpha_1} * \dket{V_{\emptyset,k_1}^{\to k_2}}^{A_{k_1}^O\alpha_1,A_{k_2}^I\alpha_2} * \dket{V_{\{k_1\},k_2}^{\to k_3}}^{A_{k_2}^O\alpha_2,A_{k_3}^I\alpha_3} * \ \cdots \notag \\[-3mm]
    & \hspace{20mm} \cdots \ * \dket{V_{\{k_1,\ldots,k_{N-2}\},k_{N-1}}^{\to k_N}}^{A_{k_{N-1}}^O\alpha_{N-1},A_{k_N}^I\alpha_N} * \dket{V_{\{k_1,\ldots,k_{N-1}\},k_N}^{\to F}}^{A_{k_N}^O\alpha_N,F} = \ket{w_\text{QC-QC}}, \label{eq:recover_wQCQC}
\end{align}
composing \cref{eq:dket_node_Vn1,eq:dket_node_V1,eq:dket_node_VN1,eq:dket_node_Ak}
so that we indeed recover the process vector of the QC-QC under consideration, as given in \cref{eq:Choi_V_QCQC}.
For the detailed calculation, see \cref{sec:composing-process}.

Having successfully reconstructed the process vector $\ket{w_\text{QC-QC}}$ representing a pure QC-QC, we still need to generalise to the mixed case to capture generic QC-QCs.
As discussed in \cref{sec:QCQC-formal}, we obtain the generalisation to mixed circuits directly by splitting $\HH^F = \HH^{F'} \otimes \HH^{\alpha_F}$:
\begin{equation}
    W = \Tr_{\alpha_F} \ketbra{w}{w} \, .
\end{equation}
As in the respective skeletal supermap, the output space of $\mathbf{V}_{N+1}$ is not sectorised, discarding $\alpha_F$ does not require any special care in comparison to conventional quantum circuits.

We have thus proven our main result, which we recall here:
\main*

\section{Alternative Routed Quantum Circuit Descriptions}
\label{sec:variations}

The construction of \cref{sec:QCQCs_as_RQCs} provides a systematic approach to obtaining any $N$-slot QC-QC by fleshing out the skeletal supermap associated with the generic routed graph $\cG_{\text{QC-QC}(N)}$. 
It is important to note, however, that this construction does not provide a \textit{unique} way to obtain any given QC-QC as a routed quantum circuit.
Not only can any given QC-QC supermap have different decompositions into internal operations $V_{\K_{n}\setminus k,k}^{\to \ell}$, but the same routed supermap can also be obtained by fleshing out skeletal supermaps associated with different routed graphs.
Indeed, as a case in point, Ref.~\cite{Vanrietvelde2023} gave a routed circuit decomposition of the quantum switch different from what one would obtain using our construction from the graph $\cG_{\text{QC-QC}(2)}$ (compare Fig.~6 from Ref.~\cite{Vanrietvelde2023} to \cref{fig:routed_graphs_QCQC_N_2} here), and similarly for what the authors of~\cite{Vanrietvelde2023} called the ``Grenoble process'' (introduced in \cite{Wechs2021}) in contrast to $\cG_{\text{QC-QC}(3)}$ (compare Fig.~21 from Ref.~\cite{Vanrietvelde2023} to \cref{fig:routed_graphs_QCQC_N_3} here).

In this section, we introduce some alternative routed quantum circuit decompositions for QC-QCs, demonstrating that they indeed yield the same sets of processes.
First, we provide an alternative treatment of the  QC-QC ancillary systems $\alpha_n$ in \cref{sec:var-ancillary}.
In \cref{sec:var-split}, we split apart the $\mathbf{V}$-nodes according to their branch structure, before rewriting them as pre- and post-processings for individual agents in \cref{sec:grouping}.
In certain cases, this allows for a less verbose representation that focuses on the connectivity between the individual agents. 
We then apply these alternatives to reconstruct previous routed quantum circuit decompositions for both the quantum switch and the Grenoble process.
Finally, in \cref{sec:less-arrows}, we sketch how certain restrictions regarding the index values practically attained allow one, for certain QC-QCs, to omit arrows from the routed graph.

\subsection{An alternative treatment of the QC-QC ancillary systems}
\label{sec:var-ancillary}

The first such variation concerns a more distinctive treatment of the ancillary systems, which is perhaps closer to how they are treated in the QC-QC framework.
Recall that in that framework, an ancillary memory system $\alpha_n$ is passed between the internal isometries $\tilde{V}_n$ and $\tilde{V}_{n+1}$, to which the agents' operations $A_k$ do not have access, see \cref{fig:QCQC}.
However, for the skeletal supermap associated with the routed graph $\cG_{\text{QC-QC}(N)}$, in \cref{subsubsec:skeletal_QCQC} we attach Hilbert spaces to the wires such that the ancillary system $\alpha_n$ is passed to the agent nodes $\mathbf{A}_k$ from each node $\mathbf{V}_n$.
The requirement that the agents' operations $A_k$ cannot act on these ancillary degrees of freedom was instead enforced by the choice of fleshing out for the $\mathbf{A}$-nodes (see \cref{fig:Fleshing_out_A_nodes}).

A natural, alternative definition of the routed graph is also possible, in which the separation of these ancillary systems from the agents' nodes is reflected in the structure of the graph and enforced by the skeletal supermap.
Matching more closely the original formulation of QC-QCs, one could add further direct arrows $\mathbf{V}_n \to \mathbf{V}_{n+1}$ (for $1\le n \le N$) to the routed graph $\cG_{\text{QC-QC}(N)}$, each associated with a single index value.\footnote{
    For instance, $\Ind_{\mathbf{V}_n \to \mathbf{V}_{n+1}} = \{ \alpha_n\}$, the choice of value bearing no further relevance -- it just in general should not be one-dimensional (so as to allow for non-trivial ancillary systems).}
Going forward, we will drop these indices, just as for the open unindexed arrows in the routed graph, and denote the resulting (valid) routed graph as $\cG_{\text{QC-QC}(N)}^\alpha$.
The additional arrows can then be associated directly with the Hilbert spaces $\HH^{\alpha_n}$ in the skeletal supermap.
For the $N=2$ and $N=3$ cases, whose routed graphs are shown explicitly for our initial construction in \cref{fig:routed_graphs_QCQC_N_2_3}, this would correspond simply to adding the direct arrows $\mathbf{V}_1\to \mathbf{V}_2$, $\mathbf{V}_2\to \mathbf{V}_3$ and (for $N=3$) $\mathbf{V}_3\to \mathbf{V}_4$.
The associated branch graph is then similarly augmented with additional solid arrows, connecting compatible branches of $\mathbf{V}_n$ and $\mathbf{V}_{n+1}$ directly.\footnote{
    This follows as we obtain additional strong parent relations via the branch-linking values
    $\mathtt{LinkVal}(\mathbf{V}_n^{\K_{n-1}}, \mathbf{V}_{n+1}^{\K_{n}}) = \{ \alpha_n \}$
    if $\K_{n-1} \subsetneq \K_n$,
    with $\mathtt{LinkVal}(\mathbf{V}_n^{\K_{n-1}}, \mathbf{V}_{n+1}^{\K_{n}}) = \emptyset$ otherwise, with the arrow $\mathbf{V}_n \xrightarrow{\alpha_n} \mathbf{V}_{n+1}$ being of non-trivial dimension.
    The weak parent relations remain unchanged.}

With this alternative treatment of the ancillary systems, the $\bar\alpha$ wire is removed from the fleshing out for the nodes $\mathbf{A}_k$, given in \cref{eq:def_J_Ak_out} and illustrated in \cref{fig:Fleshing_out_A_nodes}.
The fleshing out is thereby reduced to withholding the routing information encoded in $C_n^{(\prime)}$ from the operations $A_k$, while still producing an unsectorised slot.
Considering this, we note that one could also refrain from fleshing out the $\mathbf{A}$-node slots entirely, as the agent's access to the control system $\bar{C}^{(k)}$ is restricted due to sectorial constraints.
Specifically, we would obtain a generalisation of any QC-QC as a \emph{routed supermap} -- with its agents $A_k$ performing routed operations dependent on the state of $\bar C^{(k)}$, which may not be modified (by e.g.\ a measurement in a basis incompatible with the sectorisation).%
\footnote{
    We may then consider $\bigotimes_{n=0}^{N} \dket{\textbf{V}_{n+1}}$ as the Choi representation of this \emph{routed supermap} (cf.\ \cref{def:routed-supermap}) generalisation of QC-QCs, while additionally bearing in mind the routes for the open slots associated with $\mathbf{A}_k$.
}

Ultimately, both representations are valid and allow one to construct any QC-QC via appropriate fleshing outs.%
\footnote{Furthermore, both routed graphs generate the same set of routed supermaps, obtained from fleshing out a skeletal supermap associated with the respective routed graph.}
However, as is generally the case, certain representations may be more natural depending on the physical situation, e.g.\ the specific implementation being considered.
For example, in a quantum optics experiment where ancillary information is stored in a photonic degree of freedom (presumed to be) inaccessible to the agent (even though it passes through their labs), the first construction might seem more appropriate.
By contrast, the second picture would more closely reflect a setup where the respective degree of freedom is stored elsewhere and never physically coincident with the agent, for instance in a physical system directly transferred between the internal isometries $\tilde{V}_n$.%
\footnote{A similar idea to physically prevent agents from accessing the control qubit is proposed in \cite{Costa2025}.}
A similar distinction is also discussed in \cite{Ormrod2023}, where the authors argue for the importance of accounting for all physical transformations \textit{in principle possible} that an agent might perform when characterising the causal structure of an experiment, rather than imposing artificial restrictions on their operations.
Nonetheless, as our generic routed graph is not intended to represent any particular physical realisation of the indefinite causal order in a QC-QC, and as we do not analyse any specific causal structure here, these distinctions are of no immediate importance to this work.

\subsection{Splitting the internal nodes according to their branch structure}
\label{sec:var-split}

We continue by introducing an alternative representation to the generic routed graph $\cG_{\text{QC-QC}(N)}$ of \cref{def:generic-graph}, which splits apart the internal nodes $\mathbf{V}_{n+1}$ according to their respective branch structure.
Intuitively, the idea is to split each node $\mathbf{V}_{n+1}$ into one individual node per branch $\mathbf{V}^{\K_n}_{n+1}$.
Thus, the connectivity of the agents via the internal nodes becomes more apparent in the routed graph picture than for the original generic routed graph $\cG_{\text{QC-QC}(N)}$, at the cost of increasing the number of internal nodes exponentially.
\begin{definition}[{$\mathcal{G}^\text{split}_{\textup{QC-QC}(N)}$}]
    \label{def:internal-split}
    The \emph{split-node generic routed graph} $\mathcal{G}^\text{split}_{\textup{QC-QC}(N)}$
    for an $N$-slot QC-QC is defined as follows:
    \begin{itemize}
        \item There are 2 types of nodes: ``\textbf{A}-nodes'' $\textbf{A}_k$, for $k=1,\ldots,N$, \:\ and ``$\mathbf{\hat V}$-nodes'' $\mathbf{\hat V}^{\K_n}_{n+1}$, for $n=0,\ldots,N$, with $\K_n \subseteq \mathcal{N}$ and $\abs{\K_n} = n$. (In total, there are $N + 2^N$ nodes.)
        \item The indexed graph involves the following indexed arrows:
        \begin{itemize}
            \item open-ended arrows $\ \xrightarrow{\ \emptyset \ } \mathbf{\hat V}^\emptyset_1 \ $ and $\ \mathbf{\hat V}^\cN_{N+1}\xrightarrow{\N\cup F} \ \ $ (whose indices which we will usually drop in practice),
            \item
            $\forall\,n=1,\ldots,N, \ \forall\,k=1,\ldots,N, \ \forall\,\K_n \ni k,$
            \begin{equation}
                \mathbf{\hat V}^{\K_n \backslash k}_n
                \xrightarrow{\ \XX_{\K_n \backslash k}^k\ }
                \mathbf{A}_k
                \qquad \text{and} \qquad
                \mathbf{A}_k
                \xrightarrow{\ \XX_{\K_n \backslash k}^{k\ \prime} \ }
                \mathbf{\hat V}^{\K_n}_{n+1},
                \label{eq:split-arrows}
            \end{equation}
            with a given index $\XX_{\K_n \backslash k}^{k\ (\prime)}$ taking either $\XX_{\K_n \backslash k}^{k\ (\prime)} = \K_n$ as its value, or $\XX_{\K_n \backslash k}^{k\ (\prime)} = \ravnothing$.
        \end{itemize}
        
        \item The routes are specified such that
        \begin{itemize}
            \item For each node $\textbf{A}_k$, the lists of all input and output index values $({\XX_{\K_n \backslash k}^k})_{\K_n\ni k} \in \Ind^\text{in}_{\mathbf{A}_k}$ and $(\XX_{\K_n \backslash k}^{k\ \prime})_{\K_n\ni k} \in \Ind^\text{out}_{\mathbf{A}_k}$ must be equal -- from now on we will then simply drop the primes -- and must further satisfy $\exists!\,n,\exists!\,\K_n, \XX_{\K_n \backslash k}^k\neq\ravnothing$.
            Hence, the practical input and output sets of each node $\textbf{A}_k$ consist of (the same) lists of index values $(\ravnothing,\ldots,\ravnothing,\K_n,\ravnothing,\ldots,\ravnothing)$, for some $n$ and some set $\K_n \ni k$ at the position corresponding to the arrow coming from node $\mathbf{\hat V}^{\K_n \backslash k}_n$ and to the arrow going to node $\mathbf{\hat V}_{n+1}^{\K_n}$.
            
            \item For each node $\mathbf{\hat V}_{n+1}^{\K_n}$, the lists of all input and output index values $(\XX_{\K_n \backslash k}^k)_{k\in\K_n} \in \Ind^\text{in}_{\mathbf{\hat V}^{\K_n}_{n+1}}$ and $({\XX_{\K_n}^\ell})_{\ell\notin\K_n} \in \Ind^\text{out}_{\mathbf{\hat V}^{\K_n}_{n+1}}$ must satisfy
            \begin{align}
                \text{for } n =0 \text{ and } n=N: \quad & \exists!\,\ell, \XX_\emptyset^\ell\neq\ravnothing \ \text{ and } \ \exists!\,k, \XX_{\cN\setminus k}^k\neq\ravnothing, \text{ respectively.} \\[2mm]
                \text{for } \emptyset \subsetneq \K_n \subsetneq \cN: \quad & \left\{ \begin{tabular}{rl}
                    $\text{either}$ & $\forall\,k, \XX_{\K_n \backslash k}^k=\ravnothing, \ \forall\,\ell, \XX_{\K_n}^\ell=\ravnothing$ \\[1mm]
                    $\text{or}$ & $\exists!\,k, \XX_{\K_n \backslash k}^k\neq\ravnothing, \ \exists!\,\ell, \XX_{\K_n}^\ell\neq\ravnothing$
                \end{tabular} \right.
                .
            \end{align}
            Hence, the practical input sets of each node $\mathbf{\hat V}_{n+1}^{\K_n}$, for $0 < n \leq N$ and $\emptyset \subsetneq \K_n \subseteq \cN$, consist either of lists of all-$\ravnothing$ index values, or of lists of index values of the form $(\ravnothing,\ldots,\ravnothing,\K_n,\ravnothing,\ldots,\ravnothing)$, with the set $\K_n$ attached to the arrow coming from some node $\mathbf{A}_k$, with $k\in\K_n$.
            Similarly, the practical output sets of each node $\mathbf{\hat V}_{n+1}^{\K_{n}}$, for $0 \leq n < N$ and $\emptyset \subseteq \K_{n} \subsetneq \N$, consist either of lists of all-$\ravnothing$ index values, or of lists of index values of the form $(\ravnothing,\ldots,\ravnothing,\K_{n}\cup\ell,\ravnothing,\ldots,\ravnothing)$, with the set $\K_{n}\cup\ell$ attached to the arrow going to some node $\mathbf{A}_\ell$, such that $\ell\notin\K_{n}$.
            Lists of all-$\ravnothing$ index values at the input are related to lists of all-$\ravnothing$ index values at the output; lists of index values the form $(\ravnothing,\ldots,\ravnothing,\K_n,\ravnothing,\ldots,\ravnothing)$ at the input are related to lists of index values the form $(\ravnothing,\ldots,\ravnothing,\K_n\cup\ell,\ravnothing,\ldots,\ravnothing)$ at the output.
        \end{itemize}
        
        Both conditions are analogous to the non-split case (\cref{def:generic-graph}), when replacing the subscripts $n$ of the index labels with some $\K_{n} \setminus k$ and adapting the constraints in an appropriate manner.
        They can be condensed into the global index constraint
        \begin{align}
            & \exists!\,(k_1,k_2,\ldots,k_N), \XX_\emptyset^{k_1} = \{k_1\},\, \XX_{\{k_1\}}^{k_2} = \{k_1,k_2\},\, \ldots,\, \XX_{\mathcal{N} \setminus k_N}^{k_N} = \{k_1,k_2,\ldots,k_N\}=\N, \notag \\
            & \qquad \text{and all other } \XX_{\K_{n} \setminus k}^k=\ravnothing. 
        \end{align}
        Just like \cref{eq:global_cstr_QCQC}, this again captures the understanding that, along each global path, all operations $A_k$ are applied once and only once, in a given order.
    \end{itemize}
\end{definition}

Contrasting this with the definition of the generic routed graph $\cG_{\text{QC-QC}(N)}$ in \cref{def:graph}, we see that splitting the $\mathbf{V}$-nodes forces the index \emph{labels} to carry additional information to ensure their uniqueness across the graph.
By contrast, we see that the index \emph{values} remain virtually unchanged,
such that the value $\K_n$ does not encode any additional information other than being non-$\ravnothing$.
With the values being binary, we could alternatively pick $\XX^k_{\K_n} \in \{0, 1\}$, and write the global index constraints as follows:
\begin{equation}
    \label{eq:simple-values}
    \exists!\,(k_1,k_2,\ldots,k_N): \quad
    \XX_\emptyset^{k_1} =
    \XX_{\{k_1\}}^{k_2} =
    \cdots =
    \XX_{\mathcal{N} \setminus k_N}^{k_N} = 1,
\end{equation}
while all other $\XX_{\K_n}^k = 0$.
This would reproduce index values closer to the presentation of \cite{Vanrietvelde2023}.
However, to keep representations as close to \cref{def:graph} as possible, we will stick to the more redundant convention in this work.

\begin{figure}
	\centering
	\begin{subfigure}[b]{0.4\textwidth}
		\centering
		\!\!\!\!\begin{tikzpicture}[transform shape]
			\node (bottom) at (0,-0.3) {};
			\node (V0) at (0,0.5) {$\mathbf{V}_1$};
			\node[text=NavyBlue] (A) at (-2,2) {$\mathbf{A}$};
			\node (V1t) at (0,2.3) {$\mathbf{V}_2^A$};
            \node (V1b) at (0,1.7) {$\mathbf{V}_2^B$};
			\node[text=NavyBlue] (B) at (2,2) {$\mathbf{B}$};
			\node (V2) at (0,3.5) {$\mathbf{V}_3$};
			\node (top) at (0,4.3) {};
			\node (i1) at (0,-1.7) [align=center] {
                $\exists! \, (k_1, k_2) \, :$\\[1mm]
                $\XX_\emptyset^{k_1} = \{ k_1 \}, \;\ \XX_{k_1}^{k_2} = \{k_1, k_2 \} \mathrel{(=} \{A, B \}),$\\[1mm]
                while all other $\XX^k_{\K_n\backslash k} = \ravnothing$};
			\node (i2) at (0,-3) {};
			
			\begin{scope}[
				every node/.style={fill=white,circle,inner sep=0pt}
				every edge/.style=routedarrow]
				\path [->] (bottom) edge (V0);
				\path [->] (V0) edge[bend left] node[below] {$\scriptstyle{\XX^A_\emptyset}$} (A);
				\path [->] (V0) edge[bend right] node[below] {$\scriptstyle{\XX^B_\emptyset}$} (B);
				\path [->] (A) edge[bend left=15] node[above] {$\scriptstyle{\XX^A_\emptyset}$} (V1t);
				\path [->] (B) edge[bend left=15] node[below] {$\scriptstyle{\XX^B_\emptyset}$} (V1b);
				\path [->] (V1t) edge[bend left=15] node[above] {$\scriptstyle{\XX^B_A}$} (B);
				\path [->] (V1b) edge[bend left=15] node[below] {$\scriptstyle{\XX^A_B}$} (A);
				\path [->] (B) edge[bend right] node[above] {$\scriptstyle{\XX^B_A}$} (V2);
				\path [->] (A) edge[bend left] node[above] {$\scriptstyle{\XX^A_B}$} (V2);
			    \path [->] (V2) edge (top);
			\end{scope}
		\end{tikzpicture}
		\caption{\mbox{Split-node generic routed graph $\cG^\text{split}_{\text{QC-QC}(2)}$.}}
        \label{fig:split-routed-graph-switch}
	\end{subfigure}%
	\hfill
    \begin{subfigure}[b]{0.6\textwidth}
        \centering
    	\begin{tikzpicture}[scale=0.78, transform shape]
    		\node (bottom) at (0,-0.25) {};
    		\node (V0) at (0,0.7) {$\mathbf{V}_1$};
    		\node[text=NavyBlue] (A) at (-3.3,4.5) {$\mathbf{A}$};
    		\node (V1A) at (-1,1.7) {$\mathbf{\hat V}^A_2$};
            \node (V1B) at (0.75,2.5) {$\mathbf{\hat V}^B_2$};
            \node (V1C) at (2.5,1.7) {$\mathbf{\hat V}^C_2$};
    		\node[text=NavyBlue] (C) at (5,4.5) {$\mathbf{C}$};
    		\node[text=NavyBlue] (B) at (1.5,4.5) {$\mathbf{B}$};
    		\node (V2BC) at (-1,7.3) {$\mathbf{\hat V}^{BC}_3$};
            \node (V2AC) at (0.75,6.5) {$\mathbf{\hat V}^{AC}_3$};
            \node (V2AB) at (2.5,7.3) {$\mathbf{\hat V}^{AB}_3$};
    		\node (V3) at (0,8.3) {$\mathbf{V}_4$};
    		\node (top) at (0,9.25) {};

            \node (i) at (0,-1.25) [align=center] {
                $\exists! \, (k_1, k_2, k_3) \, :$
                $\XX_\emptyset^{k_1} = \{ k_1 \}, \;\ \XX_{\{k_1 \}}^{k_2} = \{k_1, k_2 \},$\\[1mm]
                $ \XX_{\cN \setminus k_3}^{k_2} = \{k_1, k_2, k_3 \} \mathrel{(=} \{A,B,C\} )$,
                while all other $\XX^k_{\K_n\backslash k} = \ravnothing$};
    		
    		\begin{scope}[
    			every node/.style={fill=white,circle,inner sep=0pt}
    			every edge/.style={routedarrow}]
    			\path [->] (bottom) edge (V0);
    
    			\path [->] (V0) edge[bend left=45] node[below left] {$\scriptstyle{\XX^A_\emptyset}$} (A);
    			\path [->] (V0) edge[bend right] node[right,pos=0.15] {$\scriptstyle{\XX^B_\emptyset}$} (B);
    			\path [->] (V0) edge[bend right=40] node[below right] {$\scriptstyle{\XX^C_\emptyset}$} (C);
    
    			\begin{scope}[bend angle=15]
    				\path [->] (A) edge node[below] {$\scriptstyle{\XX^A_\emptyset}$} (V1A);
    				\path [->] (B) edge node[left] {$\scriptstyle{\XX^B_\emptyset}$} (V1B);
    				\path [->] (C) edge node[below] {$\scriptstyle{\XX^C_\emptyset}$} (V1C);
    
    				\path [->] (V1A) edge[bend left] node[left,pos=0.9,font=\scriptsize] {$\XX^B_A$} (B);
    				\path [->] (V1C) edge[bend right] node[right,pos=0.9,font=\scriptsize] {$\XX^B_C$} (B);
                    \path [->] (V1B) edge[bend right=15] node[above,pos=0.6,font=\scriptsize] {$\XX^A_B$} (A);
                    \path [->] (V1C) edge[bend left] node[above,pos=0.7,font=\scriptsize] {$\XX^A_C$} (A);
    				\path [->] (V1A) edge[bend right] node[above,pos=0.75,font=\scriptsize] {$\XX^C_A$} (C);
                    \path [->] (V1B) edge[bend left=15] node[above,pos=0.55,font=\scriptsize] {$\XX^C_B$} (C);
    
    				\path [<-] (A) edge node[above left,pos=0.7] {$\scriptstyle{\XX^A_{BC}}$} (V2BC);
    				\path [<-] (B) edge node[left] {$\scriptstyle{\XX^B_{AC}}$} (V2AC);
    				\path [<-] (C) edge node[right] {$\scriptstyle{\XX^C_{AB}}$} (V2AB);

                    \path [<-] (V2BC) edge[bend right] node[left,pos=0.9,font=\scriptsize] {$\XX^B_C$} (B);
    				\path [<-] (V2AB) edge[bend left] node[right,pos=0.9,font=\scriptsize] {$\XX^B_A$} (B);
                    \path [<-] (V2AC) edge[bend left=15] node[above,pos=0.6,font=\scriptsize] {$\XX^A_C$} (A);
                    \path [<-] (V2AB) edge[bend right] node[above,pos=0.7,font=\scriptsize] {$\XX^A_B$} (A);
    				\path [<-] (V2BC) edge[bend left] node[above,pos=0.7,font=\scriptsize] {$\XX^C_B$} (C);
                    \path [<-] (V2AC) edge[bend right=15] node[above,pos=0.55,font=\scriptsize] {$\XX^C_A$} (C);
    
    			\end{scope}
    
    			\path [->] (A) edge[bend left=45] node[above left] {$\scriptstyle{\XX^A_{BC}}$} (V3);
    			\path [->] (B) edge[bend right] node[right,pos=0.85] {$\scriptstyle{\XX^B_{AC}}$} (V3);
    			\path [->] (C) edge[bend right=40] node[above right] {$\scriptstyle{\XX^C_{AB}}$} (V3);
    
    			\path [->] (V3) edge (top);
    		\end{scope}
    	\end{tikzpicture}
        \caption{Split-node generic routed graph $\cG^\text{split}_{\text{QC-QC}(3)}.$}
        \label{fig:split-routed-graph-grenoble}
    \end{subfigure}
	\begin{subfigure}[b]{0.4\textwidth}
		\centering
		\!\!\!\!\begin{tikzpicture}[transform shape]
			\node (bottom) at (0,-0.3) {};
			\node (V0) at (0,0.5) {$\mathbf{V}_1$};
			\node[text=NavyBlue] (A) at (-2,2) {$\mathbf{A'}$};
			\node[text=NavyBlue] (B) at (2,2) {$\mathbf{B'}$};
			\node (V2) at (0,3.5) {$\mathbf{V}_3$};
			\node (top) at (0,4.3) {};
			\node (i1) at (0,-1.7) [align=center] {
                $\exists! \, (k_1, k_2) \, :$\\[1mm]
                $\XX_\emptyset^{k_1} = \{ k_1 \}, \;\ \XX_{k_1}^{k_2} = \{k_1, k_2 \} \mathrel{(=} \{A, B \}),$\\[1mm]
                while all other $\XX^k_{\K_n\backslash k} = \ravnothing$};
			\node (i2) at (0,-3) {};
			
			\begin{scope}[
				every node/.style={fill=white,circle,inner sep=0pt}
				every edge/.style=routedarrow]
				\path [->] (bottom) edge (V0);
				\path [->] (V0) edge[bend left] node[below] {$\scriptstyle{\XX^A_\emptyset}$} (A);
				\path [->] (V0) edge[bend right] node[below] {$\scriptstyle{\XX^B_\emptyset}$} (B);
				\path [->] (A) edge[bend left=25] node[below] {$\scriptstyle{\XX^B_A}$} (B);
				\path [->] (B) edge[bend left=25] node[above] {$\scriptstyle{\XX^A_B}$} (A);
				\path [->] (B) edge[bend right] node[above] {$\scriptstyle{\XX^B_A}$} (V2);
				\path [->] (A) edge[bend left] node[above] {$\scriptstyle{\XX^A_B}$} (V2);
			    \path [->] (V2) edge (top);
			\end{scope}
		\end{tikzpicture}
		\caption{$\cG^\text{split}_{\text{QC-QC}(3)}$ after merging $\mathbf{V}_2^{A/B} \hookrightarrow \mathbf{A} \lor \mathbf{B}$.}
        \label{fig:group-routed-graph-switch}
	\end{subfigure}%
	\hfill
    \begin{subfigure}[b]{0.6\textwidth}
        \centering
    	\begin{tikzpicture}[scale=0.78, transform shape]
    		\node (bottom) at (0,-0.25) {};
    		\node (V0) at (0,0.7) {$\mathbf{V}_1$};
    		\node[text=NavyBlue] (A) at (-3.3,4.5) {$\mathbf{A}'$};
    		\node[text=NavyBlue] (C) at (5,4.5) {$\mathbf{C}'$};
    		\node[text=NavyBlue] (B) at (1.5,4.5) {$\mathbf{B}'$};
    		\node (V3) at (0,8.3) {$\mathbf{V}_4$};
    		\node (top) at (0,9.25) {};
            \node (topper) at (0,10) {};

            \node (i) at (0,-1.25) [align=center] {
                $\exists! \, (k_1, k_2, k_3) \, :$
                $\XX_\emptyset^{k_1} = \{ k_1 \}, \;\ \XX_{\{ k_1 \}}^{k_2} = \{k_1, k_2 \},$\\[1mm]
                $ \XX_{\cN \setminus k_3}^{k_2} = \{k_1, k_2, k_3 \} \mathrel{(=} \{A, B, C \})$,
                while all other $\XX^k_{\K_n\backslash k} = \ravnothing$};
    		
    		\begin{scope}[
    			every node/.style={fill=white,circle,inner sep=0pt}
    			every edge/.style={routedarrow}]
    			\path [->] (bottom) edge (V0);
    
    			\path [->] (V0) edge[bend left=45] node[below left] {$\scriptstyle{\XX^A_\emptyset}$} (A);
    			\path [->] (V0) edge[bend right] node[right,pos=0.15] {$\scriptstyle{\XX^B_\emptyset}$} (B);
    			\path [->] (V0) edge[bend right=40] node[below right] {$\scriptstyle{\XX^C_\emptyset}$} (C);
    
    			\begin{scope}[bend angle=25]
    				\path [->] (A) edge[bend left] node[below,font=\scriptsize] {$\XX^A_{C} \, \XX^B_A$} (B);
    				\path [->] (C) edge[bend left] node[above,font=\scriptsize] {$\XX^B_{A} \, \XX^C_B$} (B);
                    \path [->] (B) edge[bend left] node[above,font=\scriptsize] {$\XX^B_{C} \, \XX^A_B$} (A);
                    \path [->] (C) edge[bend right=45] node[below,pos=0.6,font=\scriptsize] {$\XX^C_{B} \, \XX^A_C$} (A);
    				\path [->] (A) edge[bend right=45] node[above,pos=0.4,font=\scriptsize] {$\XX^A_{B} \, \XX^C_A$} (C);
                    \path [->] (B) edge[bend left] node[below,font=\scriptsize] {$\XX^B_{A} \, \XX^C_B$} (C);    
    			\end{scope}
    
    			\path [->] (A) edge[bend left=45] node[above left] {$\scriptstyle{\XX^A_{BC}}$} (V3);
    			\path [->] (B) edge[bend right] node[right,pos=0.85] {$\scriptstyle{\XX^B_{AC}}$} (V3);
    			\path [->] (C) edge[bend right=40] node[above right] {$\scriptstyle{\XX^C_{AB}}$} (V3);
    
    			\path [->] (V3) edge (top);
    		\end{scope}
    	\end{tikzpicture}
        \caption{$\cG^\text{split}_{\text{QC-QC}(3)}$ after merging $\mathbf{V}_2^{A} \hookrightarrow \mathbf{A} \hookleftarrow \mathbf{V}_3^{BC}$ etc.}
        \label{fig:group-routed-graph-grenoble}
    \end{subfigure}
	\label{fig:group-routed-graphs}
	\caption{{\bf (a,b)} The split-node generic routed graph $\cG^\text{split}_{\text{QC-QC}(N)}$ (for $N=2,3$) according to \cref{def:internal-split}, splitting the isometries $\mathbf{V}_{n+1}$ featured in the generic routed graph $\cG_{\text{QC-QC}(N)}$ shown in \cref{fig:routed_graphs_QCQC_N_2_3} into one node $\mathbf{\hat V}_{n+1}^{\K_n}$ for each branch $\mathbf{\hat V}_{n+1}^{\K_n}$.\quad 
    {\bf(c,d)} The routed graphs obtained by merging internal isometries in $\cG^\text{split}_{\text{QC-QC}(N)}$ into agent nodes, denoted by $\hookrightarrow$, following the procedure established in \cref{sec:grouping}. We obtain a reduced number of internal nodes, 
    each of which is associated with applying not only their respective agent operations $A_k$, but also additional non-trivial \enquote{internal} pre- and post-processing encoding the action of the supermap.
    We thereby essentially recover the routed graphs shown in Figs.~6 and~21 of~\cite{Vanrietvelde2023}.
    }
	\label{fig:split-routed-graphs}
\end{figure}

The routed graphs $\mathcal{G}^\text{split}_{\textup{QC-QC}(N)}$ for $N=2$ and $N=3$ are shown in \cref{fig:split-routed-graph-switch,fig:split-routed-graph-grenoble}.

The branches $\textbf{A}_k^{\K_n\backslash k}$ of the $\mathbf{A}$-nodes maintain the form given in \cref{eq:branches_Ak}, albeit with a larger number of indices (stemming from the larger number of $\mathbf{\hat V}$-nodes) taking $\ravnothing$ as value.
For the $\mathbf{\hat V}$-nodes, we obtain a branch structure
$\Bran( \lambda_{\mathbf{\hat V^{\K_n}_{n+1}}} ) = \{ \mathbf{V}^{\K_n}_{n+1} , \mathbf{\bar V}^{\K_n}_{n+1} \}$, encoding precisely whether a given branch $\mathbf{V}^{\K_n}_{n+1}$ happens or not, with $\mathbf{V}^{\K_n}_{n+1}$ of the same form as in \cref{eq:branches_Vn1,eq:branches_V1_VN1} and with $\mathbf{\bar V}^{\K_n}_{n+1} = \big(\{(\ravnothing, \ldots, \ravnothing)\}\to\{(\ravnothing, \ldots, \ravnothing)\}\big)$
Considering the index values $\ravnothing$ to be 1-dimensional, we may then construct the branch graph for $\cG^\text{split}_{\text{QC-QC}(N)}$. 
It is closely related to the branch graph for $\cG$, with the addition of the branches $\mathbf{\bar V}^{\K_n}_{n+1}$ being subject of no strong parent/child relations as well as the same weak parent/child relations as $\mathbf{V}^{\K_n}_{n+1}$.
For the instance of $N=2$, the respective branch graph is shown in \cref{fig:split-branch-graph}. (A general comment on how node splitting affects the branch graph is presented in \cref{fig:half-switch-split}.)

\begin{figure}
    \centering
    \begin{tikzpicture}
        \node (V0) at (0,1) {$\mathbf{V}_1^\emptyset$};
        \node (A1) at (-2,3) {$\mathbf{A}^\emptyset$};
        \node (B1) at (2,3) {$\mathbf{B}^\emptyset$};
        \node (V1A) at (-2,5) {$\mathbf{V}_2^A$};
        \node[text=gray] (V1Abar) at (-3.2,5) {$\mathbf{\bar V}_2^A$};
        \node (V1B) at (2,5) {$\mathbf{V}_2^B$};
        \node[text=gray] (V1Bbar) at (3.2,5) {$\mathbf{\bar V}_2^B$};
        \node (B2) at (-2,7) {$\mathbf{B}^A$};
        \node (A2) at (2,7) {$\mathbf{A}^B$};
        \node (V2) at (0,9) {$\mathbf{V}_3^{AB}$};

        \path [->] (V0) edge (A1);
        \path [->] (A1) edge (V1A);
        \path [->] (V1A) edge (B2);
        \path [->] (B2) edge (V2);
        \path [->] (V0) edge (B1);
        \path [->] (B1) edge (V1B);
        \path [->] (V1B) edge (A2);
        \path [->] (A2) edge (V2);

        \begin{scope}[
            >={Stealth[gray]},
            every edge/.style={draw=green,dashed,very thick}]
            \path [->] (V0) edge[bend left] (A1);
            \path [->] (V0) edge (V1A);
            \path [->] (V0) edge[in=270,out=180] (V1Abar);
            \path [->] (V0) edge (B2);
            \path [->] (V0) edge[bend right] (B1);
            \path [->] (V0) edge[bend left] (V1B);
            \path [->] (V0) edge[in=270,out=0] (V1Bbar);
            \path [->] (V0) edge[bend left] (A2);
        \end{scope}

        \begin{scope}[
            >={Stealth[gray]},
            every edge/.style={draw=red,dashed,very thick}]
            \path [->] (A1) edge[bend right] (V2);
            \path [->] (V1A) edge[bend right] (V2);
            \path [->] (V1Abar) edge[in=180,out=90] (V2);
            \path [->] (B2) edge[bend left] (V2);
            \path [->] (B1) edge (V2);
            \path [->] (V1B) edge (V2);
            \path [->] (V1Bbar) edge[in=0,out=90] (V2);
            \path [->] (A2) edge[bend right] (V2);
        \end{scope}
    \end{tikzpicture}
    \caption{
        The branch graph for the split-node generic routed graph $\cG^{\text{split}}_{\text{QC-QC}(2)}$ according to \cref{def:internal-split}, featuring additional branches $\mathbf{\bar V}_{n+1}^{\K_n}$ in addition to those present in the branch graph of $\cG_{\text{QC-QC}(2)}$ shown in \cref{fig:branch-graphs}.}
    \label{fig:split-branch-graph}
\end{figure}
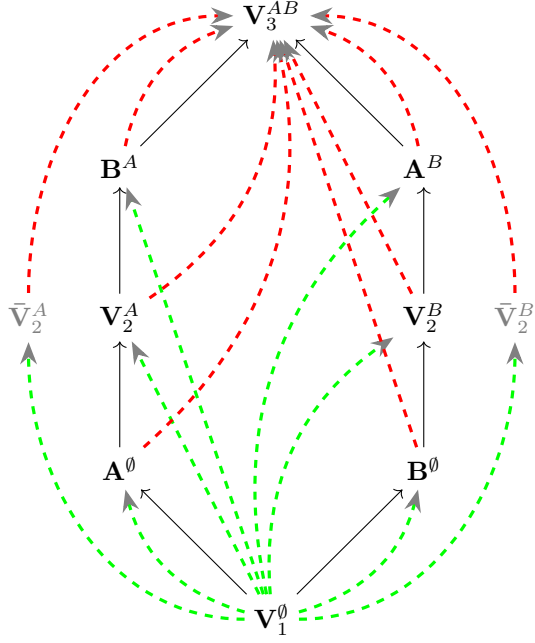

\begin{restatable}{lemma}{genericSplitGraph}
    \label{thm:split-graph}
    Considering the index values $\ravnothing$ to be 1-dimensional, the split-node generic routed graph $\cG^\text{split}_{\text{QC-QC}(N)}$ is a valid routed graph.
\end{restatable}

The proof of this statement is given in \cref{sec:proof5}.

The larger number of nodes leads to an associated skeletal supermap of a different type, featuring a larger number of slots.
As however all slots for $\mathbf{\hat V}$-nodes are fleshed out by routed isometries rather than routed supermaps, this need not impact the type of the supermap we ultimately obtain.

We continue by showing that indeed, both routed graphs generate the same class of processes by means of routed circuit decompositions.
For each decomposition for $\cG_{\text{QC-QC}(N)}$, we can specify a closely related equivalent decomposition for $\cG^\text{split}_{\text{QC-QC}(N)}$.

\begin{restatable}{lemma}{genericSplitEquivalence}
    \label{thm:split-equivalence}
    Consider the two routed graphs $\cG_{\text{QC-QC}(N)}$ and $\cG^\text{split}_{\text{QC-QC}(N)}$.
    Restricting to the case in which all $\mathbf{V}_n$ and $\mathbf{V}^{\K_n}_{n+1}$ are fleshed out by routed isometries,
    there is a canonical mapping between equivalent routed circuit decompositions (as defined in \cref{def:decomposition}) for the two routed graphs.
\end{restatable}

The proof of this statement is given in \cref{sec:proof5}, using an approach that readily generalises to a variety of routed graphs and nodes.

    For the skeletal supermap associated with $\cG^\text{split}_{\text{QC-QC}(N)}$, the Hilbert spaces associated with the arrows going in/out of these branch nodes then factorise as
    \begin{equation}
        \label{eq:fine-grained-practical}
        \HH^{\text{in}(\mathbf{\hat V}^{\K_n}_{n+1})}_\text{prac}
        = \bigoplus_{k=1}^N \,\HH_{k}^{\mathbf{A}_k\to\mathbf{\hat V}^{\K_n}_{n+1}}
        \oplus \HH_\ravnothing^{\text{in}(\mathbf{\hat V}^{\K_n}_{n+1})}
    \end{equation}
    where
    $\HH_{k}^{\mathbf{A}_k\to\mathbf{\hat V}^{\K_n}_{n+1}} \coloneqq \HH_{\K_n\backslash k,k}^{\mathbf{A}_k\to\mathbf{\hat V}_{n+1}}$ and
    $\HH_\ravnothing^{\text{in}(\mathbf{\hat V}^{\K_n}_{n+1})} \coloneqq \bigotimes_{k=1}^N \,\HH_\ravnothing^{\mathbf{A}_k\to\mathbf{\hat V}^{\K_n}_{n+1}}$, and similarly for $\HH^{\text{out}(\mathbf{\hat V}^{\K_n}_{n+1})}_\text{prac}$.

    To obtain a given QC-QC from this alternative routed graph instead of $\cG_{\text{QC-QC}(N)}$, we split the fleshing out for the $\mathbf{V}$-nodes along the same lines, obtaining
    $\dket{J_{\mathbf{\hat V}^{\K_n}_{n+1}}^\text{in}} \ast \dket{\hat{V}_{n+1}^{\K_n}} \ast \dket{J_{\mathbf{\hat V}^{\K_n}_{n+1}}^\text{out}}$, where for $n=1,\ldots,N$: 
    \begin{align}
        J_{\mathbf{V}^{\K_n}_{n+1}}^\text{in} : \ \HH^{\text{in}(\mathbf{\hat V}^{\K_n}_{n+1})}_\text{prac}
        = \bigoplus_{k=1}^N \,\HH_{k}^{\mathbf{A}_k\to\mathbf{\hat V}^{\K_n}_{n+1}}
        \oplus \HH_\ravnothing^{\text{in}(\mathbf{\hat V}^{\K_n}_{n+1})}
        \ \to & \quad
        \HH^{\hat A_n^O}\otimes\HH^{\alpha_n}\otimes\HH^{C_n'}
        \oplus \HH^{\text{in}(\mathbf{\hat V}^{\K_n}_{n+1})}_\ravnothing, \notag \\
        \ket{\psi} \in \HH_{k}^{\mathbf{A}_k\to\mathbf{\hat V}^{\K_n}_{n+1}}
        \ \mapsto & \quad
        \ket{\psi}^{\hat A_n^O\alpha_n}\otimes\ket{\K_n\backslash k,k}^{C_n'}, \\
        \ket\Omega \in \HH_\ravnothing^{\text{in}(\mathbf{\hat V}^{\K_n}_{n+1})} \ \mapsto & \quad \ket\Omega \, ,
        \label{eq:vacuum-isometry}
    \end{align}
    where $\ket{\Omega}$ is the single state in the 1-dimensional sector
    $\HH_\ravnothing^{\text{in}(\mathbf{\hat V}^{\K_n}_{n+1})}$,
    and similarly define $\dket{J_{\mathbf{\hat V}^{\K_n}_{n+1}}^\text{out}}$,
    while 
    \begin{equation}
        \dket{\hat{V}_{n+1}^{\K_n}} = 
        \dket{\tilde{V}_{n+1}^{\K_n}} \oplus \dket{\bar{V}_{n+1}^{\K_n}}
        =
        \!\!
	    \sum_{k \in \K_n, \ell \notin \K_n}  \!\!\!\!\!\! \dket{\tilde V^{\to \ell}_{\K_{n} \setminus k,k}}^{\tilde A_n^O \alpha_n \tilde A_{n+1}^I \alpha_{n+1}}
        \oplus \id^{\mathrm{in}(\mathbf{\hat V}^{\K_n}_{n+1})\to\mathrm{out}(\mathbf{\hat V}^{\K_n}_{n+1})}_{\mathbf{\bar V}^{\K_n}_{n+1}} \, .
    \end{equation}

While in principle, we could apply an analogous technique for the routed graph $\cG_{\text{QC-QC}(N)}^\alpha$, we would obtain a large number of arrows in the resulting graph in this case, connecting $\mathbf{\hat V}^{\K_{n} \setminus k}_n$ and $\mathbf{\hat V}^{\K_{n}}_{n+1}$ respectively.

Moreover, an analogous rewriting could be performed for the agent nodes $\mathbf{A}_k$, obtaining nodes $\mathbf{A}^{\K_n}_k$ for each $\K_n \subseteq \cN \setminus k$.
However, fleshing out the resulting slots for $\mathbf{A}^{\K_n}_k$ with monopartite supermaps (in the sense of a routed circuit decomposition according to \cref{def:decomposition}) would produce a supermap of a different type (than for $\cG_{\text{QC-QC}(N)}$), increasing the number of slots within the routed supermap from $N$ to $2^N$.
Furthermore, the resulting slot spaces are in direct sum form, including 1-dimensional Hilbert spaces $\HH^{\text{in/out}( \mathbf{\hat A}^{\K_n}_\ravnothing)}$, preventing us from obtaining a unsectorised slot of non-trivial dimension by factorising.
Therefore, no fleshing out for such $\mathbf{\hat A}$-nodes could recover a meaningful unrouted supermap.

Furthermore, one may wonder whether the idea of splitting a node according to its branch structure can be generalised to arbitrary routed graphs.
Indeed, this seems possible in many cases, first replacing a given node $\mathbf{N}$ with a copy $\mathbf{\hat N}^\beta$ for each of its branches $\mathbf{N}^\beta$, endowed with a copy of the respective arrows, and then adjusting their indices such that for each $\mathbf{\hat N}^\beta$, we have $\lambda_{\mathbf{\hat N^\beta}} = \lambda_{\mathbf{N}^\beta} \oplus \id_{\mathbf{\bar N}^\beta}$.
Here, the identity acts on a 1-dimensional subspace corresponding to the case that $\mathbf{N}^\beta$ does not happen.
A simple example for this is illustrated in \cref{fig:half-switch}.
However, in cases where the branch structure is dictated by correlating multiple incoming indices\footnote{
    Consider e.g.\ a routed subgraph $\mathbf{A} \xrightarrow{i} \mathbf{X} \xleftarrow{j} \mathbf{B}$ with $i, j \in \{ 1, 2 \}$ and no 1-dimensional index values, where $\Bran(\lambda_\mathbf{X}) = \{ \mathbf{X}^\gamma, \mathbf{X}^{\delta} \}$, $\Ind^\text{in}_{\mathbf{X}^\gamma} = \{ (1,1); (2,2) \}$ and $\Ind^\text{in}_{\mathbf{X}^{\delta}} = \{ (1,2); (2,1) \}$.}
(and similarly for outgoing indices), there seems to be no straightforward way to adjust the index structure accordingly, as this hinders an association of individual arrow index values to a given branch.
The precise conditions under which the index structure can be adjusted like this and a generic procedure to do so remain important open questions, especially for characterising a minimal set of routed circuit decompositions to obtain all (routed) processes obtainable from routed quantum circuits.

\begin{figure}
	\centering
    \begin{subfigure}[b]{0.17\textwidth}
		\centering
		\begin{tikzpicture}[transform shape]
			\node (bottom) at (0,-0.3) {};
			\node (V0) at (0,0.5) {$\mathbf{V}$};
			\node (A) at (0,2) {$\mathbf{N}$};
			\node (V2) at (0,3.5) {$\mathbf{W}$};
			\node (top) at (0,4.3) {};
			\node (i1) at (0,-0.8) [align=center] {
                $i+j+k=1$};
			
			\begin{scope}[
				every node/.style={fill=white,circle,inner sep=0pt}
				every edge/.style=routedarrow]
				\path [->] (bottom) edge (V0);
				\path [->] (V0) edge node[right] {$\scriptstyle{i\,j\,k}$} (A);
				\path [->] (A) edge node[right] {$\scriptstyle{i\,j\,k}$} (V2);
			    \path [->] (V2) edge (top);
			\end{scope}
		\end{tikzpicture}
		\caption{Before splitting}
        \label{fig:half-switch-condensed}
	\end{subfigure}%
    \hfill
    \begin{subfigure}[b]{0.28\textwidth}
		\centering
		\begin{tikzpicture}[transform shape]
			\node (bottom) at (0,-0.3) {};
			\node (V0) at (0,0.5) {$\mathbf{V}$};
			\node (A) at (-1.5,2) {$\mathbf{\hat N}^{i=1}$};
			\node (B) at (0,2) {$\mathbf{\hat N}^{j=1}$};
            \node (C) at (1.5,2) {$\mathbf{\hat N}^{k=1}$};
			\node (V2) at (0,3.5) {$\mathbf{W}$};
			\node (top) at (0,4.3) {};
			\node (i1) at (0,-0.8) [align=center] {
                $i+j+k=1$};
			
			\begin{scope}[
				every node/.style={fill=white,circle,inner sep=0pt}
				every edge/.style=routedarrow]
				\path [->] (bottom) edge (V0);
				\path [->] (V0) edge[bend left] node[below] {$\scriptstyle{i}$} (A);
				\path [->] (V0) edge node[right] {$\scriptstyle{j}$} (B);
                \path [->] (V0) edge[bend right] node[below] {$\scriptstyle{k}$} (C);
				\path [->] (C) edge[bend right] node[above] {$\scriptstyle{k}$} (V2);
                \path [->] (B) edge node[right] {$\scriptstyle{j}$} (V2);
				\path [->] (A) edge[bend left] node[above] {$\scriptstyle{i}$} (V2);
			    \path [->] (V2) edge (top);
			\end{scope}
		\end{tikzpicture}
		\caption{After splitting}
        \label{fig:half-switch-split}
	\end{subfigure}%
    \hfill
    \begin{subfigure}[b]{0.55\textwidth}
        \centering
        \begin{tikzpicture}
            \node (left) at (-4,3) {};
            \node (V0) at (0,1) {$\mathbf{V}$};
            \node (A1) at (-2.8,3) {$\mathbf{N}^{i=1}$};
            \node[text=gray] (A1bar) at (-1.4,3) {$\mathbf{\bar N}^{i=1}$};
            \node (A2) at (0,3) {$\mathbf{N}^{j=1}$};
            \node[text=gray] (A2bar) at (1.4,3) {$\mathbf{\bar N}^{j=1}$};
            \node (A3) at (2.7,3) {$\mathbf{N}^{k=1}$};
            \node[text=gray] (A3bar) at (4,3) {$\mathbf{\bar N}^{k=1}$};
            \node (V2) at (0,5) {$\mathbf{W}$};

            \node (i1) at (0,0) {};

            \path [->] (V0) edge (A1);
            \path [->] (V0) edge (A2);
            \path [->] (V0) edge (A3);
            \path [->] (A3) edge (V2);
            \path [->] (A2) edge (V2);
            \path [->] (A1) edge (V2);

            \begin{scope}[
    			>={Stealth[gray]},
    			every edge/.style={draw=green,dashed,very thick}]
                \path [->] (V0) edge[bend left] (A1);
                \path [->] (V0) edge (A1bar);
                \path [->] (V0) edge[bend left] (A2);
                \path [->] (V0) edge (A2bar);
                \path [->] (V0) edge[bend right] (A3);
                \path [->] (V0) edge[bend right] (A3bar);
            \end{scope}

            \begin{scope}[
    			>={Stealth[gray]},
    			every edge/.style={draw=red,dashed,very thick}]
                \path [->] (A1) edge[bend left] (V2);
                \path [->] (A1bar) edge (V2);
                \path [->] (A2) edge[bend right] (V2);
                \path [->] (A2bar) edge (V2);
                \path [->] (A3) edge[bend right] (V2);
                \path [->] (A3bar) edge[bend right] (V2);
            \end{scope}
        \end{tikzpicture}
        \caption{Composite picture of the branch graphs}
        \label{fig:half-switch-branch-graph}
    \end{subfigure}%
    \caption{
        A general example for a node-splitting procedure, obtaining the routed graph shown in {\bf (b)} from the routed graph shown in {\bf (a)}.
        While for the latter, we have $\Bran(\lambda_\mathbf{N}) = \{ \mathbf{N}^{i=1}, \mathbf{N}^{j=1}, \mathbf{N}^{k=1} \}$, each of these branches is represented by an individual node $\mathbf{\hat N^\beta}$ with $\Bran(\lambda_\mathbf{\hat N^\beta}) = \{ \mathbf{\bar N}^{\beta}, \mathbf{N}^{\beta}\}$ for the former.
        While in this example, $\mathbf{N}$ has only a single parent / child, this is not generally required this to be the case.
        \newline
        \indent\hspace{1em} In {\bf (c)}, a composite representation of the branch graph of both circuits are shown: They are identical, except that after the node splitting we obtain an additional branch $\mathbf{\bar N}^\beta$ complementing each original branch $\mathbf{N}^\beta$.
        It has the same weak parents as $\mathbf{N}^\beta$, and hence, features the same green (and, analogously, red) arrows. Accordingly, this branch does not add any non-trivial information to the branch graph and does not affect the presence of weak loops.
        Note that when considering the associated supermap, whether we flesh out $\mathbf{N}$ or the nodes $\mathbf{\hat N}$ with monopartite supermaps amounts to different supermaps:
        For {\bf (b)}, we obtain a single agent performing a (potentially routed) operation, while for {\bf (a)} we obtain three agents performing a (potentially routed) operation, dependent on the specific fleshing out.
        By contrast, if we flesh out these nodes with isometries (i.e.\ with no slots), the resulting maps are identical.
    }
    \label{fig:half-switch}

\end{figure}
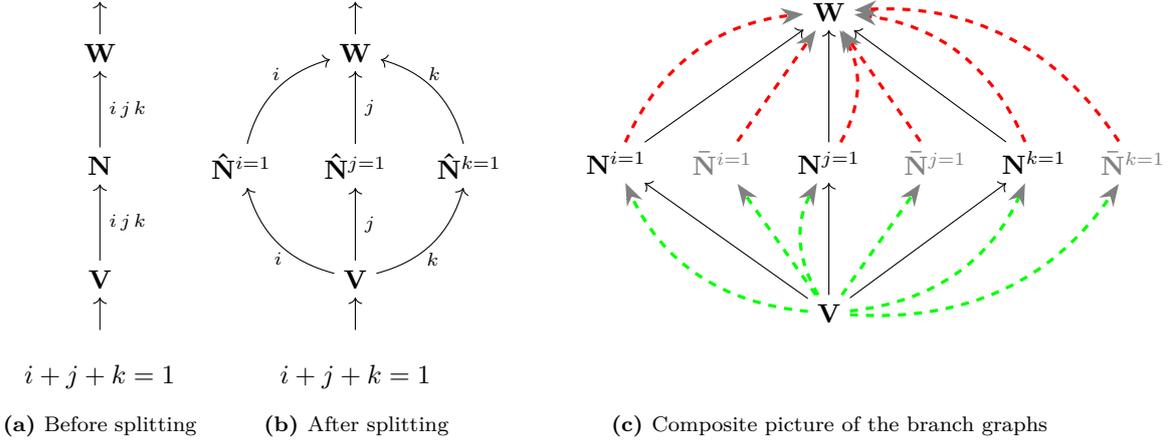

\subsection{Merging nodes in the routed graph}
\label{sec:grouping}

If we compare our construction of $\cG_{\text{QC-QC}(N)}$ (as well as $\cG^\alpha_{\text{QC-QC}(N)}$ and $\cG^\text{split}_{\text{QC-QC}(N)}$) with the routed graphs described in~\cite{Vanrietvelde2023} for the examples of the quantum switch (shown there in Fig.~6) and the Grenoble process (in Fig.~16), one immediately notices that the latter have fewer nodes than those that we provide for the same scenarios.
Indeed, while we can identify our nodes $\mathbf{V}_{1}$ and $\mathbf{V}_{N+1}$ with their nodes $P$ and $F$, respectively, the other nodes in the routed graphs of~\cite{Vanrietvelde2023} correspond more naturally to the agent nodes $\mathbf{A}_k$, and there is no clear equivalent for the remaining \enquote{internal} $\mathbf{V}$-nodes of our representation.

In this section, we will demonstrate that starting from the split-node generic routed graph
$\cG^\text{split}_{\text{QC-QC}(N)}$, we indeed recover their routed graphs and routed circuit decompositions as alternative representations -- albeit at the cost of more complex fleshing outs for the remaining nodes.
Indeed, by merging internal nodes into agent nodes, one can effectively recast the isometries as pre- and post-processing for the respective nodes.
However, we will find that a decomposition with no internal nodes other than $P$ and $F$ is generally not possible for generic QC-QCs when $N>3$.

In order to recover the routed circuits of~\cite{Vanrietvelde2023} from our approach, we will first introduce a slightly more general procedure, demonstrating that the idea of node merging is applicable to a wider variety of routed graphs, before applying it to the generic routed graph which forms the focus of our study.

\begin{figure}
    \begin{subfigure}{0.28\textwidth}
        \centering
        \begin{tikzpicture}[scale=0.9, transform shape]
            \node (V) at (0,4.5) {$\textbf{A}$};
            \node (Vpal) at (-1.5,3) {$\cdots$};
            \node (Vpar) at (1.5,3) {$\cdots$};
            \node (Vchl) at (-1,6) {$\cdots$};
            \node (Vchr) at (1,6) {$\cdots$};
            \node (A) at (0,2.7) {\textbf{V}};
            \node (Apal) at (-1,1.5) {$\cdots$};
            \node (Apar) at (1,1.5) {$\cdots$};

            \begin{scope}[
                every node/.style={fill=white,circle}
                every edge/.style=routedarrow]

                \path [->] (V) edge node[below left] {\scriptsize $\cdot$} (Vchl);
                \path [->] (V) edge node[below right] {\scriptsize $\cdot$} (Vchr);
                \path [<-] (V) edge[bend right] node[above left] {\scriptsize $\cdot$} (Vpal);
                \path [<-] (V) edge node[right] {\footnotesize $k$} (A);
                \path [<-] (V) edge[bend left] node[above right] {\scriptsize $\cdot$} (Vpar);
                \path [->] (Apal) edge node[below right] {\scriptsize $\cdot$} (A);
                \path [->] (Apar) edge node[below left] {\scriptsize $\cdot$} (A);
            \end{scope}
        \end{tikzpicture}
        \caption{Routed graph $\cG$ with \textbf{A} as sole child of \textbf{V}  (cf.~\cref{thm:var-split}).}
        \label{fig:merge-one-before}
    \end{subfigure}
    \hspace{0.04\textwidth}
    \begin{subfigure}{0.32\textwidth}
        \centering
        \begin{tikzpicture}[scale=0.9, transform shape]
            \node (A1) at (-1,4.5) {$\textbf{A}^\beta$};
            \node (A2) at (1.25,4.5) {$\mathbf{A}\mkern-1mu^\gamma$};
            
            \node (V0) at (-1.75,2.7) {$\textbf{V}^\delta$};
            \node (V1) at (-0.25,2.7) {$\textbf{V}^\varepsilon$};
            \node (V2) at (1.25,2.7) {$\textbf{V}^\eta$};
            
            \node (top) at (-1,6) {};
            \node (bottom) at (-1,1.5) {};

            \path [->] (V0) edge (A1);
            \path [->] (V1) edge (A1);
            \path [->] (V2) edge (A2);

            \begin{scope}[
            >={Stealth[gray]},
            every edge/.style={draw=red,dashed,very thick}]
                \path [->] (V0) edge[bend left] (A1);
                \path [->] (V1) edge[bend right] (A1);
            \end{scope}
        \end{tikzpicture}
        \caption{Schematic excerpt of a branch graph conforming with \cref{thm:var-split}.}
        \label{fig:merge-one-branch}
    \end{subfigure}
    \hspace{0.04\textwidth}
    \begin{subfigure}{0.28\textwidth}
        \centering
        \begin{tikzpicture}[scale=0.9, transform shape]
            \node (V) at (0,4.5) {$\mathbf{A}'$};
            \node (Vpal) at (-1.5,3) {$\cdots$};
            \node (Vpar) at (1.5,3) {$\cdots$};
            \node (Vchl) at (-1,6) {$\cdots$};
            \node (Vchr) at (1,6) {$\cdots$};
            \node (Apal) at (-1,1.5) {$\cdots$};
            \node (Apar) at (1,1.5) {$\cdots$};

            \begin{scope}[
                every node/.style={fill=white,circle}
                every edge/.style=routedarrow]

                \path [->] (V) edge node[below left] {\scriptsize $\cdot$} (Vchl);
                \path [->] (V) edge node[below right] {\scriptsize $\cdot$} (Vchr);
                \path [<-] (V) edge[bend right] node[above left] {\scriptsize $\cdot$} (Vpal);
                \path [<-] (V) edge[bend left] node[above right] {\scriptsize $\cdot$} (Vpar);
                \path [->] (Apal) edge node[below right] {\scriptsize $\cdot$} (V);
                \path [->] (Apar) edge node[below left] {\scriptsize $\cdot$} (V);
            \end{scope}
        \end{tikzpicture}
        \caption{Routed graph $\cG^{\text{merge}\uparrow}$ after merging \textbf{V} into \textbf{A}.}
        \label{fig:merge-one-complete}
    \end{subfigure}
    \par\bigskip
    \begin{subfigure}{0.28\textwidth}
        \centering
        \begin{tikzpicture}[scale=0.9, transform shape]
            \node (V) at (0,4.5) {$\textbf{V}$};
            \node (Vpal) at (-1.5,4) {$\cdots$};
            \node (Vpar) at (1.5,4) {$\cdots$};
            \node (Vchl) at (-1,6) {$\cdots$};
            \node (Vchr) at (1,6) {$\cdots$};
            \node (A) at (0,2.7) {\textbf{A}};
            \node (Apal) at (-1,1.5) {$\cdots$};
            \node (Apar) at (1,1.5) {$\cdots$};

            \begin{scope}[
                every node/.style={fill=white,circle}
                every edge/.style=routedarrow]

                \path [->] (V) edge node[below left] {\scriptsize $\cdot$} (Vchl);
                \path [->] (V) edge node[below right] {\scriptsize $\cdot$} (Vchr);
                \path [->] (A) edge[bend left] node[below left] {\scriptsize $\cdot$} (Vpal);
                \path [->] (A) edge node[right] {\footnotesize $k$} (V);
                \path [->] (A) edge[bend right] node[below right] {\scriptsize $\cdot$} (Vpar);
                \path [->] (Apal) edge node[below right] {\scriptsize $\cdot$} (A);
                \path [->] (Apar) edge node[below left] {\scriptsize $\cdot$} (A);
            \end{scope}
        \end{tikzpicture}
        \caption{Routed graph $\cG$ with \textbf{A} as sole parent of \textbf{V} (cf.~\cref{thm:var-split-two}).}
        \label{fig:merge-two-before}
    \end{subfigure}
    \hspace{0.04\textwidth}
    \begin{subfigure}{0.32\textwidth}
        \centering
        \begin{tikzpicture}[scale=0.9, transform shape]
            \node (A1) at (-1,2.7) {$\textbf{A}^\beta$};
            \node (A2) at (1.25,2.7) {$\mathbf{A}\mkern-1mu^\gamma$};
            
            \node (V0) at (-1.75,4.5) {$\textbf{V}^\delta$};
            \node (V1) at (-0.25,4.5) {$\textbf{V}^\varepsilon$};
            \node (V2) at (1.25,4.5) {$\textbf{V}^\eta$};
            
            \node (top) at (-1,6) {};
            \node (bottom) at (-1,1.5) {};

            \path [<-] (V0) edge (A1);
            \path [<-] (V1) edge (A1);
            \path [<-] (V2) edge (A2);

            \begin{scope}[
            >={Stealth[gray]},
            every edge/.style={draw=green,dashed,very thick}]
                \path [<-] (V0) edge[bend right] (A1);
                \path [<-] (V1) edge[bend left] (A1);
            \end{scope}
        \end{tikzpicture}
        \caption{Schematic excerpt of a branch graph conforming with \cref{thm:var-split-two}.}
        \label{fig:merge-two-branch}
    \end{subfigure}
    \hspace{0.04\textwidth}
    \begin{subfigure}{0.32\textwidth}
        \centering
        \begin{tikzpicture}[scale=0.9, transform shape]
            \node (Vpal) at (-1.5,4) {$\cdots$};
            \node (Vpar) at (1.5,4) {$\cdots$};
            \node (Vchl) at (-1,6) {$\cdots$};
            \node (Vchr) at (1,6) {$\cdots$};
            \node (A) at (0,2.7) {$\mathbf{A}'$};
            \node (Apal) at (-1,1.5) {$\cdots$};
            \node (Apar) at (1,1.5) {$\cdots$};

            \begin{scope}[
                every node/.style={fill=white,circle}
                every edge/.style=routedarrow]

                \path [->] (A) edge node[below left] {\scriptsize $\cdot$} (Vchl);
                \path [->] (A) edge node[below right] {\scriptsize $\cdot$} (Vchr);
                \path [->] (A) edge[bend left] node[below left] {\scriptsize $\cdot$} (Vpal);
                \path [->] (A) edge[bend right] node[below right] {\scriptsize $\cdot$} (Vpar);
                \path [->] (Apal) edge node[below right] {\scriptsize $\cdot$} (A);
                \path [->] (Apar) edge node[below left] {\scriptsize $\cdot$} (A);
            \end{scope}
        \end{tikzpicture}
        \caption{Routed graph $\cG^{\text{merge}\downarrow}$ after merging \textbf{V} into \textbf{A}.}
        \label{fig:merge-two-complete}
    \end{subfigure}
    
    \caption{
        Routed and branch graph architectures which allow for node merging according to \cref{thm:merge-one}.
        For the branch graph excerpts, we only show the most relevant arrows for merging, which will no longer be featured in the branch graph obtained afterwards.
        In the associated branch graph, we observe that each branch of $\mathbf{V}$ is associated with exactly one branch of $\mathbf{A}$.
        For details regarding the derivation of the branch graph, see the proof of \cref{thm:var-split} in \cref{sec:proof5}.
    }
    \label{fig:merge}
\end{figure}
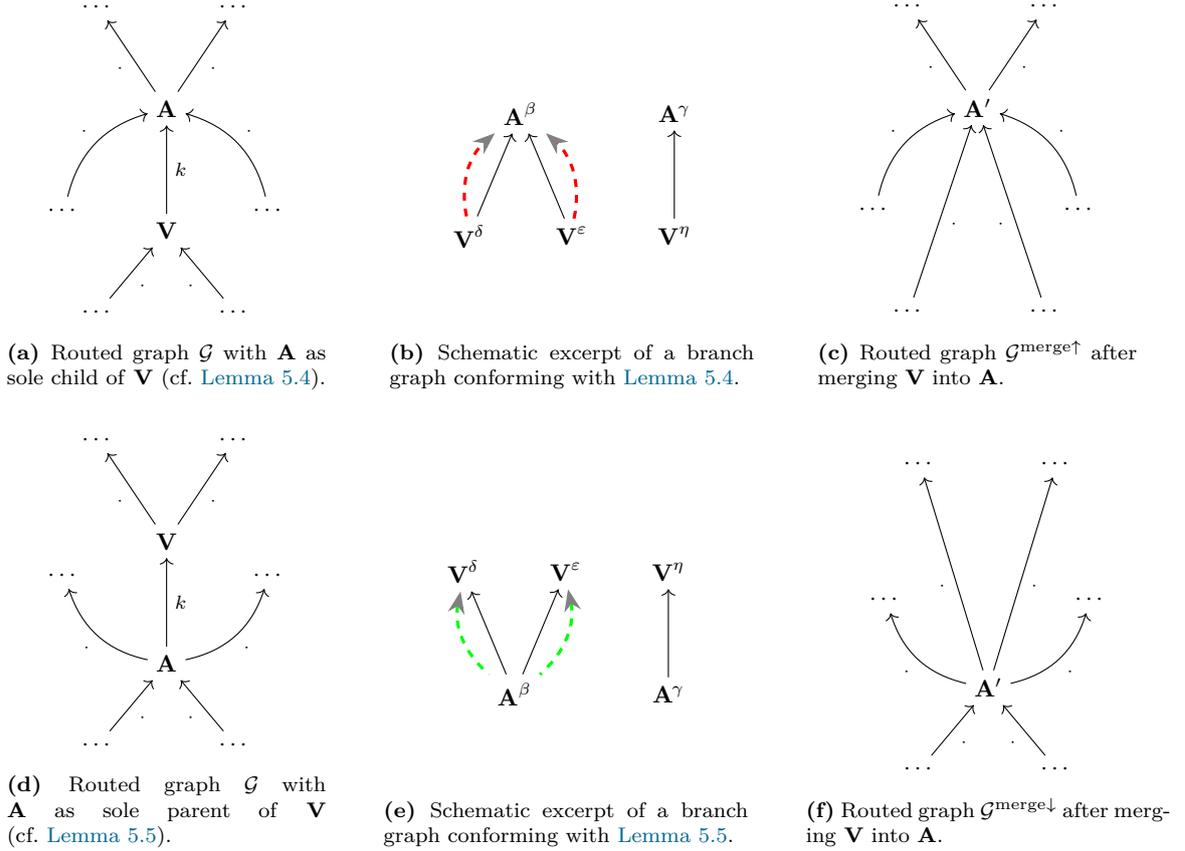

\begin{restatable}{lemma}{generalMergeOne}
    \label{thm:var-split}
    Let $\cG = (\Gamma, (\lambda_\mathbf{N})_{\mathbf{N}\in \mathtt{Nodes}_\Gamma})$ be a valid routed graph without any weak loops,
    and $\mathbf{A}, \mathbf{V} \in \mathtt{Nodes}_\Gamma$ such that $\mathbf{A}$ is the sole child of $\mathbf{V}$ (cf.\ \cref{fig:merge-one-before}),
    and the associated branch structure is such that
    \begin{equation}
        \forall \, \mathbf{A}^\beta \in \Bran(\lambda_\mathbf{A}) \quad
        \exists! \, \mathbf{V}^\gamma \in \Bran(\lambda_{\mathbf{V}}) : \quad
        \mathtt{LinkVal} ( \mathbf{V}^\gamma , \mathbf{A}^\beta ) \neq \emptyset \label{eq:merge-condition} 
    \end{equation}
    Consider another routed graph $\cG^{\text{merge}\uparrow}$ obtained by merging $\mathbf{V}$ into $\mathbf{A}$, as follows:
    \begin{itemize}
        \item the two nodes being replaced by a single node $\mathbf{A}'$, with incoming arrows $\mathrm{in}(\mathbf{A}') = \mathrm{in}(\mathbf{V}) \cup \big( \mathrm{in}(\mathbf{A}) \setminus \mathrm{out}(\mathbf{V}) \big)$ and outgoing arrows $\mathrm{out}(\mathbf{A}') = \mathrm{out}(\mathbf{A})$, as shown in \cref{fig:merge-one-complete},
        \item the resulting node $\mathbf{A}'$ being endowed with the route $\lambda_{\mathbf{A}'} = \lambda_\mathbf{A} \ast \lambda_{\mathbf{V}}$.
    \end{itemize}
    Then $\cG^{\text{merge}\uparrow}$ is a valid routed graph.
\end{restatable}
The statement of this and the following lemma, as well as schematic excerpts of the associated branch graphs, are illustrated in \cref{fig:merge}.
\begin{restatable}{lemma}{generalMergeTwo}
    \label{thm:var-split-two}
    Let $\cG = (\Gamma, (\lambda_\mathbf{N})_{\mathbf{N}\in \mathtt{Nodes}_\Gamma})$ be a valid routed graph without any weak loops,
    and $\mathbf{A}, \mathbf{V} \in \mathtt{Nodes}_\Gamma$ such that $\mathbf{A}$ is the sole parent of $\mathbf{V}$ (cf.\ \cref{fig:merge-two-before})
    and the associated branch structure is such that
    \begin{equation}
        \forall \, \mathbf{A}^\beta \in \Bran(\lambda_\mathbf{A}) \quad
        \exists! \, \mathbf{V}^\gamma \in \Bran(\lambda_{\mathbf{V}}) : \quad
        \mathtt{LinkVal} ( \mathbf{A}^\beta, \mathbf{V}^\gamma ) \neq \emptyset \label{eq:merge-condition-two}
    \end{equation}
    Consider another routed graph $\cG^{\text{merge}\downarrow}$ obtained by merging $\mathbf{V}$ into $\mathbf{A}$, similarly to \cref{thm:var-split}. 
    Then $\cG^{\text{merge}\downarrow}$ is a valid routed graph.
\end{restatable}

The proofs of these lemmas are given in \cref{sec:proof5}.
While the naming of the nodes $\mathbf{A}$ and $\mathbf{V}$ are highly suggestive with regard to the generic routed graphs and its alternative versions, they are not intended as restrictions with regard to the scope of these two lemmas.

Note that the merged routed graphs $\cG^{\text{merge}\downarrow}$ and $\cG^{\text{merge}\uparrow}$ are essentially created by merging the branches $\mathbf{A}^\beta$ and $\mathbf{V}^\gamma$.
In the branch graph, the respective branches are merged along their respective strong parent relations.
The branches of the merged node feature a subset of the union of the weak parent/child relations of the original nodes, with weak parent relations \enquote{within} merged branches vanishing.
Additionally, due to the potentially reduced number of branches and the coarse-graining of the image of the choice function, generally some green arrows \emph{to} $\mathbf{V},\mathbf{A}$ as well as some red arrows \emph{from} these nodes may vanish from the routed graph under this coarse-graining.\footnote{However, for the generic routed graph as our main case of interest, this will not be the case.}

The resulting routed graphs are not only valid, but indeed generate, in a precise sense, the same class of processes by means of routed circuit representations.

\begin{restatable}{proposition}{generalMergeThree}
    \label{thm:merge-one}
    Any process with a routed circuit decomposition (as defined in \cref{def:decomposition}) according to the routed graph $\cG$, where $\mathbf{V}$ is fleshed out with a routed isometry (an \enquote{internal} node), can also be expressed with a routed circuit decomposition according to $\cG^{\text{merge}\uparrow}$ / $\cG^{\text{merge}\downarrow}$, related to $\cG$ via \cref{thm:var-split} / \cref{thm:var-split-two}, and vice versa.
\end{restatable}

The proof of this lemma is given in \cref{sec:proof5}.
Essentially, the merge composes the channels and supermaps fleshing out the skeletal supermap associated with $\cG$ using the link product.

For $\cG^{\text{merge}\downarrow}$, we would like to point out that the dimension of $\HH^{\mathbf{A} \to \mathbf{V}}$ could generally restrict the information flow from $\HH^{\mathrm{in}(\mathbf{A})}$ to $\HH^{\mathrm{out}(\mathbf{V})}$. 
Therefore, the equivalence of both representations does not necessarily hold on the level of the individual skeletal supermaps associated with $\cG$ and $\cG^{\text{merge}\downarrow}$, respectively.
However, as the same restriction may alternatively be introduced on the level of the fleshing out, this fact does not impair the fact that both routed graphs indeed generate the same set of routed circuit decompositions.

In practice, when merging an internal node $\mathbf{V}$ into an agent node $\mathbf{A}$, we will denote this as $\mathbf{V} \hookrightarrow \mathbf{A}$, merging the internal node \enquote{into} the agent node, which will then be denoted as $\mathbf{A}'$.

We now continue by applying these results to our generic routed graph $\cG_{\text{QC-QC}(N)}$.
For $N=2$ and $N=3$, we see that a rewriting in this fashion allows to get rid of all internal nodes in $\cG_{\text{QC-QC}(N)}$, except for $\mathbf{V}_1$ and $\mathbf{V}_{N+1}$.

Considering $\cG^\text{split}_{\text{QC-QC}(3)}$ as an example, we observe that both $\mathbf{\hat V}_2^{k_1}$ and $\mathbf{\hat V}_3^{\K_2}$ have agent nodes as unique parent/child nodes.
Therefore, we may absorb $\mathbf{\hat V}_2^{k_1} \hookrightarrow \mathbf{A}_{k_1}$ and  $\mathbf{\hat V}_3^{\K_2} \hookrightarrow \mathbf{A}_{\N \setminus \K_2}$.
The resulting routed graph is illustrated in \cref{fig:group-routed-graph-grenoble} (denoting again the agents $(\mathbf{A}_1, \mathbf{A}_2, \mathbf{A}_3)$ as $(\mathbf{A}, \mathbf{B}, \mathbf{C})$).%
\footnote{
    To see explicitly how the composition of branches plays out, consider the case of $\mathbf{\hat V}_2^{B} \hookrightarrow \mathbf{B}$ in detail.
    Here,
    $\Bran(\lambda_{\mathbf{\hat V}_2^{B}}) = \{ \mathbf{V}_2^{B}, \mathbf{\bar V}_2^{B} \}$ and
    $\Bran(\lambda_\mathbf{B}) = \{ \mathbf{B}^{\K_n} \mid \K_n \subset \cN \setminus \mathbf{B} \}$.
    Then, $\mathbf{V}_2^{B}$ and $\mathbf{B}^\emptyset$ are composed, while the 1-dimensional branch is trivially composed with all other branches of $\mathbf{B}$.
    This leaves their respective operations unchanged, as the operation for $\mathbf{\bar V}_2^{B}$,
    specified by \cref{eq:vacuum-isometry}, acts identically (and on one-dimensional tensor factors, cf.\ \cref{eq:1d-factor}).
}
For $N=2$, analogous arguments apply.

Indeed, this allows us to ultimately recover the routed graph examples from \cite{Vanrietvelde2023}.\footnote{
    By contrast, introducing additional arrows to represent ancillary systems, as described in \cref{sec:var-ancillary}, introduces additional parents for the $\mathbf{V}$-nodes, which prevents us from applying the same technique for simplification: each node $\mathbf{V}^{\K_n}_{n+1}$ features additional parents $\mathbf{V}^{\K_n \setminus k}_{n}$.
}
\begin{example}[Quantum switch]
    \label{ex:quantum-switch}
    To represent the quantum switch without any internal node to represent $V_2$ explicitly, we take \cref{fig:split-routed-graph-switch} as a starting point. 
    Due to the simplicity of the routed graph, we have some freedom for the rewrite:
    Both for $\mathbf{\hat V}_2^A$ and $\mathbf{\hat V}_2^B$, we may choose to absorb them either into $\mathbf{A}$ or $\mathbf{B}$, i.e., both \cref{thm:var-split} and \cref{thm:var-split-two} could be applied we to recover the quantum switch as presented in Fig.~6 of \cite{Vanrietvelde2023}.
    For \cref{fig:group-routed-graph-switch} and this example, we choose $\mathbf{V}_2^A \hookrightarrow \mathbf{A}$ and $\mathbf{V}_2^B \hookrightarrow \mathbf{B}$.\footnote{
        This results in $\mathbf{A'} \xrightarrow{X^B_A} \mathbf{B'}$, with the alternative resulting in $\mathbf{A'} \xrightarrow{X^A_\emptyset} \mathbf{B'}$, and equally when exchanging $A$ / $B$.}

    Regarding the fleshing out, we consider the resulting node $\mathbf{B}'$ as an example.
    Here, we take the composition of the fleshing outs $\dket{\mathbf{B}}$ and $\dket{\mathbf{\hat V}_2^B}$ from $\cG^\mathrm{split}_{\text{QC-QC}(N)}$ (according to the link product) along the node merging as the new fleshing out:
    \begin{align}
        \dket{\mathbf{B}'}
        &= \dket{\mathbf{B}} \ast \dket{\mathbf{\hat V}_2^B} \notag \\
        &= \dket{J_{\textbf{B}}^\text{in}}
            \ast \dket{J_{\textbf{B}}^\text{out}}
            \ast \dket{J_{\textbf{V}^{B}_{2}}^\text{in}} 
            \ast \dket{\hat{V}_{2}^{B}}
            \ast \dket{J_{\textbf{V}^{B}_{2}}^\text{out}} \, .
    \end{align}
    and similarly for $\dket{\mathbf{A}}$.
    This reconstructs the routed graph of Fig.\ 6 of \cite{Vanrietvelde2023} when relabelling their indices accordingly.\footnote{
        Specifically, one translates their index values as follows:
        \begin{align}
            i,l = 1 \to \XX_1^A = \{ A \} \qquad
            j,k = 1 \to \XX_1^B = \{ B \} \qquad
            m = 1 \to \XX_2^A = \{A, B\} \qquad
            n = 1 \to \XX_2^B = \{A, B\} \, .
        \end{align}
        For all of their indices, if they take the value $0$, we set the associated index $\XX_n^k = \ravnothing$.
    }
\end{example}

\begin{figure}
    \centering
    \includegraphics{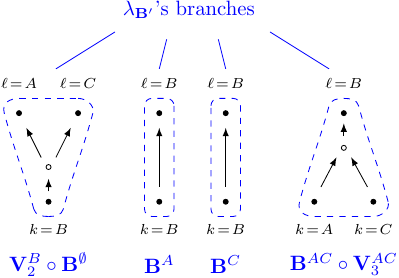}
    \caption{
    The branch structure for $\mathbf{B}'$ in \cref{fig:split-routed-graph-grenoble},
    given by 
    $\lambda_{\mathbf{B}'} = \lambda_{\mathbf{\hat{V}}_3^{AC}} \ast \lambda^{\vphantom{AC}}_{\mathbf{B}} \ast \lambda_{\mathbf{\hat V}_2^B}$    
    with a fine-graining into the branches of the original routed graph being shown in \cref{fig:branch-graphs}.
    The branch structures for $\mathbf{A}'$ and $\mathbf{C}'$ are given analogously by permutation.
    Due to the node merging, non-trivial index relations and bifurcation choices are carried over to the agent node $\mathbf{B}'$, reproducing the branch structure given for the nodes of the Grenoble process in Fig.~23 of \cite{Vanrietvelde2023}.}
    \label{fig:branches-grouping}
\end{figure}

\begin{example}[Grenoble process]
    \label{ex:Grenoble}
    The Grenoble process is a tripartite process with dynamical causal order, and was originally introduced in \cite{Wechs2021}.
    We use the transformation techniques introduced in this section to reconstruct the routed quantum circuit of the Grenoble process given in \cite{Vanrietvelde2023}.
    To reconstruct their routed graph, we take \cref{fig:split-routed-graph-grenoble} as as a starting point.
    We merge
    $\mathbf{\hat V}_2^A \hookrightarrow \mathbf{A} \hookleftarrow \mathbf{\hat V}_3^{BC}$,
    $\mathbf{\hat V}_2^B \hookrightarrow \mathbf{B} \hookleftarrow \mathbf{\hat V}_3^{AC}$,
    $\mathbf{\hat V}_2^C \hookrightarrow \mathbf{C} \hookleftarrow \mathbf{\hat V}_3^{AB}$,
    obtaining \cref{fig:group-routed-graph-grenoble}.
    In this graph, the arrows connecting the agents feature multiple indices, as several node mergings induce an arrow between those. 
    
    Renaming $\mathbf{V}_1$ to $P$ and $\mathbf{V}_4$ to $F$, we obtain the same skeletal supermap as in \cite{Vanrietvelde2023} (up to a renaming of the index values and Hilbert spaces).
    For example, for the merged node $\mathbf{B}'$ we obtain the branch structure shown in \cref{fig:branches-grouping}, reflecting Fig.~23 of \cite{Vanrietvelde2023}.

    Regarding the fleshing out, we again consider the node $\mathbf{B}'$ as an example.
    We take the composition of the fleshing outs $\dket{\mathbf{B}}$, $\dket{\mathbf{\hat V}_2^B}$ and $\dket{\mathbf{\hat V}_3^{AC}}$ along the node merging as the new fleshing out:
    \begin{align}
        \dket{\mathbf{B}'}
        &= \dket{\mathbf{\hat{V}}_3^{AC}}
            \ast \dket{\mathbf{B}}
            \ast \dket{\mathbf{\hat V}_2^B} \notag \\
        &= \dket{J_{\textbf{V}^{AC}_{3}}^\text{in}} 
            \ast \dket{\hat{V}_{3}^{AC}}
            \ast \dket{J_{\textbf{V}^{AC}_{3}}^\text{out}}
            \ast \dket{J_{\textbf{B}}^\text{in}}
            \ast \dket{J_{\textbf{B}}^\text{out}}
            \ast \dket{J_{\textbf{V}^{B}_{2}}^\text{in}} 
            \ast \dket{\hat{V}_{2}^{B}}
            \ast \dket{J_{\textbf{V}^{B}_{2}}^\text{out}} \, .
    \end{align}
    How this fleshing out relates to the fleshing out specified in Fig.~25 of \cite{Vanrietvelde2023} is described in detail in \cref{sec:Grenoble-reconstruction}.
\end{example}

For larger $N$, the possibility of merging nodes into agent nodes in $\cG_{\text{QC-QC}(N)}$ applies analogously for the internal nodes $\mathbf{V}^{\K_1}_2$ and $\mathbf{V}^{\K_{N-1}}_N$.
For the remaining internal nodes, however, the technique is no longer applicable, as the remaining $\mathbf{V}$-nodes in $\cG^\text{split}_{\text{QC-QC}(N)}$ will feature multiple parents and children respectively.
However, a specific fleshing out of this routed graph (according to a specific QC-QC implementation) may potentially be compatible with a more fine-grained branch structure.
Imposing such a branch structure can effectively restrict the set of processes supported by the corresponding routed graph.
In \cref{sec:local}, we investigate this fact in more detail and reproduce an example of a 4-partite process \cite{Salzger2022, Lukas} which cannot be entirely fine-grained in such a fashion.

Note that we do not reconstruct here the representation of the quantum 3-switch given in Section 4.1 of \cite{Vanrietvelde2023}.
Doing so would require further adjustments to the index constraints on the routed graph, restricting the set of supported processes even further. We leave this for future work.

\subsection{Removing arrows in the routed graph}
\label{sec:less-arrows}

In addition to considering alternatives to the generic routed graph $\cG_{\text{QC-QC}(N)}$ to obtain arbitrary QC-QCs, we may consider ways to simplify the routed graph when the QC-QC has additional structure yielding simplified connectivity.

Consider a given fleshing out of a skeletal supermap for the generic routed graph $\cG_{\text{QC-QC}(N)}$.
If for a given arrow $\mathbf{M} \to \mathbf{N}$, the only sector that can be populated (i.e.\ whose index value can be selected) is $\HH^{\mathbf{M}\to\mathbf{N}}_\ravnothing$ (cf.\ \cref{subsubsec:skeletal_QCQC}), we may remove the respective arrow from the routed graph -- if for the adjoint of the graph, the respective space is also never populated.
This follows as in this case, the respective 1-dimensional tensor factor in the formal spaces $\HH^{\mathrm{out}(\mathbf{M})}$ and $\HH^{\mathrm{in}(\mathbf{N})}$ may be omitted (up to an isomorphism), with the index associated to the arrow being constantly $\ravnothing$ for all remaining bifurcation choices $\vec{\bm{\ell}}$ (see \cref{subsubsec:biunivoc_GQCQC}) -- both for the original and the adjoint graph.
The branch graph and the skeletal supermap simplify accordingly:
in comparison to the original branch graph, the strong parent relations associated with the respective arrows (i.e.,  branches $\mathbf{M}^\gamma$ and $\mathbf{N}^\beta$) disappear, as they require an arrow from $\mathbf{M} \to \mathbf{N}$ in the first place.
Furthermore, due to the bifurcation choice for the branch happening being removed, it never happens\footnote{And therefore, all potential further bifurcation choices made within this branch never matter.} and all green arrows may be removed as well.
To finally remove the respective red arrows as well, the same argument is applied to the adjoint graph.
We can then remove all branches no longer associated with any arrows.

For the remaining weak arrows in the graph, one needs to check whether the corresponding functional dependence in the choice function remains.
For any routed graphs (even beyond the generic routed graph), each node will no longer have weak parents if only a single branch of it remains, or if in the resulting graph a given branch $\mathbf{M}^\gamma$ is a weak parent of just one branch of a given node $\mathbf{N}$: a non-trivial choice always requires multiple possibilities (for the branch which happens). 

Specifically for a given implementation of a QC-QC (and therewith, a choice of $\{ V_{\K_{n}\setminus k,k}^{\to \ell} \}$ yielding the respective process according to \cref{eq:Choi_V_QCQC}), arrows may be removed in particular if for some $\ell,n$, the condition $V_{\K_{n}\setminus k,k}^{\to \ell} = 0$ is satisfied for all $\K_{n}, k$.
Then in the image of the isometry given by $\dket{\mathbf{V}_{n+1}}$, only the sector $\HH^{\mathbf{V}_{n+1}\to\mathbf{A}_\ell}_\ravnothing$ will be populated, and $\mathcal{X}^\ell_{n+1} = \ravnothing$ for any choice of operations $A_k$ and past state $\ket{\psi}^P$.
By extension, the sector $\HH^{\mathbf{A}_\ell\to\mathbf{V}_{n+2}}_\ravnothing$ is fixed as well.
Hence, both the arrows $\mathbf{V}_{n+1} \xrightarrow{\mathcal{X}^\ell_{n+1}} \mathbf{A}_\ell$ and $\mathbf{A}_\ell \xrightarrow{\mathcal{X}^\ell_{n+1}} \mathbf{V}_{n+2}$ may be omitted from the routed graph, and we may modify $\tilde{V}_{n+2}$ such that for all $m \not\in \K_{n} \cup \ell$ and all $\K_{n} \not\ni \ell$, $V_{\K_{n},\ell}^{\to m} = 0$ holds -- which then ensures the analogous property for the adjoint of the graph.
For the associated branch graph, incoming green dashed arrows will generally always vanish if the respective branch is now certain either to happen or not to happen, and analogously for the red dashed arrows.
In such a QC-QC, there is always an implementation in which $A_\ell$ never acts in the $(n+1)$th position, which allows to remove the corresponding arrows.

\begin{example}
    Consider a bipartite process compatible with a fixed order $A \prec B$.
    Then clearly, there exist implementations of this process where $V^{\to B}_{\emptyset,\emptyset} = V^{\to A}_{\emptyset,B} = 0$.
    Hence, for any choice of agent operations, $\XX^B_1 = \XX^A_2 = \ravnothing$, and there exists an implementation in which $V^{\to F}_{B,A} = 0$ as well.   
    Accordingly, we may remove the arrows
    \vspace{-.1cm}
    \begin{equation}
        \mathbf{V}_{1} \xrightarrow{\mathcal{X}^B_{\emptyset}} \mathbf{B}, \;\
        \mathbf{B} \xrightarrow{\mathcal{X}^B_{\emptyset}} \mathbf{V}_2, \;\
        \mathbf{V}_{2} \xrightarrow{\mathcal{X}^A_{B}} \mathbf{A} \;\ \text{and} \;\
        \mathbf{A} \xrightarrow{\mathcal{X}^A_{B}} \mathbf{V}_3 \; ,
    \end{equation}
    from the graph in \cref{fig:routed_graphs_QCQC_N_2}, leaving us with the acyclic routed graph
    \vspace{-.1cm}
    \begin{equation}
        \longrightarrow \mathbf{V}_1
        \xrightarrow{\XX^A_\emptyset} \mathbf{A}
        \xrightarrow{\XX^A_\emptyset} \mathbf{V}_2
        \xrightarrow{\XX^B_A} \mathbf{B} 
        \xrightarrow{\XX^B_A} \mathbf{V}_3
        \longrightarrow \!\quad .
    \end{equation}
    Here, the indices always take the values $\XX^A_1=\{A\}, \ \XX^B_2=\{A,B\}$, and could be omitted from the graph.
    The resulting branch graph is then almost identical to the routed graph, as for each node, there is only a single branch, which always happens.
    \begin{equation}
        \mathbf{V}_1^\emptyset
        \xrightarrow{\XX^A_\emptyset} \mathbf{A}^\emptyset
        \xrightarrow{\XX^A_\emptyset} \mathbf{V}_2^A
        \xrightarrow{\XX^B_A} \mathbf{B}^A
        \xrightarrow{\XX^B_A} \mathbf{V}_3^{AB}.
    \end{equation}
\end{example}

\section{Discussion}
\label{sec:discussion}

In providing a representation of arbitrary QC-QCs in terms of routed quantum circuits, we have established an application of the latter framework well beyond the particular examples previously considered, extending it to a wide class of processes.
Our construction notably provides a more fine-grained picture of the logical structure of QC-QCs, with the control of the causal order being represented by indices $\mathcal{X}_n^k$ in the routed graph.

The fact that we achieved this by specifying just a single template for the routed graph and a generic fleshing out is testament to the ability of the framework to condense entire classes of processes in terms of their underlying information-processing structure, which is concisely encapsulated in the routed graph.
In addition, the more fine-grained branch graph it brings forth may be closely related to the underlying causal structure of the respective process, as has been tentatively demonstrated in \cite{Ormrod2023}.

\subsection{Contrasting routed circuit decompositions}

In Section~4 of \cite{Vanrietvelde2023}, explicit constructions were given for two instances of QC-QCs: the quantum switch \cite{Chiribella2013} and the Grenoble process \cite{Wechs2021}.
These constructions differ from our generic routed graph $\cG_{\mathrm{QC-QC}(N)}$, with fewer nodes: internal nodes corresponding to the agents $A_k$, as well as nodes for the past and future $P$ and $F$, which play roles analogous to our internal nodes $\mathbf{V}_1$ and $\mathbf{V}_{N+1}$.
By featuring additional internal nodes $\mathbf{V}_n$, our construction appears to yield a significantly more verbose routed graph representation.
In \cref{sec:grouping}, we discussed how one can, under certain conditions, nonetheless recover the structure of their routed graph from ours by merging nodes. 
However, for $N>3$, a representation without further internal (i.e., non-agent) nodes can no longer generally be found, as we exemplify in \cref{sec:local}.
The generalisation of the routed graphs given in~\cite{Vanrietvelde2023}, without additional internal nodes, is thus unable to capture the full expressivity of QC-QCs.

Even in the $N=2$ and $N=3$ cases, the more compact routed graph representation of~\cite{Vanrietvelde2023} comes at the cost of a more complex fleshing out for the agent nodes $\mathbf{A}_k$, which must not only act as an adapter for allowing the local operations of the agents to be plugged in, but also includes the, in general non-trivial, (internal) routing of the process. 

However, with our generic routed graph (modified to represent the ancillary systems explicitly in $\cG_{\mathrm{QC-QC}(N)}^\alpha$), we can mostly forgo the necessity of a fleshing out to describe the class of QC-QCs.
Indeed, for the $\mathbf{V}$-nodes the routes actually restrict the associated slots precisely to valid QC-QC isometries, and any routed isometries following the routes constitutes a valid fleshing out.
For the $\mathbf{A}$-nodes, as observed in \cref{sec:var-ancillary}, allowing the agents to perform any routed operation\footnote{With the skeletal supermap associated with the routed graph from \cref{sec:var-ancillary}, that is, to prevent the agents from accessing the ancillary systems $\alpha_n$.} following the routes essentially gives them just the power to perform routed operations dependent on the causal order -- which in fact mirrors many physical implementations suggested so far more closely, while still conveying a form of ICO \cite{Ormrod2023, delaHamette2024}.
One can moreover show that each such routed process may still be simulated as a conventional (non-routed) QC-QC, potentially decreasing the relevance of this distinction.

Indeed, consider a process where a coherent copy of the control system $\bar{C}^{(k)}$ is passed to each agent $A_k$.
Then, they may perform their operation controlled on the set of agents $\K_n \setminus k$ who acted previously, as in a routed process, or perform a measurement on the causal order, destroying the coherence of the quantum control in the QC-QC.
Any $N$-slot QC-QC process may be amended to pass a copy of the control to the agents in such a fashion, increasing the dimensions of $\HH^{A_k^{I/O}}$ accordingly.
After their operation, the coherent copy may then be undone -- presuming, the agents did not modify the copy -- and be traced out without harming the coherence of the control.
If the agents modify their copy of the control however, they are transcending the set of operations which would be allowed in a routed switch, and may collapse the action of the process to a QC-CC (with \emph{classical} control of causal order).

Finally, we briefly comment on the index structure we utilised.
While in \cite{Vanrietvelde2023}, all indices used in practice only take values $\{ 0, 1 \}$, we use the index value to encode subsets $\K_n \subseteq \cN$ (or $\ravnothing$) -- with the notable exception of \cref{eq:simple-values}.
Thereby, we demonstrate how a more complex set of possible index values may be helpful to obtain an adequate routed graph.

\subsection{Characterising classes of processes}

Several previous works have proposed subclasses of QC-QCs with specific properties of interest. 
These include a restriction to classical control (\textsf{QC-CC}) \cite{Wechs2021}, to superpositions of fixed orders (\textsf{Sup} \cite{Liu2023, Mothe2024}, or the more general \textsf{QC-supFO} class \cite{Mothe2025}), and to non-influenceable causal orders (\textsf{QC-NIQC}), capturing a weak kind of dynamicality of causal order~\cite{Mothe2025}.

Typically, the restrictions on these subclasses can be characterised as restrictions on the internal operations of a QC-QC. 
Our routed quantum circuit construction provides a new perspective from which to study these classes, as these restrictions may correspond to additional constraints on the routed graph.
For instance, modifying the index structure to capture the full order $(k_1, \ldots, k_N)$ from $\mathbf{V}_1$ all the way to $\mathbf{V}_{N+1}$ should yield a superposition of fixed orders.
Generally, constraining the routed graph will involve introducing additional structure into the indices, effectively fine-graining the sectorisation of the systems to store additional information on the causal order, and imposing restrictions on the branch structure of the internal nodes $\mathbf{V}_{n+1}$.
The resulting graphs can then potentially be simplified following techniques analogous to those presented in \cref{sec:variations}.

Conversely, imposing constraints directly on the routed graph may be an interesting avenue to study new types of QC-QCs,
a potential example being the subclass of processes which can be represented without involving internal transformations beyond local pre- and post-processing in-between the agents (cf.\ \cref{sec:local}).
When assuming the operations associated with each branch of the respective nodes to be localised in spacetime, this property in particular affects the embeddability of the associated circuit into spacetime \cite{VVR, Salzger2025}.

In addition to capturing these classes in terms of routed quantum circuits, the properties of these representations may also serve as starting points to generalise the concepts they capture to non-QC-QC processes.
For instance, in the language of routed quantum circuits, green dashed arrows in the branch graph may indicate whether a node has an influence on the order of later nodes:
beyond QC-QCs, this has been showcased for both the Lugano process \cite{Vanrietvelde2023}, and more recently, for a process with ``unilateral determination of causal order'' \cite{Mejdoub2025}.
Studying this relation in detail would potentially allow the notion of dynamicality, which has so far been formalised for QC-QCs only~\cite{Mothe2025}, to be studied in processes beyond QC-QCs.

\subsection{Consequences for the routed quantum circuit framework}

Last but not least, we consider the impact of our results on the general status of the routed quantum circuit framework.
By extending the number of processes with an established routed quantum circuit decomposition from a small handful to the entire class of QC-QCs, we confirm the applicability of the framework to a wide class of physically implementable processes.
Moreover, recent results connecting QC-QCs with the causal box framework \cite{Portmann2017, Salzger2022, deBank2025} indicate that QC-QCs appear to be the most general class of processes compatible (and implementable) within a fixed (acyclic) spacetime.
Combined with our results, this would imply that skeletal supermaps associated with routed graphs (without weak loops) -- or even the subset of routed graphs from our construction -- are capable of expressing any quantum process implementable in a fixed spacetime.\footnote{However, this alone does not imply the ability to interpret each (valid) skeletal supermap as a blueprint for a non-routed process!}
This would provide some additional support the first main conjecture of \cite{Vanrietvelde2023}, concerning the hypothesis that any unitary process may be obtained from fleshing out some skeletal supermap associated with a routed graph. 
The results of \cite{Lukas} further indicate that these statements may also be generalisable from fixed acyclic spacetimes to any settings where all agents (independently) record a monotonically increasing time, e.g.\ capturing gravitational quantum switch scenarios.
Further progress in this direction, extending the result beyond processes implementable within a fixed spacetime, would most likely require significant progress in the study of causal decompositions of unitary channels \cite{Lorenz2021, vanrietvelde2025qca}.

Considering the fact that QC-QCs do not violate causal inequalities \cite{Wechs2021}, our results are in agreement with the second major conjecture of~\cite{Vanrietvelde2023} -- that a skeletal supermap violates causal inequalities if and only if its branch graph features weak loops.
However, as our procedure only yields a small subset of valid routed graphs without weak loops, we are unable to provide any stronger statement in that regard.

The joint perspective obtained by studying QC-QCs and routed quantum circuits may also be valuable to provide further insight on the composability of processes.
Indeed, although the validity conditions for routed graphs provide a condition to affirm the validity of compositions, a systematic study of the resulting family of routed graphs has been lacking so far.
Restricting to compositions of routed graphs which represent faithful reconstructions of QC-QCs may help both to render the scope of this study more manageable and to restrict oneself to physically admissable implementations \cite{Chardonnet2025, Simmons2024, Kissinger2019, PRXQuantum.4.040307, Apadula2024}.
Conversely, the clear graphical representation of connectivity may also allow one to restrict the class of processes attainable by composing a restricted set of \enquote{elementary} processes.

\section{Acknowledgements}

We thank Hlér Kristjánsson, Ved Kunte, Raphaël Mothe, Pierre Pocreau, Matthias Salzger, V.~Vilasini and Emilien de Bank for helpful discussions.

This research was funded in part by l'Agence Nationale de la Recherche (ANR) projects ANR-15-IDEX-02 and ANR-22-CE47-0012, the PEPR integrated project EPiQ ANR-22-PETQ-0007 as part of Plan France 2030, and the STEP2 grant ANR-22-EXES-0013. For the purpose of open access, the authors have applied a CC-BY public copyright license to any Author Accepted Manuscript (AAM) version arising from this submission.


\begin{thebibliography}{10}

\bibitem{Reichenbach1956}
Hans Reichenbach.
\newblock ``The direction of time''.
\newblock \href{https://dx.doi.org/10.2307/2216858}{University of California
  Press}. Berkeley~(1956).

\bibitem{Spirtes1993}
Peter Spirtes, Clark Glymour, and Richard Scheines.
\newblock ``Causation, prediction, and search''.
\newblock \href{https://dx.doi.org/10.1007/978-1-4612-2748-9}{Springer New
  York}. NY~(1993).

\bibitem{Pearl2009}
Judea Pearl.
\newblock ``Causality: Models, reasoning and inference''.
\newblock \href{https://dx.doi.org/10.1017/CBO9780511803161}{Cambridge
  University Press}. ~(2009).
\newblock 2nd edition.

\bibitem{Wood2015}
Christopher~J. Wood and Robert~W. Spekkens.
\newblock ``The lesson of causal discovery algorithms for quantum correlations:
  {C}ausal explanations of {B}ell-inequality violations require fine-tuning''.
\newblock \href{https://dx.doi.org/10.1088/1367-2630/17/3/033002}{New J. Phys.
  {\bf 17}, 033002}~(2015).

\bibitem{PhysRevX.7.031021}
John-Mark~A. Allen, Jonathan Barrett, Dominic~C. Horsman, Ciar\'an~M. Lee, and
  Robert~W. Spekkens.
\newblock ``Quantum common causes and quantum causal models''.
\newblock \href{https://dx.doi.org/10.1103/PhysRevX.7.031021}{Phys. Rev. X {\bf
  7}, 031021}~(2017).

\bibitem{arxiv.1906.10726}
Jonathan Barrett, Robin Lorenz, and Ognyan Oreshkov.
\newblock ``Quantum causal models''~(2019).
\newblock  \href{http://arxiv.org/abs/1906.10726}{arXiv:1906.10726}.

\bibitem{Hardy2009}
Lucien Hardy.
\newblock ``Quantum gravity computers: On the theory of computation with
  indefinite causal structure''.
\newblock \href{https://dx.doi.org/10.1007/978-1-4020-9107-0_21}{Pages
  379--401}.
\newblock Springer Netherlands. ~(2009).
\newblock
  \href{http://arxiv.org/abs/quant-ph/0701019}{arXiv:quant-ph/0701019}.

\bibitem{Brukner2015}
{\v C}aslav Brukner.
\newblock ``Bounding quantum correlations with indefinite causal order''.
\newblock \href{https://dx.doi.org/10.1088/1367-2630/17/8/083034}{New J. Phys.
  {\bf 17}, 083034}~(2015).

\bibitem{Costa2016}
Fabio Costa and Sally Shrapnel.
\newblock ``Quantum causal modelling''.
\newblock \href{https://dx.doi.org/10.1088/1367-2630/18/6/063032}{New J. Phys.
  {\bf 18}, 063032}~(2016).

\bibitem{Barrett2021}
Jonathan Barrett, Robin Lorenz, and Ognyan Oreshkov.
\newblock ``Cyclic quantum causal models''.
\newblock \href{https://dx.doi.org/10.1038/s41467-020-20456-x}{Nature
  Commun.{\bf 12}}~(2021).

\bibitem{VVR}
V.~Vilasini and Renato Renner.
\newblock ``Embedding cyclic information-theoretic structures in acyclic
  space-times: No-go results for indefinite causality''.
\newblock \href{https://dx.doi.org/10.1103/physreva.110.022227}{Phys. Rev. A
  {\bf 110}, 022227}~(2024).
\newblock  \href{http://arxiv.org/abs/2203.11245}{arXiv:2203.11245}.

\bibitem{Adlam2023}
Emily Adlam.
\newblock ``Is there causation in fundamental physics? {N}ew insights from
  process matrices and quantum causal modelling''.
\newblock \href{https://dx.doi.org/10.1007/s11229-023-04160-z}{Synthese {\bf
  201}, 152}~(2023).

\bibitem{Vilasini24}
V.~Vilasini and Renato Renner.
\newblock ``Fundamental limits for realizing quantum processes in spacetime''.
\newblock \href{https://dx.doi.org/10.1103/PhysRevLett.133.080201}{Phys. Rev.
  Lett. {\bf 133}, 080201}~(2024).
\newblock  \href{http://arxiv.org/abs/2408.13387}{arXiv:2408.13387}.

\bibitem{Chiribella2008_Supermap}
Giulio Chiribella, Giacomo~Mauro D'Ariano, and Paulo Perinotti.
\newblock ``Transforming quantum operations: Quantum supermaps''.
\newblock \href{https://dx.doi.org/10.1209/0295-5075/83/30004}{Europhys. Lett.
  {\bf 83}, 30004}~(2008).
\newblock  \href{http://arxiv.org/abs/0804.0180}{arXiv:0804.0180}.

\bibitem{Oreshkov2012}
Ognyan Oreshkov, Fabio Costa, and {\v{C}}aslav Brukner.
\newblock ``Quantum correlations with no causal order''.
\newblock \href{https://dx.doi.org/10.1038/ncomms2076}{Nature Commun. {\bf 3},
  1092}~(2012).
\newblock  \href{http://arxiv.org/abs/1105.4464}{arXiv:1105.4464}.

\bibitem{Zych2019}
Magdalena Zych, Fabio Costa, Igor Pikovski, and {\v C}aslav Brukner.
\newblock ``Bell's theorem for temporal order''.
\newblock \href{https://dx.doi.org/10.1038/s41467-019-11579-x}{Nature Commun.
  {\bf 10}, 3772}~(2019).

\bibitem{Moller2021}
Nat{\'a}lia~S. M{\'o}ller, Bruna Sahdo, and Nelson Yokomizo.
\newblock ``Quantum switch in the gravity of earth''.
\newblock \href{https://dx.doi.org/10.1103/physreva.104.042414}{Phys. Rev. A
  {\bf 104}, 042414}~(2021).
\newblock  \href{http://arxiv.org/abs/2012.03989}{arXiv:2012.03989}.

\bibitem{Baumann2022}
Veronika Baumann, Marius Krumm, Philippe~Allard Gu\'erin, and \ifmmode
  \check{C}\else~\v{C}\fi{}aslav Brukner.
\newblock ``Noncausal {P}age-{W}ootters circuits''.
\newblock \href{https://dx.doi.org/10.1103/PhysRevResearch.4.013180}{Phys. Rev.
  Res. {\bf 4}, 013180}~(2022).

\bibitem{AC2024}
Anne-Catherine de~la Hamette, Viktoria Kabel, Marios Christodoulou, and {\v
  C}aslav Brukner.
\newblock ``Quantum diffeomorphisms cannot make indefinite causal order
  definite''~(2024).
\newblock  \href{http://arxiv.org/abs/2211.15685}{arXiv:2211.15685}.

\bibitem{Sahdo2024}
Bruna Sahdo.
\newblock ``Indefinite order in the interface of quantum mechanics and
  gravity''.
\newblock Master's thesis.
\newblock Federal University of Minas Gerais.
\newblock ~(2023).
\newblock  \href{http://arxiv.org/abs/2310.02290}{arXiv:2310.02290}.

\bibitem{Chiribella2013}
Giulio Chiribella, Giacomo~Mauro D'Ariano, Paolo Perinotti, and Benoit Valiron.
\newblock ``Quantum computations without definite causal structure''.
\newblock \href{https://dx.doi.org/10.1103/PhysRevA.88.022318}{Phys. Rev. A
  {\bf 88}, 022318}~(2013).
\newblock  \href{http://arxiv.org/abs/0912.0195}{arXiv:0912.0195}.

\bibitem{Araujo2017}
Mateus Ara{\'u}jo, Adrien Feix, Miguel Navascu{\'e}s, and {\v C}aslav Brukner.
\newblock ``A purification postulate for quantum mechanics with indefinite
  causal order''.
\newblock \href{https://dx.doi.org/10.22331/q-2017-04-26-10}{Quantum {\bf 1},
  10}~(2017).

\bibitem{Lugano2014}
\"Amin Baumeler, Adrien Feix, and Stefan Wolf.
\newblock ``Maximal incompatibility of locally classical behavior and global
  causal order in multiparty scenarios''.
\newblock \href{https://dx.doi.org/10.1103/PhysRevA.90.042106}{Phys. Rev. A
  {\bf 90}, 042106}~(2014).
\newblock  \href{http://arxiv.org/abs/1403.7333}{arXiv:1403.7333}.

\bibitem{Baumeler2016}
\"{A}min Baumeler and Stefan Wolf.
\newblock ``The space of logically consistent classical processes without
  causal order''.
\newblock \href{https://dx.doi.org/10.1088/1367-2630/18/1/013036}{New J. Phys.
  {\bf 18}, 013036}~(2016).

\bibitem{Chiribella2012}
Giulio Chiribella.
\newblock ``Perfect discrimination of no-signalling channels via quantum
  superposition of causal structures''.
\newblock \href{https://dx.doi.org/10.1103/physreva.86.040301}{Phys. Rev. A
  {\bf 86}, 040301}~(2012).
\newblock  \href{http://arxiv.org/abs/1109.5154}{arXiv:1109.5154}.

\bibitem{Feix2015}
Adrien Feix, Mateus Ara{\'u}jo, and {\v C}aslav Brukner.
\newblock ``Quantum superposition of the order of parties as a communication
  resource''.
\newblock \href{https://dx.doi.org/10.1103/physreva.92.052326}{Phys. Rev. A
  {\bf 92}, 052326}~(2015).
\newblock  \href{http://arxiv.org/abs/1508.07840}{arXiv:1508.07840}.

\bibitem{Guerin2016}
Philippe~Allard Gu{\'e}rin, Adrien Feix, Mateus Ara{\'u}jo, and {\v C}aslav
  Brukner.
\newblock ``Exponential communication complexity advantage from quantum
  superposition of the direction of communication''.
\newblock \href{https://dx.doi.org/10.1103/physrevlett.117.100502}{Phys. Rev.
  Lett. {\bf 117}, 100502}~(2016).
\newblock  \href{http://arxiv.org/abs/1605.07372}{arXiv:1605.07372}.

\bibitem{Colnaghi2012}
Timoteo Colnaghi, Giacomo~Mauro D'Ariano, Stefano Facchini, and Paolo
  Perinotti.
\newblock ``Quantum computation with programmable connections between gates''.
\newblock \href{https://dx.doi.org/10.1016/j.physleta.2012.08.028}{Phys. Lett.
  A {\bf 376}, 2940--2943}~(2012).
\newblock  \href{http://arxiv.org/abs/1109.5987}{arXiv:1109.5987}.

\bibitem{Facchini2015}
Stefano Facchini and Simon Perdrix.
\newblock ``Quantum circuits for the unitary permutation problem''.
\newblock In Theory and Applications of Models of Computation.
\newblock \href{https://dx.doi.org/10.1007/978-3-319-17142-5_28}{Pages
  324--331}.
\newblock Springer International Publishing~(2015).
\newblock  \href{http://arxiv.org/abs/1405.5205}{arXiv:1405.5205}.

\bibitem{Araujo2014}
Mateus Ara{\'u}jo, Fabio Costa, and {\v C}aslav Brukner.
\newblock ``Computational advantage from quantum-controlled ordering of
  gates''.
\newblock \href{https://dx.doi.org/10.1103/physrevlett.113.250402}{Phys. Rev.
  Lett. {\bf 113}, 250402}~(2014).
\newblock  \href{http://arxiv.org/abs/1401.8127}{arXiv:1401.8127}.

\bibitem{Taddei2021}
M{\'a}rcio~M. Taddei, Jaime Cari{\~n}e, Daniel Mart{\'\i}nez, Tania
  Garc{\'\i}a, Nayda Guerrero, Alastair~A. Abbott, Mateus Ara{\'u}jo, Cyril
  Branciard, Esteban~S. G{\'o}mez, Stephen~P. Walborn, Leandro Aolita, and
  Gustavo Lima.
\newblock ``Computational advantage from the quantum superposition of multiple
  temporal orders of photonic gates''.
\newblock \href{https://dx.doi.org/10.1103/prxquantum.2.010320}{PRX Quantum
  {\bf 2}, 010320}~(2021).

\bibitem{Porcreau2025}
Alastair~A. Abbott, Mehdi Mhalla, and Pierre Pocreau.
\newblock ``Classical and quantum query complexity of {B}oolean functions under
  indefinite causal order''.
\newblock In Proc. 22nd Int. Conf. Quantum Phys. Log. (QPL 2025).
\newblock to appear in EPTCS~(2025).
\newblock  \href{http://arxiv.org/abs/2506.05187}{arXiv:2506.05187}.

\bibitem{Portmann2017}
Christopher Portmann, Christian Matt, Ueli Maurer, Renato Renner, and Bj{\"o}rn
  Tackmann.
\newblock ``Causal boxes: Quantum information-processing systems closed under
  composition''.
\newblock \href{https://dx.doi.org/10.1109/TIT.2017.2676805}{IEEE Trans. Inf.
  Theory {\bf 63}, 3277--3305}~(2017).
\newblock  \href{http://arxiv.org/abs/1512.02240}{arXiv:1512.02240}.

\bibitem{Wechs2021}
Julian Wechs, Hippolyte Dourdent, Alastair~A. Abbott, and Cyril Branciard.
\newblock ``Quantum circuits with classical versus quantum control of causal
  order''.
\newblock \href{https://dx.doi.org/10.1103/PRXQuantum.2.030335}{PRX Quantum
  {\bf 2}, 030335}~(2021).

\bibitem{PBS}
Alexandre Cl{\'e}ment and Simon Perdrix.
\newblock ``{PBS}-calculus: A graphical language for coherent control of
  quantum computations''.
\newblock In 45th International Symposium on Mathematical Foundations of
  Computer Science (MFCS 2020).
\newblock \href{https://dx.doi.org/10.4230/LIPICS.MFCS.2020.24}{Volume 170 of
  Leibniz International Proceedings in Informatics (LIPIcs)}.
\newblock Schloss Dagstuhl - Leibniz-Zentrum f{\"u}r Informatik~(2020).
\newblock  \href{http://arxiv.org/abs/2002.09387}{arXiv:2002.09387}.

\bibitem{PBSGrenoble}
Cyril Branciard, Alexandre Cl{\'e}ment, Mehdi Mhalla, and Simon Perdrix.
\newblock ``Coherent control and distinguishability of quantum channels via
  {PBS}-diagrams''.
\newblock In 46th International Symposium on Mathematical Foundations of
  Computer Science (MFCS 2021).
\newblock \href{https://dx.doi.org/10.4230/LIPICS.MFCS.2021.22}{Leibniz
  International Proceedings in Informatics (LIPIcs)}. Schloss Dagstuhl -
  Leibniz-Zentrum f{\"u}r Informatik~(2021).
\newblock  \href{http://arxiv.org/abs/2103.02073}{arXiv:2103.02073}.

\bibitem{Purves2021}
Tom Purves and Anthony~J. Short.
\newblock ``Quantum theory cannot violate a causal inequality''.
\newblock \href{https://dx.doi.org/10.1103/physrevlett.127.110402}{Phys. Rev.
  Lett. {\bf 127}, 110402}~(2021).
\newblock  \href{http://arxiv.org/abs/2101.09107}{arXiv:2101.09107}.

\bibitem{Arrighi2023}
Pablo Arrighi, Christopher Cedzich, Marin Costes, Ulysse R{\'{e}}mond, and
  Beno{\^{\i}}t Valiron.
\newblock ``Addressable quantum gates''.
\newblock \href{https://dx.doi.org/10.1145/3581760}{{ACM} Trans. Quantum
  Comput.}~(2023).
\newblock  \href{http://arxiv.org/abs/2109.08050}{arXiv:2109.08050}.

\bibitem{Vanrietvelde2023}
Augustin Vanrietvelde, Nick Ormrod, Hl{\'e}r Kristj{\'a}nsson, and Jonathan
  Barrett.
\newblock ``Consistent circuits for indefinite causal order''~(2023).
\newblock  \href{http://arxiv.org/abs/2206.10042}{arXiv:2206.10042}.

\bibitem{Chardonnet2025}
Kostia Chardonnet, Marc de~Visme, Beno{\^\i}t Valiron, and Renaud Vilmart.
\newblock ``The {Many-Worlds Calculus}''.
\newblock \href{https://dx.doi.org/10.46298/lmcs-21(2:13)2025}{Log. Methods
  Comput. Sci {\bf 21}, 2}~(2025).

\bibitem{Salzger2022}
Matthias Salzger.
\newblock ``Connecting indefinite causal order processes to composable quantum
  protocols in a spacetime''.
\newblock \href{https://dx.doi.org/10.48550/arXiv.2304.06735}{Master's thesis}.
\newblock ETH Zurich.
\newblock ~(2022).
\newblock  \href{http://arxiv.org/abs/2304.06735}{arXiv:2304.06735}.

\bibitem{Salzger2024}
Matthias Salzger and V~Vilasini.
\newblock ``Mapping indefinite causal order processes to composable quantum
  protocols in a spacetime''.
\newblock \href{https://dx.doi.org/10.1088/1367-2630/ad9d6f}{New J. Phys. {\bf
  27}, 023002}~(2025).

\bibitem{Vanrietvelde2021}
Augustin Vanrietvelde, Hl{\'e}r Kristj{\'a}nsson, and Jonathan Barrett.
\newblock ``Routed quantum circuits''.
\newblock \href{https://dx.doi.org/10.22331/q-2021-07-13-503}{Quantum {\bf 5},
  503}~(2021).

\bibitem{Hler2019}
Giulio Chiribella and Hl{\'e}r Kristj{\'a}nsson.
\newblock ``Quantum {S}hannon theory with superpositions of trajectories''.
\newblock \href{https://dx.doi.org/10.1098/rspa.2018.0903}{Proc. Roy. Soc. A
  {\bf 475}, 20180903}~(2019).

\bibitem{Chiribella2008}
G.~Chiribella, G.~M. D'Ariano, and P.~Perinotti.
\newblock ``Quantum circuit architecture''.
\newblock \href{https://dx.doi.org/10.1103/PhysRevLett.101.060401}{Phys. Rev.
  Lett. {\bf 101}, 060401}~(2008).
\newblock  \href{http://arxiv.org/abs/0712.1325}{arXiv:0712.1325}.

\bibitem{Chiribella2009}
Giulio Chiribella, Giacomo~Mauro D'Ariano, and Paolo Perinotti.
\newblock ``Theoretical framework for quantum networks''.
\newblock \href{https://dx.doi.org/10.1103/PhysRevA.80.022339}{Phys. Rev. A
  {\bf 80}, 022339}~(2009).
\newblock  \href{http://arxiv.org/abs/0904.4483}{arXiv:0904.4483}.

\bibitem{Coecke_Kissinger_2017}
Bob Coecke and Aleks Kissinger.
\newblock ``Picturing quantum processes: A first course in quantum theory and
  diagrammatic reasoning''.
\newblock \href{https://dx.doi.org/10.1017/9781316219317}{Cambridge University
  Press}. ~(2017).

\bibitem{Mejdoub2025}
Ilyass Mejdoub and Augustin Vanrietvelde.
\newblock ``Unilateral determination of causal order in a cyclic
  process''~(2025).
\newblock  \href{http://arxiv.org/abs/2506.18540}{arXiv:2506.18540}.

\bibitem{Choi1975}
Man-Duen Choi.
\newblock ``Completely positive linear maps on complex matrices''.
\newblock \href{https://dx.doi.org/10.1016/0024-3795(75)90075-0}{Linear Algebra
  Appl. {\bf 10}, 285--290}~(1975).

\bibitem{Costa2025}
Fabio Costa and Lee Rozema.
\newblock ``Scalable quantum control of causal order with single use of
  operations''~(2025).
\newblock In preparation; presented at Relativistic Quantum Information - North
  2024.

\bibitem{Ormrod2023}
Nick Ormrod, Augustin Vanrietvelde, and Jonathan Barrett.
\newblock ``Causal structure in the presence of sectorial constraints, with
  application to the quantum switch''.
\newblock \href{https://dx.doi.org/10.22331/q-2023-06-01-1028}{Quantum {\bf 7},
  1028}~(2023).

\bibitem{Lukas}
Lukas Schmitt.
\newblock ``Operational causality in the {P}age-{W}ootters framework''.
\newblock \href{https://dx.doi.org/10.3929/ethz-b-000722642}{Master's thesis}.
\newblock ETH Zurich.
\newblock ~(2023).

\bibitem{delaHamette2024}
Anne-Catherine de~la Hamette, Viktoria Kabel, and {\v C}aslav Brukner.
\newblock ``What an event is not: unravelling the identity of events in quantum
  theory and gravity''~(2024).
\newblock  \href{http://arxiv.org/abs/2404.00159}{arXiv:2404.00159}.

\bibitem{Liu2023}
Qiushi Liu, Zihao Hu, Haidong Yuan, and Yuxiang Yang.
\newblock ``Optimal strategies of quantum metrology with a strict hierarchy''.
\newblock \href{https://dx.doi.org/10.1103/physrevlett.130.070803}{Phys. Rev.
  Lett. {\bf 130}, 070803}~(2023).
\newblock  \href{http://arxiv.org/abs/2203.09758}{arXiv:2203.09758}.

\bibitem{Mothe2024}
Rapha{\"e}l Mothe, Cyril Branciard, and Alastair~A. Abbott.
\newblock ``Reassessing the advantage of indefinite causal orders for quantum
  metrology''.
\newblock \href{https://dx.doi.org/10.1103/physreva.109.062435}{Phys. Rev. A
  {\bf 109}, 062435}~(2024).
\newblock  \href{http://arxiv.org/abs/2312.12172}{arXiv:2312.12172}.

\bibitem{Mothe2025}
Rapha{\"e}l Mothe, Alastair~A. Abbott, and Cyril Branciard.
\newblock ``Correlations and quantum circuits with dynamical causal
  order''~(2025).
\newblock  \href{http://arxiv.org/abs/2507.07992}{arXiv:2507.07992}.

\bibitem{Salzger2025}
Matthias Salzger and John~H. Selby.
\newblock ``A decompositional framework for process theories in
  spacetime''~(2025).
\newblock  \href{http://arxiv.org/abs/2411.08266}{arXiv:2411.08266}.

\bibitem{deBank2025}
Emilien de~Bank, V.~Vilasini, and Cyril Branciard.
\newblock ``A fine-grained perspective on indefinite causal order: Modelling
  agents in spacetime''~(2025).
\newblock In preparation.

\bibitem{Lorenz2021}
Robin Lorenz and Jonathan Barrett.
\newblock ``Causal and compositional structure of unitary transformations''.
\newblock \href{https://dx.doi.org/10.22331/q-2021-07-28-511}{Quantum {\bf 5},
  511}~(2021).

\bibitem{vanrietvelde2025qca}
Augustin Vanrietvelde, Octave Mestoudjian, and Pablo Arrighi.
\newblock ``Causal decompositions of 1{D} quantum cellular automata''~(2025).
\newblock  \href{http://arxiv.org/abs/2506.22219}{arXiv:2506.22219}.

\bibitem{Simmons2024}
Will Simmons and Aleks Kissinger.
\newblock ``A complete logic for causal consistency''~(2024).
\newblock  \href{http://arxiv.org/abs/2403.09297}{arXiv:2403.09297}.

\bibitem{Kissinger2019}
Aleks Kissinger and Sander Uijlen.
\newblock ``A categorical semantics for causal structure''.
\newblock \href{https://dx.doi.org/10.23638/lmcs-15(3:15)2019}{Log. Methods
  Comput. Sci. {\bf 15}, 3}~(2019).

\bibitem{PRXQuantum.4.040307}
Eleftherios-Ermis Tselentis and \"Amin Baumeler.
\newblock ``Admissible causal structures and correlations''.
\newblock \href{https://dx.doi.org/10.1103/PRXQuantum.4.040307}{PRX Quantum
  {\bf 4}, 040307}~(2023).

\bibitem{Apadula2024}
Luca Apadula, Alessandro Bisio, and Paolo Perinotti.
\newblock ``No-signalling constrains quantum computation with indefinite causal
  structure''.
\newblock \href{https://dx.doi.org/10.22331/q-2024-02-05-1241}{Quantum {\bf 8},
  1241}~(2024).

\end{thebibliography}


\appendix

\clearpage

\begin{center}
    {\LARGE \bf{Appendix}}
\end{center}

\section{Derivation of the choice function \texorpdfstring{for $\mathcal{G}_{\mathrm{QC-QC}(N)}$}{}}
\label{sec:choice-function}

Here we will evaluate explicitly the choice relation $\Lambda_{\mathcal{G}_{\text{QC-QC}(N)}}$, given by \cref{eq:choice_rel_QCQC_link_prod}:
\begin{align}
    \Lambda_{\mathcal{G}_{\text{QC-QC}(N)}} & = \lambda_{\textbf{A}_1}^\text{aug} * \cdots * \lambda_{\textbf{A}_N}^\text{aug} * \lambda_{\textbf{V}_1}^\text{aug} * \cdots * \lambda_{\textbf{V}_{N+1}}^\text{aug} . \label{eq:choice_rel_QCQC_link_prod_2}
\end{align}

For this, let us start by explicitly writing the Boolean matrix elements of the different augmented relations in the graph:
\begin{itemize}
    \item For $\lambda_{\textbf{A}_k}^\text{aug}$, according to \cref{eq:A_routes_new,eq:aug_rel_Ak_new}, one can write (using the label $\K_{n_k-1}^{[\textbf{A}_k]}$ rather than $\K_n\backslash k$ for the branch that happens at node $\textbf{A}_k$, to distinguish it further below from the labels at the other nodes):
    \begin{align}
        \left(\lambda_{\textbf{A}_k}^\text{aug}\right)_{(\XX_n^k)_{n=1}^N}^{(\XX_n^k)_{n=1}^N,\K_{n_k-1}^{[\textbf{A}_k]}} = \delta_{\XX_{n_k}^k,\K_{n_k-1}^{[\textbf{A}_k]}\cup k} \ \Big( \prod_{n\neq n_k} \delta_{\XX_n^k,\ravnothing} \Big). \label{eq:aug_rel_Ak_Bool_terms}
    \end{align}
    
    \item For $\lambda_{\textbf{V}_{n+1}}^\text{aug}$ ($n=1,\ldots,N-1$), according to \cref{eq:V_routes_2_new,eq:aug_rel_Vn1_new}, one can write (here for ease of notations we keep the label $\K_n$ for the branch that happens at node $\textbf{V}_{n+1}$, rather than introducing $\K_n^{[\textbf{V}_{n+1}]}$):
    \begin{align}
    		& \left(\Lambda^\text{aug}_{\textbf{V}_{n+1}}\right)_{(\XX_n^k)_{k=1}^N,\vec{\bm{\ell}}_n=(\ell^{(\K_n')})_{\K_n'}}^{(\XX_{n+1}^\ell)_{\ell=1}^N,\K_n} \notag \\
    		& \quad = \Bigg( \sum_{k\in\K_n} \Big[ \delta_{\XX_n^k,\K_n} \ \Big( \prod_{k'\in\N\backslash k} \delta_{\XX_n^{k'},\ravnothing} \Big) \Big] \Bigg) \Bigg( \delta_{\XX_{n+1}^{\ell^{(\K_n)}},\K_n\cup\ell^{(\K_n)}} \ \Big( \prod_{\ell'\in\N\backslash \ell^{(\K_n)}} \delta_{\XX_{n+1}^{\ell'},\ravnothing} \Big)\Bigg)\,, \label{eq:aug_rel_Vn1_Bool_terms}
    \end{align}
    while for $\lambda_{\textbf{V}_1}^\text{aug}$ and $\lambda_{\textbf{V}_{N+1}}^\text{aug}$, according to \cref{eq:V_routes_1_new,eq:aug_rel_V1_new,eq:aug_rel_VN1_new}, one can write (ignoring the trivial inputs and outputs, see \cref{fig:augmented_rel_A_Vs_new}):
    \begin{align}
        \left(\Lambda^\text{aug}_{\textbf{V}_1}\right)_{\ell^{(\emptyset)}}^{(\XX_1^\ell)_{\ell=1}^N} & = \delta_{\XX_1^{\ell^{(\emptyset)}},\ell^{(\emptyset)}} \ \Big( \prod_{\ell'\in\N\backslash \ell^{(\emptyset)}} \delta_{\XX_1^{\ell'},\ravnothing} \Big) \label{eq:aug_rel_V1_Bool_terms} \\
        \text{and} \quad \left(\Lambda^\text{aug}_{\textbf{V}_{N+1}}\right)_{(\XX_N^k)_{k=1}^N} & = \sum_{k\in\N} \Big[ \delta_{\XX_N^k,\N} \ \Big( \prod_{k'\in\N\backslash k} \delta_{\XX_N^{k'},\ravnothing} \Big) \Big]\,. \label{eq:aug_rel_VN1_Bool_terms}
    \end{align}
\end{itemize}

\medskip

Recalling from \cref{subsubsec:compose_rel} how the ``link product'' $*$ in \cref{eq:choice_rel_QCQC_link_prod_2} evaluates, we can write the Boolean matrix elements of the choice relation as
\begin{align}
    & \left(\Lambda_{\mathcal{G}_{\text{QC-QC}(N)}}\right)_{\vec{\bm{\ell}}=(\ell^{(\K_n')})_{n=0,\ldots,N-1;\K_n'}}^{(\K_{n_k-1}^{[\textbf{A}_k]})_{k\in\N},(\K_n)_{n=1}^{N-1}} \notag \\
    & \qquad = \boolSum_{(\XX_n^k)_{n;k}} \left( \prod_{k\in\N} \left(\lambda_{\textbf{A}_k}^\text{aug}\right)_{(\XX_n^k)_{n=1}^N}^{(\XX_n^k)_{n=1}^N,\K_{n_k-1}^{[\textbf{A}_k]}} \right) \notag \\[-2mm]
    & \hspace{30mm} \left(\Lambda^\text{aug}_{\textbf{V}_1}\right)_{\ell^{(\emptyset)}}^{(\XX_1^\ell)_{\ell=1}^N} \left( \prod_{n=1}^{N-1} \left(\Lambda^\text{aug}_{\textbf{V}_{n+1}}\right)_{(\XX_n^k)_{k=1}^N,\vec{\bm{\ell}}_n=(\ell^{(\K_n')})_{\K_n'}}^{(\XX_{n+1}^\ell)_{\ell=1}^N,\K_n} \right) \left(\Lambda^\text{aug}_{\textbf{V}_{N+1}}\right)_{(\XX_N^k)_{k=1}^N} .\label{eq:choice_rel_QCQC}
\end{align}

To evaluate this expression, using \cref{eq:aug_rel_Ak_Bool_terms,eq:aug_rel_Vn1_Bool_terms,eq:aug_rel_V1_Bool_terms,eq:aug_rel_VN1_Bool_terms} above, one may start by noticing that for $m=1,\ldots,N$,
\begin{align}
    & \Bigg( \delta_{\XX_m^{\ell^{(\K_{m-1})}},\K_{m-1}\cup\ell^{(\K_{m-1})}} \ \Big( \prod_{\ell'\in\N\backslash \ell^{(\K_{m-1})}} \delta_{\XX_m^{\ell'},\ravnothing} \Big)\Bigg) \Bigg( \sum_{k\in\K_m} \Big[ \delta_{\XX_m^k,\K_m} \ \Big( \prod_{k'\in\N\backslash k} \delta_{\XX_m^{k'},\ravnothing} \Big) \Big] \Bigg) \notag \\
    & \qquad = \delta_{\K_m,\K_{m-1}\cup\ell^{(\K_{m-1})}} \ \delta_{\XX_m^{\ell^{(\K_{m-1})}},\K_{m-1}\cup\ell^{(\K_{m-1})}} \ \Big( \prod_{\ell'\in\N\backslash \ell^{(\K_{m-1})}} \delta_{\XX_m^{\ell'},\ravnothing} \Big) \,.
\end{align}
From this one easily obtains, recursively, that
\begin{align}
    & \left(\Lambda^\text{aug}_{\textbf{V}_1}\right)_{\ell^{(\emptyset)}}^{(\XX_1^\ell)_{\ell=1}^N} \left( \prod_{n=1}^m \left(\Lambda^\text{aug}_{\textbf{V}_{n+1}}\right)_{(\XX_n^k)_{k=1}^N,\vec{\bm{\ell}}_n=(\ell^{(\K_n')})_{\K_n'}}^{(\XX_{n+1}^\ell)_{\ell=1}^N,\K_n} \right) \notag \\
    & = \delta_{\K_1,\ell^{(\emptyset)}} \, \delta_{\K_2,\K_1\cup \ell^{(\K_1)}} \, \ldots \, \delta_{\K_m,\K_{m-1}\cup \ell^{(\K_{m-1})}} \ \prod_{n=0}^m \Bigg( \delta_{\XX_{n+1}^{\ell^{(\K_n)}},\K_n\cup\ell^{(\K_n)}} \ \Big( \prod_{\ell'\in\N\backslash \ell^{(\K_n)}} \!\!\delta_{\XX_{n+1}^{\ell'},\ravnothing} \Big) \!\Bigg),
\end{align}
and that the third line in \cref{eq:choice_rel_QCQC} takes the same form as above (for $m=N-1$), so that (recalling also \cref{eq:aug_rel_Ak_Bool_terms})
\begin{align}
    & \left(\Lambda_{\mathcal{G}_{\text{QC-QC}(N)}}\right)_{\vec{\bm{\ell}}=(\ell^{(\K_n')})_{n=0,\ldots,N-1;\K_n'}}^{(\K_{n_k-1}^{[\textbf{A}_k]})_{k\in\N},(\K_n)_{n=1}^{N-1}} \notag \\
    & \quad = \delta_{\K_1,\ell^{(\emptyset)}} \, \delta_{\K_2,\K_1\cup \ell^{(\K_1)}} \, \ldots \, \delta_{\K_{N-1},\K_{N-2}\cup \ell^{(\K_{N-2})}} \notag \\
    & \hspace{8mm} \boolSum_{(\XX_n^k)_{n;k}} \! \Bigg[ \prod_{k\in\N} \Bigg( \delta_{\XX_{n_k}^k,\K_{n_k-1}^{[\textbf{A}_k]}\cup k} \ \Big( \prod_{n\neq n_k} \delta_{\XX_n^k,\ravnothing} \Big) \Bigg) \Bigg] \Bigg[ \prod_{n=0}^{N-1} \Bigg( \delta_{\XX_{n+1}^{\ell^{(\K_n)}},\K_n\cup\ell^{(\K_n)}} \ \Big( \prod_{\ell'\in\N\backslash \ell^{(\K_n)}} \!\!\delta_{\XX_{n+1}^{\ell'},\ravnothing} \Big) \!\Bigg) \Bigg]. \label{eq:choice_rel_QCQC_v2}
\end{align}

Recalling the recursive definition of $\K_1^{\vec{\bm{\ell}}} \coloneqq \ell^{(\emptyset)}, \K_{n+1}^{\vec{\bm{\ell}}} \coloneqq \K_n^{\vec{\bm{\ell}}}\cup \ell^{(\K_n^{\vec{\bm{\ell}}})}$ we introduced in \cref{eq:Kn1ell}, one can see that the above expression can only be nonzero if the $\K_n$'s satisfy the same recursive constraints, i.e.\ if $\K_n = \K_n^{\vec{\bm{\ell}}}$ for all $n=1,\ldots,N-1$.
The last product term in square brackets above then restricts the Boolean sum to only index values $\XX_n^k$ such that, for any $n=1,\ldots,N$, $\XX_n^k = \K_{n-1}\cup\ell^{(\K_{n-1})} = \K_n^{\vec{\bm{\ell}}}$ if $k = \ell^{(\K_{n-1})} = \ell^{(\K_{n-1}^{\vec{\bm{\ell}}})}$ and $\XX_n^k = \ravnothing$ otherwise.
For the first product term in square bracket to be nonzero, one then needs these to satisfy $\XX_{n_k}^k \neq \ravnothing$: this requires that for each $k\in\N$, $n_k$ be such that $k = \ell^{(\K_{n_k-1}^{\vec{\bm{\ell}}})}$, i.e., $n_k$ must be equal to $n_k^{\vec{\bm{\ell}}}$ as we defined in \cref{eq:Kn1ell}. In that case $\XX_{n_k^{\vec{\bm{\ell}}}}^k = \K_{n_k^{\vec{\bm{\ell}}}}^{\vec{\bm{\ell}}}$, and the first Kronecker delta in the last line of \cref{eq:choice_rel_QCQC_v2} requires $\K_{n_k^{\vec{\bm{\ell}}}}^{\vec{\bm{\ell}}} = \K_{n_k-1}^{[\textbf{A}_k]}\cup k$ -- or equivalently, $\K_{n_k-1}^{[\textbf{A}_k]} = \K_{n_k^{\vec{\bm{\ell}}}-1}^{\vec{\bm{\ell}}}$ -- to be nonzero.
Reciprocally, satisfying these various conditions indeed gives an overall value of 1 for the expression in \cref{eq:choice_rel_QCQC_v2} above.

All in all, we find that we can write
\begin{align}
    & \left(\Lambda_{\mathcal{G}_{\text{QC-QC}(N)}}\right)_{\vec{\bm{\ell}}=(\ell^{(\K_n')})_{n=0,\ldots,N-1;\K_n'}}^{(\K_{n_k-1}^{[\textbf{A}_k]})_{k\in\N},(\K_n)_{n=1}^{N-1}} = \Bigg( \prod_{k\in\N} \delta_{\K_{n_k-1}^{[\textbf{A}_k]},\K_{n_k^{\vec{\bm{\ell}}}-1}^{\vec{\bm{\ell}}}} \Bigg) \Bigg( \prod_{n=1}^{N-1} \delta_{\K_n,\K_n^{\vec{\bm{\ell}}}} \Bigg) . \label{eq:choice_rel_QCQC_v3}
\end{align}
As one sees, for any fixed input value $\vec{\bm{\ell}}=(\ell^{(\K_n')})_{n=0,\ldots,N-1;\K_n'}$ of the choice relation, there exists a unique output value $\big((\K_{n_k-1}^{[\textbf{A}_k]})_{k\in\N},(\K_n)_{n=1}^{N-1}\big)$ for which all Kronecker deltas above give the value 1, so that the Boolean matrix element is 1. The choice relation thus defines a proper function, as described in \cref{eq:choice_function_QCQC} of the main text.

\section{Obtaining a QC-QC by fleshing out a skeletal supermap}
\label{sec:composing-process}

We provide here the detailed calculations that allowed us to obtain
\begin{align}
    & \Big( \bigotimes_{n=0}^N \dket{\textbf{V}_{n+1}} \Big) * \Big( \bigotimes_{k\in\N} \dket{\textbf{A}_k} \Big) \notag \\[-1mm]
    & = \sum_{(k_1,k_2,\ldots,k_N)} \dket{V_{\emptyset,\emptyset}^{\to k_1}}^{P,A_{k_1}^I\alpha_1} * \dket{V_{\emptyset,k_1}^{\to k_2}}^{A_{k_1}^O\alpha_1,A_{k_2}^I\alpha_2} * \dket{V_{\{k_1\},k_2}^{\to k_3}}^{A_{k_2}^O\alpha_2,A_{k_3}^I\alpha_3} * \ \cdots \notag \\[-4mm]
    & \hspace{20mm} \cdots \ * \dket{V_{\{k_1,\ldots,k_{N-2}\},k_{N-1}}^{\to k_N}}^{A_{k_{N-1}}^O\alpha_{N-1},A_{k_N}^I\alpha_N} * \dket{V_{\{k_1,\ldots,k_{N-1}\},k_N}^{\to F}}^{A_{k_N}^O\alpha_N,F} = \ket{w_\text{QC-QC}}, 
\end{align}
as in \cref{eq:recover_wQCQC}.

To prove this, we first show recursively that, for any $m=0,\ldots, N-1$,
\begin{align}
    & \Big( \bigotimes_{n=0}^m \dket{\textbf{V}_{n+1}} \Big) * \Big( \bigotimes_{k\in\N} \dket{\textbf{A}_k} \Big) \notag \\
    & = \sum_{(k_1,k_2,\ldots,k_m)} \dket{V_{\emptyset,\emptyset}^{\to k_1}}^{P,A_{k_1}^I\alpha_1} * \dket{V_{\emptyset,k_1}^{\to k_2}}^{A_{k_1}^O\alpha_1,A_{k_2}^I\alpha_2} * \ \cdots \ * \dket{V_{\{k_1,\ldots,k_{m-2}\},k_{m-1}}^{\to k_m}}^{A_{k_{m-1}}^O\alpha_{m-1},A_{k_m}^I\alpha_m} \notag \\[-3mm]
    & \hspace{24mm} * \bigoplus_{k_{m+1}} \dket{V_{\{k_1,\ldots,k_{m-1}\},k_m}^{\to k_{m+1}}}^{A_{k_m}^O\alpha_m,A_{k_{m+1}}^I\alpha_{m+1}} * \dket{\id}^{A_{k_{m+1}}^O\alpha_{m+1},\HH_{\{k_1,\ldots,k_m\},k_{m+1}}^{\textbf{A}_{k_{m+1}}\to\textbf{V}_{m+2}}} \notag \\[-3mm]
    & \hspace{95mm} * \Big( \bigotimes_{k\in\N\backslash\{k_1,\ldots,k_{m+1}\}} \dket{\textbf{A}_k} \Big) . \label{eq:recover_wQCQC_recursively_for_m}
\end{align}

Indeed, starting from $m=0$, one has (changing the dummy index $\ell$ in \cref{eq:dket_node_V1} to $k_1$)
\begin{align}
    & \dket{\textbf{V}_1} * \Big( \bigotimes_{k\in\N} \dket{\textbf{A}_k} \Big) = \bigoplus_{k_1\in\N} \dket{V_{\emptyset,\emptyset}^{\to k_1}}^{P,\HH_{\emptyset,k_1}^{\textbf{V}_1\to\textbf{A}_{k_1}}} * \Big( \dket{\textbf{A}_{k_1}} \otimes \bigotimes_{k\in\N\backslash k_1} \dket{\textbf{A}_k} \Big) \notag \\
    & \quad = \bigoplus_{k_1} \dket{V_{\emptyset,\emptyset}^{\to k_1}}^{P,\HH_{\emptyset,k_1}^{\textbf{V}_1\to\textbf{A}_{k_1}}} \!\! * \Big( \bigoplus_{n=1}^N \bigoplus_{\K_{n}\ni k_1} \!\dket{\id}^{\HH_{\K_{n}\setminus k_1,k_1}^{\textbf{V}_n\to\textbf{A}_{k_1}},A_{k_1}^I\alpha_n} * \dket{\id}^{A_{k_1}^O\alpha_n,\HH_{\K_{n}\setminus k_1,k_1}^{\textbf{A}_{k_1}\to\textbf{V}_{n+1}}} \!\Big) * \Big( \bigotimes_{k\in\N\backslash k_1} \!\! \dket{\textbf{A}_k} \!\Big) \notag \\
    & \quad = \bigoplus_{k_1} \dket{V_{\emptyset,\emptyset}^{\to k_1}}^{P,A_{k_1}^I\alpha_1} * \dket{\id}^{A_{k_1}^O\alpha_1,\HH_{\emptyset,k_1}^{\textbf{A}_{k_1}\to\textbf{V}_{2}}} * \Big( \bigotimes_{k\in\N\backslash k_1} \dket{\textbf{A}_k} \Big) ,
\end{align}
as in \cref{eq:recover_wQCQC_recursively_for_m} -- with the understanding that for $m=0$, only the direct sum over $k_1$ was left, and that the link product of all $\dket{V_{\cdot,\cdot}^\cdot}$'s reduced to just $\dket{V_{\emptyset,\emptyset}^{\to k_1}}^{P,A_{k_1}^I\alpha_1}$.

Supposing that \cref{eq:recover_wQCQC_recursively_for_m} holds for some value $m=0,\ldots, N-2$, one can then write (similarly changing the dummy index $\ell$ in \cref{eq:dket_node_Vn1} to $k_{m+2}$)
\begin{align}
    & \Big( \bigotimes_{n=0}^{m+1} \dket{\textbf{V}_{n+1}} \Big) * \Big( \bigotimes_{k\in\N} \dket{\textbf{A}_k} \Big) = \Big( \bigotimes_{n=0}^m \dket{\textbf{V}_{n+1}} \Big) * \dket{\textbf{V}_{m+2}} * \Big( \bigotimes_{k\in\N} \dket{\textbf{A}_k} \Big) \notag \\[1mm]
    & = \sum_{(k_1,k_2,\ldots,k_m)} \dket{V_{\emptyset,\emptyset}^{\to k_1}}^{P,A_{k_1}^I\alpha_1} * \dket{V_{\emptyset,k_1}^{\to k_2}}^{A_{k_1}^O\alpha_1,A_{k_2}^I\alpha_2} * \ \cdots \ * \dket{V_{\{k_1,\ldots,k_{m-2}\},k_{m-1}}^{\to k_m}}^{A_{k_{m-1}}^O\alpha_{m-1},A_{k_m}^I\alpha_m} \notag \\[-3mm]
    & \hspace{24mm} * \bigoplus_{k_{m+1}} \dket{V_{\{k_1,\ldots,k_{m-1}\},k_m}^{\to k_{m+1}}}^{A_{k_m}^O\alpha_m,A_{k_{m+1}}^I\alpha_{m+1}} * \dket{\id}^{A_{k_{m+1}}^O\alpha_{m+1},\HH_{\{k_1,\ldots,k_m\},k_{m+1}}^{\textbf{A}_{k_{m+1}}\to\textbf{V}_{m+2}}} \notag \\[-3mm]
    & \hspace{85mm} * \dket{\textbf{V}_{m+2}} * \Big( \bigotimes_{k\in\N\backslash\{k_1,\ldots,k_{m+1}\}} \dket{\textbf{A}_k} \Big) \notag \\[1mm]
    & = \sum_{(k_1,k_2,\ldots,k_m)} \cdots * \bigoplus_{k_{m+1}} \dket{V_{\{k_1,\ldots,k_{m-1}\},k_m}^{\to k_{m+1}}}^{A_{k_m}^O\alpha_m,A_{k_{m+1}}^I\alpha_{m+1}} * \dket{\id}^{A_{k_{m+1}}^O\alpha_{m+1},\HH_{\{k_1,\ldots,k_m\},k_{m+1}}^{\textbf{A}_{k_{m+1}}\to\textbf{V}_{m+2}}} \notag \\[-2mm]
    & \hspace{10mm} * \bigoplus_{\substack{\K_{m+1},k\in\K_{m+1},\\ k_{m+2}\notin\K_{m+1}}} \dket{V_{\K_{m+1}\setminus k,k}^{\to k_{m+2}}}^{\HH_{\K_{m+1}\backslash k,k}^{\textbf{A}_k\to\textbf{V}_{m+2}},\HH_{\K_{m+1},k_{m+2}}^{\textbf{V}_{m+2}\to\textbf{A}_{k_{m+2}}}} * \Big( \dket{\textbf{A}_{k_{m+2}}} \otimes \!\!\!\bigotimes_{\substack{k\in\N\backslash\{k_1,\ldots\\ \qquad\quad\ldots,k_{m+2}\}}} \!\!\!\dket{\textbf{A}_k} \Big) \notag \\[1mm]
    & = \sum_{(k_1,k_2,\ldots,k_m)} \cdots * \bigoplus_{k_{m+1}} \dket{V_{\{k_1,\ldots,k_{m-1}\},k_m}^{\to k_{m+1}}}^{A_{k_m}^O\alpha_m,A_{k_{m+1}}^I\alpha_{m+1}} * \dket{\id}^{A_{k_{m+1}}^O\alpha_{m+1},\HH_{\{k_1,\ldots,k_m\},k_{m+1}}^{\textbf{A}_{k_{m+1}}\to\textbf{V}_{m+2}}} \notag \\[-1mm]
    & \hspace{10mm} * \bigoplus_{k_{m+2}} \dket{V_{\{k_1,\ldots,k_m\},k_{m+1}}^{\to k_{m+2}}}^{\HH_{\{k_1,\ldots,k_m\},k_{m+1}}^{\textbf{A}_{k_{m+1}}\to\textbf{V}_{m+2}},\HH_{\{k_1,\ldots,k_{m+1}\},k_{m+2}}^{\textbf{V}_{m+2}\to\textbf{A}_{k_{m+2}}}} \notag \\[-1mm]
    & \hspace{3mm} * \Big( \bigoplus_{n=1}^N \bigoplus_{\K_{n}\ni k_{m+2}} \dket{\id}^{\HH_{\K_{n}\setminus k_{m+2},k_{m+2}}^{\textbf{V}_n\to\textbf{A}_{k_{m+2}}},A_{k_{m+2}}^I\alpha_n} * \dket{\id}^{A_{k_{m+2}}^O\alpha_n,\HH_{\K_{n}\setminus k_{m+2},k_{m+2}}^{\textbf{A}_{k_{m+2}}\to\textbf{V}_{n+1}}} \Big) * \Big( \!\!\!\!\!\bigotimes_{\substack{k\in\N\backslash\{k_1,\ldots\\ \qquad\quad\ldots,k_{m+2}\}}} \!\!\!\!\!\!\dket{\textbf{A}_k} \Big) \notag 
\end{align}
\begin{align}
    & = \sum_{(k_1,k_2,\ldots,k_m)} \cdots * \sum_{k_{m+1}} \dket{V_{\{k_1,\ldots,k_{m-1}\},k_m}^{\to k_{m+1}}}^{A_{k_m}^O\alpha_m,A_{k_{m+1}}^I\alpha_{m+1}} \notag \\[-1mm]
    & \hspace{10mm} * \bigoplus_{k_{m+2}} \dket{V_{\{k_1,\ldots,k_m\},k_{m+1}}^{\to k_{m+2}}}^{A_{k_{m+1}}^O\alpha_{m+1},A_{k_{m+2}}^I\alpha_{m+2}} * \dket{\id}^{A_{k_{m+2}}^O\alpha_{m+2},\HH_{\{k_1,\ldots,k_{m+1}\},k_{m+2}}^{\textbf{A}_{k_{m+2}}\to\textbf{V}_{m+3}}} \notag \\[-3mm]
    & \hspace{95mm} * \Big( \bigotimes_{k\in\N\backslash\{k_1,\ldots,k_{m+2}\}} \!\!\!\dket{\textbf{A}_k} \Big) ,
\end{align}
so that we obtain the same form as \cref{eq:recover_wQCQC_recursively_for_m}, for the value of $m+1$ instead of $m$ -- which indeed proves recursively that \cref{eq:recover_wQCQC_recursively_for_m} holds for all $m=0,\ldots,N-1$.
For $m=N-1$ more specifically, we obtain:
\begin{align}
    & \Big( \bigotimes_{n=0}^{N-1} \dket{\textbf{V}_{n+1}} \Big) * \Big( \bigotimes_{k\in\N} \dket{\textbf{A}_k} \Big) \notag \\
    & = \sum_{(k_1,\ldots,k_{N-1})} \!\!\dket{V_{\emptyset,\emptyset}^{\to k_1}}^{P,A_{k_1}^I\alpha_1} * \dket{V_{\emptyset,k_1}^{\to k_2}}^{A_{k_1}^O\alpha_1,A_{k_2}^I\alpha_2} * \cdots * \dket{V_{\{k_1,\ldots,k_{N-3}\},k_{N-2}}^{\to k_{N-1}}}^{A_{k_{N-2}}^O\alpha_{N-2},A_{k_{N-1}}^I\alpha_{N-1}} \notag \\[-3mm]
    & \hspace{24mm} * \bigoplus_{k_N} \dket{V_{\{k_1,\ldots,k_{N-2}\},k_{N-1}}^{\to k_N}}^{A_{k_{N-1}}^O\alpha_{N-1},A_{k_N}^I\alpha_N} * \dket{\id}^{A_{k_N}^O\alpha_N,\HH_{\{k_1,\ldots,k_{N-1}\},k_N}^{\textbf{A}_{k_N}\to\textbf{V}_{N+1}}} .
\end{align}

\medskip

With this in place, we can then write
\begin{align}
    & \Big( \bigotimes_{n=0}^N \dket{\textbf{V}_{n+1}} \Big) * \Big( \bigotimes_{k\in\N} \dket{\textbf{A}_k} \Big) = \Big( \bigotimes_{n=0}^{N-1} \dket{\textbf{V}_{n+1}} \Big) * \Big( \bigotimes_{k\in\N} \dket{\textbf{A}_k} \Big) * \dket{\textbf{V}_{N+1}} \notag \\
    & = \sum_{(k_1,\ldots,k_{N-1})} \!\!\dket{V_{\emptyset,\emptyset}^{\to k_1}}^{P,A_{k_1}^I\alpha_1} * \dket{V_{\emptyset,k_1}^{\to k_2}}^{A_{k_1}^O\alpha_1,A_{k_2}^I\alpha_2} * \cdots * \dket{V_{\{k_1,\ldots,k_{N-3}\},k_{N-2}}^{\to k_{N-1}}}^{A_{k_{N-2}}^O\alpha_{N-2},A_{k_{N-1}}^I\alpha_{N-1}} \notag \\[-3mm]
    & \hspace{24mm} * \bigoplus_{k_N} \dket{V_{\{k_1,\ldots,k_{N-2}\},k_{N-1}}^{\to k_N}}^{A_{k_{N-1}}^O\alpha_{N-1},A_{k_N}^I\alpha_N} * \dket{\id}^{A_{k_N}^O\alpha_N,\HH_{\{k_1,\ldots,k_{N-1}\},k_N}^{\textbf{A}_{k_N}\to\textbf{V}_{N+1}}} \notag \\[-2mm]
    & \hspace{35mm} * \bigoplus_{k\in\N} \dket{V_{\N\backslash k,k}^{\to F}}^{\HH_{\N\backslash k,k}^{\textbf{A}_k\to\textbf{V}_{N+1}},F} \notag \\[1mm]
    & = \sum_{(k_1,\ldots,k_{N-1})} \!\!\dket{V_{\emptyset,\emptyset}^{\to k_1}}^{P,A_{k_1}^I\alpha_1} * \dket{V_{\emptyset,k_1}^{\to k_2}}^{A_{k_1}^O\alpha_1,A_{k_2}^I\alpha_2} * \cdots * \dket{V_{\{k_1,\ldots,k_{N-3}\},k_{N-2}}^{\to k_{N-1}}}^{A_{k_{N-2}}^O\alpha_{N-2},A_{k_{N-1}}^I\alpha_{N-1}} \notag \\[-3mm]
    & \hspace{24mm} * \sum_{k_N} \dket{V_{\{k_1,\ldots,k_{N-2}\},k_{N-1}}^{\to k_N}}^{A_{k_{N-1}}^O\alpha_{N-1},A_{k_N}^I\alpha_N} * \dket{V_{\{k_1,\ldots,k_{N-1}\},k_N}^{\to F}}^{A_{k_N}^O\alpha_N,F},
\end{align}
which gives precisely \cref{eq:recover_wQCQC}.

\section{Proofs for \cref{sec:variations}}
\label{sec:proof5}

\genericSplitGraph*
\begin{proof}
    The split-node generic routed graph is a modified version of the generic routed graph, featuring multiple nodes $\mathbf{\hat V}^{\K_n}_{n+1}$ in place of a single original node $\mathbf{V}_{n+1}$, being endowed with a subset of its arrows.
    However, all existing branches are preserved under this change.
    
    The modification preserves the property of univocality: For the additional branches $\mathbf{\bar V}_{n+1}^{\K_n}$ with their singleton input/output index sets, no non-trivial bifurcation choice $\Ind^\text{out}_{\mathbf{\bar V}_{n+1}^{\K_n}}$ is made. Therefore, the choice function for $\cG_{\text{QC-QC}(N)}$ (cf.\ \cref{eq:choice-function}) remains practically unchanged:
    Whenever $\Lambda_{(\Gamma,(\lambda_\mathbf{N})_\mathbf{N})}$ evaluates to branch $\mathbf{V}_{n+1}^{\K_n}$ for node $\mathbf{V}_{n+1}$, the new choice relation $\Lambda_{(\Gamma',(\lambda'_\mathbf{N})_\mathbf{N})}$ will evaluate to branch $\mathbf{V}_{n+1}^{\K_n}$ for node $\mathbf{\hat V}_{n+1}^{\K_n}$, and to branch $\mathbf{\bar V}_{n+1}^{\K'_n}$ for any node $\mathbf{\bar V}_{n+1}^{\K'_n}$ with $\K_n' \neq \K_n$.
    An analogous argument applies for the adjoint for the graph, confirming bi-univocality of the resulting graph.

    For the branch graph, the existing nodes $\mathbf{V}_{n+1}^{\K_n}$ and the arrows connected therewith remain unchanged, while we receive new arrows $\mathbf{\bar V}_{n+1}^{\K_n}$ endowed with the same incoming green arrows and outgoing red arrows as $\mathbf{V}_{n+1}^{\K_n}$.
    By contrast, the $\mathbf{V}_{n+1}^{\K_n}$ do not participate in any strong parent/child relationship, as they feature only a single, one-dimensional index value.
    Hence, the partial order of arrows, independent of colour, in the branch graph is preserved from the generic routed graph, and the rewriting does not introduce any loops into the new branch graph.
\end{proof}

\genericSplitEquivalence*
\begin{proof}
    Any routed isometry $V$ fleshing out the slot for a node $\mathbf{V}_{n+1}$ needs to follow the branched route $\lambda_{\mathbf{V}_{n+1}}$ and therefore be of the form
    $V = \bigoplus_{\K_n} V^{\K_n}$
    with 
    $V^{\K_n}: \HH^{\text{in}(\mathbf{V}_{n+1})}_{\K_n} \to \HH^{\text{out}(\mathbf{V}_{n+1})}_{\K_n}$.
    Therefore, we can take these $V^{\K_n}$ individually to flesh out the slots associated with the branch nodes $\mathbf{\hat V}_{n+1}^{\K_n}$, with
    $\hat V^{\K_n} \coloneqq  V^{\K_n} \oplus \id^{\mathrm{in}(\mathbf{\hat V}^{\K_n}_{n+1})\to\mathrm{out}(\mathbf{\hat V}^{\K_n}_{n+1})}_{\mathbf{\bar V}^{\K_n}_{n+1}}$, following the respective routes.
    Therefore, any routed circuit decomposition based on $\cG$ can equivalently be obtained using $\cG^\text{split}$, provided that the slot for $\mathbf{V}_{n+1}$ is fleshed out with a routed isometry rather than an routed superisometry.
    
    Conversely, consider arbitrary routed circuit compositions for $\cG^\text{split}$.
    While the most general fleshing out for the slot associated with $\mathbf{\hat V}_{n+1}^{\K_n}$ would be
    \begin{equation}
        \label{eq:split-isometry}
        \hat V^{\K_n} = V^{\K_n} \oplus\; e^{i\varphi_{\K_n}} \cdot \,\id^{\mathrm{in}(\mathbf{\hat V}^{\K_n}_{n+1})\to\mathrm{out}(\mathbf{\hat V}^{\K_n}_{n+1})}_{\mathbf{\bar V}^{\K_n}_{n+1}}\, ,
    \end{equation}    
    we can restrict to $\phi_{\K_n} = 0$, as a given phase $\phi_{\K_n}$ could always be absorbed into the remaining $V^{\K'_n}$.
    Therefore, any routed circuit decomposition based on $\cG^\text{split}_{\text{QC-QC}(N)}$ can equivalently be obtained using $\cG_{\text{QC-QC}(N)}$ as well (again, given that the slot for $\mathbf{V}_{n+1}$ is fleshed out with a routed isometry rather than an routed superisometry).
\end{proof}

\generalMergeOne*
\begin{proof}
    By construction, every $k \in \Ind^\text{prac}_{\mathbf{V} \to \mathbf{A}} = \bigsqcup_{\mathbf{V}^\gamma} \Ind_{\mathbf{V}^\gamma}^\text{out}$ can be associated with a branch $\mathbf{V}^\gamma$. 
    \Cref{eq:merge-condition} additionally ensures that each $k$ is associated with exactly one branch $\mathbf{A}^\beta$ and
    $\bigsqcup_{\mathbf{A}^\beta : \mathtt{LinkVal} ( \mathbf{V}^\gamma , \mathbf{A}^\beta ) \neq \emptyset}\Ind_{\mathbf{A}^\beta}^\text{in} = \Ind_{\mathbf{V}^\gamma}^\text{out}$.

    Concatenating the branches $\mathbf{A}^\beta$ and $\mathbf{V}^\gamma$ (or rather, their respective part of the relation) to a single branch $\mathbf{A}'^\gamma$ within a joint node $\mathbf{A}'$ yields a \enquote{completely connected} (sub-)relation
    $\Ind_{\mathbf{V}^\gamma}^\text{in} \to \bigsqcup_{\mathbf{A}^\beta} \Ind_{\mathbf{A}^\beta}^\text{out}$, which may correspond to a branch (if the entire relation turns out to be branched).
    Concatenating all branches in this fashion, altogether we indeed confirm the resulting relation $\lambda_{\mathbf{A}'} = \lambda_\mathbf{V} \ast \lambda_\mathbf{A}$ to be branched.

    We now consider how this changes the choice function $\Lambda_\mathcal{G}$:
    For the arguments, the bifurcation choices $\Ind^\text{out}_{\mathbf{V}^\gamma}$ and $\Ind^\text{out}_{\mathbf{A}^\beta}$ are affected, while in its image, multiple branches of $\mathbf{A}^\beta$ may be coarse-grained to a single value of the choice function.
    However, due to the structure of the graph, any choice made at $\mathbf{V}^\gamma$ is completely washed out at $\mathbf{A}^\beta$ anyway, keeping the arguments of the choice function practically unchanged for determining the remaining indices.
    Therefore, the choice function remains a function, and the graph is univocal.

    As the scenario is inherently asymmetric, we also consider univocality for the adjoint graph $\cG^\top$ explicitly.
    In this graph, we have $\bigsqcup_{\mathbf{A}^\beta : \mathtt{LinkVal} ( \mathbf{A}^\beta , \mathbf{V}^\gamma ) \neq \emptyset}\Ind_{\mathbf{A}^\beta}^\text{out} = \Ind_{\mathbf{V}^\gamma}^\text{in}$, and a branched route $\lambda^\top_{\mathbf{A'}}$.
    However, the arguments of the choice function change in a different manner, with any bifurcation choice made within the branch $\mathbf{V}^\gamma$ being preserved through $\mathbf{A}^\beta$, and therefore leaving the choice function practically invariant once more.

    Finally, if the original graph did not feature any loops (which means that the branch graph is acyclic even ignoring any colours), contracting across an arrow can only introduce loops into the graph if another directed path exists between the respective nodes in the branch graph.
    As $\mathbf{A}$ is the sole child of $\mathbf{V}$, this can not be the case here.
    Therefore, the routed graph is valid.
\end{proof}

\generalMergeTwo*
\begin{proof}
    This is the adjoint scenario for \cref{thm:var-split}, so the proof follows by analogy.
\end{proof}

\generalMergeThree*
\begin{proof}
    The backward direction of the proof is obvious: 
    For any fleshing out of a mergeable node pair $\mathbf{A},\mathbf{V}$ in $\cG$, we obtain an associated fleshing out for $\cG^{\text{merge}\uparrow/\downarrow}$ by fleshing out the resulting node with $\dket{\mathbf{A}} \ast \dket{\mathbf{V}}$, composing the repesctive routed (super)maps along the composition of the respective routes.

    For the forward direction, it remains to show that splitting the merged node into two, \enquote{undoing} the merge, does not restrict the set of possible operations due to a insufficient dimension for (some) sectors of $\HH^{\mathbf{V}\to\mathbf{A}}$ / $\HH^{\mathbf{A}\to\mathbf{V}}$, which could lead to a more restricted set of processes.
    To do so, we remind of the fact that the practical (and therefore operationally relevant) Hilbert space $\HH^{\text{in/out}(\mathbf{N})}_\text{prac}$ of a node $\mathbf{V}$ is given by the direct sum of all individual branch spaces $\HH^{\text{in/out}(\mathbf{N})}_{\mathbf{V}^\gamma}$.
    
    \textbf{I. Proof for $\cG^{\mathrm{merge}\uparrow}$:}
    For any fleshing out of $\mathbf{V}$ (with a routed isometry), we have
    $
        \forall \, \mathbf{V}^\gamma \in \Bran(\lambda_{\mathbf{V}}) : \
        \dim(\HH^{\text{out}(\mathbf{V})}_\mathbf{V^\gamma}) \geq \dim(\HH^{\text{in}(\mathbf{V})}_\mathbf{V^\gamma})
    $.
    Therefore, no valid dimension assignment for $\HH^{\mathbf{V}\to\mathbf{A}}$ (with
    $
        \dim(\HH^{\text{out}(\mathbf{V})}_\mathbf{V^\gamma}) =
        \dim \big( \bigoplus_{k \in \Ind^\text{out}_{\mathbf{V}^\gamma}} \HH^{\mathbf{V}\to\mathbf{A}}_k \big)
    $)
    can restrict the set of possible operations of form
    $\HH^{\text{in}(\mathbf{V})} \xrightarrow{\lambda_{\mathbf{A} \ast \mathbf{V}}} \HH^{\text{out}(\mathbf{A})}$.

    \textbf{II. Proof for $\cG^{\mathrm{merge}\downarrow}$:}

    First, consider the case where the skeletal supermap specifies that $\mathbf{V}$ must be fleshed out by a routed unitary, which entails that the dimension assignment for the skeletal supermap is such that
    $
        \forall \, \mathbf{V}^\gamma \in \Bran(\lambda_{\mathbf{V}}) : \
        \dim(\HH^{\text{out}(\mathbf{V})}_\mathbf{V^\gamma})
        = \dim(\HH^{\text{in}(\mathbf{V})}_\mathbf{V^\gamma})
    $.
    In this case, the dimension for $\dim(\HH^{\text{in}(\mathbf{V})}_\mathbf{V^\gamma}) = \dim \Big( \bigoplus_{k \in \Ind^\text{in}_{\mathbf{V}^\gamma}} \HH^{\mathbf{V}\to\mathbf{A}}_k \Big)$ is clearly sufficient to not restrict the set of possible operations
    $\HH^{\text{in}(\mathbf{A})} \xrightarrow{\lambda_{\mathbf{A} \ast \mathbf{V}}} \HH^{\text{out}(\mathbf{V})}$.
    However, for the routed circuit decomposition as a whole, the restriction to a routed unitary is actually without loss of generality:
    Specifying a monopartite routed superisometry to flesh out $\mathbf{A}$, we may choose it so that it prepares some independent state and sends it to $\mathbf{V}$, allowing to recover any (routed) isometry by composition of this state with the (routed) unitary.
    The specifics of this argument are analogous to the proof of \cref{thm:superisometry}.
\end{proof}

\newpage
\section{Reconstruction of the routed circuit decomposition of the Grenoble process \texorpdfstring{in \cite{Vanrietvelde2023}}{}}
\label{sec:Grenoble-reconstruction}

In Figure 25 of \cite{Vanrietvelde2023}, a routed quantum circuit decomposition of the Grenoble process is displayed.
While we outlined the general procedure to obtain this representation from the $\cG_{\text{QC-QC}(N)}$ in \cref{ex:Grenoble}, we did not specify a detailed mapping between the two representations.

Here, we provide the detailed translations required to relate both representations.
For the index values, we have:
\begin{align}
    l=1 \to \XX_\emptyset^A = \{A\} \,, \qquad
    m=1 \to \XX_\emptyset^B = \{B\} \,, \qquad
    n=1 \to \XX_\emptyset^C = \{C\} \,, \notag \\
    l_1=1 \to \XX_A^B = \{A, B\} \,, \quad
    m_1=1 \to \XX_B^C = \{B, C\} \,, \quad
    n_1=1 \to \XX_C^A = \{C, A\} \,, \notag \\
    l_2=1 \to \XX_A^C = \{A, C\} \,, \quad
    m_2=1 \to \XX_B^A = \{A, B\} \,, \quad
    n_2=1 \to \XX_C^B = \{B, C\} \,, \notag \\
    f=1 \to \XX_{BC}^A = \{A, B, C\} \,, \qquad
    g=1 \to \XX_{AC}^B = \{A, B, C\} \,, \qquad
    h=1 \to \CC_{AB}^C = \{A, B, C\} \,.
\end{align}
For all of their indices, if they take the value $0$, we set the associated index $\XX_{\K_n}^k = \ravnothing$.

Slightly modifying the circuit by routing $E^g$ trivially through $W_B$, for the Hilbert spaces we have:
\begin{align}
    \HH^{P_C} \otimes \HH^{P_T} \quad &\sim \quad \HH^P , \\
    \HH^{F_T} \quad &\sim \quad \HH^F , \\
    \HH^{F_A} \otimes \HH^{F_C} \quad &\sim \quad \HH^{\alpha_F} , \\
    \HH^{E^g} \otimes \HH^{Q^{gm l_1 n_2}} \quad &\sim \quad \HH^{\bar \alpha} \otimes \HH^{\bar C^{(k)}} , \\
    \HH^{E^g} \otimes \HH^{S^g} \quad &\sim \quad \HH^{\mathbf{C} \to \mathbf{V}_4} , \\
    \HH^{E^{g=1}} \otimes \HH^{S^{g=1}} \quad &\sim \quad \HH^{\alpha_3} \otimes \HH^{C^{(B)}_3} , \\
    \HH^{X^{m l_1}} \quad &\sim \quad \HH^{\mathbf{B}\to\mathbf{C}} , \\
    \HH^{Y^{m_2 n_2}} \quad &\sim \quad \HH^{\mathbf{B}\to\mathbf{A}}
\end{align}
etc. Similarly, we identify
\begin{align}
    U_P \quad &\sim \quad \dket{\mathbf{V}_1} \,,  \qquad\qquad
    \textit{\AE}_F \quad \sim \quad \dket{\mathbf{V}_4} \,, \\
    V_B \quad &\sim \quad \dket{J_{\textbf{V}^{AC}_{3}}^\text{in}} 
        \ast \dket{\tilde{V}_{3}^{AC}}
        \ast \dket{J_{\textbf{V}^{AC}_{3}}^\text{out}}
        \ast \dket{J_{\textbf{B}}^\text{in}} \,, \\
    W_B \quad &\sim \quad \dket{J_{\textbf{B}}^\text{out}}
        \ast \dket{J_{\textbf{V}^{B}_{2}}^\text{in}} 
        \ast \dket{\tilde{V}_{2}^{B}}
        \ast \dket{J_{\textbf{V}^{B}_{2}}^\text{out}}
\end{align}
etc.
Indeed we see that system $E^g$ being withheld from the agent $B$ ensures that its output dimension does not change, even though $2 = \dim(\HH^{\alpha_3}) > \dim(\HH^{\alpha_2}) = 1$, while the agent's system dimensions are constant.
While our notations may differ significantly, we reproduce essentially the same representation for the process.
\newpage

\section{Restricting to agent-localised processing}
\label{sec:local}

In \cref{sec:grouping}, we have demonstrated how to obtain an alternative routed circuit decomposition with fewer internal nodes by modifying our generic routed graph $\cG_{\text{QC-QC}(N)}$, eliminating $\mathbf{V}_2$ and $\mathbf{V}_N$.
Based on this, we consider an example which demonstrates how a removal of further internal nodes (between $\mathbf{V}_2$ and $\mathbf{V}_{N}$) generally fails, and under which conditions we may recover the possibility to get rid of further internal nodes.

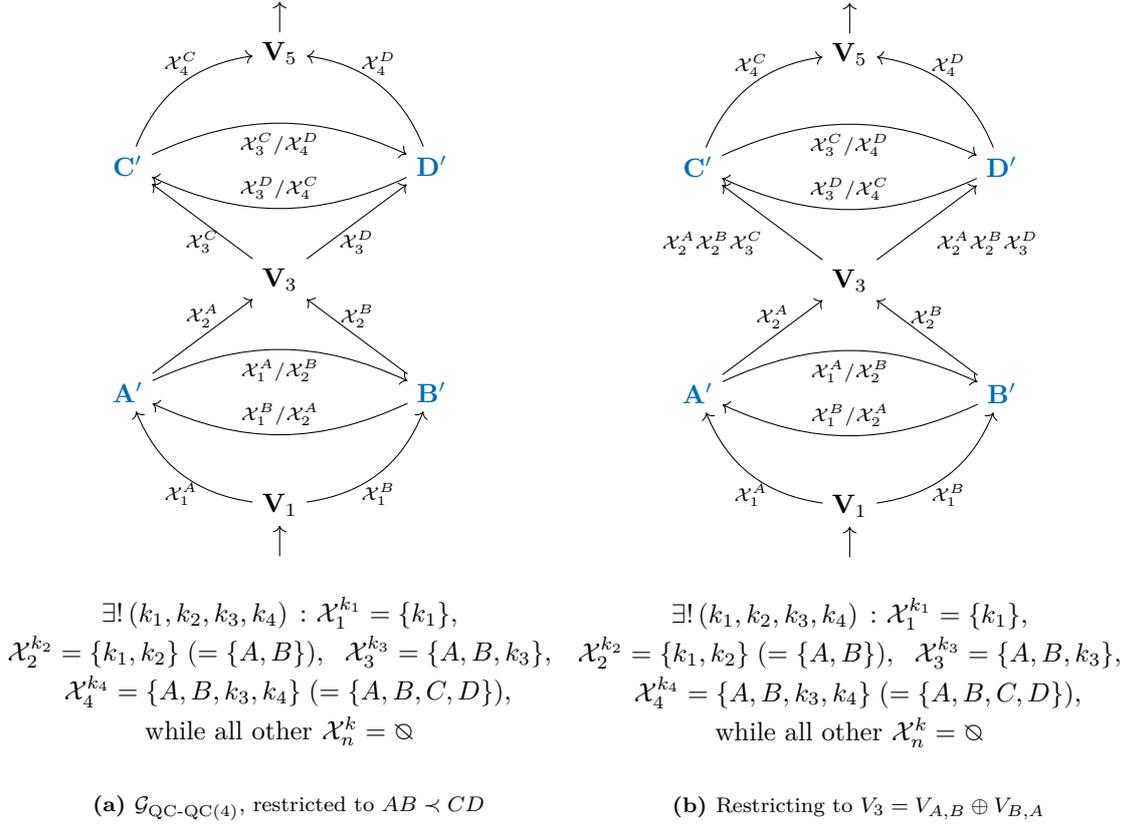
\begin{figure}
	\centering
	\begin{subfigure}[b]{0.48\textwidth}
		\centering
		\!\!\!\!\begin{tikzpicture}[transform shape]
			\node (bottom) at (0,-0.3) {};
			\node (V0) at (0,0.5) {$\mathbf{V}_1$};
			\node[text=NavyBlue] (A) at (-2,2) {$\mathbf{A}'$};
			\node[text=NavyBlue] (B) at (2,2) {$\mathbf{B}'$};
			\node (V2) at (0,3.5) {$\mathbf{V}_3$};
            \node[text=NavyBlue] (C) at (-2,5) {$\mathbf{C}'$};
			\node[text=NavyBlue] (D) at (2,5) {$\mathbf{D}'$};
            \node (V4) at (0,6.5) {$\mathbf{V}_5$};
			\node (top) at (0,7.3) {};
			\node (i1) at (0,-1.7) [align=center] {
                $\exists! \, (k_1, k_2, k_3, k_4) \, :$
                $\XX_1^{k_1} = \{ k_1 \},$\\[1mm]
                $\XX_2^{k_2} = \{k_1, k_2 \} \mathrel{(=} \{A, B \}),\;\ \XX_3^{k_3} = \{ A, B, k_3 \},$\\[1mm]
                $ \;\ \XX_4^{k_4} = \{A,B,k_3,k_4 \} \mathrel{(=} \{A, B,C,D \}),$\\[1mm]
                while all other $\XX^k_n = \ravnothing$};
			\node (i2) at (0,-3) {};
			
			\begin{scope}[
				every node/.style={fill=white,circle,inner sep=0pt}
				every edge/.style=routedarrow]
				\path [->] (bottom) edge (V0);
                
				\path [->] (V0) edge[bend left] node[below] {$\scriptstyle{\XX^A_1}$} (A);
				\path [->] (V0) edge[bend right] node[below] {$\scriptstyle{\XX^B_1}$} (B);
				\path [->] (A) edge[bend left=25] node[below] {$\scriptstyle{\XX^A_1 / \XX^B_2}$} (B);
				\path [->] (B) edge[bend left=25] node[above] {$\scriptstyle{\XX^B_1 / \XX^A_2}$} (A);
				\path [->] (B) edge node[above] {$\scriptstyle{\XX^B_2}$} (V2);
				\path [->] (A) edge node[above] {$\scriptstyle{\XX^A_2}$} (V2);

                \path [->] (V2) edge node[below] {$\scriptstyle{\XX^C_3}$} (C);
				\path [->] (V2) edge node[below] {$\scriptstyle{\XX^D_3}$} (D);
				\path [->] (C) edge[bend left=25] node[below] {$\scriptstyle{\XX^C_3 / \XX^D_4}$} (D);
				\path [->] (D) edge[bend left=25] node[above] {$\scriptstyle{\XX^D_3 / \XX^C_4}$} (C);
				\path [->] (D) edge[bend right] node[above] {$\scriptstyle{\XX^D_4}$} (V4);
				\path [->] (C) edge[bend left] node[above] {$\scriptstyle{\XX^C_4}$} (V4);
                
			    \path [->] (V4) edge (top);
			\end{scope}
		\end{tikzpicture}
		\caption{$\cG_\text{QC-QC(4)}$, restricted to $AB \prec CD$}
        \label{fig:Zurich}
	\end{subfigure}%
	\hfill
    \begin{subfigure}[b]{0.48\textwidth}
		\centering
		\!\!\!\!\begin{tikzpicture}[transform shape]
			\node (bottom) at (0,-0.3) {};
			\node (V0) at (0,0.5) {$\mathbf{V}_1$};
			\node[text=NavyBlue] (A) at (-2,2) {$\mathbf{A}'$};
			\node[text=NavyBlue] (B) at (2,2) {$\mathbf{B}'$};
			\node (V2) at (0,3.5) {$\mathbf{V}_3$};
            \node[text=NavyBlue] (C) at (-2,5) {$\mathbf{C}'$};
			\node[text=NavyBlue] (D) at (2,5) {$\mathbf{D}'$};
            \node (V4) at (0,6.5) {$\mathbf{V}_5$};
			\node (top) at (0,7.3) {};
			\node (i1) at (0,-1.7) [align=center] {
                $\exists! \, (k_1, k_2, k_3, k_4) \, :$
                $\XX_1^{k_1} = \{ k_1 \},$\\[1mm]
                $\XX_2^{k_2} = \{k_1, k_2 \} \mathrel{(=} \{A, B \}),\;\ \XX_3^{k_3} = \{ A, B, k_3 \},$\\[1mm]
                $ \;\ \XX_4^{k_4} = \{A,B,k_3,k_4 \} \mathrel{(=} \{A, B,C,D \}),$\\[1mm]
                while all other $\XX^k_n = \ravnothing$};
			\node (i2) at (0,-3) {};
			
			\begin{scope}[
				every node/.style={fill=white,circle,inner sep=0pt}
				every edge/.style=routedarrow]
				\path [->] (bottom) edge (V0);
                
				\path [->] (V0) edge[bend left] node[below] {$\scriptstyle{\XX^A_1}$} (A);
				\path [->] (V0) edge[bend right] node[below] {$\scriptstyle{\XX^B_1}$} (B);
				\path [->] (A) edge[bend left=25] node[below] {$\scriptstyle{\XX^A_1 / \XX^B_2}$} (B);
				\path [->] (B) edge[bend left=25] node[above] {$\scriptstyle{\XX^B_1 / \XX^A_2}$} (A);
				\path [->] (B) edge node[above] {$\scriptstyle{\XX^B_2}$} (V2);
				\path [->] (A) edge node[above] {$\scriptstyle{\XX^A_2}$} (V2);

                \path [->] (V2) edge node[below left] {$\scriptstyle{\XX^A_2 \XX^B_2 \XX^C_3}$} (C);
				\path [->] (V2) edge node[below right] {$\scriptstyle{\XX^A_2 \XX^B_2 \XX^D_3}$} (D);
				\path [->] (C) edge[bend left=25] node[below] {$\scriptstyle{\XX^C_3 / \XX^D_4}$} (D);
				\path [->] (D) edge[bend left=25] node[above] {$\scriptstyle{\XX^D_3 / \XX^C_4}$} (C);
				\path [->] (D) edge[bend right] node[above] {$\scriptstyle{\XX^D_4}$} (V4);
				\path [->] (C) edge[bend left] node[above] {$\scriptstyle{\XX^C_4}$} (V4);
                
			    \path [->] (V4) edge (top);
			\end{scope}
		\end{tikzpicture}
		\caption{Restricting to $V_3 = V_{A,B} \oplus V_{B,A}$}
        \label{fig:not-Zurich}
	\end{subfigure}%
	\hfill
    \begin{subfigure}[b]{0.33\textwidth}
	\end{subfigure}%
	\caption{
        Two routed graphs for 4-partite processes with $AB \prec CD$. While {\bf (a)} can represent any QC-QC with this order, representation {\bf (b)} is only applicable if $\mathbf{V}_3$ admits a more fine-grained branch structure into two branches $\mathbf{V}_{A,B}$ and $\mathbf{V}_{B,A}$.
        \Cref{ex:Zurich} can not be represented with the latter graph.}
	\label{fig:routed-graphs-four}
\end{figure}

\begin{figure}
    \centering
    \!\!\!\!\begin{tikzpicture}[transform shape]
        \node (bottom) at (0,-0.3) {};
        \node (V0) at (0,0.5) {$\mathbf{V}_1$};
        \node[text=NavyBlue] (A) at (-2,2) {$\mathbf{A}''$};
        \node[text=NavyBlue] (B) at (2,2) {$\mathbf{B}''$};
        \node[text=NavyBlue] (C) at (-2,5) {$\mathbf{C}'$};
        \node[text=NavyBlue] (D) at (2,5) {$\mathbf{D}'$};
        \node (V4) at (0,6.5) {$\mathbf{V}_5$};
        \node (top) at (0,7.3) {};
        \node (i1) at (0,-1.7) [align=center] {
            $\exists! \, (k_1, k_2, k_3, k_4) \, :$
            $\XX_1^{k_1} = \{ k_1 \},$\\[1mm]
            $\XX_2^{k_2} = \{k_1, k_2 \} \mathrel{(=} \{A, B \}),\;\ \XX_{3,k_2}^{k_3} = \{ A, B, k_3 \},$\\[1mm]
            $ \;\ \XX_4^{k_4} = \{A,B,k_3,k_4 \} \mathrel{(=} \{A, B,C,D \}),$\\[1mm]
            while all other $\XX^k_n,\, \XX^{k_3}_{3,k_2} = \ravnothing$};
        \node (i2) at (0,-3) {};
        
        \begin{scope}[
            every node/.style={fill=white,circle,inner sep=0pt}
            every edge/.style=routedarrow]
            \path [->] (bottom) edge (V0);
            
            \path [->] (V0) edge[bend left] node[below] {$\scriptstyle{\XX^A_1}$} (A);
            \path [->] (V0) edge[bend right] node[below] {$\scriptstyle{\XX^B_1}$} (B);
            \path [->] (A) edge[bend left=25] node[below] {$\scriptstyle{\XX^A_1 / \XX^B_2}$} (B);
            \path [->] (B) edge[bend left=25] node[above] {$\scriptstyle{\XX^B_1 / \XX^A_2}$} (A);
            
            \path [->] (B) edge[bend right=10] node[above right] {$\scriptstyle{\XX^D_{3,B}}$} (D);
            \path [->] (A) edge[bend left=10] node[above left] {$\scriptstyle{\XX^C_{3,A}}$} (C);
            \path [->] (B) edge node[right,pos=0.4] {$\scriptstyle{\XX^C_{3,B}}$} (C);
            \path [->] (A) edge node[left,pos=0.4] {$\scriptstyle{\XX^D_{3,A}}$} (D);

            \path [->] (C) edge[bend left=25] node[below] {$\scriptstyle{\XX^C_3 / \XX^D_4}$} (D);
            \path [->] (D) edge[bend left=25] node[above] {$\scriptstyle{\XX^D_3 / \XX^C_4}$} (C);
            \path [->] (D) edge[bend right] node[above] {$\scriptstyle{\XX^D_4}$} (V4);
            \path [->] (C) edge[bend left] node[above] {$\scriptstyle{\XX^C_4}$} (V4);
            
            \path [->] (V4) edge (top);
        \end{scope}
    \end{tikzpicture}
    \caption{
        A routed graph which may represent all 4-partite processes which feature order $AB \prec CD$ and support only agent-localised transformations.
        Any process supported by this routed graph can be understood as a wiring of independent pre- and post-processing (with memory) localised to the individual agents.
        This representation is obtained starting from \cref{fig:not-Zurich}, by additional assuming a more fine-grained branch structure imposing $V_3 = V_{A,B} \oplus V_{B,A}$ for the internal node $\mathbf{V}_3$.
        We split the respective node into two nodes $\mathbf{\hat V}_{B,A}$ and $\mathbf{\hat V}_{A,B}$ according to \cref{sec:var-split}, and merge the resulting nodes into their respective parents $\mathbf{A}'$ and $\mathbf{B}'$ afterwards to obtain $\mathbf{A}''$ and $\mathbf{B}''$.
        One would obtain the same final result for the routed graph by assuming a branch structure imposing $V_3 = V^{\to C} \oplus V^{\to D}$ instead, albeit with $\mathbf{C}$ and $\mathbf{D}$ featuring the $''$ instead. }
    \label{fig:local-graph}
\end{figure}

\begin{example}[{Zurich process \cite{Salzger2022, Lukas}}]
    \label{ex:Zurich}
    The \emph{Zurich process}, introduced in \cite{Salzger2022} as double quantum switch (and later independently proposed in Appendix~B of \cite{Lukas}),
    is given as follows:
    \begin{align}
    	V^{\rightarrow k_1}_{\emptyset, \emptyset} &= \tfrac{1}{\sqrt{2}} \ket{\psi}^{A^I_{k_1}} \\
    	V^{\rightarrow k_2}_{\emptyset, k_1} &= \mathbb{1}^{A^O_{k_1} \rightarrow A^I_{k_1}} \\
    	V^{\rightarrow k_3}_{\{k_1\}, k_2} &= \begin{cases}
    		-\frac{1}{\sqrt{2}} \mathbb{1}^{A^O_{k_2} \rightarrow A^I_{k_3}}, &k_2=B \text{ and } k_3=C \\
    		\frac{1}{\sqrt{2}} \mathbb{1}^{A^O_{k_2} \rightarrow A^I_{k_3}}, &\text{else}
    	\end{cases} \\
    	V^{\rightarrow k_4}_{\{k_1, k_2\}, k_3} &= \mathbb{1}^{A^O_{k_3} \rightarrow A^I_{k_4}} \\
    	V^{\rightarrow F}_{\{k_1, k_2, k_3\}, k_4} &= \mathbb{1}^{A^O_{k_4} \rightarrow F} \otimes \ket{k_4 \text{ mod } 2}^{\alpha_F}
    \end{align}
    where it is assumed that $k_1, k_2 \in \{A,B\}$ and $k_3, k_4 \in \{C, D\}$, while all other internal operations are assumed to vanish. Hence, it can be represented using the routed graph depicted in \cref{fig:Zurich}.
    Intuitively, we can see the process as two subsequent quantum switches, with the operation performed at $\tilde{V}_3$ acting as a basis change of the control system from the computational basis to $\{\ket{+}, \ket{-}\}$.\footnote{
        In particular, this implies that even dependent on isometries $A, B, C, D$ chosen by the agents, no POVM $\mathcal{M}(A,B,C,D)$ on $\HH^F \otimes \HH^{\alpha_F}$ can return the entire causal order in the process, while this is possible for the quantum switch or the Grenoble process.
    }
    
    Having only a single branch, $\tilde{V}_3$ does not admit a direct sum decomposition $\tilde{V}_3 = V_{A,B} \oplus V_{B,A}$ or $\tilde{V}_3 = V^{\to C} \oplus V^{\to D}$ into independent isometries.
    Indeed, the respective maps $V_{k_1,k_2}^{\to k_3}$ do not constitute isometries on the respective subspaces when considered in isolation.

    This raises the question whether we can capture relevant subclasses of processes, allowing for such a decomposition, by adjusting the index constraints to produce a more fine-granular branch structure.
    For instance, the structure $V_{A,B} \oplus V_{B,A}$ suggests that at the third agent $C$ or $D$, the second agent is encoded in $\HH^{\text{in}(\mathbf{C}/\mathbf{D})}$\footnote{
        In absence of an ancillary system $\alpha_3$, it will be encoded in $\HH^{C^I/D^I}$. If the respective information is passed through an ancillary, it will be passed further along through the process instead.}.
    This is not the case for the Zurich process, but would allow for many simpler 4-partite processes, including a concatenation of two quantum switches controlled by the same basis or by two separate, uncorrelated qubits.

    Encoding this branch structure using additional index values and constraints upon them, we obtain \cref{fig:not-Zurich} from \cref{fig:Zurich}.
    The changed branch structure then allows to perform further node splitting (according to \cref{thm:split-equivalence}) and node merging (according to \cref{thm:merge-one}), getting rid of $\mathbf{V}_3$ in the routed graph similarly as for $\mathbf{V}_2$.
    Remarkably, we would have received the same end result for the routed graph if we had chosen a branch structure imposing $V^{\to C} \oplus V^{\to D}$ instead, ultimately performing a merge of $\mathbf{\hat V}^{\to C} \hookrightarrow \mathbf{C}'$ and $\mathbf{\hat V}^{\to D} \hookrightarrow \mathbf{D}'$.
    This reflects the fact that after merging, the routed graph recovers the symmetry of the original routed graph \cref{fig:Zurich} under flipping all arrows.
    
    This suggests a more general class of processes with an arbitrary number of agents, which can be represented without any need of internal nodes other than $\mathbf{V}_1$ and $\mathbf{V}_{N+1}$.
    The action of the internal isometries is then restricted to 
    local pre- and post-processing\footnote{with a side-channel $\bar\alpha$} at the respective agents:
    At each node, the next agent is independently chosen, without interference from any other agents.
    Altogether, this amounts of a process which decomposes into a set of quantum 1-combs (and two routed isometries) wired together in a routed circuit.
    We specify the generic routed graph $\mathcal{G}^\text{local}_{\textup{QC-QC}(N)}$ capturing implementations of this form as follows:

    \begin{definition}[{$\mathcal{G}^\text{local}_{\textup{QC-QC}(N)}$}]
        \label{def:local-graph}
        The routed graph $\mathcal{G}^\text{local}_{\textup{QC-QC}(N)}$, representing \emph{agent-localised processes} with minimal internal nodes,
        is defined as follows:
        \begin{itemize}
            \item There are 2 types of nodes: ``\textbf{A}-nodes'' $\textbf{A}_k$, for $k=1,\ldots,N$, \:\ and ``\textbf{V}-nodes'' $\textbf{V}_1$ and $\mathbf{V}_{n+1}$.
            \item The indexed graph involves the following indexed arrows:
            \begin{itemize}
                \item open-ended arrows $\quad \longrightarrow \mathbf{V}_1$ and $\mathbf{V}_{N+1}\longrightarrow \quad$, 
                \item $\forall\,k=1,\ldots,N$,
                \begin{equation}
                    \textbf{V}_1\xrightarrow{\ \XX_1^k\ }\textbf{A}_k \quad \text{and} \quad \textbf{A}_k\xrightarrow{\ {\XX_N^k}\ }\textbf{V}_{N+1},
                    \label{eq:local-arrows}
                \end{equation}
                with the indices ${\XX_n^k}$ taking one of two values each: either ${\XX_1^k} = k$ or ${\XX_N^k} = \cN \setminus k$ as its value, or ${\XX_n^k} = \ravnothing$.
                The latter value will later be considered 1-dimensional.
                \item $\forall\,n=2,\ldots,N-1 , \quad \forall\,k=1,\ldots,N , \quad \forall\,\ell=1,\ldots,N$,
                \begin{equation}
                    \textbf{A}_k\xrightarrow{\ \XX_{n,k}^\ell\ }\textbf{A}_\ell \,,
                \end{equation}
                with a given index ${\XX_{n,k}^\ell}$ taking either ${\XX_{n,k}^\ell} = \K_n$ with $k,\ell \in \K_n$ as its value, or ${\XX_n^k} = \ravnothing$.
                The latter value will later be considered 1-dimensional.
            \end{itemize}

            \item The routes are specified such that
            \begin{itemize}
                \item for each node $\mathbf{A}_k$, the lists of all input and output index values $(\XX_{n,j}^k)_{n,j} \in \Ind^\text{in}_{\mathbf{A}_{k}}$ and $({\XX_{n,k}^\ell})_{n,\ell} \in \Ind^\text{out}_{\mathbf{A}_{k}}$ must satisfy that there is exactly one $1 \le n \le N-1$ such that
                \begin{equation}
                    \exists!\,j, \XX_{n,j}^k \neq\ravnothing, \;\ \exists!\,\ell, \XX_{n+1,k}^\ell\neq\ravnothing; \quad \text{ and } \ \XX_{n+1,k}^\ell=\XX_{n,j}^k\cup \ell \,.
                \end{equation}
                For all other values of $n$, all indices $\XX_{n,j}^k$ and $\XX_{n,k}^\ell$ are $\ravnothing$.
                \item for $\mathbf{V}_1$, $\XX_\emptyset^\ell \in \Ind^\text{out}_{\mathbf{V}_{1}}$ must satisfy
                $\exists!\,\ell, \ \XX_\emptyset^\ell\neq\ravnothing$.
                \item for $\mathbf{V}_N$, $\XX_\cN^k \in \Ind^\text{in}_{\mathbf{V}_{N+1}}$ must satisfy
                $\exists!\,k, \ \XX_{\cN \setminus k}^k\neq\ravnothing$.
            \end{itemize} 
            They condense to the \emph{global index constraints} that 
            $\exists!\,(k_1,k_2,\ldots,k_N):$
            \begin{equation*}
                \XX_{1}^{k_1} = \{k_1\}, \; \XX_{2,k_1}^{k_2} = \{k_1,k_2\}, \; \ldots, \; \XX_{N-1,k_{N-2}}^{k_{N-1}} = \{k_1,k_2,\ldots,k_{N-1}\}, \; \XX_{N}^{k} = \cN
            \end{equation*}
            while all other $\XX_{n,k}^\ell=\ravnothing$.
            This captures the understanding that, along each global path, all operations $A_k$ are applied once and only once, in a given order.
        \end{itemize}
    \end{definition}
    For $N=4$, this graph is shown in \cref{fig:local-graph}.

    Actually, one may simplify $\mathcal{G}^\text{local}_{\textup{QC-QC}(N)}$ by removing the label $n$ for each index $\XX_{n,k}^\ell$, as generally $\abs{\XX_{n,k}^\ell} = n$, further estranging $\cG^\text{local}_{\text{QC-QC}(N)}$ from the original generic routed graph $\mathcal{G}_{\textup{QC-QC}(N)}$.
    As expected, we confirm that up to relabelling, $\mathcal{G}^\text{local}_{\textup{QC-QC}(2)}$ and $\mathcal{G}^\text{local}_{\textup{QC-QC}(3)}$ reproduce the generic routed graphs we received after node merging in \cref{sec:grouping}.
    
    We conjecture that there is an implementation-agnostic SDP characterisation of the subclass of QC-QCs which may be implemented in this fashion, and leave proving this conjecture for future work.
    
\end{example}

\end{document}